\documentclass{article}
\usepackage{graphicx} 
\usepackage{fullpage}

\usepackage{float}
\usepackage[dvipsnames,svgnames,table]{xcolor}
\usepackage[color=green!40, textsize=small]{todonotes}
\definecolor{linkblue}{RGB}{20,60,120}
\definecolor{citeblue}{RGB}{60,90,140}
\usepackage[colorlinks=true,linkcolor=linkblue,citecolor=citeblue]{hyperref}

\usepackage[T1]{fontenc}

\usepackage{amsthm}

\usepackage{thmtools}
\usepackage{mathtools}
\usepackage{tikz}
\usetikzlibrary{arrows.meta,positioning,decorations.markings}
\usepackage{subcaption}
\usetikzlibrary{calc}
\usepackage{standalone}
\usepackage{enumitem}
\tikzset{
  circnode/.style={
    draw,
    circle,
    minimum size=4mm,
    inner sep=0pt
  }
} 

\tikzset{
    mid arrow/.style={
        postaction={decorate},
        decoration={
            markings,
            mark=at position 0.9 with {\arrow{Stealth}}
        }
    }
}

\usepackage{array}
\usepackage{xspace}

\usepackage{geometry}
\usepackage{amsmath,amsthm,amssymb}
\usepackage[capitalise]{cleveref}
\crefname{equation}{}{}

\usepackage{setspace} 
\usepackage{orcidlink}

\newtheorem{theorem}{Theorem}[section]
\newtheorem{problem}[theorem]{Problem}

\newtheorem{proposition}[theorem]{Proposition}
\newtheorem{corollary}[theorem]{Corollary}
\newtheorem{observation}[theorem]{Observation}
\newtheorem{fact}[theorem]{Fact}
\newtheorem{conjecture}[theorem]{Conjecture}
\newtheorem{lemma}[theorem]{Lemma}
\theoremstyle{definition}
\newtheorem{definition}[theorem]{Definition}
\newtheorem{example}[theorem]{Example}
\theoremstyle{remark}
\newtheorem{remark}[theorem]{Remark}
\newtheorem{claim}[theorem]{Claim}
\newtheorem{assumption}[theorem]{Assumption}
\newtheorem{notation}[theorem]{Notation}

\usepackage{graphicx} 
\usepackage{xspace}

\newcommand{\cA}{\ensuremath{\mathcal{A}}}

\newcommand{\cC}{\ensuremath{\mathcal{C}}}
\newcommand{\cD}{\ensuremath{\mathcal{D}}}

\newcommand{\cF}{\ensuremath{\mathcal{F}}}

\newcommand{\cI}{\ensuremath{\mathcal{I}}}
\newcommand{\cJ}{\ensuremath{\mathcal{J}}}
\newcommand{\cK}{\ensuremath{\mathcal{K}}}

\newcommand{\cP}{\ensuremath{\mathcal{P}}}

\newcommand{\cR}{\ensuremath{\mathcal{R}}}
\newcommand{\cS}{\ensuremath{\mathcal{S}}}

\newcommand{\cV}{\ensuremath{\mathcal{V}}}
\newcommand{\cW}{\ensuremath{\mathcal{W}}}
\newcommand{\cY}{\ensuremath{\mathcal{Y}}}
\newcommand{\cX}{\ensuremath{\mathcal{X}}}

\newcommand{\dcS}{\ensuremath{\vec{\mathcal{S}}}}
\newcommand{\dcY}{\ensuremath{\vec{\mathcal{Y}}}}
\newcommand{\dcX}{\ensuremath{\vec{\mathcal{X}}}}

\newcommand{\bN}{\ensuremath{\mathbb{N}}}

\newcommand{\nf}{nearly full\xspace}
\newcommand{\Nf}{Nearly full\xspace}

\newcommand{\N}{\mathbb{N}}
\newcommand{\Z}{\mathbb{Z}}
\newcommand{\Q}{\mathbb{Q}}
\newcommand{\Qpl}{\mathbb{Q}_+}
\newcommand{\R}{\mathbb{R}}

\newcommand{\primef}{\textsf{prime}}

\newenvironment{poc}{\begin{proof}[Proof of {C}laim.]}{\end{proof}}

\usepackage{soul}

  \newcommand{\SuggestChange}[2]{{\color{red} \relax\ifmmode\text{\st{$#1$}}\else \st{#1}\fi}\ {\color{blue} #2}}

\usepackage{enumitem,amssymb}
\newlist{todolist}{itemize}{2}
\setlist[todolist]{label=$\square$}
\usepackage{pifont}
\newcommand{\maxpl}{\oplus^{\mathrm{max}}}
\newcommand{\maxmu}{\otimes^{\mathrm{max}}}
\newcommand{\minpl}{\oplus^{\mathrm{min}}}
\newcommand{\minmu}{\otimes^{\mathrm{min}}}

\newcommand{\Rmax}{\mathbb{R}_{\mathrm{max}}}
\newcommand{\Zmax}{\mathbb{Z}_{\mathrm{max}}}
\newcommand{\Rmin}{\mathbb{R}_{\mathrm{min}}}
\newcommand{\Zmin}{\mathbb{Z}_{\mathrm{min}}}

\newcommand{\Nd}{\N^{\underline{d}}}
\newcommand{\Ninj}[1]{\N^{\underline{#1}}}
\newcommand{\Qpld}{\Qpl^{\underline{d}}}

\newcommand{\ClauseDigraph}{\Delta}

\newcommand{\CC}{\mathcal{C}}
\newcommand{\DD}{\mathcal{D}}

\usepackage{braket}
\newenvironment{claimproof}[1][Proof of Claim]{\begin{proof}[#1]}{\end{proof}}
\newtheorem*{claim*}{Claim}

\newcommand{\shift}{\textsf{shift}}
\newcommand{\func}{\textsf{func}}
\newcommand{\Esh}{E^\shift}
\newcommand{\Efn}{E^\func}
\newcommand{\opt}{\textsf{opt}}

\newcommand{\com}[1]{{#1}^\bullet}
\newcommand{\twin}[1]{\overline{#1}}
\newcommand{\eps}{\varepsilon}

\newcommand{\idx}{\textsf{idx}}
\newcommand{\val}{\textsf{val}}

\definecolor{darkgreen}{RGB}{0,130,0}
\definecolor{darkred}{RGB}{130,0,0}

\newcommand{\chibounded}{\textcolor{darkgreen}{$\chi$-bounded}\xspace}
\newcommand{\chiunbounded}{\textcolor{darkred}{$\chi$-unbounded}\xspace}

\newcommand{\chiboundedness}{\textcolor{darkgreen}{$\chi$-boundedness}\xspace}
\newcommand{\chiunboundedness}{\textcolor{darkred}{$\chi$-unboundedness}\xspace}

\newcommand{\chifbounded}{\textcolor{darkgreen}{$\chi_f$-bounded}\xspace}

\newcommand{\chifboundedness}{\textcolor{darkgreen}{$\chi_f$-boundedness}\xspace}

\usepackage[most]{tcolorbox}

\definecolor{samplingboxcolor}{HTML}{FCF5E5} 
\definecolor{titleboxcolor}{HTML}{EAEAEA}   
\definecolor{bordercolor}{rgb}{0.7, 0.7, 0.7} 

\newtcolorbox[auto counter, number within=section]{boxsampling}[1][]{
	colback=samplingboxcolor,  
	colframe=bordercolor,      
	coltitle=black,            
	colbacktitle=titleboxcolor,
	fonttitle=\bfseries\large, 
	title=#1,                  
	rounded corners,           
	boxrule=0.75pt,            
	enhanced                   
}

\tcolorboxenvironment{example}{
  enhanced,
  breakable,
  sharp corners,
  colback=blue!2,
  colframe=blue!60,
  boxrule=0pt,
  borderline west={1pt}{0pt}{blue!65},
  left=8pt,
  right=7pt,
  top=7pt,
  bottom=7pt,
  before skip=10pt,
  after skip=10pt
}

\tcolorboxenvironment{definition}{
  enhanced,
  breakable,
  sharp corners,
  colback=violet!2,
  colframe=violet!60,
  boxrule=0pt,
  borderline west={1pt}{0pt}{violet!65},
  left=8pt,
  right=7pt,
  top=7pt,
  bottom=7pt,
  before skip=10pt,
  after skip=10pt
}

\usepackage{etoolbox}
\newtoggle{anonymous}
\togglefalse{anonymous}

\title{Set-defined graph classes: $\chi$-boundedness meets tropical algebra
}
\date{}

\iftoggle{anonymous}{
	\author{Anonymous author(s)}
}{%

\author{Sarosh Adenwalla\thanks{School of Computer Science \& Informatics, University of Liverpool, UK, \texttt{sarosh.adenwalla@liverpool.ac.uk}, \orcidlink{0009-0009-8582-1281}} 
\quad 
Samuel Braunfeld\thanks{The Czech Academy of Sciences, Institute of Computer Science, Pod Vod\'{a}renskou v\v{e}\v{z}\'{\i} 2, 182 00 Prague, Czech Republic, \texttt{braunfeld@cs.cas.cz}, \orcidlink{0000-0003-3531-9970}} 
\quad 
Tom{\'a{\v{s}}} Hons\thanks{Charles University, Computer Science Institute, Malostransk\'{e} N\'{a}m\v{e}st\'{i} 25, 150 00 Prague, Czech Republic; The Czech Academy of Sciences, Institute of Computer Science, Pod Vod\'{a}renskou v\v{e}\v{z}\'{\i} 2, 182 00 Prague, Czech Republic, \texttt{honst@iuuk.mff.cuni.cz}, \orcidlink{0009-0001-7202-8681}} 
\quad
John Sylvester\thanks{School of Computer Science \& Informatics, University of Liverpool, UK, \texttt{john.sylvester@liverpool.ac.uk}, \orcidlink{0000-0002-6543-2934}} 
\quad 
Viktor Zamaraev\thanks{School of Computer Science \& Informatics, University of Liverpool, UK, \texttt{viktor.zamaraev@liverpool.ac.uk}, \orcidlink{0000-0001-5755-4141}}}
}
\begin{document}

\maketitle

\thispagestyle{empty}
\setcounter{page}{0}

\begin{center}
   \textit{Dedicated to Jarik Ne\v{s}et\v{r}il in celebration of his 80th birthday.}
\end{center}

\begin{abstract}
We study set-defined graph classes: hereditary classes of graphs in which vertices are assigned fixed-length numerical tuples and adjacency depends only on equality patterns among tuple coordinates. These classes arise in a range of areas, including structural graph theory, communication complexity, logic, adjacency labeling, bioinformatics, and quantum computing.
We investigate the structural complexity of set-defined classes by asking when they are $\chi$-bounded, i.e. when chromatic number is controlled by clique number throughout the class.

Our main results give structural and algorithmic characterizations of $\chi$-boundedness in set-defined classes. 
First, we establish a decomposition theorem showing that every graph in a set-defined class can be partitioned into a number of parts bounded polynomially in its clique number, so that each part induces the union of a constant number of shift-colorable graphs, that is, graphs admitting a homomorphism to a shift graph. 
This decomposition implies that every $\chi$-bounded set-defined graph class is polynomially $\chi$-bounded and that every set-defined graph class is polynomially fractionally $\chi$-bounded. It also identifies bounded unions of shift-colorable graphs as the fundamental obstruction to $\chi$-boundedness in set-defined classes.

Second, for full set-defined graph classes, that is, classes containing all graphs realizable by a fixed Boolean rule on equality patterns, we prove a stronger dichotomy: every such class is either polynomially $\chi$-bounded or contains the class of $D$-dimensional shift graphs for some $D \geq 2$.
This implies that the Gy\'arf\'as--Sumner conjecture holds for full set-defined graph classes.
Furthermore, we provide an algorithm that, given a Boolean-function description of a full set-defined class, decides $\chi$-boundedness of the class. 
The algorithm reduces $\chi$-boundedness to tropical feasibility, with correctness proved via the duality between tropical feasibility and winning strategies in mean payoff games.
Conversely, we show that every integer system of tropical inequalities, and hence every mean payoff game, can be encoded in strongly polynomial time as a set-defined graph class whose non-$\chi$-boundedness is equivalent to their feasibility. 
All these results extend to a broader family of set-defined classes.

As one application of our techniques, we obtain a complete characterization of $\chi$-boundedness for classes of generalized Johnson graphs.
\end{abstract}

\newpage

\thispagestyle{empty}
\setcounter{page}{0}
\begin{spacing}{1}
{\small
\tableofcontents
}
\end{spacing}
\newpage

\section{Introduction}
The study of graph classes provides a unifying framework for understanding how structural restrictions influence combinatorial and computational complexity.
A central theme is to distinguish graph classes in which global behaviour is controlled by local structure from those that admit genuinely wild phenomena, such as graphs of bounded clique number and arbitrarily large chromatic number.

Among the many approaches to defining graph classes, we focus here on set-defined graph classes. Informally, these are classes in which each graph arises from assigning every vertex a label consisting of a fixed-length tuple of numbers and adjacency is determined by a fixed rule that inspects only which coordinates of the two labels are equal.
In other words, two vertices are connected precisely when their labels satisfy a prescribed equality pattern, independent of the specific numerical values. When a single such rule can generate all graphs in a class, we call the class set-defined.

Set-defined graph classes have appeared under other names independently in several settings. In structural graph theory, set-defined classes arise as classes of finite induced subgraphs of graphs definable in the pure equality structure $(\mathbb N,=)$; they were studied in \cite{JiangNesOdM20} as natural classes that are both edge-stable and semi-algebraic, two broad families of classes that exhibit strong forms of the regularity lemma \cite{MS14,APPRRS05,FPS14}. In adjacency labeling, they are precisely the classes admitting equality-based labeling schemes, equivalently the equality fragment of logical labeling schemes \cite{Cha2023,HWZ25}. In communication complexity, their bipartite analogues  correspond to communication problems admitting deterministic constant-cost protocols with access to the Equality oracle, or equivalently to Boolean matrices of bounded blocky rank \cite{HHH23,HWZ25}, linking them to operator theory, harmonic analysis, cryptography, circuit complexity, and fine-grained complexity. 
We discuss these connections and the broader significance of set-defined classes for theoretical computer science and related fields in~\Cref{sec:set-defined-significance}.

{Despite the syntactic simplicity of the definition, set-defined classes form a remarkably broad family. 
They include classes of bounded degree, bounded degeneracy, structurally bounded expansion, edge-stable classes of bounded twin-width, edge-stable classes of bounded merge-width,
and well-studied classes such as shift graphs, Kneser graphs $K(n,k)$ and Johnson graphs $J(n,k)$ for fixed $k$. They can also be viewed as a dense analogue of classes of bounded degeneracy \cite{JiangNesOdM20}. 

Most of these set-defined graph classes are known to be \chibounded, meaning that their chromatic number is bounded by a function of their clique number. Thus, within such a class, excluding large cliques also bounds coloring complexity, a local-to-global principle with important connections to theoretical computer science as outlined in~\Cref{sec:chiisTCS}.
However, set-definedness alone does not guarantee \chiboundedness: shift graphs have clique number two and arbitrarily large chromatic number. This contrast highlights the diversity of behaviors encompassed by the set-defined framework and raises our central question: which set-defined graph classes are \chibounded?

In this work, we provide structural and algorithmic characterizations of \chiboundedness within the family of set-defined graph classes. These characterizations identify shift-type graphs as the fundamental obstructions, have consequences for classical problems such as the Erd\H{o}s--Hajnal problem on subgraphs of large girth and large chromatic number and the Gy\'arf\'as--Sumner conjecture, and reveal a two-way algorithmic connection of \chiboundedness with tropical feasibility and mean payoff games.

\subsection{Our results}

One of our main messages is that the sole ``reason'' for a set-defined class to not be \chibounded are shift-colorable graphs of large chromatic number. To explain this more precisely we first introduce the necessary definitions and state our decomposition theorem. 

We state most of our results for digraphs rather than undirected graphs, which we may view as symmetric digraphs. Although we were initially concerned with undirected graphs, the proofs naturally pass through digraphs, and the statements for digraphs are more general, sharper, and immediately imply the corresponding statement for undirected graphs by symmetrizing the edge relation. Our digraphs are assumed to be without loops and without multiple edges in the same direction.

A set of (di)graphs is a \emph{hereditary class} if it is closed under taking induced sub(di)graphs and under isomorphism. Unless stated otherwise, ``class'' refers to a hereditary class.

\begin{boxsampling}[\small Set-defined classes]
Given $d \in \mathbb{N}$, a Boolean function $f : \{0,1\}^{d^2} \rightarrow \{0,1\}$, and a finite set $V \subseteq \mathbb{N}^d$, the \emph{realisation of $f$ over $V$} is the digraph $G=(V,E)$ with $(u,w) \in E$ if and only if 
\[
    f(Q_{u,w}(1,1), Q_{u,w}(1,2), \ldots, Q_{u,w}(d,d-1), Q_{u,w}(d,d)) = 1,
\]
where $Q_{u,w}(i,j) = 1$ if $u_i = w_j$, and $Q_{u,w}(i,j) = 0$ if $u_i \neq w_j$. 
A digraph is \emph{$(d,f)$-set-defined} if it is isomorphic to a realisation of $f$ over some $V \subseteq \mathbb{N}^d$.
An (undirected) graph is \emph{$(d,f)$-set-defined} if it is the underlying graph of a $(d,f)$-set-defined \emph{digraph}.
A (di)graph is \emph{$d$-dimensional} set-defined if it is $(d,f)$-set-defined for some $f$.

\begin{definition}[Set-defined classes]
    A hereditary class of (di)graphs is (\emph{$d$-dimensional}) \emph{set-defined} if there exists $d \in \mathbb{N}$ and a Boolean function $f : \{0,1\}^{d^2} \rightarrow \{0,1\}$ such that every (di)graph in the class is $(d,f)$-set-defined. Such a class is \emph{full} if it consists of all $(d,f)$-set-defined (di)graphs. We denote by $\cX_f$ and $\dcX_f$ the full classes of $(d,f)$-set-defined graphs and digraphs, respectively. A set-defined class is \emph{\nf} if it is either full or consists of all $(d,f)$-set-defined (di)graphs realizable on tuples without repeated entries and with the same order type (in the sense of Section \ref{sec:prelim}).
\end{definition}
\end{boxsampling}

For integers $n,d \in \mathbb{N}$ with $d \ge 2$, the \emph{canonical $d$-dimensional shift digraph}, denoted $\vec{S}(n,d)$, is the digraph with vertex set $\{(u_1,\dots,u_d)\in [n]^d : u_1<\cdots<u_d\}$ in which there is a directed edge from $(u_1,\dots,u_d)$ to $(w_1,\dots,w_d)$ if and only if $u_k = w_{k-1}$ for every $k \in \{2,\ldots,d\}$.
The \emph{canonical $d$-dimensional shift graph} $S(n,d)$ is the underlying graph of $\vec{S}(n,d)$.
A (di)graph is called a \emph{$d$-dimensional shift (di)graph} if it is an induced sub(di)graph of a canonical $d$-dimensional shift (di)graph.
When the dimension is not specified, we mean $d=2$. In particular, a \emph{shift digraph} is an induced subdigraph of $\vec{S}(n,2)$ for some $n$.

Let $G$ and $H$ be (di)graphs. A \emph{homomorphism} from $G$ to $H$ is a map $\varphi \colon V(G) \to V(H)$ such that, whenever $(u,w)\in E(G)$, we have $(\varphi(u),\varphi(w))\in E(H)$.
We say that a (di)graph $G$ is \emph{shift-colorable} if there exists a homomorphism from $G$ to a shift (di)graph.
We are now ready to state our first main result.

\definecolor{coordone}{RGB}{40,90,180}
\definecolor{coordtwo}{RGB}{180,60,40}

\newcommand{\vlabel}[2]{%
$\big(\textcolor{coordone}{#1},\textcolor{coordtwo}{#2}\big)$%
}

\captionsetup[subfigure]{font=small}
\begin{figure}[t]
\centering

\begin{subfigure}[t]{0.31\textwidth}
\centering
\begin{tikzpicture}[
    scale=0.85,
    every node/.style={font=\scriptsize},
    vertex/.style={circle, draw, fill=black, inner sep=1.6pt},
    lab/.style={font=\scriptsize, align=center}
]
    \node[vertex] (v12) at (1,2) {};
    \node[vertex] (v13) at (1,3) {};
    \node[vertex] (v14) at (1,4) {};
    \node[vertex] (v15) at (1,5) {};
    
    \node[vertex] (v23) at (2,3) {};
    \node[vertex] (v24) at (2,4) {};
    \node[vertex] (v25) at (2,5) {};
    
    \node[vertex] (v34) at (3,4) {};
    \node[vertex] (v35) at (3,5) {};
    
    \node[vertex] (v45) at (4,5) {};
    
    \draw[-{Stealth}]  (v12)--(v23);
    \draw[-{Stealth}]  (v12)--(v24);
    \draw[-{Stealth}]  (v12)--(v25);
    
    \draw[-{Stealth}]  (v13)--(v34);
    \draw[bend left=12,mid arrow] (v13) to (v35);
    
    \draw[-{Stealth}]  (v14)--(v45);
    
    \draw[-{Stealth}]  (v23)--(v34);
    \draw[-{Stealth}]  (v23)--(v35);
    
    \draw[-{Stealth}]  (v24)--(v45);
    
    \draw[-{Stealth}]  (v34)--(v45);
    
    \node[lab,left=2pt of v12] {\vlabel{1}{2}};
    \node[lab,left=2pt of v13] {\vlabel{1}{3}};
    \node[lab,left=2pt of v14] {\vlabel{1}{4}};
    \node[lab,left=2pt of v15] {\vlabel{1}{5}};
    
    \node[lab,below=2pt of v23] {\vlabel{2}{3}};
    \node[lab,above=2pt of v25] {\vlabel{2}{5}};
    
    \node[lab,below=2pt of v34] {\vlabel{3}{4}};
    \node[lab,above=2pt of v35] {\vlabel{3}{5}};
    
    \node[lab,above=2pt of v45] {\vlabel{4}{5}};
\end{tikzpicture}
\caption{\small Shift graph: directed edge from $(x_1,x_2)$ to $(y_1,y_2) $ if 
\textcolor{coordtwo}{$x_2$}
=
\textcolor{coordone}{$y_1$}}
\end{subfigure}
\hfill
\begin{subfigure}[t]{0.31\textwidth}
\centering
\begin{tikzpicture}[
    scale=0.85,
    every node/.style={font=\scriptsize},
    vertex/.style={circle, draw, fill=black, inner sep=1.6pt},
    lab/.style={font=\scriptsize, align=center}
]

\node[vertex] (a) at (0,1.2) {};
\node[vertex] (b) at (1.0,2.0) {};
\node[vertex] (c) at (1.0,0.4) {};
\node[vertex] (d) at (3.2,1.6) {};
\node[vertex] (e) at (4.2,1.6) {};

\draw[{Stealth}-{Stealth}] (a)--(b);
\draw[{Stealth}-{Stealth}] (b)--(c);
\draw[{Stealth}-{Stealth}] (a)--(c);
\draw[{Stealth}-{Stealth}] (d)--(e);

\node[lab, left=0pt of a] {\vlabel{1}{4}};
\node[lab, above=0pt of b] {\vlabel{1}{7}};
\node[lab, below=0pt of c] {\vlabel{1}{2}};
\node[lab, left=0pt of d] {\vlabel{3}{8}};
\node[lab, right=0pt of e] {\vlabel{3}{5}};

\end{tikzpicture}
\caption{Equivalence graph:  edge if the \\ \textcolor{coordone}{first coordinates} are equal}
\end{subfigure}
\hfill
\begin{subfigure}[t]{0.31\textwidth}
\centering
\begin{tikzpicture}[
    scale=0.85,
    every node/.style={font=\scriptsize},
    vertex/.style={circle, draw, fill=black, inner sep=1.6pt},
    lab/.style={font=\scriptsize, align=center}
]

\node[vertex] (c) at (1.8,1.3) {};
\node[vertex] (a) at (0.2,2.2) {};
\node[vertex] (b) at (0.2,0.4) {};
\node[vertex] (d) at (3.4,2.2) {};
\node[vertex] (e) at (3.4,0.4) {};

\draw[-{Stealth}] (c)--(a);
\draw[-{Stealth}] (c)--(b);
\draw[-{Stealth}] (d)--(c);
\draw[-{Stealth}] (e)--(c);

\node[lab, above=1pt of c] {\vlabel{1}{2}};
\node[lab, above=1pt of a] {\vlabel{2}{3}};
\node[lab, below=1pt of b] {\vlabel{2}{4}};
\node[lab, above=1pt of d] {\vlabel{5}{1}};
\node[lab, below=1pt of e] {\vlabel{6}{1}};

\end{tikzpicture}
\caption{Tree: directed edge from $(x_1,x_2)$ to $(y_1,y_2) $ if 
\textcolor{coordtwo}{$x_2$}
=
\textcolor{coordone}{$y_1$}}
\end{subfigure}

\vspace{0.6em}

\caption{
Examples of graphs from set-defined classes: shift graphs, equivalence graphs (i.e.\ graphs in which every component is a complete graph), trees.  
Each vertex is represented by a vector of natural numbers, and adjacency depends only on which coordinates of two vectors are equal.
}
\label{fig:set-defined-examples}
\end{figure}
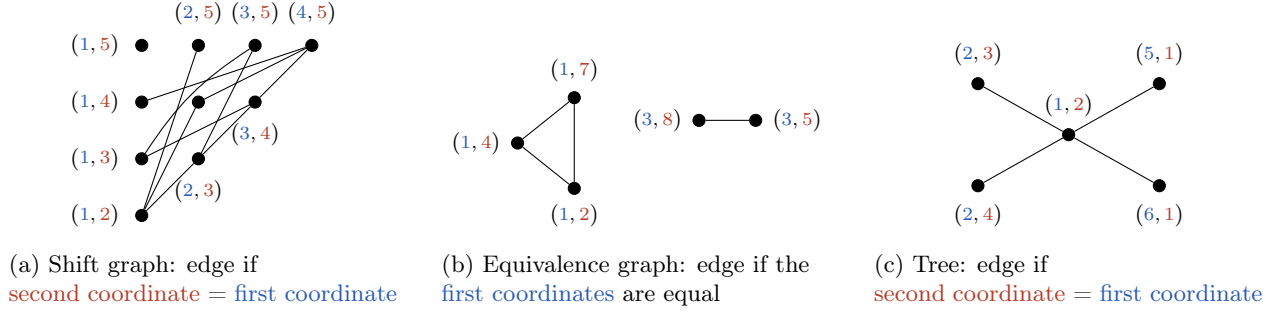

\begin{restatable}{theorem}{MainDecomposition}\label{th:main-decomposition}
    Let $\dcX$ be a set-defined class of digraphs. There exists a constant $c \in \mathbb{N}$ and a polynomial function $\tau \colon \mathbb{N} \to \mathbb{N}$ such that for every digraph $G=(V,E) \in \dcX$:
    \begin{enumerate}
        \item[1)] the vertex set $V$ can be partitioned into at most $\tau(\omega(G))$ parts $V_1, V_2, \ldots, V_r$, and
        \item[2)] for each $i \in [r]$, the induced subdigraph $G[V_i]$ is the union of at most $c$ shift-colorable digraphs.
    \end{enumerate}
\end{restatable}

We write $\chi(G)$ and $\omega(G)$ for the chromatic and clique numbers of $G$ if $G$ is a graph, and of the underlying graph of $G$ if $G$ is a digraph.
A hereditary class $\cX$ of (di)graphs is \chibounded if there exists a function $h \colon \bN \to \bN$ such that $\chi(G) \le h(\omega(G))$ for every $G \in \cX$.
If $h$ can be chosen to be polynomial (resp.\ linear), then $\cX$ is \emph{polynomially} (resp.\ \emph{linearly}) \chibounded. If $\cX$ is not \chibounded, we say that it is \chiunbounded.

The following corollary of \cref{th:main-decomposition} shows that \chiboundedness and polynomial \chiboundedness coincide in set-defined classes. It also identifies a specific structure for the graphs witnessing \chiunboundedness.

\begin{restatable}{corollary}{polyChiVsUnionsShiftSubgraphs}
\label{cor:poly-chi-vs-unions-shift-subgraphs}
    Let $\dcX$ be a set-defined class of digraphs. There exists a constant $c \in \mathbb{N}$ such that either $\dcX$ is polynomially \chibounded or it contains digraphs of arbitrarily large chromatic number that are unions of at most $c$ shift-colorable digraphs, and thus $\dcX$ is \chiunbounded.
\end{restatable}

Esperet conjectured that every \chibounded hereditary graph class is polynomially \chibounded~\cite{Esp17}. Although this conjecture is false in general~\cite{BDW24}, it remains natural to identify structural conditions under which the implication holds. Pollyanna classes provide a framework for studying this question~\cite{CCDO26}: a graph class is \emph{Pollyanna} if its intersection with any \chibounded graph class is polynomially \chibounded. It follows from \cref{cor:poly-chi-vs-unions-shift-subgraphs} that every set-defined graph class is Pollyanna. 
We note that the degree of the binding polynomial cannot, however, be bounded uniformly over all \chibounded set-defined classes. Indeed, no fixed degree suffices even for classes that are Boolean combinations of equivalence graphs~\cite{ABSZ24}.

Every \chiunbounded graph class contains graphs of arbitrarily large chromatic number and bounded clique number. \Cref{cor:poly-chi-vs-unions-shift-subgraphs} shows that, for set-defined classes, these witnesses can be chosen with additional structure: each is the union of a bounded number of shift-colorable graphs, and hence of a bounded number of triangle-free graphs. By Ramsey's theorem, such a decomposition guarantees bounded clique number. However, bounded clique number alone does not guarantee such a decomposition: for every $r,t \geq 1$, there exist graphs of clique number $t$ that cannot be expressed as the union of at most $r$ graphs, each with clique number strictly less than $t$~\cite{NR76}.
Furthermore, as discussed below, this structure rules out several standard constructions as witnesses to \chiunboundedness in set-defined classes.

Another corollary of \cref{th:main-decomposition} is that all set-defined classes are polynomially \emph{fractionally} \chibounded, i.e. the fractional chromatic number of graphs in such classes is upper bounded by a polynomial function of their clique number. 

\begin{restatable}{corollary}{CorFracChi}
\label{prop:frac chi}
            Let $\cX$ be a set-defined graph class. Then $\cX$ is polynomially fractionally \chibounded.
\end{restatable}
    One implication of this result is that \chiunboundedness of a set-defined graph class cannot be witnessed by many standard constructions of triangle-free graphs of unbounded chromatic number, because they have unbounded fractional chromatic number and thus are not fractionally $\chi$-bounded.
    This includes Mycielski graphs \cite{larsen1995fractional}, Zykov graphs \cite{charbit2011fractional}, and Burling graphs \cite{walczak2015triangle}.

    We further note that establishing polynomial fractional \chiboundedness of certain set-defined graph classes was required as a key ingredient in the efficient quantum algorithm for the shadow tomography problem developed in~\cite{King2025triply}. 
    \cref{prop:frac chi} recovers this property for those classes and gives a fractional $\chi$-binding polynomial of lower degree than that in~\cite{King2025triply}. We refer the interested reader to \cref{sec:fractional} for further details.

The preceding results apply to arbitrary set-defined graph classes. We now turn to \emph{\nf} set-defined classes, for which we obtain stronger structural results and a decidability theorem. This setting includes many natural classes such as Kneser or Johnson graphs or the generalised Johnson graph of \cref{thm:johnson_graphs_main}, where vertices represent $k$-sets and edges are determined by some intersection condition.

For \nf set-defined classes we show that \chiunboundedness can be witnessed by the class of $D$-dimensional shift (di)graphs for some $D$, rather than unions of shift-colorable (di)graphs. 
In the following statement, given a digraph, its symmetrization is obtained by adding, for every directed edge, the corresponding directed edge in the opposite direction.

\begin{restatable}{theorem}{MainSetDefined}\label{th:main-set-defined}
    Let $\dcX$ be a \emph{\nf} set-defined class of digraphs. Then either $\dcX$ is polynomially \chibounded or, for some $D \geq 2$, class $\dcX$ contains the class $\dcS_D$ of $D$-dimensional shift digraphs or the class of symmetrized $D$-dimensional shift digraphs.
\end{restatable}

As a corollary of this result we deduce that \nf set-defined graph classes satisfy the Gy\'arf\'as--Sumner conjecture (see \cref{sec:discussion} for further discussion).

\begin{restatable}[Gy\'arf\'as--Sumner conjecture for \nf set-defined classes]{corollary}{GS}\label{cor:GS}
    Let $\cX$ be a \emph{\nf} set-defined graph class. Then either $\cX$ is polynomially \chibounded or $\cX$ contains all forests.    
\end{restatable}

\Nf set-defined classes provide a natural setting for decidability questions. Any decidability result would require a finite presentation for the classes it applies to. The natural finite piece of information associated to a set-defined class is the formula defining its edge relation, and the canonical class associated to a given formula is the full class. We provide an algorithm that, given a Boolean-function description of a \nf set-defined class, decides if the class is \chibounded or not.

\begin{restatable}{theorem}{Decision}\label{th:full-set-defined-decision}
        There exists an algorithm that, given a Boolean function $f : \{0,1\}^{d^2} \rightarrow \{0,1\}$, decides whether each of the \nf set-defined classes of digraphs defined by $f$ is \chibounded or not.
\end{restatable}

The key step behind \cref{th:main-set-defined,th:full-set-defined-decision} is a reduction of \chiboundedness for \nf set-defined classes to the non-existence of finite solutions of certain systems of tropical inequalities. This provides the algorithmic criterion of \cref{th:full-set-defined-decision} and reveals a natural tropical-algebraic structure underlying the problem.

Through the standard correspondence between systems of tropical inequalities and mean payoff games, this criterion can also be reformulated in game-theoretic terms: the relevant tropical systems have finite solutions if and only if the associated mean payoff games admit winning strategies from every starting position. Thus the \chiboundedness problem for \nf set-defined classes reduces to solving a family of mean payoff game instances.

We show that the correspondence also goes in the opposite direction. Namely, we prove that every integer system of tropical inequalities, and hence every mean payoff game, can be encoded in strongly polynomial time as a set-defined graph class whose \chiunboundedness is equivalent to their feasibility.

\begin{theorem}[informal; see \cref{thm:tropics_to_chi_unboundedness} for the formal statement]
\label{th:from-topical-to-chi-bounded}
There is a strongly polynomial-time algorithm that, given an arbitrary system of integer tropical inequalities, produces a description of a set-defined class $\dcX$ such that the system has a finite solution if and only if $\dcX$ is \chiunbounded.
\end{theorem}

While the main focus of the present work is to establish structural theorems that apply uniformly across set-defined graph classes, we also illustrate our techniques in \cref{sec:johnson}, where we give a complete characterisation of \chiboundedness for classes of generalized Johnson graphs.

Let $d \in \N$ and $\cA \subseteq [0,d-1]$. The class $\cJ(d,\cA)$ of \emph{$(d,\cA)$-generalized Johnson graphs} consists of all graphs admitting an assignment of a set $S_v$ of size $d$ to each vertex $v$ (all assigned sets are distinct) such that vertices $u,v$ are adjacent if and only if $|S_u \cap S_v| \in \cA$. 
Examples within this family include the classes of Kneser graphs, generalized Kneser graphs, Johnson graphs, and line graphs of $k$-uniform hypergraphs. 

\begin{restatable}{theorem}{johnson}\label{thm:johnson_graphs_main}
    Let $d \in \N$ and $\cA \subseteq [0,d-1]$.
    The class $\cJ(d,\cA)$ is \chiunbounded if and only if there exist $a,b,c \in [0,d]$ such that
    \begin{enumerate}[label=(\roman*)]
        \item $a \leq b < c$,
        \item $a,b,c \notin \cA$,
        \item $a+1 \in \cA$, and
        \item $c-b \in [1, \lfloor \frac{d-b}{2} \rfloor]$.
    \end{enumerate}
    Moreover, if $\cJ(d,\cA)$ is \chiunbounded, then $\cS_3 \subseteq \cJ(d,\cA)$.
    If the conditions above can be satisfied with $a=b$, then $\cS_2 \subseteq \cJ(d,\cA)$.
\end{restatable}

We note that the \chiunboundedness of classes $\cJ(d,\cA)$, where $d$ is even and 
$\cA = \{i \in [d] : i \text{ is odd}\}$ was established in~\cite[Theorem~5.7]{janzer2026sunflowers}. Our \cref{thm:johnson_graphs_main} recovers this result and strengthens it by identifying two-dimensional shift graphs as witnesses to  \chiunboundedness (see \cref{ex:odd-johnson}). 

Finally, one may ask whether the class of graphs defined by a `typical' function is \chibounded. We show that, with high probability in $d$, the full set-defined class of digraphs given by a random $d$-dimensional function contains the class $\dcS_2$ of $2$-dimensional shift digraphs, and is thus \chiunbounded.
 
\begin{restatable}{theorem}{randomclass}\label{thm:randomclass}
 	Let $f:\{0,1\}^{d^2}\to\{0,1\}$ be chosen uniformly at random. Then  	$\lim\limits_{d \to \infty} \mathbb{P}(\dcS_2 \subseteq \mathcal \dcX_f) =1$. 
\end{restatable}

\subsection{Discussion \& open problems}
\label{sec:discussion}

\subsubsection{$\chi$-boundedness of hereditary graph classes}\label{sec:chiisTCS}

The local-to-global control expressed by \chiboundedness has concrete algorithmic consequences, including implications for approximate vertex coloring noted already in~\cite{Gya87}. A simple illustration, discussed in~\cite{abrishami2026burling}, is the $(k,\ell)$-coloring problem~\cite{garey1976complexity}, a prototypical promise constraint satisfaction problem~\cite{krokhin2022invitation}: for fixed $k<\ell$, one is given a graph and needs to decide whether it can be colored with $k$ colors or cannot even be colored with $\ell$ colors, under the promise that one of these two cases holds. 
In a \chibounded class with a $\chi$-binding function $f$, this can be solved in polynomial time for $\ell > f(k)$. 
Indeed, it suffices to test whether the input graph contains a $(k+1)$-clique: under the promise, its presence identifies the second case, whereas its absence implies
that the chromatic number of the graph is at most $f(k)<\ell$, ruling out the second case. These algorithmic consequences motivate several structural questions about \chiboundedness. We begin with the obstructions to \chiboundedness in general set-defined classes.

\paragraph{Obstructions to \chiboundedness in general set-defined classes.} 
We have identified unions of a bounded number of shift-colorable graphs as canonical witnesses to \chiunboundedness in general set-defined classes. A natural question is whether shift graphs themselves suffice. For full set-defined classes, this holds in a strong form: every \chiunbounded full set-defined class contains \emph{all} $d$-dimensional shift graphs for some $d$. Although this strong form fails for arbitrary set-defined graph classes, as shown in \cref{ex:shift_diag}, we conjecture that shift graphs still witness \chiunboundedness in this setting.

\begin{conjecture} \label{conj:general-set-defined}
Let $\dcX$ be a set-defined class that is \chiunbounded. Then $\dcX$ contains a class of shift digraphs, or symmetrized shift digraphs, of unbounded chromatic number.
\end{conjecture}

\paragraph{Gy\'arf\'as--Sumner conjecture.}
A proof of \cref{conj:general-set-defined} would, by \cref{lem:GS}, establish the Gy\'arf\'as--Sumner conjecture for all set-defined classes. 
This conjecture is a central structural problem concerning \chiboundedness of hereditary graph classes. It asserts that, for every forest $F$, the class of graphs with no induced copy of $F$ is \chibounded. We restate it below in an equivalent form that is more convenient for our discussion.

\begin{conjecture}[Gy\'arf\'as--Sumner]
\label{conj:GS}
    Every hereditary class is \chibounded or contains the class of forests. 
\end{conjecture}

The conjecture has been open since the 1970s, and despite substantial recent progress, much remains unknown; see, for example, the survey \cite{SS20}. 

Our \cref{cor:GS} shows that the Gy\'arf\'as--Sumner conjecture holds for \emph{\nf} set-defined classes. On the other hand, it is not clear whether the conjecture holds for arbitrary set-defined classes. By \cref{th:main-decomposition}, this question reduces to classes consisting of unions of a bounded number of shift-colorable graphs. However, \cref{conj:GS} 
already appears non-trivial even for the simplest such classes, namely subgraphs of shift graphs and unions of two shift graphs, which motivates the following more specific problems.

\begin{problem}\label{pr:GS-subgraphs-of-shift-graphs}
    Is it true that for every tree $T$ there exists $c = c(T)$ such that any subgraph $G$ of a shift graph with $\chi(G) \geq c$ contains $T$ as an induced subgraph?
\end{problem}

\begin{problem}\label{pr:GS-union-of-shift-graphs}
    Is it true that for every tree $T$ there exists $c=c(T)$ such that the union of any two graphs, each isomorphic to a shift graph of chromatic number at least $c$, contains $T$ as an induced subgraph?
\end{problem}

\paragraph{Linear $\chi$-boundedness.}
As discussed in~\cite{Gya87}, a proof that a graph class is linearly \chibounded often yields a polynomial-time constant-factor approximation algorithm for vertex coloring within that class. This algorithmic connection motivated Gy{\'a}rf{\'a}s to ask for useful sufficient conditions for a class to be linearly \chibounded.

Linear \chiboundedness is known for many set-defined graph classes, including classes of structurally bounded expansion \cite{NdMPRS21}, edge-stable classes of bounded twin-width \cite{GPT22}, edge-stable classes of bounded merge-width \cite{DGP26}, edge-stable classes excluding both a path and the bipartite complement of a path as semi-induced subgraphs \cite{dMPS25}, and classes that are Boolean combinations of proper subclasses of equivalence graphs \cite{ABSZ24}.
These examples motivate the search for a characterisation of linearly \chibounded set-defined graph classes.  As a first step, we ask the following question for \nf set-defined classes.

\begin{problem}\label{pr:linearly-chi-bounded}
    Which \emph{\nf} set-defined graph classes are linearly \chibounded?
\end{problem}

\subsubsection{Subgraphs of large girth and large chromatic number}

\begin{remark}
Since the first version of this paper appeared, a negative solution to the Erd\H{o}s--Hajnal conjecture has been found with substantial assistance from AI tools~\cite{KK26arc,KK26six}; see~\cite{NW26} for an exposition of the construction, its improvement, and connections to the existing literature. We nevertheless retain the remainder of this subsection essentially as it appeared in the original version, for historical context since the problem posed here identified precisely the class in which the counterexamples were subsequently found.
\end{remark}

A long-standing conjecture of Erd\H{o}s and Hajnal asserts that, for every
$k,g \in \mathbb{N}$, every graph of sufficiently large chromatic number
contains a subgraph with chromatic number at least $k$ and girth at least
$g$~\cite{Erd73}. R{\"o}dl proved the special case of $g=4$, showing that
every graph of sufficiently large chromatic number contains a triangle-free
subgraph of large chromatic number~\cite{Rod77}. The conjecture is known to hold
for canonical $d$-dimensional shift graphs for every fixed $d \geq 2$
~\cite{Enj23,TW18} and for Kneser graphs~\cite{MW22b}; it has also
recently been announced for the sequence of Burling graphs~\cite{PTW26}.

Our \cref{th:main-decomposition} identifies shift-colorable graphs as the
central case of this problem within set-defined classes. More precisely, it
shows that every graph of sufficiently large chromatic number in a set-defined class contains a shift-colorable subgraph of
large chromatic number. Since shift-colorable graphs are triangle-free,
this gives a structural strengthening of R{\"o}dl's theorem in the
set-defined setting. More importantly, to establish the
Erd\H{o}s--Hajnal conjecture for every set-defined graph class, it suffices to establish it for shift-colorable graphs.

The problem already appears difficult for much more special classes, such as subgraphs of shift graphs, or even shift graphs themselves, that is, induced subgraphs of canonical shift graphs.

\begin{problem}[Erd\H{o}s--Hajnal conjecture for shift graphs]
\label{pr:EH-shift}
    Is it true that, for all $k,g \in \mathbb{N}$, there exists $f(k,g)$
    such that every shift graph $G$ with $\chi(G) \geq f(k,g)$ contains
    a subgraph $H$ with $\chi(H) \geq k$ and girth at least $g$?
\end{problem}

We note that the existing proof of the Erd\H{o}s–Hajnal conjecture for canonical shift graphs relies crucially on their canonical structure \cite{Enj23, TW18}. Thus, resolving \cref{pr:EH-shift} would require fundamentally new ideas.

We also note that no set-defined class can contain graphs of arbitrarily large girth and chromatic number, so the Erd\H{o}s--Hajnal conjecture will always require passing to subgraphs in this setting. This is because if a set-defined class has girth at least 5, then it can contain no $K_3$ or $K_{2,2}$, and so must have bounded degeneracy and thus bounded chromatic number, by \cite[Theorem 35]{JiangNesOdM20}. This also serves as a sanity check on our goal of finding some $\cS_k$ in every \chiunbounded full set-defined class, since every $\cS_k$ has girth 4.

\subsubsection{Solving tropical inequalities and mean payoff games}\label{sssec:solving_MPG_overview}

Tropical linear algebra provides an algebraic framework for representing and analyzing many natural problems in combinatorial optimization, including all-pairs shortest paths, set covering, and exact cycle covers~\cite{butkovic2003max,Butkovic2010}.
It has also played an important role in establishing new conditional and unconditional bounds in fine-grained complexity theory~\cite{williams2014faster,korhonen2021lower}.
Mean payoff games are zero-sum, two-player games whose decision problem lies in the complexity class $\mathsf{NP} \cap \mathsf{coNP}$~\cite{ZWICK1996343}.
A formal equivalence between mean payoff games and tropical linear systems was established in~\cite{akian2012tropical}, and the algorithmic consequences of this equivalence remain an active area of research~\cite{allamigeon2015tropicalizing,akian2022tropical}.

A central open question is whether mean payoff games (or, equivalently, tropical linear systems) can be solved in polynomial time. Despite the development of algorithms based on a variety of approaches, classical known general algorithms have pseudopolynomial running times~\cite{GURVICH198885,ZWICK1996343,bjorklund2007combinatorial,Dhingra2006}. A major breakthrough for parity games, a closely related but technically simpler class, was a discovery of a quasipolynomial-time algorithm~\cite{Calude2017}. Adapting this approach has yielded improved pseudopolynomial algorithms for mean payoff games; however, this line of attack alone is unlikely to produce even a quasipolynomial-time algorithm in this setting~\cite{fijalkow_et_al}.
Two recent preprints report quasipolynomial-time algorithms for mean payoff games. An OpenAI preprint presents a deterministic quasipolynomial-time algorithm~\cite{OAI_MPG}. Soleymani and Farina give a quasipolynomial-time algorithm for approximating fixed points of $\ell_\infty$-contractions; through known reductions, their result also yields a quasipolynomial-time algorithm for mean payoff games~\cite{soleymani2026quasipolynomialalgorithmellinftycontractions}. These developments leave the question of polynomial-time solvability open.

Our Theorem~\ref{th:from-topical-to-chi-bounded} gives a strongly polynomial-time reduction from systems of tropical inequalities, and hence from mean payoff games, to the \chiboundedness problem for set-defined graph classes. This reduction opens the possibility of bringing structural and algorithmic tools from graph theory---particularly the theory of set-defined classes---to bear on these problems. A natural direction for future work is to determine whether this connection can yield new structural characterizations or algorithmic advances for mean payoff games. For example, sufficient conditions for \chiboundedness, such as bounded degeneracy, may translate into tractability criteria for previously unrecognized classes of mean payoff games and tropical linear systems.

\subsection{Overview of proof strategies}
\label{ssec:overview_of_proof_strategies}

Each technical section in the paper begins with a more detailed outline of its proofs. 
Here we give a high-level overview of the main proof strategies and how the 
different components fit together.

Let $\dcX$ be a set-defined class of digraphs and let 
$f : \{0,1\}^{d^2} \to \{0,1\}$ be a Boolean function such that 
$\dcX \subseteq \dcX_f$.

\paragraph{Decomposition via DNF and clause simplification.} At a high level, the proof of our decomposition theorem (\cref{th:main-decomposition}), which is developed in \cref{sec:decomposition-theorem}, proceeds through a sequence of vertex and edge partitions that gradually reveal the desired structure of digraphs in $\dcX_f$. The starting point is a representation of $f$ in \emph{full} disjunctive normal form (DNF). This representation allows us to view any digraph in $\dcX_f$ as the edge-disjoint union of subdigraphs realised by the individual clauses of the DNF.
We then show that vertex sets of digraphs in $\dcX_f$ admit a partition into at most $d^{2d}$ 
parts so that each part induces a digraph belonging to a set-defined class whose defining DNF has a simplified clause structure. This reduces the analysis to the study of such digraph classes and allows us to classify clauses into several types.
For every class defined by a simplified DNF, we prove that the edges arising from all clause types except one form a subdigraph $H$ with polynomially \chibounded structure. In particular, the vertex set of $H$ admits a partition into independent sets $V_1, V_2, \dots, V_r$, where $r \le \tau(\omega(H))$ for some polynomial $\tau$.
The edges arising from each of the remaining clauses, which we call \emph{path clauses}, form 
a shift-colorable digraph, and hence the edges produced by all path clauses collectively form, inside each $V_i$, a union of shift-colorable digraphs. 
This yields the structural decomposition stated in 
\cref{th:main-decomposition}.

\paragraph{Reduction to path clauses over sets with functional constraints.}
In \cref{sec:refinements-full-set-defined}, we revisit the proof of the decomposition theorem in the special case of \nf set-defined classes and refine its main reductions so as to preserve the information needed for the stronger dichotomy of \cref{th:main-set-defined}. More precisely, the reductions are arranged to keep track of \chiboundedness equivalences between the original class and the classes produced along the way.
These reductions rely on an additional feature of the partition $V_1,\dots,V_r$ constructed in the proof of \cref{th:main-decomposition}: each part satisfies a family of \emph{functional constraints}, where each such constraint specifies that one coordinate of the vertices in the set is determined by some of the others. In the setting of \nf set-defined classes, this extra structure implies that the monotone closure of the class contains every digraph that can be realised by each path clause on vertex sets satisfying the corresponding functional constraints. Since each path clause produces a shift-colorable subdigraph, and hence a triangle-free subdigraph, the class is \chibounded if and only if every class corresponding to a path clause has bounded chromatic number (\cref{th:from-full-set-defined-to-shift-colorable}). This reduces the problem to the following question: given a collection of functional constraints and a path clause, does the class of all digraphs realised by that path clause on sets of vertices satisfying those constraints have bounded chromatic number?

\paragraph{Tropical and game-theoretic dichotomy.}
To answer this question, in \cref{sec:chi-dichotomy}, we encode the functional constraints as tropical inequalities and reduce the question to the existence of finite solutions in two systems of tropical inequalities.
We show that when both systems have finite solutions, such solutions can be used to construct a realisation of every $D$-dimensional shift digraph, for some $D \in \bN$, over sets satisfying the functional constraints.
On the other hand, when one of the systems does not have a finite solution, we show that the digraphs in the class have bounded chromatic number. For this latter result, we exploit the known correspondence between systems of tropical inequalities and mean payoff games. 
Starting from a system of tropical inequalities, we construct an associated game digraph and apply a sequence of transformations to obtain a dependency digraph that encodes dependencies between vertex coordinates. Using graph out-covering projections, we relate walks in the game digraph to walks in the dependency digraph.
This correspondence allows us to deduce structural constraints on directed paths in digraphs generated by the path clause. In particular, under the assumption that the tropical system has no finite solution (equivalently, that the corresponding game does not admit a strategy ensuring a player a non-negative payoff for every starting position), we show that in any sufficiently long directed path the coordinates of the second vertex are uniquely determined by those of the first. As a consequence, these digraphs cannot contain a directed tree consisting of two long edge-disjoint directed paths emanating from a common root. Known results then imply that such digraphs have bounded chromatic number.

\paragraph{From subgraphs to induced subgraphs.}
The case when one of the path clauses can realise all $D$-dimensional shift digraphs implies that the class is \chiunbounded. However, this does not immediately imply that the class contains all $D$-dimensional shift digraphs, since other path clauses may contribute additional edges and the $D$-dimensional shift digraphs appear only as subdigraphs of digraphs in the class. 
Using an additional argument, in \cref{sec:induced-shift-graphs}, we refine the union of digraphs produced by the path clauses and show that such digraphs contain, as \emph{induced} subdigraphs, all $D'$-dimensional shift digraphs (or their symmetrizations) for some $D' \geq D$. This yields \cref{th:main-set-defined}.

\paragraph{Deciding $\chi$-boundedness in \nf set-defined classes.}
In \cref{sec:deciding-chiboundedness}, we combine \cref{th:from-full-set-defined-to-shift-colorable} with the reduction established in \cref{sec:chi-dichotomy} to prove \cref{th:full-set-defined-decision}, that is, the decidability of \chiboundedness for \nf set-defined classes. More precisely, by \cref{th:from-full-set-defined-to-shift-colorable}, deciding whether, for a given Boolean function $f \colon \{0,1\}^{d^2} \to \{0,1\}$, the class $\dcX_f$ is \chibounded reduces to checking bounded chromatic number for finitely many classes defined by a path clause on vertex sets satisfying a collection of functional constraints. Moreover, for each such class, a finite description in terms of a path clause and a collection of functional constraints can be computed from $f$. 
By \cref{thm:tropical_dichotomy}, bounded chromatic number for each of these classes can in turn be decided from their descriptions by constructing and solving a pair of finite systems of tropical inequalities. Since finite systems of tropical inequalities are algorithmically solvable, this yields \cref{th:full-set-defined-decision}.

\paragraph{Connections to tropical algebra.}
In Section~\ref{sec:compressed_matrices}, we strengthen the connection between \chiboundedness of set-defined classes and the solvability of tropical inequalities, or equivalently, of mean payoff games. Building on the tropical representation developed in Section~\ref{sec:chi-dichotomy}, we introduce a more efficient compressed representation of a path clause together with its functional constraints by a pair of tropical systems of much smaller dimension. This compressed form still fully captures the relevant \chiboundedness question: the corresponding class has unbounded chromatic number if and only if both tropical systems admit finite solutions. Conversely, the class of tropical systems arising from this representation is flexible enough that, after a simple normalization, an arbitrary tropical system can be encoded in it. Since the corresponding set-defined class can then be reconstructed in strongly polynomial time, this yields Theorem~\ref{th:from-topical-to-chi-bounded}.

\subsection{Related work} \label{ssec:related_work}

Two lines of work are particularly relevant to our study of canonical obstructions to \chiboundedness in sufficiently tame graph classes.

The first concerns the conjecture stating that for a fixed graph $F$, if a hereditary graph class $\cC$ contains no induced subdivision of $F$ and is \chiunbounded, then $\cC$ must contain the class of all Burling graphs \cite[Conjecture 2.1]{chudnovsky2021induced}. Although this conjecture was made with the suggestion that it should be disproved, in \cite{abrishami2026burling} it is proved under an additional hypothesis that in particular holds for classes of string graphs (i.e.\ intersection graphs of curves in the plane). In that setting, there is a single minimal obstruction to \chiboundedness, under induced subgraph containment. 
Our setting has no minimum \chiunbounded class, since the $d$-dimensional shift graphs form an infinite decreasing chain under inclusion, and they are all \chiunbounded but their intersection is a class of bipartite graphs.
Moreover, it is not even true that every \chiunbounded class of shift graphs contains the class of $d$-dimensional shift graphs for some $d$ (Example \ref{ex:shift_diag}).

The second line of work comes from infinite graph theory. It concerns the conjecture that if an infinite graph has uncountable chromatic number, then it must contain a homomorphic image of every finite $d$-dimensional shift graph for some $d$ \cite{taylor1970combinatorial} (or every finite $d$-dimensional shift graph as a subgraph \cite[Problem 2]{erdos1972some}). Although this conjecture was disproved  \cite[Theorem 4]{hajnal1984must}, it was recently revisited under additional model-theoretic assumptions in \cite{halevi2022infinite, halevi2023infinite, halevi2025infinite}. In addition to sharing the theme of cliques and shift graphs as the canonical witnesses of large chromatic number, the analysis of two of the three cases considered in \cite{halevi2022infinite} (i.e.~the $\omega$-stable and superstable cases) reduces to full set-defined graph classes arising from graphs defined on a totally indiscernible set. It is then shown that, if such a class has unbounded chromatic number, its monotone closure contains the class of finite $d$-dimensional shift graphs. Thus these results are close in spirit to ours, but our focus is finer as we are concerned with induced subgraph containment rather than subgraph containment.

\section{Broader significance of set-defined graphs}\label{sec:set-defined-significance}
In this section we highlight various contexts in which set-defined graph classes arise.

\paragraph{Graph Theory.} 
In \cite{JiangNesOdM20}, a set-defined graph class is defined as a subclass of the class of finite induced subgraphs of a graph definable in $(\mathbb{N}, =)$. Since this structure admits quantifier elimination, there exists a fixed dimension $d$ such that every vertex can be associated with an element of $\mathbb{N}^d$, and adjacency between two vertices is determined by a Boolean combination of equalities among the corresponding tuple entries. This definition coincides precisely with that of graph classes admitting \emph{equality-based labeling schemes (EBLS)}~\cite{HWZ25}, or, equivalently, those admitting the equality fragment of logical labeling schemes~\cite{Cha2023}.

In~\cite{JiangNesOdM20}, it was shown that numerous well-studied graph classes are set-defined, including classes of bounded shrub-depth, classes of bounded expansion, and, more generally, classes of bounded degeneracy. Moreover, set-defined graph classes can be viewed as a dense analogue of classes of bounded degeneracy: a class has bounded degeneracy if and only if it is both weakly sparse\footnote{A class is weakly sparse if there exists $t \in \bN$ such that no graph in the class contains a complete bipartite graph $K_{t,t}$ as a subgraph.} and set-defined. 
It was also shown that classes of structurally bounded expansion are set-defined. The fact that all these classes are set-defined also follows from independent results showing that classes of bounded degeneracy~\cite{Cha2023,HWZ25} and classes of structurally bounded expansion~\cite{EHK22} admit EBLS.

Further classes are known to be set-defined. In particular, the structural result of~\cite{GPT22} was used in~\cite{HWZ25} to prove that edge-stable classes of bounded twin-width admit EBLS. Weakly sparse small classes are set-defined since they have bounded expansion~\cite{BDSZ25}, 
edge-stable classes of bounded merge-width are set-defined because they have structurally bounded expansion \cite{DGP26},
and edge-stable unit disk graphs are because they admit EBLS~\cite{HZ24}. 
Any edge-stable class excluding both a path and the bipartite complement of a path as semi-induced subgraphs is set-defined~\cite{dMPS25}. 
Edge-stable classes of bipartite graphs excluding a subdivision of a star and its bipartite complement as induced subgraphs are set-defined because they admit EBLS~\cite{HZ24}.  Finally, classes of graphs that are Boolean combinations of equivalence graphs are also set-defined~\cite{ABSZ24}.

\paragraph{Communication complexity and beyond.}
Set-defined graph classes admit a natural interpretation within communication complexity. 
A close connection between adjacency labeling schemes for hereditary graph classes and communication protocols in the randomized public-coin two-party communication model was established in~\cite{Har20}, using a probabilistic variant of labeling schemes.
In the important special case of constant-cost communication protocols, there is a tight correspondence: a communication problem admits such a protocol if and only if the associated hereditary graph class admits probabilistic adjacency labels of constant size, and vice versa~\cite{HWZ25}. 
Since set-defined graph classes admit equality-based labeling schemes, they in particular admit probabilistic adjacency labels of constant size, and hence correspond to constant-cost communication problems \cite{HWZ25}.

Constant-cost communication problems themselves form a rich and non-trivial class, with deep connections to areas such as operator theory and harmonic analysis \cite{HHH23}, differential privacy \cite{FX14}, and cryptography \cite{AN25}. They have been the subject of sustained recent study \cite{HHH23,FHHH24,HH24,HZ24,HWZ25,GHR25,AKL25} and are widely believed to exhibit strong but still poorly understood structural properties.
Set-defined graph classes correspond to a particularly well-structured subfamily of such problems: namely,  the communication problems that admit constant-cost \emph{deterministic} protocols with access to the Equality oracle~\cite{HWZ25}. 
To make this correspondence explicit, it is convenient to adopt the standard matrix representation of communication problems. In this representation, communication problems are modeled as (not necessarily symmetric) Boolean matrices, and permuting rows or columns does not affect the communication complexity of the corresponding problem. Interpreted graph-theoretically, such matrices are precisely the bipartite adjacency matrices of bigraphs, that is, bipartite graphs with an ordered bipartition. Accordingly, communication problems admitting constant-cost deterministic protocols with access to the Equality oracle can be associated with set-defined classes of bigraphs. We therefore refer to the corresponding classes of bipartite adjacency matrices as \emph{set-defined classes of matrices}.

The simplest examples of set-defined matrices are \emph{blocky matrices}.
These are blowups of partial permutation matrices, i.e.\ matrices obtained from permutation matrices by duplicating some rows and/or columns, and adding zero rows and/or zero columns; their corresponding bigraphs are disjoint unions of complete bipartite graphs.
Importantly, blocky matrices form a ``complete'' class for the family of set-defined classes of matrices in the sense that any matrix in such a class can be expressed as a fixed Boolean combination of blocky matrices.
In \cite{HHH23}, the \emph{blocky rank} of a matrix $M$ is defined as a minimum integer $r$ such that $M$ can be expressed as a linear combination (over $\mathbb{C}$) of $r$ blocky matrices. In particular, matrices of blocky rank 1 are exactly blocky matrices. It is not hard to show that for a Boolean matrix $M$, the blocky rank of $M$ is functionally equivalent to the minimum number of blocky matrices whose Boolean combination is equal to $M$ (the latter parameter is known as functional blocky rank \cite{AY24}). 
Consequently, classes of set-defined matrices are precisely those classes of Boolean matrices with bounded blocky rank.
This characterization yields a communication-complexity interpretation established in \cite{HHH23}: bounded blocky rank exactly captures communication problems admitting constant-cost deterministic protocols with access to the Equality oracle.

Beyond communication complexity, blocky rank has applications in cryptography~\cite{AN25},  circuit complexity~\cite{jukna2006graph,AY24}, and fine-grained complexity~\cite{williams2024orthogonal}.
We refer the reader to \cite{hambardzumyan2026spiky} for more details.

A further connection arises through the $\gamma_2$-norm, a matrix norm that was introduced to communication complexity by Linial and Shraibman \cite{LS07}.
The logarithm of the $\gamma_2$-norm lower bounds the cost of deterministic communication protocols with access to the Equality oracle \cite{HHH23}. Therefore, by the above discussion, classes of set-defined matrices have bounded $\gamma_2$-norm. 
The converse was conjectured in~\cite{HHH23}, where it was shown to be equivalent to an important open question in operator theory about characterization of idempotent Schur multipliers in terms of contractive idempotents. The conjecture has since been proved in~\cite{BGHJN26}. Thus, a class of Boolean matrices is set-defined if and only if its matrices have bounded $\gamma_2$-norm. 

\paragraph{Logical perspective.}
In \cite{JiangNesOdM20} set-defined graph classes were introduced as those included in the set of finite induced subgraphs of a graph definable in $\mathbb{N}$, considered as an infinite set just equipped with equality (hence the name). 
This positions them within the landscape of definability in finite model theory. They have an orthogonal relationship to first-order (FO) transductions, which are the central tool for understanding the complexity of FO model checking. FO transductions are able to use a non-deterministic coloring and arbitrary first-order formulas, but define the edge relation on singletons. In contrast, set-defined graph classes cannot use  coloring and only use quantifier-free formulas, but can define the edge relation on tuples.

While graphs definable in an infinite set equipped with equality may seem limited, the ability to pass to a subclass means that the definition of set-defined classes would be 
the same if one more broadly considered graphs interpretable in the theory of equality, or even trace definable in the theory of equality in the sense of \cite{walsberg2021notes}. So the infinite structures giving rise to set-defined classes are richer than they first appear. The structures trace definable in the theory of equality have recently received attention in the context of infinite-domain constraint satisfaction problems \cite{bodirsky2025structures, bodirsky2025taking}, and Bodor has shown (personal communication) that they include a large class of structures studied by model-theorists (the $\omega$-categorical $\omega$-stable structures with disintegrated algebraic closure \cite{lachlan1987structures}).

\paragraph{Applications.}
Some specific set-defined graph classes play an important role in a variety of application areas.

De~Bruijn graphs of dimension $d$ are defined as follows: the vertices are all $d$-tuples over a fixed alphabet, and there is an edge from a vertex $u$ to a vertex $v$ if and only if the $(d-1)$-suffix of $u$ coincides with the $(d-1)$-prefix of $v$. These graphs arise naturally in several domains. In bioinformatics, they were used by \cite{IW95} and later \cite{PTW01} to address the problem of assembling sequencing reads into a genome, leading to practical tools \cite{ZB08} and to the successful assembly of large-scale human genome sequences \cite{LZR10}. Beyond bioinformatics, De~Bruijn graphs are used in the distributed hash table protocol Koorde \cite{KK03} and in time-series forecasting \cite{CKB25}.

The class of $d$-dimensional shift graphs is defined similarly to $d$-dimensional De~Bruijn graphs, with the key difference that vertices are $d$-tuples of strictly increasing integers. These graphs form a classical example of triangle-free graphs with large chromatic number \cite{EH68} and play an important role in a number of mathematical fields including poset theory \cite{FHRT92}, Ramsey theory \cite{DLR95}, and model theory \cite{halevi2022infinite}.

Another class of set-defined graphs is defined over $d$-tuples, where two vertices are adjacent if and only if the corresponding tuples share a common element in different positions. Lower bounds on the chromatic number of these graphs were used in \cite{AFCK25} to obtain tight bounds on the size of data structures for monotone minimal perfect hashing.

Finally, several graph classes of interest in discrete mathematics, consisting of graphs whose vertices correspond to $k$-element sets and whose edges are defined by prescribed relations between these sets, are also set-defined. Notable examples include generalized Kneser graphs, generalized Johnson graphs, line graphs, and line graphs of $k$-uniform hypergraphs. Certain generalized Johnson graphs also arise in quantum computing, where they are used in an efficient algorithm for shadow tomography~\cite{King2025triply}.

\section{Preliminaries} \label{sec:prelim}
For natural numbers $a,b$ with $a \leq b$, we denote by $[a,b]$ the set $\{ a, a+1, \ldots, b \}$. In the case $a=1$, we use the shorthand $[b]$ for $[1,b]$.
A \emph{discrete partition} of a set $S$ is a partition in which every element of $S$ forms a singleton set.
The multicolour Ramsey number $R(p_1,p_2,\ldots,p_t)$ is the minimum $n$ such that in any edge colouring of $K_n$ with $t$ colours, there exists a monochromatic clique of size $p_i$ for some $i \in [t]$. When all $p_i$ are equal to the same value $p$, we write $R_t(p)$ as a shorthand for the corresponding multicolour Ramsey number.
For a set $I \subseteq \bN$ we denote by $\bN^I$ the set of functions from $I$ to $\bN$, and by $\Ninj{I}$ the set of injective functions from $I$ to $\bN$; when $I=[d]$ for some $d \in \bN$, we write $\bN^d$ and $\Ninj{d}$ instead. We sometimes refer to elements of $\bN^I$ and $\Ninj{I}$ as \emph{vectors} and \emph{injective vectors}, respectively, indexed by the elements of $I$.
For $v \in \bN^I$ and non-empty set $S \subseteq I$, we denote by $v_{|S}$ the restriction of $v$ to $S$, and by $v(S)$ we denote the set of values of $v$ on $S$, i.e.\ $v(S) := \{ v_i : i \in S \}$. For brevity, when $S$ consists of a single element $i$, we write $v_i$ instead of $v_{|\{i\}}$. For a set $V \subseteq \bN^I$ we define $V_{|S} := \{ v_{|S} : v \in V \}$.
We say that $u,w \in \bN^I$ have the same \emph{order type} if $u_i < u_j \iff w_i < w_j$ for all $i,j \in I$.
A set $V$ is \emph{order-uniform} if all its elements are of the same order type. We say that $V$ is \emph{injective} if $V \subseteq \Ninj{I}$.

All undirected graphs in this paper are simple, i.e.\ without loops or multiple edges. Unless stated otherwise, directed graphs (digraphs) are without multiple edges in the same direction and without loops. If there is a directed edge from a vertex $u$ to a vertex $v$ in a digraph, we sometimes denote such an edge as $u \mapsto v$. The \emph{underlying graph} of a digraph $G$ is a simple graph obtained from $G$ after ignoring edge orientations and collapsing multiple edges to single edges. 
Given a (di)graph $G$, we write $V(G)$ for its vertex set, and $E(G)$ for its edge set. A (di)graph $H$ is a \emph{sub(di)graph} of $G$ if $V(H) \subseteq V(G)$ and $E(H) \subseteq E(G)$, and $H$ is an \emph{induced sub(di)graph} of $G$ if $V(H) \subseteq V(G)$, and $E(H)$ consists exactly of the edges in $E(G)$ with both endpoints in $V(H)$.
A digraph $G$ is connected if its underlying graph is connected. A connected component of $G$ is a maximal connected subgraph of $G$. 
For an undirected graph $G$, its \emph{clique number} $\omega(G)$ is the size of the largest complete subgraph in $G$, and its \emph{chromatic number} $\chi(G)$ is the smallest number of colours needed to colour the vertices of $G$ so that no two adjacent vertices share the same colour. 
For a digraph $G$, its clique number $\omega(G)$ and  chromatic number $\chi(G)$ are defined as the corresponding parameters of the underlying graph of $G$. 

A \emph{homomorphism} from a digraph $G$ to a digraph $H$ is a function $\rho : V(G) \rightarrow V(H)$ that preserves edges, i.e.\ $(u,v) \in E(G)$ implies $(\rho(u), \rho(v)) \in E(H)$ for every ordered pair of vertices $u,v \in V(G)$. We say that $G$ is \emph{homomorphic to $H$} if there exists a homomorphism from $G$ to $H$.
It is easy to see that if $G$ is homomorphic to $H$, then $\chi(G) \leq \chi(H)$ and $\omega(G) \leq \omega(H)$.

\subsection{Set-defined graphs}

Given $d \in \mathbb{N}$, a Boolean function $f : \{0,1\}^{d^2} \rightarrow \{0,1\}$, and a set $V \subseteq \mathbb{N}^d$, the \emph{realisation of $f$ over $V$} is the digraph $G=(V,E)$ with $(u,w) \in E$ if and only if 
\[
    f(Q_{u,w}(1,1), Q_{u,w}(1,2), \ldots, Q_{u,w}(d,d-1), Q_{u,w}(d,d)) = 1,
\]
where $Q_{u,w}(i,j) = 1$ if $u_i = w_j$, and $Q_{u,w}(i,j) = 0$ if $u_i \neq w_j$. If $F$ is a Boolean formula that represents $f$, we will also say that $G$ is a realisation of $F$ over $V$. 
If $G$ is the realisation of $F$ on $V$, we will also say that \emph{$G$ is induced by $F$ on $V$}.
For convenience, we will occasionally abuse notation and write $f(u,w)$ or $F(u,w)$ to denote $f(Q_{u,w}(1,1), Q_{u,w}(1,2), \ldots, Q_{u,w}(d,d))$. 
We will denote the variables of the function $f$ by $q_{i,j}$, and interpret them as indicators of equality between the $i$-th coordinate of the first vertex and the $j$-th coordinate of the second vertex. That is, the existence of an edge from $u$ to $w$ is determined by the value of $f$ when the variables $q_{i,j}$ are set to $Q_{u,w}(i,j)$ for all $i,j \in [d]$.

A digraph is a \emph{$(d,f)$-set-defined digraph} if it is isomorphic to a realisation of $f$ over some $V \subseteq \mathbb{N}^d$.
An (undirected) graph is a \emph{$(d,f)$-set-defined graph} if it is the underlying graph of a $(d,f)$-set-defined \emph{digraph}.
When analysing the structure of an arbitrary $(d,f)$-set-defined (di)graph $G=(V,E)$, we will often assume, without loss of generality, that $V \subseteq \bN^d$ and $G$ is a realisation of $f$ over $V$.

\begin{definition}[Set-defined classes]
    A class of digraphs is \emph{set-defined} if there exists $d \in \mathbb{N}$ and a Boolean function $f : \{0,1\}^{d^2} \rightarrow \{0,1\}$ such that every digraph in the class is $(d,f)$-set-defined. Such a class is \emph{full} if it consists of all $(d,f)$-set-defined digraphs.
    We denote by 
    \begin{enumerate}
        \item[$\dcX_f$] the full class of $(d,f)$-set-defined digraphs;
        \item[$\dcX_f'$] the class of $(d,f)$-set-defined digraphs that admit a realisation of $f$ over an \emph{injective} set $V \subseteq \Nd$;
        \item[$\dcY_f$] the class of $(d,f)$-set-defined digraphs that admit a realisation of $f$ over an \emph{injective order-uniform} set $V \subseteq \Nd$.
    \end{enumerate}
    A class is \emph{\nf} if it is of the form $\dcX_f$ or $\dcY_f$ for some Boolean function $f$. From the definition, we have $\dcY_f \subseteq \dcX_f' \subseteq \dcX_f$.
    If $F$ is a Boolean formula representing $f$, we also write $\dcX_F$, $\dcX_F'$, and $\dcY_F$ to denote the classes $\dcX_f$, $\dcX_f'$, and $\dcY_f$, respectively.
\end{definition}

\subsection{Shift graphs}

Because shift (di)graphs will play a central role in our results, we record some basic facts and notation concerning them.

We recall the \emph{canonical $d$-dimensional shift digraph} $\vec{S}(n,d)$ from the introduction, and the corresponding undirected graph $S(n,d)$. An alternative definition is to define $\vec{S}(n,2)$ to be the directed line graph of a transitive tournament of size $n$, and then to define $\vec{S}(n,d+1)$ to be the directed line graph of $\vec{S}(n,d)$ (see \cite[Lemma 2.20]{hell2004graphs}). (We recall that the directed line graph of a digraph $G = (V, E)$ is the digraph with vertex set $E$ and an edge from $(a, b)$ to $(c, d)$ if $b = c$.)

For a fixed $d$, the induced subdigraphs of the digraphs $\vec{S}(n,d)$ for all $n \geq 1$ form the class of \emph{$d$-dimensional shift digraphs}, which we denote $\dcS_d$, and the corresponding undirected graphs form the class of \emph{$d$-dimensional shift graphs}, which we denote $\cS_d$.

\begin{fact} \label{fact:shift}
    \begin{enumerate}
        \item Each $\cS_d$ has odd-girth $2d+1$, and in particular is triangle-free.
        \item Each $\dcS_d$ has unbounded chromatic number. In particular\footnote{The superscript $d-1$ indicates an iterated logarithm, e.g. $\log^{(2)}_2 n =\log_2\log_2 n$}, $\chi(\vec{S}(n,d)) \geq \log^{(d-1)}_2 n $.
        \item For each $d$, $\dcS_d \supsetneq \dcS_{d+1}$.
    \end{enumerate}
\end{fact}
\begin{proof}
    \begin{enumerate}
        \item The lower bound is given in \cite[Lemma 2.22]{hell2004graphs}. That $S(2d+2, d)$ contains $C_{2d+1}$ can be seen directly.
        \item This follows from the definition of $\vec{S}(n,d)$ as the iterated directed line graph of a tournament, and the lower bound on the chromatic number of a directed line graph of a digraph in \cite[Lemma 2.21]{hell2004graphs}.
        \item From the definition of $\vec{S}(n,d)$ in terms of iterated directed line graphs, we see that for every $d \geq 2$, the $(d-1)$-iterated directed line graph of any acyclic digraph is an induced subgraph of $\vec{S}(n, d)$ for some $n$. For $d \geq 2$, $\vec{S}(n,d+1)$ is the $(d-1)$-iterated directed line graph of the acyclic digraph $\vec{S}(n,2)$, and so is in $\dcS_d$. Thus $\dcS_d \supseteq \dcS_{d+1}$, and the strictness of the inclusion follows from the odd-girth in part 1.
        \qedhere
    \end{enumerate}
\end{proof}

The last point of Fact \ref{fact:shift} justifies our convention of calling $\dcS_2$ the class of shift digraphs.

Finally, we note the following example referred to in Section \ref{ssec:related_work}.

\begin{example} \label{ex:shift_diag}
    We give an example of a hereditary class $\cS$ of shift graphs with unbounded chromatic number that does not contain $\cS_d$ for any $d$. Let $t_i(x)$ be the tower function defined by $t_1(x) = x$ and $t_{i+1}(x) = 2^{t_i(x)}$. For each $k \geq 2$, let $G_k = S(t_k(k),k)$, and let $\cS$ be the hereditary closure of $\{G_k : k \geq 2\}$. By Fact \ref{fact:shift}(2), $\chi(G_k) \geq k$. By (the proof of) Fact \ref{fact:shift}(1), $S(2d+2, d)$ contains $C_{2d+1}$ and so is not an induced subgraph of $S(n', d')$ for any $d' > d$ and any $n'$. Let $n_d  = \max(2d+2, t_d(d)+1)$. Then for each $d \geq 2$, $S(n_d, d) \not\in \cS$.

    Similar reasoning gives continuum-many such classes that form an antichain under inclusion. Note that for $k \geq 3$, $G_k$ is the largest graph in $\cS$ that contains $C_{2k+1}$, and so $\{G_k : k \geq 3\} \subset \cS$ forms an antichain of maximal elements under the induced subgraph embedding. 
    Given $X \subset \{k \in \N : k \geq 3\}$, let $\cS_X = \cS - \{G_k : k \in X\}$. Then each $\cS_X$ is still a hereditary class, and has unbounded chromatic number. Also, $\cS_X$ and $\cS_{X'}$ are incomparable if $X$ and $X'$ are, and there exist continuum-many inclusion-incomparable subsets of $\N$. However, we note that the single \chiunbounded class $\cS_\N$ is contained in all these classes.
\end{example}

A different construction of such classes is given in Section \ref{sec:stability}.

\subsection{Boolean functions, DNFs, and clause digraphs}

Let $f : \{0,1\}^{d^2} \rightarrow \{0,1\}$ be a Boolean function with variables $q_{i,j}$, $i,j \in [d]$. Let $F = \bigvee_{i=1}^k C_i$ be a propositional formula in disjunctive normal form (DNF) representing $f$, where each $C_i$ is a conjunction of literals (i.e.\ variables or their negations). We will refer to these conjunctions as \emph{clauses} of $F$. A clause is \emph{full}\footnote{In Boolean algebra, this is also known as \emph{minterm}.}
if it contains exactly one literal for each of the $d^2$ variables. A DNF is \emph{full} if each of its clauses is full. It is a basic fact that any Boolean function admits a full DNF.
We say that a full DNF or a full clause is \emph{$d$-dimensional} if it is over $d^2$ variables.

Given a full clause $C$, the \emph{clause digraph} $\ClauseDigraph_C$ of $C$ is a  digraph, possibly with loops, on vertex set $[d]$, where for every $i,j \in [d]$ there is an edge from $i$ to $j$ if and only if $q_{i,j}$ appears positively in $C$.
Observe that any digraph on $[d]$ defines a full clause over the variables $q_{i,j}$, $i, j \in [d]$. That is, there is a natural one-to-one correspondence between full clauses and $d$-vertex digraphs.

Let $\cP$ be a partition of $[d]$. We say that a digraph on vertex set $[d]$ is \emph{$\cP$-compatible} if the in-neighborhood of every vertex is an empty set or a $\cP$-class and the out-neighborhood of every vertex is an empty set or a $\cP$-class. 
A full clause $C$ is \emph{$\cP$-compatible} if its clause digraph $\ClauseDigraph_C$ is $\cP$-compatible.
If every vertex of $\ClauseDigraph_C$ has at most one in-neighbor and at most one out-neighbor, we call $C$ \emph{injective}\footnote{Such a digraph $\ClauseDigraph_C$ can be viewed as the functional digraph of a \textit{partial injective} function; we omit the word ``partial'' for brevity.}.
Clearly, in the clause digraph of an injective clause every connected component corresponds to either a directed path or a directed cycle.
We say that an injective clause $C$ is
\begin{enumerate}
  \item \emph{connected} if the underlying graph of $\ClauseDigraph_C$ is connected;
    \item \emph{discrete} if every connected component of the underlying graph of $\ClauseDigraph_C$ is a single-vertex graph;

    \item a \emph{path} clause if every non-trivial connected component (i.e.\ connected component with at least 2 vertices) of $\ClauseDigraph_C$ is a directed path, and there exists at least one non-trivial connected component;

    \item a \emph{cyclic} clause if $\ClauseDigraph_C$ contains a directed cycle with at least 2 vertices;

    \item an \emph{acyclic} clause if it is not cyclic. 
\end{enumerate}

See \cref{fig:clause-type} for an illustration of different types of clause digraphs. A DNF is injective, discrete, acyclic, if it is full and each of its clauses is respectively injective, discrete, acyclic.

Given an injective clause $C$, we say that $i \in [d]$ is a \emph{loop coordinate with respect to $C$}, if $i$ is an isolated vertex with a loop in $\ClauseDigraph_C$. A path clause is called \emph{loopless} if it contains no loop coordinates.

\begin{figure}[htbp]
    \centering
    \begin{subfigure}[t]{0.24\textwidth}
        \captionsetup{margin={-1cm,0cm}}
        \centering
        \begin{tikzpicture}[
    ->,
    >=Stealth,
    circnode/.append style={inner sep=0pt, minimum size=5pt},
    every label/.append style={font=\scriptsize, inner sep=1pt},
]

\path[use as bounding box] (-0.6,-1.8) rectangle (2.6,0.6);

\node[circnode,label=left:$1$] (d1) at (0,0) {};
\node[circnode,label=right:$2$] (d2) at (1,0) {};

\node[circnode,label=left:$3$] (d3) at (0,-1.2) {};
\node[circnode,label=right:$4$] (d4) at (1,-1.2) {};

\draw (d1) -- (d2);
\draw (d2) -- (d3);
\draw (d3) -- (d4);

\end{tikzpicture}
        \caption{Connected, path, loopless}
    \end{subfigure}
    \hspace{-0.5cm}
    \begin{subfigure}[t]{0.24\textwidth}
        \captionsetup{margin={-1.2cm,0cm}}
        \centering
        \begin{tikzpicture}[ ->, >=Stealth,    
    circnode/.append style={inner sep=0pt, minimum size=5pt}, 
    every label/.append style={font=\scriptsize, inner sep=1pt}, 
] 
    \path[use as bounding box] (-0.6,-1.8) rectangle (2.6,0.6);
    
    \node[circnode,label=left:$1$] (b1) at (0,0) {}; 
    \node[circnode,label=right:$2$] (b2) at (1,0) {}; 
    \node[circnode,label=left:$3$] (b3) at (0,-1.2) {}; \node[circnode,label=right:$4$] (b4) at (1,-1.2) {}; 
    
    \draw (b1) edge[loop above] (b1); 
    \draw (b2) edge[loop above] (b2); 
    \draw (b3) edge[loop above] (b3);
\end{tikzpicture}
        \caption{Discrete}
    \end{subfigure}
    \hspace{-0.5cm}
    \begin{subfigure}[t]{0.24\textwidth}
        \centering







\begin{tikzpicture}[
    ->,
    >=Stealth,
    circnode/.append style={inner sep=0pt, minimum size=5pt},
    every label/.append style={font=\scriptsize, inner sep=1pt},
]

\path[use as bounding box] (-0.5,-1.65) rectangle (2.45,0.45);

\node[circnode,label={[xshift=-1pt]left:$1$}] (a1) at (0,0) {};
\node[circnode,label={[yshift=1pt]right:$2$}] (a2) at (1,0) {};
\node[circnode,label={[xshift=1pt]right:$3$}] (a3) at (2,0) {};

\node[circnode,label={[xshift=-1pt]left:$4$}] (a4) at (0,-1.2) {};
\node[circnode,label={[yshift=-1pt]right:$5$}] (a5) at (1,-1.2) {};
\node[circnode,label={[xshift=1pt]right:$6$}] (a6) at (2,-1.2) {};

\draw (a1) -- (a2);
\draw (a2) -- (a5);
\draw (a5) -- (a4);
\draw (a4) -- (a1);

\draw (a6) -- (a3);

\end{tikzpicture}
        \caption{Cyclic}
    \end{subfigure}
    \hspace{0.5cm}
    \begin{subfigure}[t]{0.24\textwidth}
        \centering
        \begin{tikzpicture}[
    ->,
    >=Stealth,
    circnode/.append style={inner sep=0pt, minimum size=5pt},
    every label/.append style={font=\scriptsize, inner sep=1pt},
]

\path[use as bounding box] (-0.6,-1.8) rectangle (2.6,0.6);

\node[circnode,label=above:$1$] (c1) at (0,0) {};
\node[circnode,label=above:$2$] (c2) at (1,0) {};
\node[circnode,label=above:$3$] (c3) at (2,0) {};

\node[circnode,label=below:$4$] (c4) at (0,-1.2) {};
\node[circnode,label=below:$5$] (c5) at (1,-1.2) {};
\node[circnode,label=below:$6$] (c6) at (2,-1.2) {};

\draw (c1) -- (c2);
\draw (c2) -- (c3);

\draw (c5) -- (c4);

\draw (c6) edge[loop above] (c6);

\end{tikzpicture}
        \caption{Path, acyclic}
    \end{subfigure}
\caption{Examples of clause digraphs corresponding to different types of injective clauses.}
\label{fig:clause-type}
\end{figure}
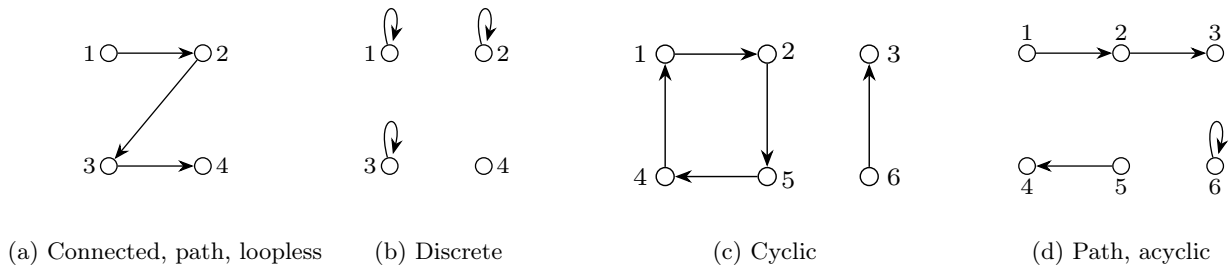

Let $C$ be a clause over the variables $q_{i,j}$, $i,j \in [d]$, and let $I \subseteq [d]$. We denote by $C_{|I}$ the clause obtained from $C$ by removing literals corresponding to variables $q_{i,j}$ with at least one of the two indices $i$ and $j$ being outside $I$.

\begin{lemma} \label{lem:subclause}
	Let $C$ be a clause, let $I \subseteq [d]$,  and let $V \subseteq \bN^d$ be such that for all distinct $u, v \in V$, if $u_{|I} = v_{|I}$, then $C$ does not hold on $(u,v)$.
    Let $G$ be the realisation of $C$ over $V$ and $G'$ the realisation of $C_{|I}$ over $V_{|I}$.
    Then $G$ is homomorphic to $G'$.
\end{lemma}
\begin{proof}
    We claim that $v \mapsto v_{|I}$ is a homomorphism from $G$ to $G'$.
    Let $u,v$ be arbitrary vertices in $V$. If $u_{|I} = v_{|I}$, then by assumption $(u,v)$ is not an edge of $G$.
    If there is an edge $(u,v)$ in $G$, then all literals of $C$ corresponding to the variables $q_{i,j}$, with $i,j \in I$, must be satisfied by $(u,v)$. This means that all literals of $C_{|I}$ must be satisfied by $(u_{|I},v_{|I})$, i.e.\ $(u_{|I},v_{|I})$ is an edge in $G'$.
\end{proof}

\subsection{Boolean functions of graphs and graph classes}\label{sec:bool-fun-graph-classes}

Let $G, H_1, H_2, \ldots, H_k$ be $n$-vertex (di)graphs on the same vertex set $V$. 
A (di)graph $G$ is said to be a \emph{Boolean combination} of $H_1, H_2, \ldots, H_k$, if there exists a Boolean function
$f : \{0,1\}^k \rightarrow \{0,1\}$ such that 
\begin{equation}\label{eq:fun-def}
    G(a,b) = f\left( H_1(a,b), H_2(a,b), \ldots, H_k(a,b) \right)
\end{equation}
holds for all distinct $a,b \in V$. 
In other words, the adjacency matrix of $G$ is obtained by applying $f$ to the adjacency matrices of $H_1, H_2, \ldots, H_k$ entry-wise. 
Abusing notation, we will sometimes write $G = f(H_1, H_2, \ldots, H_k)$.
For example, we write $G = H_1 \vee H_2 \vee \cdots \vee H_k$ to denote that $G$ is the disjunction (i.e. the union) of $H_1, H_2, \ldots, H_k$.

\begin{lemma}\label{lem:omega-of-union}
    Let $G=(V,E)$, and $H_i=(V,E_i)$, $i \in [k]$, be such that $G = \bigvee_{i=1}^k H_i$. 
    If $\omega(H_i) \leq c_i$ for every $i \in [k]$, then $\omega(G) < R(c_1+1, c_2+1, \ldots, c_k+1)$.
\end{lemma}
\begin{proof}
    Assign to every edge of $G$ one of the colours in $[k]$ depending on which of the graphs $H_1$, $H_2$, \ldots, $H_{k-1}$, or $H_k$ contributed that edge to the union (breaking ties arbitrarily). Then, by Ramsey's theorem, if $\omega(G) \geq R(\omega(H_1)+1, \omega(H_2)+1, \ldots, \omega(H_k)+1)$, then there exists $i \in [k]$ such that $H_i$ contains a clique of size $\omega(H_i)+1$, which is clearly impossible.
\end{proof}

For graph classes $\cX_1, \cX_2, \ldots, \cX_k$ and a Boolean function
$f : \{0,1\}^k \rightarrow \{0,1\}$, we write $f(\cX_1, \cX_2, \ldots, \cX_k)$ to denote the class of graphs $\{G=f(H_1, H_2, \ldots, H_k) : H_i=(V,E_i) \in \cX_i \}$.
It is easy to see that if all $\cX_i$, $i \in [k]$, are hereditary
then $f(\cX_1, \cX_2, \ldots, \cX_k)$ is also hereditary. 
We denote 
$\bigvee_{i=1}^k \cX_i := \big\{ H_1 \vee H_2 \vee \cdots \vee H_k : H_i=(V,E_i) \in \cX_i, i \in [k] \big\}$, and call this class the union of $\cX_1, \cX_2, \ldots, \cX_k$.

\begin{observation}\label{lem:union}
    Let $F_1, F_2, \ldots, F_k$ be Boolean formulas and $F = \bigvee_{i=1}^k F_i$. Let $V \subseteq \bN^d$ and $G_i$ be the realisation of $F_i$ over $V$ for every $i \in [k]$. Then the realisation of $F$ over $V$ is equal to $\bigvee_{i=1}^k G_i$. In particular, $\cX_F \subseteq \bigvee_{i=1}^k \cX_{F_i}$ and $\cY_F \subseteq \bigvee_{i=1}^k \cY_{F_i}$.
\end{observation}

\subsection{\texorpdfstring{$\chi$}{chi}-boundedness}

A class $\cX$ is \emph{\chibounded} if there exists a function $h : \bN \rightarrow \bN$ such that $\chi(H)\leq h(\omega(H))$ holds for every induced subgraph $H$ of a graph $G \in \cX$; 
in this case, the function $h$ is called a \emph{$\chi$-binding function} for $\cX$.
If $h$ can be chosen to be polynomial (resp.\ linear), then $\cX$ is \emph{polynomially} (resp.\ \emph{linearly}) \chibounded. If $\cX$ is not \chibounded, we say that it is \chiunbounded.
When we consider a function which is $\chi$-binding for some class of graphs, without loss of generality, we always assume that such a function is non-decreasing.
We will use the following result from \cite{Gya87}, showing that a disjunction of \chibounded classes is also \chibounded.

\begin{lemma}[{\cite[Proposition 5.1(a)]{Gya87}}]
\label{lem:union-chi-bounded}
    Let $\cX_1, \cX_2, \ldots, \cX_k$ be $\chi$-bounded classes of graphs with
    $\chi$-binding functions $h_1, h_2, \ldots, h_k$ respectively. Then the class $\bigvee_{i=1}^k \cX_i$ is \chibounded and $g(w)=\prod_{i=1}^k h_i(w)$ is a suitable $\chi$-binding function. 
\end{lemma}

\subsection{Functional sets}\label{sec:functional-sets}

Let $d \in \bN$ and $L \subsetneq [d]$.
Denote $\ell := |L|$, $\overline{L} := [d] \setminus L$, and let $\lambda \in \overline{L}$.
A set $U \subseteq \bN^d$ is \emph{$(L,\lambda)$-functional}, if for every $v \in U$, the value $v_{\lambda}$ is determined by the vector $v_{|L}$, i.e.\ there exists a function $\varphi \colon \bN^{\ell} \rightarrow \bN$ such that $v_{\lambda} = \varphi(v_{|L})$ holds for every $v \in U$.
Note, if $U$ is $(\emptyset, \lambda)$-functional, then all vectors in $U$ have the same $\lambda$-th coordinate.
A set $U \subseteq \bN^d$ is \emph{$L$-functional} if it is $(L,\lambda)$-functional for some $\lambda \in \overline{L}$.

More generally, let $k \in \bN$, and $L_s \subsetneq [d]$ and $\lambda_s \in \overline{L_s}$ for $s \in [k]$. A set $U \subseteq \bN^d$ is \emph{$((L_1,\lambda_1),(L_2,\lambda_2), \ldots, (L_k,\lambda_k))$-functional} if it is $(L_s,\lambda_s)$-functional for every $s \in [k]$, in such a case we refer to $((L_1,\lambda_1),(L_2,\lambda_2), \ldots, (L_k,\lambda_k))$ as \emph{functional constraints} of $U$. A set $U \subseteq \bN^d$ is 
$(L_1, L_2, \ldots, L_k)$-functional if it is $((L_1,\lambda_1),(L_2,\lambda_2), \ldots, (L_k,\lambda_k))$-functional for some $\lambda_s \in \overline{L_s}$, $s \in [k]$.

\subsection{Tropical algebra and mean payoff games}
\label{sec:tropical-prelim}

We use the term tropical algebra to speak about two related algebraic objects: the \emph{max-plus algebra} and \emph{min-plus algebra}.
Both notions will be used.
However, since they are dual to each other, we primarily present the results in the language of max-plus algebra (in compliance with the majority of literature we cite).
We mention the appropriate statements for min-plus algebra at the end of this subsection.
For a comprehensive background on tropical algebra and its connection to combinatorics see \cite{Butkovic2010} or \cite{Joswig2021}.

The max-plus algebra is the semiring $(\R \cup \{-\infty\}, \maxpl, \maxmu, -\infty, 0)$, where the `addition' $\maxpl$ is defined as $\max\{\cdot,\cdot\}$, `multiplication' $\maxmu$ as $+$, with $-\infty$ being the neutral element to $\maxpl$, and $0$ being the neutral element to $\maxmu$. 
For example,
\[
    (5 \maxmu 3) \maxpl (6 \maxmu -\infty) = 8 \maxpl -\infty = 8.
\]
We use $\Rmax = \R \cup \{-\infty\}$ and $\Zmax = \Z \cup \{-\infty\}$ for brevity. For a matrix $A \in \Rmax^{m \times n}$ and a vector $x \in \Rmax^n$, we write $A \maxmu x$ to denote the max-plus matrix-vector multiplication, which is given by $A \maxmu x=(\max_{j\in [n]}\{A_{i,j} + x_i\})_{i \in [m]} $.
If no confusion arises, we may write $Ax$ instead of $A \maxmu x$.

Our main focus is on the max-plus inequalities of the form
\[
    Ax \leq Bx
    ,
\]
where $A, B \in \Rmax^{m \times n}$ are fixed matrices and $x \in \Rmax^n$ is a vector of variables.
The (ordinary) inequality between the resulting $m$-dimensional vectors should hold entry-wise.

\begin{observation}\label{obs:shifting_solution}
    If $x^*$ solves a max-plus system $Ax \leq Bx$, then any $x^* + \alpha$, $\alpha \in \R$ solves it as well.
    The same holds true for min-plus systems.
\end{observation}

It was observed by \cite{Dhingra2006} and \cite{akian2012tropical} that such systems are closely related to \emph{mean payoff games}.
Informally, a mean payoff game is a two-player perfect information game on a finite bipartite graph between a \emph{column} and \emph{row} player controlling the columns and rows of the matrix $A$ and $B$, respectively.

To define this game precisely, we first introduce the bipartite weighted digraph associated with a pair of matrices $(A,B)$.
Let $\Gamma = \Gamma(A,B)$ be a bipartite weighted digraph with parts $\cR = [m]$ and $\cC = [n]$.
For each $i \in \cR$ and $j \in \cC$, unless $A_{ij} = -\infty$, $\Gamma$ has a directed edge $(j,i) \in E(\Gamma)$ with weight $-A_{ij}$.
Similarly, for each $i \in \cR$ and $j \in \cC$, unless $B_{ij} = -\infty$, $\Gamma$ has a directed edge $(i,j)$ with weight $B_{ij}$.
The weight of an edge $e$ is denoted by $w(e)$. 
We will always work under the following standard assumption.

\begin{assumption}\label{assump:valid_moves_exist}
    Each column of $A$ as well as each row of $B$ contains a finite entry.
    This is equivalent to the property that each vertex of the graph $\Gamma(A,B)$ has at least one outgoing edge.
\end{assumption}

Given an initial vertex $s_{0} \in \cR \cup \cC$, the \emph{mean payoff game} $(\Gamma,s_0)$ associated with the system $Ax \leq Bx$ is an infinite game between the \emph{row} and the \emph{column} player defined as follows.
The row player controls the vertices in $\cR$ (indexed by rows of the matrices), while the column player controls the vertices in $\cC$ (indexed by columns of the matrices).
The game starts by placing a token on the initial vertex $s_0$.
In each round, the player controlling the currently occupied vertex $s$ moves the token along an outgoing edge of $s$, say $(s,s')$, and the value $w(s,s')$ is transferred from the column player to the row player (the value might be negative).
The game continues with the token placed on the vertex $s'$ and never stops.
\cref{assump:valid_moves_exist} guarantees that each vertex has at least one outgoing edge, so the game may always continue.
An infinite sequence $W = (s_0, s_1, s_2, \dots)$, where $(s_i,s_{i+1})$ is an edge in $\Gamma(A,B)$, is called a \emph{realization} of the game $(\Gamma,s_0)$.

The goal of the row player is to maximize the long-term average payoff from the column player, while the column player tries to minimize this value.
Formally, given a realization $W = (s_0, s_1, s_2, \dots)$ of the game $(\Gamma,s_0)$, we define the \emph{mean payoff} of the row player as 
\begin{equation}\label{eq:row_payoff}
    \nu_\textrm{R}(W) = \liminf_{N \to \infty} \frac{1}{N} \sum_{k=1}^N w(s_{k-1}, s_k),
\end{equation}
while the \emph{mean loss} of the column player is defined as
\begin{equation}\label{eq:column_loss}
    \nu_\textrm{C}(W) = \limsup_{N \to \infty} \frac{1}{N} \sum_{k=1}^N w(s_{k-1}, s_k).
\end{equation}
Note that, in general, these values do not have to be the same.

\begin{remark}
    Equivalently, some authors (such as \cite{akian2012tropical}) define the graph $\Gamma$ with $w(j,i) = A_{ij}$ instead of $-A_{ij}$.
    Then, instead of one-directional payments, they consider mutual payments between the row and the column player, i.e.\ the player controlling vertex $s$ receives the value $w(s,s')$ from the other player.
    Our approach, which is closer to \cite[Section~9]{Joswig2021} (but with an exchanged role of $A$ and $B$), is more convenient for us due to the fact that the payoff of the row player over a sequence of moves is equal to the sum of the weights of the traversed edges.
\end{remark}

\begin{example}[Taken from {\cite[Example~9.1]{Joswig2021}}]\label{ex:game_graph_example0}
    Consider the matrices
    \begin{align*}
        A =
        \begin{pmatrix}
         3      & -\infty \\
        -3      & -\infty \\
        -2      & 1
        \end{pmatrix}
        \quad \text{and} \quad
        B =
        \begin{pmatrix}
         1      & -\infty \\
        -3      & 2       \\
        -\infty & 4
        \end{pmatrix}
        ,
    \end{align*}
    corresponding to the system
    \begin{align*}
        3+x_1 &\leq 1+x_1, \\
        -3+x_1 &\leq \max\{-3+x_1,2+x_2\}, \\
        \max\{-2+x_1,1+x_2\} &\leq 4+ x_2.
    \end{align*}

        The corresponding game graph $\Gamma$ is displayed in Figure~\ref{fig:game_graph-ex0}.
\end{example}

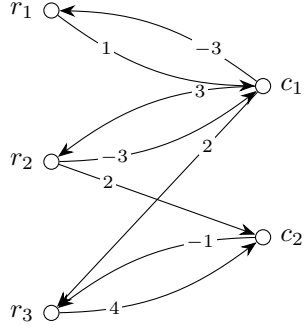
\begin{figure}
    \centering
    \begin{tikzpicture}[xscale=.8,
	dot/.style={circle,fill,inner sep=1.8pt},
	weight/.style={font=\scriptsize, fill=white, inner sep=1pt},
	>= {Stealth[scale=1.25]}
	]

    	\tikzset{
			circnode/.style={
				draw,
				circle,
				minimum size=0.2cm,
				inner sep=0pt
			}
		}
        
	\def\Xleft{-3.5cm}
	\def\Xmid{0cm}
	\def\Xright{3.5cm}
	
	\def\vertsep{2cm}
	
	\node[circnode,label=right:$c_1$] (Cx1) at (\Xmid,{2*\vertsep}) {};
	\node[circnode,label=right:$c_2$] (Cx2) at (\Xmid,{1*\vertsep}) {};
	
	\node[circnode,label=left:$r_1$] (Ls1) at (\Xleft,{2.5*\vertsep}) {};
	\node[circnode,label=left:$r_2$] (Ls2) at (\Xleft,{1.5*\vertsep}) {};
	\node[circnode,label=left:$r_3$] (Ls3) at (\Xleft,{0.5*\vertsep}) {};
	
	\newcommand{\edge}[4]{%
		\ifnum#3=0
		\draw[->] (#1) -- (#2);
		\else
		\draw[->] (#1) -- node[pos=#4, weight] {$#3$} (#2);
		\fi
	}
	
	\newcommand{\bedge}[6]{%
		\ifnum#3=0
		\draw[->, bend #5=15] (#1) to (#2);
		\else
		\draw[->, bend #5=15] (#1) to node[pos=#4, weight, #6] {$#3$} (#2);
		\fi
	}
	
	\def\pos{0.25}
	
	\bedge{Cx1}{Ls1}{-3}{\pos}{right}{}
	\bedge{Cx1}{Ls2}{3}{\pos}{right}{}
	\edge {Cx1}{Ls3}{2}{\pos}
	
	\bedge{Cx2}{Ls3}{-1}{\pos}{right}{}
	
	\bedge{Ls1}{Cx1}{1}{\pos}{right}{}
	\bedge{Ls2}{Cx1}{-3}{\pos}{right}{}
	\edge {Ls2}{Cx2}{2}{\pos}
	
	\bedge{Ls3}{Cx2}{4}{\pos}{right}{swap}
	
\end{tikzpicture}
    \caption{Game graph $\Gamma$ corresponding to \cref{ex:game_graph_example0}, with $\mathcal{R}=\{r_1,r_2,r_3\}$ and $\mathcal{C}=\{c_1,c_2\}$. }
    \label{fig:game_graph-ex0}
\end{figure}

A \emph{strategy} of a player is a function $\sigma$ which, given a finite sequence of moves $W$ that ends in a player's vertex $s$, produces an outneighbor $\sigma(W)$ of $s$.
We say that the player follows the strategy $\sigma$ if, upon occurrence of the sequence $W$, the player chooses the edge $(s, \sigma(W))$.
A strategy is \emph{positional} if it depends only on the current position of the token.
That is, if it satisfies $\sigma(W) = \sigma(W')$ for any two sequences $W$ and $W'$ that end in the same vertex.

A fundamental theorem of mean payoff games is the following.

\begin{theorem}[\cite{Ehrenfeucht1979}\cite{GURVICH198885}]\label{thm:positional_strategy}
    There are positional strategies $\tau$ and $\sigma$ for the row and column player, respectively, such that for each $s_0 \in V(\Gamma)$, there is a value $\nu(\Gamma, s_0) \in \R$ satisfying that in the mean payoff game $(\Gamma,s_0)$ 
    \begin{enumerate}[label=(\roman*)]
        \item the mean payoff of the row player is at least $\nu(\Gamma, s_0)$ when following $\tau$, and
        \item the mean loss of the column player is at most $\nu(\Gamma, s_0)$ when following $\sigma$.
    \end{enumerate}
\end{theorem}

We say that the strategy of a player ensuring them in the game $(\Gamma,s_0)$ the value $\nu(\Gamma, s_0)$ is \emph{optimal}.
Observe that the mean payoff~\eqref{eq:row_payoff} for the row player is always at most the mean loss~\eqref{eq:column_loss} of the column player; hence, Theorem~\ref{thm:positional_strategy} implies that equality can be obtained by optimal positional strategies.

A state $s_0 \in V(\Gamma)$ is \emph{winning} for the row, or column player if $\nu(\Gamma, s_0) > 0$, or $\nu(\Gamma, s_0) < 0$, respectively.
Otherwise, we say that $s_0$ is a \emph{draw-state}.%
\footnote{Some authors (such as \cite{akian2012tropical}) define a state $s_0$ as winning for the row player if $\nu(\Gamma, s_0) \geq 0$.}
The natural language constructions apply, e.g. a state is non-losing for a player if it is  a draw or winning state for the player.

The key connection between tropical algebras and mean payoff games is formalized in the following statement. See~\cite[Corollary~3.4]{akian2012tropical}, where it is attributed to~\cite{Dhingra2006}.

\begin{theorem}\label{thm:connection_of_tropical_and_mean_payoff}
    The max-plus system $Ax \leq Bx$ has a finite solution $x \in \R^n$ if and only if the row player does not lose the game on $\Gamma$ from any starting state. That is, if for every $s_0 \in V(\Gamma)$ we have that 
    $
        \nu(\Gamma, s_0) \geq 0.
    $
\end{theorem}

In this paper, we are interested in the consequence of the fact that the system $Ax \leq Bx$ has no finite solution.
The contrapositive of Theorem~\ref{thm:connection_of_tropical_and_mean_payoff} gives the following corollary (also see \cite[Theorem~9.25]{Joswig2021}).

\begin{corollary}\label{cor:no_solution_implies_negative_mean_payoff}
    If the max-plus system $Ax \leq Bx$ has no finite solution $x \in \R^n$, then the set of winning states for the column player is non-empty.
\end{corollary}

This can be seen as a tropical variant of Farkas' lemma, giving an existential witness for the non-existence of a solution to a system of inequalities.

Let $X \subseteq V(\Gamma)$ be the set of winning states for the column player and assume that $X$ is non-empty.
We will consider the restriction $\Gamma[X]$ of the game graph $\Gamma$ to $X$.
Naturally, the row player has no edge from $X$ to $V(\Gamma) \setminus X$, while the column player (if it has such edges) chooses not to use them.
This yields the following standard observation.

Let $S = X \cap \cC$ and $T = X \cap \cR$.
Let $\sigma_{|S}$ be the restriction of the column player's optimal positional strategy $\sigma: \cC \to \cR$ (in a game on $\Gamma$) to the domain $S$.
Clearly both $S$ and $T$ are non-empty, and $\sigma_{|S}(s) \in T$ for all $s \in S$.

\begin{observation}\label{obs:restricted_game_graph}
    The strategy $\sigma_{|S} : S \to T$ witnesses that all states of the game on $\Gamma[X]$ are winning for the column player.
\end{observation}

The following is a simple exercise from mathematical analysis.

\begin{observation}\label{obs:negative_mena_loss_implies_divergence}
    Let $W = (s_0, s_1, \dots)$ be a realization of a mean payoff game with $\nu_C(W) < 0$. Then,
\[
	\lim_{N \to \infty} \sum_{k=1}^N w(s_{k-1}, s_k) = -\infty.
\]
\end{observation}

\subsubsection{Min-plus algebra}\label{sec:min_plus_prelims}

Min-plus algebra is the semiring $(\R \cup \{\infty\}, \minpl, \minmu, \infty, 0)$, with $\minpl$ defined as $\min\{\cdot,\cdot\}$ and $\minmu$ as $+$.
Let $\Rmin = \R \cup \{\infty\}$ and $\Zmin = \Z \cup \{\infty\}$.
Note that the min-plus algebra is isomorphic to the max-plus algebra by the mapping $x \mapsto -x$.
Hence, for $A, B \in \Rmin^{m \times n}$ and $x \in \Rmin^n$, we have
\[
    A \minmu x \geq B \minmu x
\]
if and only if we have
\[
    A' \maxmu x' \leq B' \maxmu x'
\]
for $A', B' \in \Rmax^{m \times n}$, and $x' \in \Rmax^{n}$ defined as $A' = -A, B' = -B$, and $x' = -x$.

This allows us to construct the game graph $\Gamma_{\min}(A,B)$ corresponding to the min-plus inequality $A \minmu x \geq B \minmu x$ via the corresponding max-plus inequality $A' \maxmu x' \leq B' \maxmu x'$.
That is, let $\Gamma_{\max}(A',B')$ be the game graph for $A' \maxmu x' \leq B' \maxmu x'$ as above.
Then we define $\Gamma_{\min}(A,B)$ for the min-plus inequality $A \minmu x \geq B \minmu x$ to be the graph $\Gamma_{\max}(A',B')$ with reversed signs of all edge-weights.

The change of signs essentially switches the roles of the row and column player.
That is, while in the `max-plus game', the column player tries to minimize the total weight of the walk, in the `min-plus game', the column player (who is still defined as the one controlling the part $\cC$ of $\Gamma_{\min}$ indexed by the columns of $A$) tries to maximize the weight.
Hence, we define the mean payoff (instead of mean loss) of the column player from a realization $W = (s_0, s_1, s_2, \dots)$ as
\[
    \rho_C(W)  = \liminf_{N \to \infty} \frac{1}{N} \sum_{k=1}^N w(s_{k-1}, s_k)
    .
\]
Moreover, we define $\rho(\Gamma_{\min}, s_0)$ as $\rho_C(W)$, where $W$ is the realization starting in $s_0$ produced by optimal strategies of the row and column players, existence of which is guaranteed by Theorem~\ref{thm:positional_strategy}.
We say that a state $s_0$ is winning for the column player if $\rho(\Gamma_{\min}, s_0) > 0$.

With this, we may state the min-plus analog of Corollary~\ref{cor:no_solution_implies_negative_mean_payoff}.
The proof follows directly from definitions.

\begin{corollary}\label{cor:min_plus_no_solution_implies_positive_mean_payoff}
    If the min-plus system $Ax \geq Bx$ has no finite solution $x \in \R^n$, then the set of winning states for the column player is non-empty.
\end{corollary}

\subsubsection{Algorithms}\label{sssec:algorithms}

A central objective in the study of mean payoff games is the development of efficient solvers, a task essentially equivalent to solving tropical systems of the form $Ax \leq Bx$.
It is well established that this problem can be solved in pseudopolynomial time~\cite{ZWICK1996343}; however, the complexity of such algorithms is polynomial relative to the matrix dimensions and the actual numeric values of the weights.

Despite significant ongoing research~\cite{fijalkow_et_al, dorfman_et_al}, whether a truly polynomial-time algorithm exists remains a major open problem.
Such an algorithm would require a running time polynomial in the matrix size and the binary representation of the weights, i.e.\ the logarithm of their absolute values.

In Section~\ref{sec:compressed_matrices}, we present a strongly polynomial algorithm reducing the problem of deciding the existence of a finite solution to $Ax \leq Bx$ to the problem of deciding $\chi$-boundedness of a certain set-defined class.
Such an algorithm runs in polynomial time with respect to the matrix size in the arithmetic model (i.e. using constant time per arithmetic operation), using at most polynomial space with respect to the total size of the input.

\subsection{Graph out-covering projections}\label{sec:covering-projections}

Let $G$ and $H$ be two multidigraphs.
We say that $f: V(G) \cup E(G) \to V(H) \cup E(H)$, mapping vertices to vertices and edges to edges, is a graph out-covering projection if
\begin{enumerate}[label=(\roman*)]
    \item $f$ respects endpoints. That is, if $e$ is an edge from $u$ to $v$, then $f(e)$ is an edge from $f(u)$ to $f(v)$.
    \item $f$ is bijective on out-edges at each vertex. That is, if $E^+(v)$ denotes the set of outgoing edges from $v$, then for every $v \in V(G)$, $f$ is a bijection between $E^+(v)$ and $E^+(f(v))$.
\end{enumerate}
Moreover, we may require that $f$ preserves some more information such as unary or binary predicates.

Note that this is similar to the standard definition of graph covering projection~\cite{Kratochvil_97} (also known as a locally bijective homomorphism), except that we only require bijectivity for out-edges and not for in-edges of a vertex.

Clearly, an out-covering projection is a homomorphism from $G$ to $H$ as it maps edges to edges.
Moreover, the property of local bijection allows to assign unique preimages to edges from $H$ upon fixing their source vertex in $G$, which leads to the following correspondence between walks in $G$ and $H$ known as the \emph{unique walk lifting property}~\cite{FKKN14}.
Denote by $\cW_G(v)$ the set of all walks from a vertex $v \in V(G)$ in $G$.

\begin{observation}\label{obs:walks_via_covering_projection}
    Suppose $f: G \to H$ is an out-covering projection. 
    Then $f$ lifts to a bijection $F$ between $\cW_G(v)$ and $\cW_H(f(v))$ for any $v \in V(G)$.
\end{observation}

\section{Decomposition theorem for set-defined graph classes}\label{sec:decomposition-theorem}

In this section, we prove our main decomposition result for set-defined graph classes.

\MainDecomposition*

Before proving this theorem, we first derive \cref{cor:poly-chi-vs-unions-shift-subgraphs} from it.

\polyChiVsUnionsShiftSubgraphs*
\begin{proof}
Let $c \in \mathbb{N}$ and $\tau \colon \mathbb{N}\to\mathbb{N}$ be the constant and the polynomial function given by \cref{th:main-decomposition} for the class $\dcX$.
Suppose that digraphs in $\dcX$ which are unions of at most $c$ shift-colorable digraphs have bounded chromatic number, that is, there exists a constant $M \in \mathbb{N}$ such that every such digraph has chromatic number at most $M$. Then, by \cref{th:main-decomposition}, any $G \in \dcX$ admits a partition $V(G)=V_1\cup \dots \cup V_r$ with $r \le \tau(\omega(G))$ such that each $G[V_i]$ has chromatic number at most $M$. Hence,
$\chi(G)\le M\cdot \tau(\omega(G))$, and thus $\dcX$ is polynomially \chibounded. 
\end{proof}

\subsection{Outline of the proof}\label{sec:reduction-outline}

We begin by outlining the proof of \cref{th:main-decomposition}, which is carried out in \cref{sec:reduction-to-matching-clauses,sec:multiple-full-matching,sec:proof-decomposition-theorem}. We then describe refinements for full set-defined graph classes that will be needed for the \chiboundedness decision procedure.

Let $\dcX$ be a $d$-dimensional set-defined class, and let $f \colon \{0,1\}^{d^2} \to \{0,1\}$ be a Boolean function such that $\dcX \subseteq \dcX_f$. The proof of \cref{th:main-decomposition} proceeds through the following steps:

\begin{enumerate}
    \item In \cref{sec:reduction-to-matching-clauses}, we prove \cref{th:decomposition-uniform-Y}, which reduces the analysis to classes definable by injective DNFs over injective order-uniform sets. To do so, we first partition the vertex set according to the equality pattern of the coordinates, and show in \cref{th:compatible_to_full_matching} that each part induces a digraph belonging to a class $\dcX_H'$ for some injective DNF $H$ over at most $(d')^2$ variables, where $d' \leq d$. We then refine this partition further according to the order type of the vertices. Altogether, this yields at most $d^{2d}$ parts, each inducing a digraph in some class $\dcY_H$ for an injective DNF $H$, thereby reducing the analysis to such classes.

    \item In \cref{sec:multiple-full-matching}, we obtain a decomposition result for classes $\dcY_H$, where $H$ is an injective DNF.
    \begin{enumerate}
        \item In \cref{sec:reduction-to-ayclic}, we reduce the analysis to acyclic DNFs. To this end, we show that cyclic clauses of $H$ do not contribute any edges to digraphs in $\dcY_H$ (\cref{lemma:cyclicchrom}). Hence $H$ can be replaced, without changing the class $\dcY_H$, by the full DNF consisting only of the acyclic clauses of $H$ (\cref{th:reduction-to-acyclic}).
        Once we pass to an acyclic DNF, only two types of clauses remain: discrete clauses and path clauses.

        \item In \cref{sec:all-discrete-clauses}, we treat discrete DNFs, that is, DNFs consisting only of discrete clauses. We show that, for every such DNF $H$, each digraph $G$ in $\dcY_H$ admits a partition into independent sets whose number is polynomially bounded in terms of the clique number of~$G$. Moreover, we construct this partition so that the independent sets satisfy certain functional constraints. Although this additional property is not needed for the proof of \cref{th:main-decomposition}, it is crucial for the later analysis of \nf set-defined graph classes in \cref{sec:refinements-full-set-defined}.

        \item In \cref{sec:path-clauses}, we consider DNFs consisting of a single path clause $P$. We show in \cref{lem:path-clause-shift-colorable} that every digraph in $\dcY_P$ is shift-colorable.

        \item In \cref{sec:decomposition-injective}, we combine the previous two steps. We first partition the vertex set so that the discrete clauses contribute no edges within the parts, and then observe that, on each part, the remaining edges come only from path clauses. Since each path clause defines a shift-colorable digraph, we obtain the decomposition theorem for classes defined by injective DNFs (\cref{th:decomposition-injective}): every digraph in $\dcY_H$ admits a partition into polynomially many parts, in terms of the clique number, such that each part induces a subdigraph that is the union of boundedly many shift-colorable digraphs.
    \end{enumerate}

    \item In \cref{sec:proof-decomposition-theorem}, we combine \cref{th:decomposition-uniform-Y} and \cref{th:decomposition-injective} to obtain the desired decomposition for arbitrary $d$-dimensional set-defined classes, namely \cref{th:main-decomposition}.
\end{enumerate}

In \cref{sec:refinements-full-set-defined}, we derive a refinement of the previous argument that enables a reduction of the \chiboundedness of a \nf set-defined class to the solvability of systems of tropical inequalities, developed in \cref{sec:chi-dichotomy}. 
Starting from $\dcX_f$, we first reduce to finitely many subclasses defined by acyclic DNFs (\cref{lem:fromXtoY_H}). Next, for each acyclic DNF $H$, we reduce the analysis to finitely many DNFs consisting of a single path clause of $H$ together with all discrete clauses of $H$ (\cref{lem:path+discrete-clauses}). We then reduce further to the case in which the path clause is loopless (\cref{lem:loop-elimination}). Finally, we pass to classes of the form $\dcY_{P,Z}$, where $P$ is a loopless path clause and $Z$ encodes the relevant functional constraints arising from the discrete clauses (\cref{th:diaganal-to-path}). This yields \cref{th:from-full-set-defined-to-shift-colorable}, which serves as the input for the dichotomy proved in \cref{sec:chi-dichotomy}. The reduction from \cref{th:from-full-set-defined-to-shift-colorable} and the dichotomy from \cref{sec:chi-dichotomy} are then combined in \cref{sec:deciding-chiboundedness} into a decision procedure for \chiboundedness of \nf set-defined graph classes.

\subsection{Reduction to classes defined by injective DNFs}\label{sec:reduction-to-matching-clauses}

In this section, we show how to partition the vertex set of digraphs in a $d$-dimensional set-defined class into subsets whose number depends only on $d$, such that each subset induces a digraph belonging to a set-defined class specified by an injective DNF. We first introduce the partition via these new classes, and then show that they can be defined by injective DNFs.

Let $f : \{0,1\}^{d^2} \rightarrow \{0,1\}$ be a Boolean function, and let $\cP$ be a partition of $[d]$. We say that $v \in \bN^d$ is \emph{$\cP$-homogeneous} if $v_i = v_j$ $\iff$ $i$ and $j$ are in the same $\cP$-class. A set $V \subseteq \bN^d$ is \emph{$\cP$-homogeneous} if every $v \in V$ is $\cP$-homogeneous.
We denote by $\dcX_{f,\cP}$ the subclass of $\dcX_{f}$ consisting of digraphs that can be realised by $f$ over $\cP$-homogeneous sets, i.e.\
\[
    \dcX_{f,\cP} := \{ H \simeq G=(V,E) : V \subseteq \bN^d \text{ is } \cP\text{-homogeneous, and } G \text{ is a realisation of } f \text{ over } V \}.
\]
Note that, since any subset of a $\cP$-homogeneous set is also $\cP$-homogeneous, the class $\dcX_{f,\cP}$ is hereditary.

Let $G=(V,E)$ be an arbitrary digraph in $\dcX_f$, which is a realisation of $f$ over a set $V \subseteq \bN^{d}$.
For a partition $\cP$ of $[d]$, denote by $V_{\cP}$ the $\cP$-homogeneous subset of $V$. Since $\{V_{\cP} : \cP \text{ is a partition of } [d]\}$ is a partition of $V$, the vertex set of $G$ partitions into at most $d^d$ sets each of which induces a digraph from $\dcX_{f,\cP}$ for some partition $\cP$ of $[d]$.

Next, we show that a class $\dcX_{f,\cP}$ is equal to a class $\dcX_H'$ of digraphs that can be realised by an injective DNF $H$ over an injective set of vectors.
To this end, we establish two auxiliary results. 
First, we show that a full non-$\cP$-compatible clause cannot create an edge between $\cP$-homogeneous vertices (\cref{lem:non-P-matching-clause}), and therefore such clauses can be ignored in realisations of digraphs in $\dcX_{f,\cP}$. 
Second, we show that full $\cP$-compatible clauses can be replaced with injective clauses (\cref{lem:clause-to-prime}).
We then use these lemmas to derive the claimed equality between the classes (\cref{th:compatible_to_full_matching}).

\begin{lemma}\label{lem:non-P-matching-clause}
    Let $C$ be a full non-$\cP$-compatible clause,
    and $u,w \in \bN^d$ be $\cP$-homogeneous. 
    Then $C$ evaluates to 0 on $(u,w)$.
\end{lemma}
\begin{proof}
    Suppose $(u,w)$ satisfies $C$, and let $i \in [d]$ be a vertex of the clause digraph $\ClauseDigraph_C$ witnessing non-$\cP$-compatibility of $C$.
    If the out-neighborhood (resp. in-neighborhood) of $i$ in $\ClauseDigraph_C$ contains vertices from two different $\cP$-classes, then $u_i$ (resp. $w_i$) must be equal to two distinct values. 
    If the out-neighborhood (resp. in-neighborhood) of $i$ in $\ClauseDigraph_C$ contains only part of a $\cP$-class, then $u_i$ (resp. $w_i$) must be both equal and not equal to the same value.
    These contradictions prove the lemma.
\end{proof}

Let $\cP$ be a partition of $[d]$, and $I \subseteq [d]$ consist of the least element of every $\cP$-class. Let $r,t$ be not necessarily distinct elements in $I$, and $P_r, P_t$ be the $\cP$-classes containing $r$ and $t$, respectively.
For a $\cP$-compatible clause $C$ we have that $P_r \times P_t \subseteq E(\ClauseDigraph_C)$ if $(r,t) \in E(\ClauseDigraph_C)$, and $(P_r \times P_t) \cap E(\ClauseDigraph_C) = \emptyset$ if $(r,t) \not\in E(\ClauseDigraph_C)$. Furthermore, we have that every vertex in the clause digraph $\ClauseDigraph_{C_{|I}}$ of $C_{|I}$ has at most one in-neighbor and at most one out-neighbor, and thus $C_{|I}$ is injective.

\begin{lemma}\label{lem:clause-to-prime}
    Let $C$ be a full $\cP$-compatible clause over variables $q_{i,j}$, $i,j \in [d]$.
    If $u,w \in \bN^d$ are $\cP$-homogeneous, then $C(u,w)=C_{|I}(u_{|I},w_{|I})$.
\end{lemma}
\begin{proof}
    Consider arbitrary not necessarily distinct elements $r,t \in I$, and let $P_r$ and $P_t$ be the $\cP$-classes containing $r$ and $t$ respectively.
    Then, by construction, the clause digraph $\ClauseDigraph_{C_{|I}}$ contains (resp. does not contain) the edge $(r,t)$ if and only if $\ClauseDigraph_C$ contains all (resp. none) of the arcs in $P_r \times P_t$. 
    Thus, clause $C_{|I}$ contains the literal $q_{r,t}$ (respectively, $\overline{q_{r,t}}$) if and only if $C$ contains the literals $q_{i,j}$, $i \in P_r$, $j \in P_t$ (respectively, the literals $\overline{q_{i,j}}$, $i \in P_r$, $j \in P_t$).
    Therefore, since $u$ and $w$ are $\cP$-homogeneous, we have that $(u_{|I}, w_{|I})$ satisfies each literal in $C_{|I}$ if and only if $(u,w)$ satisfies each literal in $C$.
\end{proof}

\begin{lemma}\label{th:compatible_to_full_matching}
    Let $f : \{0,1\}^{d^2} \rightarrow \{0,1\}$, and let $\cP$ be a partition of $[d]$.
    Then there exists an injective DNF $H$ over $(d')^2$ variables, for some $d' \leq d$, such that  $\dcX_{f,\cP} = \dcX_{H}'$.
\end{lemma}
\begin{proof}
    Let $F$ be a full DNF representing $f$, and let $F'$ be the DNF consisting of the $\cP$-compatible clauses of $F$, i.e.\ $F$ can be written as $F' \vee \bigvee_{i=1}^{t} C_i$, where $C_1, C_2, \ldots, C_t$ are the clauses of $F$ that are not $\cP$-compatible. 
    From \cref{lem:non-P-matching-clause} we have that $F(u,w) = F'(u,w)$ holds for any $\cP$-homogeneous $u,w \in \bN^d$.

    Let $I \subseteq [d]$ consist of the least element of every $\cP$-class.
    Define $H$ to be the DNF consisting of the clauses $\{ C_{|I} : C \text{ is a clause of } F'\}$, and note that each of these clauses is injective.
    It follows from \cref{lem:clause-to-prime} that for any $\cP$-homogeneous  $u,w \in \bN^d$, $F'(u,w) = H(u_{|I},w_{|I})$. Consequently, $F(u,w) = H(u_{|I},w_{|I})$ holds for any $\cP$-homogeneous $u,w \in \bN^d$. 
    Since the $\cP$-homogeneous sets $V \subseteq \bN^d$ are in one-to-one correspondence with the injective sets $V_{|I} \subseteq \Ninj{I}$, we conclude that $\dcX_{f,\cP} = \dcX_{H}'$.
\end{proof}

Together with the discussion at the beginning of the section, \cref{th:compatible_to_full_matching} yields the desired partition. We further refine this partition to ensure that each part is order-uniform, i.e.~consists of vectors of the same order type.

\begin{theorem}\label{th:decomposition-uniform-Y}
    Let $f : \{0,1\}^{d^2} \rightarrow \{0,1\}$. There exist at most $d^{2d}$ injective DNFs $H_1, H_2, \ldots, H_r$, each over $(d')^2$ variables for some $d' \leq d$, so that every digraph $G \in \dcX_f$ admits a partition $V(G) = V_1 \cup \cdots \cup V_r$ such that $G[V_i] \in \dcY_{H_i}$, for every $i \in [r]$. Furthermore, each such $\dcY_{H_i} \subset \dcX_f$.
\end{theorem}
\begin{proof}
	Let $G=(V,E)$ be an arbitrary digraph in $\dcX_f$, which is a realisation of $f$ over a set $V \subseteq \bN^{d}$. Then, as discussed earlier, the vertex set of $G$ partitions into at most $d^d$ sets each of which induces a digraph from $\dcX_{f,\cP}$ for some partition $\cP$ of $[d]$. By \cref{th:compatible_to_full_matching}, each of these classes is equal to $\dcX_H'$ for some injective DNF $H$ over $(d')^2$ variables, for some $d' \leq d$. Now, for each digraph $G'=(V',E') \in \dcX_H'$, which is a realisation of $H$ over an injective set $V' \subseteq \Ninj{d'}$, vertex set $V'$ can be partitioned into at most $d'!$ order-uniform subsets. Thus, each such subset induces a digraph from $\dcY_H$.
	Combining the two partitions together and noting that $d'! \leq d! \leq d^d$ we obtain first part of the theorem. For the ``Furthermore'' part, for each $i$ there is some partition $\cP$, such that $\dcY_{H_i} \subset \dcX'_{H_i} = \dcX_{f, \cP} \subset \dcX_f$, where the equality is by Lemma \ref{th:compatible_to_full_matching}.
\end{proof}

\subsection{Injective DNFs}\label{sec:multiple-full-matching}

\cref{th:decomposition-uniform-Y} reduces our analysis to digraphs in  classes $\dcY_F$, where $F$ is an injective DNF.
In this section we establish a decomposition for such digraphs. We begin by reducing the problem to digraphs defined by acyclic DNFs.

\subsubsection{Reduction to acyclic DNFs}\label{sec:reduction-to-ayclic}

To reduce to acyclic DNFs, we show that none of the cyclic clauses can create an edge between vertices of the same order type.

\begin{lemma} \label{lemma:cyclicchrom}
	Let $C$ be a cyclic clause. Then any digraph in $\dcY_C$ is edgeless.
\end{lemma}
\begin{proof}
	Suppose first that the clause digraph $\ClauseDigraph_C$ is a directed cycle on $d \geq 2$ vertices. Without loss of generality assume that $E(\ClauseDigraph_C) = \{ (i,i+1) : i \in [d-1] \} \cup \{ (d,1) \}$. That is, the variables $q_{i,i+1}$, for $i \in [d-1]$, and $q_{d,1}$ appear positively, and all other variables $q_{i,j}$, $i,j \in [d]$ appear negatively in $C$.
	
	Let $G=(V,E) \in \dcY_C$ be the realisation of $C$ over an order-uniform injective set $V \subseteq \Nd$. We claim that no two vertices in $V$ can be connected by an edge. Indeed, let $u,w \in V$ and assume $(u,w) \in E$. Due to the structure of clause $C$, the latter holds if and only if $w$ is obtained from $u$ by a single cyclic shift, i.e.\ $w = (u_d, u_1, \ldots, u_{d-1})$.
	Since $u$ and $w$ have the same order type, we have $u_1 < u_2 \iff u_d < u_1$ and $u_i < u_{i+1} \iff u_{i-1} < u_i$ for every $2 \leq i \leq d-1$. These equivalences imply that either $u_1 < u_2 < \cdots < u_d < u_1$ or $u_1 > u_2 > \cdots > u_d > u_1$, both of which are contradictions.
	
	If $\ClauseDigraph_C$ has more than one connected component, and one of these components is a directed cycle on at least two vertices, then by \cref{lem:subclause} and the above, any $G \in \dcY_C$ is homomorphic to an edgeless graph, and therefore $G$ itself is edgeless.
\end{proof}

\begin{theorem}\label{th:reduction-to-acyclic}
	Let $H$ be an injective DNF and let $F$ be the DNF consisting of the acyclic clauses of $H$.
	Then $\dcY_H = \dcY_F$.
\end{theorem}
\begin{proof}
	Let $V \subseteq \Nd$ be an injective order-uniform set. By \cref{lemma:cyclicchrom}, none of the cyclic clauses in $H$ induces an edge between vertices in $V$. Hence the realisation of $H$ over $V$ coincides with the realisation of $F$ over $V$, which implies the theorem.
\end{proof}

\cref{th:reduction-to-acyclic} reduces the analysis to acyclic DNFs. Such DNFs consist of two types of clauses: discrete clauses and path clauses. 

\subsubsection{Discrete DNFs}\label{sec:all-discrete-clauses}

In this section, we consider digraphs realised by discrete DNFs, that is, DNFs in which all clauses are discrete.  

For a discrete clause $D$, we denote by $L(D) \subseteq [d]$ the set of its loop coordinates. Without loss of generality, we assume that $L(D) \neq [d]$ for all discrete clauses $D$. This is because no two distinct $v,u \in \Nd$ can satisfy such a clause, and thus it can be removed from the DNF without affecting anything. 

Let $D$ be a discrete clause, i.e.\ every connected component of the clause digraph $\ClauseDigraph_D$ is a single-vertex digraph, or, equivalently, only the variables $q_{i,i}$, $i \in [d]$ can appear positively in $D$. 
Denote $L := L(D) \subseteq [d]$ and $\overline{L} := [d] \setminus L$.
For a set $V \subseteq \Nd$ and $a \in \Ninj{L}$, we denote $V_{L,a} := \{ v \in V : v_{|L} = a \}$.
Let $G=(V,E)$ be the realisation of $D$ over a set $V \subseteq \Nd$. Let $a \in \Ninj{L}$ be fixed, and denote $G_{L,a} := G[V_{L,a}]$. 
In the next few lemmas we establish auxiliary facts about $G_{L,a}$ and $G$.

Recall that by $v(\overline{L})$ we denote the set of values of $v$ on $\overline{L}$, i.e.\ $v(\overline{L}) := \{ v_i : i \in \overline{L} \}$.

\begin{lemma}\label{lem:non-adjacent}
	Two vertices $v$ and $w$ in $G_{L,a}$ are non-adjacent if and only if $v(\overline{L}) \cap w(\overline{L}) \neq \emptyset$.
\end{lemma}
\begin{proof}
	By definition, for any $v,w \in V_{L,a}$, and any $i \in L$ and $j \in [d] \setminus \{ i \}$, we have that $v_i = w_i$, $v_i \neq w_j$, and $w_i \neq v_j$.
	These equalities and non-equalities satisfy the corresponding literals of $D$. All remaining literals are $\overline{q_{j_1,j_2}}, j_1,j_2 \in \overline{L}$.
	Therefore, neither of $(v,w)$ and $(w,v)$ satisfies $D$ if and only if there exist $j_1,j_2 \in \overline{L}$ such that $v_{j_1} = w_{j_2}$, i.e.\ $v(\overline{L}) \cap w(\overline{L}) \neq \emptyset$.
\end{proof}

\begin{lemma}\label{lem:G-L-a-partition}
	There exists a set $R \subseteq \bN$ of size at most $d \cdot \omega(G_{L,a})$ such that $V_{L,a}$ is the union of at most $d^2 \cdot \omega(G_{L,a})$ (possibly, empty) 
	sets $U_{i,r}, i \in \overline{L}, r \in R$, where
	$U_{i,r} := \{ v \in V_{L,a} : v_i = r \}$.
\end{lemma}
\begin{proof}
	Let $W$ be a maximal clique in $G_{L,a}$. By the maximality of $W$, every vertex $v$ in $V_{L,a} \setminus W$ is non-adjacent to some vertex $w \in W$, and thus, by \cref{lem:non-adjacent}, $v(\overline{L}) \cap w(\overline{L}) \neq \emptyset$. Therefore, denoting $R := \bigcup_{w \in W} w( \overline{L})$, we  conclude that $v(\overline{L}) \cap R \neq \emptyset$ holds for every $v \in V_{L,a}$, and hence $V_{L,a}$ is covered by the sets $U_{i,r} = \{ v \in V_{L,a} : v_i = r \}$, $i \in \overline{L}, r \in R$. Since $|R| \leq |\overline{L}| \cdot |W| \leq d \cdot \omega(G_{L,a})$, we conclude that the number of such sets is at most $|\overline{L}| \cdot |R| \leq d^2 \cdot \omega(G_{L,a})$.
\end{proof}

\begin{lemma}\label{lem:functional-covering}
	The set $V$ admits a covering by at most 
	$d^2 \cdot \omega(G)$ sets each of which is $L$-functional.
\end{lemma}
\begin{proof}
	Let $M = V_{|L}$.
	For $a \in M$ and the digraph $G_{L,a}$, denote by $R^a$ and $U_{i,r}^a$, $i \in \overline{L}, r \in R^a$, the sets guaranteed by \cref{lem:G-L-a-partition}.
	Since 
	\[
	V = \bigcup_{a \in M} V_{L,a} 
	\quad \text{ and } \quad V_{L,a} = \bigcup_{i \in \overline{L}} \bigcup_{r \in R^a} U_{i,r}^a,
	\]
	we have that $V$ is covered by the sets $U_{i,r}^a$, $a \in M, i \in \overline{L}, r \in R^a$. In the rest of the proof we show how to group these sets into the desired number of $L$-functional sets.
	
	Recall that $|R^a| \leq d\cdot \omega(G_{L,a}) \leq d \cdot \omega(G)$.
	For $j \in [d \cdot \omega(G)]$, denote by $r_j^a$ the $j$-th largest element in $R^a$, if it exists.
	Now, for every $i \in \overline{L}$, we define at most $d \cdot \omega(G)$ distinct $(L,i)$-functional sets. 
	The $j$-th such set is defined as 
	\[
	W_{i,j} := \bigcup_{a \in M} U_{i,r_j^a}^a,
	\]
	where we assume that $U_{i,r_j^a}^a$ is empty whenever $r_j^a$ does not exist.
	The $(L,i)$-functionality of $W_{i,j}$ is witnessed by a function $\varphi_j : \Ninj{L} \rightarrow \bN$ that maps $a$ to $r^a_j$, i.e.\ for every $v \in W_{i,j}$ we have $v_i = \varphi_j(v_{|L}) =
	r_j^{v_{|L}}$. 
	
	It is easy to see that $V = \bigcup_{i,j} W_{i,j}$ and the number of sets in the union is at most $|\overline{L}| \cdot d \cdot \omega(G) \leq d^2 \cdot \omega(G)$.
\end{proof}

\begin{lemma}\label{lem:is-functional}
	Let $G=(V,E) \in \dcY_D$ be the realisation of $D$ over a set $V \subseteq \Nd$, and let $U \subseteq V$ be an $L$-functional set.
	Then $U$ is an independent set in $G$.
\end{lemma}
\begin{proof}
	Let $\lambda \in \overline{L}$ be such that $U$ is $(L,\lambda)$-functional, and let $\varphi \colon \Ninj{L} \rightarrow \bN$ be a function witnessing the $(L,\lambda)$-functionality of $U$, i.e.\ $v_{\lambda} = \varphi(v_{|L})$ holds for every $v \in U$.
	We claim that any two vertices $v$ and $w$ in $U$ are not adjacent. Indeed, if $v_{|L} \neq w_{|L}$, then the vertices are not adjacent as they differ in at least one loop coordinate. Otherwise, if $v_{|L}  = w_{|L}$, then, due to the $(L,\lambda)$-functionality, we have $v_{\lambda} = \varphi(v_{|L}) = \varphi(w_{|L}) = w_{\lambda}$, in which case $v$ and $w$ are not adjacent either, since $\lambda$ is a non-loop coordinate.
\end{proof}

The next lemma extends \cref{lem:is-functional} to multiple discrete clauses.

\begin{lemma}\label{lem:multi-functional-independent-set}
	Let $F$ be a discrete DNF, and let $D_1, D_2, \ldots, D_k$ be the clauses of $F$ with $L_i := L(D_i)$, $i \in [k]$.
	Let $G=(V,E)$ be the realisation of $F$ over a set $V \subseteq \Nd$ and let $U \subseteq V$ be an $(L_1,L_2, \ldots, L_k)$-functional set.
	Then $U$ is an independent set in $G$.
\end{lemma}
\begin{proof}
	For each $i \in [k]$, let $H_i=(V,E_i)$ be the realisation of $D_i$ over $V$. Then $G = \bigvee_{i=1}^k H_i$.
	By definition, the $(L_1,L_2, \ldots, L_k)$-functional set $U$ is $L_i$-functional for every $i \in [k]$. Then, by \cref{lem:is-functional}, $U$ is independent in $H_i$ for every $i \in [k]$, and hence it is independent in $G$.
\end{proof}

We are now ready to prove the main result of this section.

\begin{theorem}\label{th:diagonal-multi}
	Let $F$ be a discrete DNF, $D_1, D_2, \ldots, D_k$ be the clauses of $F$, and let $L_i := L(D_i)$, $i \in [k]$.
	Then, for every realisation $G=(V,E)$ of $F$ over a set $V \subseteq \Nd$, its vertex set can be partitioned into at most $(d^2 \cdot \omega(G))^k$ independent sets each of which is $(L_1,L_2, \ldots, L_k)$-functional.
\end{theorem}
\begin{proof}
	For each $i \in [k]$, let $H_i=(V,E_i)$ be the realisation of $D_i$ over $V$.
	Then $G = \bigvee_{i=1}^k H_i$. 
	Denote by $\cP_i$ the family of at most $d^2 \cdot \omega(H_i)$ $L_i$-functional sets covering $V$, which is guaranteed to exist by \cref{lem:functional-covering}.
	Then every vertex in $V$ belongs to at least one set in $\cP_i$ for every $i \in [k]$. This implies that $V$ can be partitioned into at most $(d^2 \cdot \omega(G))^k$ $(L_1,L_2, \ldots, L_k)$-functional sets, and by \cref{lem:multi-functional-independent-set} each of these sets is independent in $G$.
\end{proof}

\subsubsection{Path clauses}\label{sec:path-clauses}

\begin{lemma} \label{lem:path-clause-shift-colorable}
	Let $P$ be a path clause. Then any digraph $G \in \dcY_P$ is shift-colorable. In particular, $\omega(G) \leq 2$.
\end{lemma}
\begin{proof}
	Suppose first that the clause digraph $\ClauseDigraph_P$ is a directed path on $d \geq 2$ vertices. Without loss of generality, assume that $E(\ClauseDigraph_P)=\{(i+1,i) : i\in[d-1]\}$.
	Thus the variables $q_{i+1,i}$, for $i\in[d-1]$, appear positively in $P$, whereas all other variables $q_{i,j}$, with $i,j\in[d]$, appear negatively.
	
	Let $G=(V,E) \in \dcY_P$ be a realisation of $P$ over an order-uniform injective set $V\subseteq \Nd$. We will show that $G$ is a shift digraph. By the definition of $P$, for any $u,w\in V$ we have $(u,w)\in E \iff u_{i+1}=w_{i}\text{ for every }i \in [d-1]$.
	
	We claim that if $G$ contains an edge, then all vectors in $V$ are either increasing or decreasing. Indeed, suppose that $(u,w)\in E$. Since $u$ and $w$ have the same order type, for every $i\in[d-2]$ we have
	$u_{i}<u_{i+1} \Longleftrightarrow w_{i}<w_{i+1}
	\Longleftrightarrow u_{i+1}<u_{i+2}$.
	Hence all adjacent comparisons in $u$ are the same, and thus either $u$ is increasing or $u$ is decreasing.
	Since all vectors in $V$ have the same order type, they are therefore all increasing or all decreasing. 
	
	From the above, if the common order type of vectors in $V$ is neither increasing nor decreasing, then $G$ has no edges, and hence it is trivially a shift digraph.
	If all vectors in $V$ are increasing, then $G$ is a shift digraph by definition.
	If all vectors in $V$ are decreasing, then $G$ is isomorphic to a shift digraph,  which can be seen by mapping each  $(v_1, v_2, \ldots, v_d) \in V$ to $((n+1)-v_1, (n+1)-v_2, \ldots, (n+1)-v_d)$, where $n$ is the largest coordinate value of a vector in $V$.
	
	Suppose now that $\ClauseDigraph_P$ has more than one connected component. Since $P$ is a path clause, one of these components is a directed path on at least two vertices. 
	By \cref{lem:subclause} and the above, any digraph $G \in \dcY_P$ is homomorphic to a shift digraph, and therefore $G$ is shift-colorable.
\end{proof}

\subsubsection{Decomposition theorem for injective DNFs}\label{sec:decomposition-injective}

\begin{theorem}\label{th:decomposition-injective}
	Let $F$ be an injective DNF on $d^2$ variables. Then, for every  $G \in \dcY_F$, the vertex set of $G$ admits a partition into at most $(d^2 \cdot \omega(G))^{2^{d^2}}$ parts, each inducing a subdigraph that is a union of at most $2^{d^2}$ shift-colorable digraphs.
\end{theorem}
\begin{proof}
	Without loss of generality, by \cref{th:reduction-to-acyclic}, we assume that $F$ is acyclic.
	Let $P_1, P_2, \ldots, P_{a}$ be the path clauses and $D_{1}, D_{2}, \ldots, D_{b}$ be the discrete clauses of $F$, i.e.\
	\[
	F := \bigvee_{i \in [a]} P_i \vee \bigvee_{j \in [b]} D_j.
	\]
	Let $G 
    \in \dcY_F$ be the realisation of $F$ over an order-uniform injective set $V \subseteq \Nd$.
	Denote by $G_D$ the realisation of $\bigvee_{j \in [b]} D_j$ over $V$, and for every $i\in[a]$ denote by $G_{P_i}$ the realisation of $P_i$ over $V$.
	
	By \cref{th:diagonal-multi}, $V$ can be partitioned into $r \leq (d^2 \cdot \omega(G_D))^b \leq (d^2 \cdot \omega(G))^b$ sets $V_1, V_2, \ldots, V_r$ each of which is an independent set in $G_D$. In particular, none of the discrete clauses contribute edges in $G$ within these sets, and thus, for each $i \in [r]$, the digraph $G[V_i]$ is the union of digraphs $G_{P_1}[V_i], G_{P_2}[V_i], \ldots, G_{P_a}[V_i]$. 
	Since any subset of an injective order-uniform set is injective and order-uniform, by \cref{lem:path-clause-shift-colorable} each of these digraphs is shift-colorable. Therefore, for every $i \in [r]$, $G[V_i]$ is a union of at most $a$ shift-colorable digraphs. 
	Finally, as the total number of clauses in $F$ is at most $2^{d^2}$, we have $a, b \leq 2^{d^2}$, and the result follows.
\end{proof}

\subsection{Proof of the decomposition theorem}\label{sec:proof-decomposition-theorem}

We now prove our decomposition theorem that we restate below for convenience.

\MainDecomposition*
\begin{proof}
	Let $f : \{0,1\}^{d^2} \rightarrow \{0,1\}$ be such that $\dcX \subseteq \dcX_{f}$. Let $G \in \dcX$ be the realisation of $f$ over an arbitrary set $V \subseteq \N^d$.
	By \cref{th:decomposition-uniform-Y},
	$V$ can be partitioned into $r \leq d^{2d}$ sets $V_1, V_2, \ldots, V_r$ such that, for every $i \in [r]$, the digraph $G[V_i]$ is in $\dcY_{H_i}$ for some injective DNFs $H_1, H_2, \ldots, H_r$ each over at most $d^2$ variables. By \cref{th:decomposition-injective}, the vertex set of each such digraph admits a partition into at most 
	$(d^2 \cdot \omega(G[V_i]))^{2^{d^2}} \leq (d^2 \cdot \omega(G))^{2^{d^2}}$ parts, each inducing a subdigraph that is a union of at most $2^{d^2}$ shift-colorable digraphs. Combining the two partitions, we obtain the theorem with $\tau(\omega) := d^{2d} \cdot (d^2 \cdot \omega)^{2^{d^2}}$ and $c := 2^{d^2}$.
\end{proof}

\subsection{Refinements for full set-defined classes}\label{sec:refinements-full-set-defined}

Recall, a path clause is \emph{loopless} if it contains no loop coordinates.
Let $P$ be a loopless path clause over $d^2$ variables.  For a tuple of pairs $Z = ((L_1,\lambda_1), (L_2,\lambda_2), \ldots, (L_k,\lambda_k))$, where $L_i \neq [d]$ and $\lambda_i \in \overline{L_i} := [d] \setminus L_i$, $i \in [k]$, we denote by $\dcY_{P,Z}$ the class of all set-defined digraphs that can be realised by the clause $P$ over order-uniform $Z$-functional sets in $\Nd$.
More formally,
\[
\dcY_{P,Z} := \{ H \simeq G=(V,E) : V \subseteq \Nd \text{ is order-uniform } Z\text{-functional, and } G \text{ is a realisation of } P \text{ over } V \}.
\]
Observe that any subset of an order-uniform $Z$-functional set is also order-uniform and $Z$-functional, which implies that $\dcY_{P,Z}$ is a hereditary class.

Our goal in this section is to reduce \chiboundedness of a \nf $d$-dimensional set-defined graph class 
to boundedness of chromatic number in a family of classes of the form $\dcY_{P,Z}$. Furthermore, we will show that the description of the latter classes in the form of pairs $(P,Z)$ can be computed from $f$ in time that depends only on $d$. Together with the decision procedure for boundedness of chromatic number in classes  $\dcY_{P,Z}$ that we establish in \cref{sec:chi-dichotomy}, this will give us a procedure for deciding \chiboundedness of \nf set-defined graph classes (see \cref{sec:deciding-chiboundedness}).

\begin{restatable}{theorem}{fromXtoYPZ}\label{th:from-full-set-defined-to-shift-colorable}
	Let $f : \{0,1\}^{d^2} \rightarrow \{0,1\}$. There exist $r \leq d^{2d} \cdot 2^{d^2} \cdot d^{2^{d^2}}$ classes $\dcY_{P_1,Z_1}, \dcY_{P_2,Z_2}, \ldots, \dcY_{P_r,Z_r}$ such that $\dcX_f$ is \chibounded if and only if each of the classes $\dcY_{P_i,Z_i}, i \in [r]$, has bounded chromatic number. Furthermore, the pairs $(P_i,Z_i), i \in [r]$ can be computed from $f$.
\end{restatable}

We prove \cref{th:from-full-set-defined-to-shift-colorable} via a sequence of lemmas that reduce the \chiboundedness of $\dcX_f$ to that of the classes $\dcY_{P_i,Z_i}$ via intermediate classes.
First, we reduce to classes defined by acyclic DNFs.

\begin{lemma}\label{lem:fromXtoY_H}
	Let $f : \{0,1\}^{d^2} \rightarrow \{0,1\}$. There exist $r \leq d^{2d}$ acyclic DNFs $F_1, F_2, \ldots, F_r$ each over at most $d^2$ variables such that 
	\begin{enumerate}
		\item[(1)] $\dcY_{F_i} \subseteq \dcX_f$ for every $i \in [r]$;\label{lem:fromXtoY_H:item1}
		
		\item[(2)] $\dcX_f$ is \chibounded if and only if for every $i \in [r]$ the class $\dcY_{F_i}$ is \chibounded;
		
		\item[(3)] the DNFs $F_1, F_2, \ldots, F_r$ can be computed from $f$.
	\end{enumerate}
\end{lemma}
\begin{proof}
	Let $H_1, H_2, \ldots, H_r$, $r \leq d^{2d}$, be the injective DNFs that are given by \cref{th:decomposition-uniform-Y} for $f$. For $i \in [r]$, let $F_i$ be the full DNF consisting of the acyclic clauses of $H_i$.
	
	From the proof of \cref{th:decomposition-uniform-Y}, each $H_i$ corresponds to some partition $\cP_i$ of $[d]$ so that $\dcX_{f,\cP_i} = \dcX_{H_i}'$. Since $\dcX_{f,\cP_i} \subseteq \dcX_f$, $\dcY_{H_i} \subseteq \dcX_{H_i}'$, and, by \cref{th:reduction-to-acyclic}, $\dcY_{H_i} = \dcY_{F_i}$, we conclude that $\dcY_{F_i} \subseteq \dcX_f$ holds for every $i \in [r]$. This proves item (1) of the statement.
	
	Now, since $\dcY_{F_i} \subseteq \dcX_f$, for every $i \in [r]$, if any of the classes $\dcY_{F_i}$ is \chiunbounded, then the class $\dcX_f$ is \chiunbounded. This establishes one direction of item (2). To prove the other direction, suppose that each $\dcY_{F_i}$ is \chibounded and $h_i$ is the corresponding $\chi$-binding function. Let $G=(V,E)$ be an arbitrary digraph in $\dcX_f$, which is a realisation of $f$ over a set $V \subseteq \bN^{d}$.
	By \cref{th:decomposition-uniform-Y}, $G$ admits a partition $V(G) = V_1 \cup \cdots \cup V_r$ with $G[V_i] \in \dcY_{H_i} = \dcY_{F_i}$ for every $i \in [r]$. Thus, we have 
	\[
	\chi(G) \leq \sum_{i=1}^{r} \chi(G[V_{i}]) \leq \sum_{i=1}^{r} h_{i}(\omega(G[V_{i}])) \leq 
	\sum_{i=1}^{r} h_{i}(\omega(G)),
	\]
	which completes the proof of item (2).
	
	We now justify item~(3). First, we construct a full DNF $F$ for $f$ from its truth table. Next, following the proof of \cref{th:compatible_to_full_matching}, we construct the injective DNFs $H_i$ from $F$ and the corresponding partitions $\cP_i$ of $[d]$. Finally, from each injective DNF $H_i$, we obtain an acyclic DNF $F_i$ by removing all cyclic clauses.
\end{proof}

Next we deal with classes defined by acyclic DNFs, and reduce them to classes defined by acyclic DNFs with exactly one path clause. We note that unlike the other reductions of this subsection which pass to a subclass of the initial class, this reduction passes to a subclass of the monotone closure of the initial class. This is not an issue for the decision procedure, but we will revisit this point in Lemma \ref{lem:full_contains_SD}, using the analysis from Section \ref{sec:chi-dichotomy}.

\begin{lemma}\label{lem:path+discrete-clauses}
	Let $F := \bigvee_{i \in [a]} P_i \vee \bigvee_{j \in [b]} D_j$ be an acyclic DNF over $d^2$ variables, where $P_1, P_2, \ldots, P_{a}$ are the path clauses and $D_{1}, D_{2}, \ldots, D_{b}$ are the discrete clauses of $F$. For each $i \in [a]$, let $F_i := P_i \vee \bigvee_{j \in [b]} D_j$.
	Then the class $\dcY_F$ is \chibounded if and only if $\dcY_{F_i}$ is \chibounded for every $i \in [a]$. 
\end{lemma}
\begin{proof}
	Since $F = \bigvee_{i=1}^a F_i$, by \cref{lem:union}, we have that $\dcY_F \subseteq \bigvee_{i=1}^a \dcY_{F_i}$.
	Thus, if, for every $i \in [a]$, the class $\dcY_{F_i}$ is \chibounded with a $\chi$-binding function $h_i$, then, by \cref{lem:union-chi-bounded}, $\dcY_F$ is \chibounded with a $\chi$-binding function $g = \prod_{i=1}^a h_i$.
	
	Now, assume that $\dcY_{F}$ is \chibounded with a $\chi$-binding function $h : \bN \rightarrow \bN$, and let $t \in [a]$. We will show that $\dcY_{F_t}$ is \chibounded. Denote $I = [a] \setminus \{t\}$ and note that $F = F_t \vee \bigvee_{i \in I} P_i$.
	
	Let $H_t$ be an arbitrary digraph in $\dcY_{F_t}$, which is a realisation of $F_t$ over an order-uniform set $V \subseteq \Nd$.
	For every $i \in I$, let $H_i$ be the realisation of $P_i$  over $V$, and let $G$ be the realisation of $F$ over $V$. Then $G = H_t \vee \bigvee_{i \in I} H_i \in \dcY_F$.
	
	For each $i \in I$, since the digraph $H_i$ belongs to the class $\dcY_{P_i}$ defined by a path clause, by \cref{lem:path-clause-shift-colorable}, we have $\omega(H_i) \leq 2$.
	Thus, by \cref{lem:omega-of-union}, $\omega(\bigvee_{i \in I} H_i) \leq c$, where $c = R_{a-1}(3)$.
	Furthermore, since $G = H_t \vee \bigvee_{i \in I} H_i$, again by \cref{lem:omega-of-union}, $\omega(G) \leq R(\omega(H_t)+1, c+1) \leq ( \omega(H_t) + c)^{c}$.
	Consequently, we have 
	\[
	\chi(H_t) \leq \chi(G) \leq h(\omega(G)) \leq h((\omega(H_t) + c)^{c}).
	\] 
	Since $H_t$ is an arbitrary digraph from $\dcY_{F_t}$, we conclude that $\dcY_{F_t}$ is \chibounded with a $\chi$-binding function $g(\omega) = h((\omega + c)^{c})$. 
\end{proof}

We now consider classes defined by acyclic DNFs with exactly one path clause, and reduce them to classes defined by DNFs with exactly one \emph{loopless} path clause.
Let $F = P \vee \bigvee_{i=1}^k D_i$ be a DNF over $d^2$ variables, where $P$ is a path clause and $D_1, D_2, \ldots, D_k$ are discrete clauses. Let $L \subset [d]$ be the set of loop coordinates of $P$, and let $S = [d] \setminus L$.
Let $D_1, D_2, \ldots, D_{t}$ be the discrete clauses of $F$ whose every coordinate in $L$ is a loop coordinate, and $D_{t+1}, D_{t+2}, \ldots, D_k$ be all the other discrete clauses of $F$. Denote 
\[
F_1 = P \vee \bigvee_{i=1}^{t} D_i
\quad \text{and} \quad
F_2 = \bigvee_{i=t+1}^{k} D_i,
\]
and note that if $t = 0$, then $F_1 = P$.
Let $P' = P_{|S}$ and $D_i' = {D_i}_{|S}$ for $i \in [t]$, i.e.\ $P', D_1', D_2', \ldots, D_{t}'$ are obtained from $P, D_1, D_2, \ldots, D_{t}$, respectively, by removing literals corresponding to variables with an index in $L$.
Denote $F_1' = P' \vee \bigvee_{i=1}^{t} D_i'$. 

\begin{lemma}\label{lem:loop-elimination}
	With the notation above, $\dcY_{F_1'} \subseteq \dcY_F$, and $\dcY_{F_1'}$ is \chibounded if and only if $\dcY_F$ is.
\end{lemma}
\begin{proof}
	We start by proving $\dcY_{F_1'} \subseteq \dcY_F$.
	Let $G'=(V',E')$ be an arbitrary digraph in $\dcY_{F_1'}$, which is a realisation of $F_1'$ over an order-uniform set $V' \subseteq \Ninj{S}$. 
	To show that $G'$ belongs to $\dcY_F$, fix an element $a \in \Ninj{L}$ such that for every $v' \in V'$ it holds that $\max v'(S) < \min a(L)$.
	Define $V := \{ v \in \Nd : v_{|S} = v' \in V' \text{ and } v_{|L} = a \}$.
    Clearly, the set $V$ is order-uniform by the choice of $a$.
	Moreover, it is easy to see that any ordered pair $(v',u') \in V' \times V'$ satisfies $F_1'$ if and only if the ordered pair $(v,u)$ satisfies $F$. Thus, the realisation of $F$ over $V$ is isomorphic to $G'$, and hence $G'$ belongs to $\dcY_F$.
	
	Since $\dcY_{F_1'} \subseteq \dcY_F$, if $\dcY_F$ is \chibounded, then so is $\dcY_{F_1'}$.
	We now prove that the \chiboundedness of $\dcY_{F_1'}$ implies that of $\dcY_F$.
	Suppose $\dcY_{F_1'}$ is \chibounded with a $\chi$-binding function $g$. We begin by showing that $\dcY_{F_1}$ is \chibounded with the same $\chi$-binding function $g$. For this we argue that every connected digraph in $\dcY_{F_1}$ belongs to $\dcY_{F_1'}$.
	Let $G=(V,E)$ be a connected digraph in $\dcY_{F_1}$, which is a realisation of $F_1$ over an order-uniform set $V \subseteq \Nd$. Since in all clauses of $F_1$ all coordinates in $L$ are loop coordinates, if $u,v \in \Nd$ are such that $u_{|L} \neq v_{|L}$, then neither $(u,v)$ nor $(v,u)$ satisfy $F_1$.
	This, together with the connectedness of $G$, implies that all 
	vertices $v \in V$ have the same restriction $v_{|L}$.
	Thus, since elements in $V$ are injective, we have that 
	$v(L) \cap u(S) = \emptyset$ and $u(L) \cap v(S) = \emptyset$ hold for any $u,v \in V$.
	Therefore, an ordered pair $(v,u)$ of vertices in $V$ satisfies $F_1$ if and only if the ordered pair $(v_{|S}, u_{|S})$ satisfies $F_1'$.
	This implies that the realisation of $F_1'$ over the order-uniform set $V' := \{ v_{|S} : v \in V \}$ is isomorphic to $G$, and thus $G$ belongs to $\dcY_{F_1'}$. Consequently, for every connected digraph $G$ in $\dcY_{F_1}$ we have that $\chi(G) \leq g(\omega(G))$. Since the chromatic number of a digraph is the maximum of the chromatic numbers of its components, and the same holds for the clique number, it is easy to see that the inequality $\chi(H) \leq g(\omega(H))$ holds for all digraphs $H \in \dcY_{F_1}$.
	
	To complete the proof, we observe that $F = F_1 \vee F_2$ and thus, by \cref{lem:union}, 
	$\dcY_F \subseteq \dcY_{F_1} \vee \dcY_{F_2}$. Therefore, since $\dcY_{F_2}$ is \chibounded (which follows from \cref{th:diagonal-multi}), we have that, by \cref{lem:union-chi-bounded}, $\dcY_F$ is \chibounded too.
\end{proof}

\begin{lemma}\label{th:diaganal-to-path}
	Let $F = P \vee \bigvee_{i=1}^k D_i$ be a DNF, where 
	$P$ is a loopless path clause and $D_1, D_2, \ldots, D_k$ are discrete clauses with $L_i := L(D_i) \neq [d]$, $i \in [k]$. Then, 
	
	\begin{enumerate}
		\item[(1)] $\dcY_{P,Z} \subseteq \dcY_F$ for every $Z = ((L_1,\lambda_1), (L_2,\lambda_2), \ldots, (L_k,\lambda_k))$ with $\lambda_i \in \overline{L_i}$, $i \in [k]$; and
		
		\item[(2)] $\dcY_F$ is \chibounded if and only if $\dcY_{P,Z}$ has bounded chromatic number for every $Z = ((L_1,\lambda_1), (L_2,\lambda_2), \ldots,$ $ (L_k,\lambda_k))$ with $\lambda_i \in \overline{L_i}$, $i \in [k]$.
	\end{enumerate}
\end{lemma}
\begin{proof}
    Denote $D = \bigvee_{i=1}^k D_i$.
	We start by proving item (1).
	Let $H_P \in \dcY_{P,Z}$ be the realisation of $P$ over an order-uniform $(L_1,L_2, \ldots, L_k)$-functional set $V \subseteq \Nd$, and $H_D$ be the realisation of $D$ over $V$. Then, by \cref{lem:multi-functional-independent-set}, $H_D$ is an edgeless digraph, and  therefore $H_P = H_P \vee H_D \in \dcY_F$, as desired. 
	
	Now, item (1) implies one direction in item (2). Indeed, if $\dcY_F$ is \chibounded with a $\chi$-binding function $h$, then any hereditary subclass of $\dcY_F$, including the classes $\dcY_{P,Z}$, is \chibounded with the same $\chi$-binding function. Since $\dcY_{P,Z} \subseteq \dcY_{P}$, by \cref{lem:path-clause-shift-colorable}, every digraph in $\dcY_{P,Z}$ has clique number at most $2$, and therefore, any digraph in $\dcY_{P,Z}$  has chromatic number at most $h(2)$.
	
	To prove the other direction in item (2), suppose  that graphs in all classes $\dcY_{P,Z}$ have  chromatic number at most $c$ for some $c \in \bN$.
	We will show that $\dcY_F$ is \chibounded.
	
	Let $G$ be an arbitrary digraph in $\dcY_F$, which is a realisation of $F$ over an order-uniform set $V \subseteq \Nd$. Denote by $H_P$ and $H_{D}$ the realisations of $P$ and $D$ over $V$, respectively. Then $H_P \in \dcY_P$, $H_D \in \dcY_D$, and, by  \cref{lem:union}, $G = H_P \vee H_D$.
	
	By \cref{th:diagonal-multi}, the set $V$ partitions into at most $(d^2 \cdot \omega(H_D))^k \leq (d^2 \cdot \omega(G))^k$ subsets each of which is $(L_1,L_2, \ldots, L_k)$-functional and independent in $H_D$. 
	Let $U$ be such a set and let $\lambda_s \in \overline{L_s}, s \in [k]$ be such that $U$ is $((L_1,\lambda_1), (L_2,\lambda_2), \ldots, (L_k,\lambda_k))$-functional. 
	By definition, $H_P[U]$ belongs to the class $\dcY_{P,Z}$ for $Z = ((L_1,\lambda_1), (L_2,\lambda_2), \ldots, (L_k,\lambda_k))$, and therefore $\chi(H_P[U]) \leq c$. Furthermore, since $U$ is an independent set in $H_D$, we have $G[U] = H_P[U] \vee H_D[U] = H_P[U]$, and hence $\chi(G[U]) = \chi(H_P[U]) \leq c$. Thus
	\[
	\chi(G) \leq \sum_{U} \chi(G[U]) \leq c \cdot (d^2 \cdot \omega(G))^k,
	\] 
	where the sum is over the $(L_1,L_2, \ldots, L_k)$-functional sets $U$ that partition $V$.  
\end{proof}

We now have everything ready to prove the main theorem of this section, which we restate below for convenience.

\fromXtoYPZ*
\begin{proof}
	By \cref{lem:fromXtoY_H}, there is a family $\cF$ of at most $d^{2d}$ acyclic DNFs each over at most $d^2$ variables such that $\dcX_f$ is \chibounded if and only if all classes $\dcY_{F}$, $F \in \cF$ are \chibounded.
	
	Fix a DNF $F \in \cF$. If $F$ has no path clauses, then it can be ignored. Indeed, in such a case $\dcY_{F}$ is  \chibounded by \cref{th:diagonal-multi}, and thus $\dcY_{F}$ does not affect \chiboundedness of $\dcX_f$. 
	So suppose that $F$ contains at least one path clause, and write
	\[
		F=\bigvee_{i\in[a]} Q_{i}\ \vee\ \bigvee_{j\in[b]} D_{j},
	\]
	where $Q_{1},\ldots,Q_{a}$ are the path clauses of $F$ and $D_{1},\ldots,D_{b}$ are its discrete clauses.
	For every $i\in[a]$, let
	\[
		F_{i}:=Q_{i}\ \vee\ \bigvee_{j\in[b]} D_{j}.
	\]
	By \cref{lem:path+discrete-clauses}, the class $\dcY_{F}$ is \chibounded if and only if each of the classes $\dcY_{F_{i}}$, $i\in[a]$, is \chibounded.
	
	Now fix $i\in[a]$. 
	By \cref{lem:loop-elimination} there exists a DNF
	\[
		F_i'=Q_{i}'\ \vee\ \bigvee_{\ell\in[k_{i}]} D_{i,\ell},
	\]
	over $d'^2$ variables for some $d' \leq d$,
	where $Q_{i}'$ is a \emph{loopless} path clause and
	$D_{i,1},\ldots,D_{i,k_i}$ are discrete clauses with
	$L(D_{i,\ell})\neq [d']$ for $\ell \in [k_i]$, such that $\dcY_{F_i}$ is \chibounded if and only if
	$\dcY_{F_i'}$ is \chibounded.
	
	By \cref{th:diaganal-to-path}, $\dcY_{F_i'}$ is \chibounded if and only if
	every class of the form $\dcY_{Q_i',Z}$ has bounded chromatic number, where 
	$Z=\bigl((L(D_{i,1}),\lambda_1),\ldots,(L(D_{i,k_{i}}),\lambda_{k_{i}})\bigr)$
	with $\lambda_\ell\in [d']\setminus L(D_{i,\ell})$ for $\ell\in[k_{i}]$.
	
	Consider all classes $\dcY_{Q_i',Z}$ arising in this way over all DNFs $F \in \cF$ and path clauses of $F$ and denote them by $\dcY_{P_1,Z_1},\dcY_{P_2,Z_2},\ldots,\dcY_{P_r,Z_r}$.
	Chaining together the equivalences above, we obtain $\dcX_f$ is \chibounded if and only if each class $\dcY_{P_i,Z_i}$, $i \in [r]$, has bounded chromatic number.
	Since each DNF over at most $d^2$ variables has at most $2^{d^2}$ clauses, it follows that $r \leq d^{2d} \cdot 2^{d^2} \cdot d^{2^{d^2}}$.
	
	Finally, we note that all pairs $(P_i,Z_i)$ can be computed from $f$. First, the acyclic DNFs in $\cF$ can be computed from $f$ by \cref{lem:fromXtoY_H}. Then, for each $F \in \cF$, one can enumerate its path clauses, apply the construction from \cref{lem:loop-elimination}, and finally enumerate all admissible tuples $Z$.
\end{proof}

\section{\texorpdfstring{$\chi$}{Chi}-boundedness dichotomy via tropical algebra and mean payoff games}\label{sec:chi-dichotomy}

In this section, we establish \cref{thm:tropical_dichotomy} which gives a correspondence between \chiboundedness of a class of digraphs defined by a path clause over functional vertex sets, and solutions to certain systems of tropical inequalities.

\begin{theorem}\label{thm:tropical_dichotomy}
    Let $\dcY_{P,Z}$ be a class of digraphs defined by a loopless $d$-dimensional path clause $P$ on $Z$-functional vertex sets in $\Nd$.
    Let $m$ be the number of maximal paths in the clause digraph $\Delta_P$ of $P$ and $k$ the number of functional constraints in $Z$.
    Then, there exist matrices $\widehat{A}, \widehat{B} \in \Zmin^{(2d-2m+k) \times d}$ and $\widetilde{A}, \widetilde{B} \in \Zmax^{(2d-2m+k) \times d}$ such that the following are equivalent:
    \begin{enumerate}[label=(\roman*)]
        \item the chromatic number of $\dcY_{P,Z}$ is unbounded;
        \label{item:unbounded_chromatic_number}
        \item $\dcY_{P,Z}$ contains the class $\dcS_D$ of $D$-dimensional shift digraphs for some $D \geq 2$;
        \label{item:contains_shift_graphs}
        \item each of the tropical inequalities $\widehat{A} \minmu x \geq \widehat{B} \minmu x$ and $\widetilde{A} \maxmu y \leq \widetilde{B} \maxmu y$ has a finite solution.
        \label{item:tropicals_have_solution}
    \end{enumerate}
\end{theorem}

We complement the qualitative dichotomy from \cref{thm:tropical_dichotomy} with quantitative bounds.
Namely, we show that 
if the class $\dcY_{P,Z}$ has unbounded chromatic number, then $\dcS_D \subseteq \dcY_{P,Z}$ for $D \leq 2d-2$ (Corollary~\ref{cor:functional_contains_SD}), and
if the class $\dcY_{P,Z}$ has bounded chromatic number, then $\chi(\dcY_{P,Z}) \leq 2d+1$ (Corollary~\ref{cor:bound_on_chromatic_number}).

In this section, we prefer simplicity of the tropical representation of the pair $(P,Z)$ over efficiency.
In \cref{sec:compressed_matrices}, we will consider a more efficient representation enabling a two-way strongly polynomial-time reduction between arbitrary mean payoff games and boundedness of chromatic number of classes $\dcY_{P,Z}$.

\subsection{Outline of the proof}\label{sec:dichotomy-outline}

In this section, we provide an outline of the proof of Theorem~\ref{thm:tropical_dichotomy}.
To fix the notation, let $P$ be a loopless path clause and $Z = ((L_1,\lambda_1), (L_2,\lambda_2), \ldots, (L_k,\lambda_k))$ a tuple of pairs, where, for $s \in [k]$, $L_s \subsetneq [d]$ and $\lambda_s \in [d] \setminus L_s$.
Our focus is on \chiboundedness of the class $\dcY_{P,Z}$ of digraphs defined by the clause $P$ on $Z$-functional vertex subsets of $\Nd$.
This is equivalent to the boundedness of chromatic number of $\dcY_{P,Z}$ as graphs in this class are triangle-free due to Lemma~\ref{lem:path-clause-shift-colorable}.

We prove three individual implications of Theorem~\ref{thm:tropical_dichotomy}.
The implication~\ref{item:contains_shift_graphs} $\Rightarrow$~\ref{item:unbounded_chromatic_number} is immediate.

The proof of the implication~\ref{item:tropicals_have_solution} $\Rightarrow$~\ref{item:contains_shift_graphs} is given in Section~\ref{ssec:construction-of-shift-graphs}.
Towards this goal, we introduce a construction of a $D$-dimensional shift digraph in the class $\dcY_{P,Z}$ using an \emph{interval representation} for the vertices' coordinates, see Definition~\ref{def:interval_system}.
Lemma~\ref{lem:interval_representation_implies_shift_graphs} states that if the pair $(P,Z)$ admits an interval representation of dimension $D$, the class $\dcY_{P,Z}$ contains the class $\dcS_D$ of $D$-dimensional shift digraphs.
We express the necessary conditions on the existence of an interval representation in terms of tropical inequalities; thus constructing the matrices $\widehat{A}, \widehat{B}, \widetilde{A}, \widetilde{B}$.
The equivalence between the existence of finite solutions to these systems and the existence of an interval representation is given by Lemma~\ref{lem:interval_representation_iff_tropical_solutions}.

In Section~\ref{ssec:bounding-the-chromatic-number}, we prove the implication \ref{item:unbounded_chromatic_number} $\Rightarrow$~\ref{item:tropicals_have_solution} by proving its contrapositive.
We assume that the max-plus system $\widetilde{A} \maxmu y \leq \widetilde{B} \maxmu y$ does not have a finite solution.
The key Lemma~\ref{lem:first_determines_second} states that in a certain projection of digraphs from $\dcY_{P,Z}$, long directed paths have the following property: coordinates of the first vertex fully determine those of the second vertex via the constraints imposed by the path clause $P$ and $Z$-functionality of the vertex set.
From this property we derive that digraphs in $\dcY_{P,Z}$ do not contain as a subdigraph a directed tree consisting of two long paths starting from a common vertex. This together with known results imply that graphs in $\dcY_{P,Z}$ have bounded chromatic number, which establishes the implication \ref{item:unbounded_chromatic_number} $\Rightarrow$~\ref{item:tropicals_have_solution}.

In order to prove Lemma~\ref{lem:first_determines_second}, we employ a natural BFS-like procedure to keep track of the coordinates of vertices in a long directed path whose values can be determined from the coordinates of the path's first vertex (see Definition~\ref{def:trackable_vertices} and Observation~\ref{obs:trackable_vertices_are_function_of_first_level}).
To show that all coordinates of the second vertex can be determined in this way, i.e.\ to prove Lemma~\ref{lem:first_determines_second}, we crucially employ Theorem~\ref{thm:connection_of_tropical_and_mean_payoff} connecting tropical inequalities and mean payoff games.
More specifically, in Section~\ref{sssec:determining_set_S}, we construct a subdigraph $\Gamma_{|S}$ of the game digraph $\Gamma(\widetilde{A}, \widetilde{B})$, where the column player is guaranteed to win.
Then, in Section~\ref{sssec:finishing_proof}, we use the optimal strategy of the column player to prove that each coordinate of the path's second vertex admits a search trajectory such that backtracking along this trajectory ends in the coordinates of the first vertex; thus, proving Lemma~\ref{lem:first_determines_second}.

In the intermediate Sections~\ref{sssec:multidigraph_cQ},~\ref{sssec:relating_cQ_and_Gamma|S},~and~\ref{sssec:utilizing_correspondence}, we develop the formal connection between coordinate tracking in the multidigraph $\cK$ representing dependencies among coordinates of vertices in the long path from Lemma~\ref{lem:first_determines_second} and the strategies in the mean payoff game on the digraph $\Gamma_{|S}$.

Finally, in Section~\ref{sssec:what_if_minplus_fails}, we explain the necessary changes in the proof of \ref{item:unbounded_chromatic_number} $\Rightarrow$~\ref{item:tropicals_have_solution} if it is the min-plus system $\widehat{A} \minmu x \geq \widehat{B} \minmu x$ that has no finite solution.

\subsubsection{Notation and conventions}\label{sssec:dichotomy_notation}

For the whole section, we fix a loopless $d$-dimensional path clause $P$ and a collection of functional constraints $Z = ((L_1,\lambda_1), \ldots, (L_k, \lambda_k))$.
For each $(L_s, \lambda_s)$, we have $L_s \subsetneq [d]$ and $\lambda_s \in [d] \setminus L_s$.
We assume that the positive literals of $P$ form a subset of the set $\{q_{c+1,c} : c \in [d-1]\}$ (which is exactly the set of positive literals of the clause defining $d$-dimensional shift graphs).
We let $\rho_1, \dots, \rho_m$ enumerate the maximal paths of the clause digraph $\Delta_P$ of $P$ (counting isolated vertices as paths).
We let $p \colon [d] \to [m]$ be the function mapping an element of $[d]$ to the index of the maximal path in $\Delta_P$ containing it.
We assume that the enumeration $\rho_1, \dots, \rho_m$ satisfies for $c, c' \in [d]$ that $p(c) \leq p(c')$.
We recall that $\dcY_{P,Z}$ denotes the set of digraphs induced by $P$ on $Z$-functional order-uniform subsets of $\Nd$.

\begin{remark}\label{rem:Q_instead_of_N}
    For convenience, the shift graphs constructed in \cref{ssec:construction-of-shift-graphs} will be represented by $d$-dimensional vectors with rational entries (instead of positive integers).
    This alteration is harmless as a finite $Z$-functional order-uniform subset $A$ of $\Qpld$ can be turned into a $Z$-functional order-uniform subset $B$ of $\Nd$ by clearing denominators.
    Since clearing denominators preserves order, it also preserves equalities and non-equalities.
    Thus, the graphs induced by a path clause $P$ on the vertex sets $A$ and $B$ are isomorphic.

    The reason why $\Q$ is more convenient than $\N$ is that it contains an infinite  sequence of infinite subsets each lying strictly above its predecessors, such as the intervals $(j,j+1)$ for $j = 0,1, \dots$.
    This allows us to construct $D$-dimensional shift graphs in $\dcY_{P,Z}$ by a uniform mapping of $D$-tuples to $d$-tuples.
    (When clearing denominators, we lose this property.)
\end{remark}

\begin{remark}\label{rem:each_discrete_clause_has_a_loop}
   Without loss of generality, we assume that $L_s \not= \emptyset$ for each $s \in [k]$.
   Indeed, if $Z$ contains a functional constraint $(\emptyset, \lambda)$, the class $\dcY_{P,Z}$ contains only edgeless digraphs.
   This is because all vectors in an $(\emptyset, \lambda)$-functional set have the same $\lambda$-th coordinate as discussed in Section~\ref{sec:functional-sets}.
   Since the path clause $P$ is loopless, the literal $q_{\lambda, \lambda}$ appears negatively in $P$.
   Therefore, $P$ is never satisfied and the resulting digraph is edgeless.
\end{remark}

\subsection{Construction of shift digraphs}\label{ssec:construction-of-shift-graphs}

Here we prove the implication \ref{item:tropicals_have_solution}~$\Rightarrow$~\ref{item:contains_shift_graphs} of Theorem~\ref{thm:tropical_dichotomy}. We begin with two examples.
The first is fairly simple, but the second contains the essential ideas of our construction.
These examples avoid the use of tropical algebra, but later examples will show how to arrive at the main ingredient of the solutions systematically using tropical algebra.

\begin{example}
	Let $d = 3$, let $P$ be a path clause with a single positive literal $q_{2,1}$, and $Z=(\{1\}, 3)$.
	 We will show that $\dcY_{P,Z}$ contains the class of 2-dimensional shift digraphs.
    
	Given $v = (s, t) \in \N^2$, let $\pi(v) = (1-2^{-s}, 1-2^{-t}, 2-2^{-s})$ (recall \cref{rem:Q_instead_of_N} that we allow here representation by rationals).
    Given a 2-dimensional shift digraph $S$ with representations for the vertices by increasing pairs of integers, let $\pi(S)$ be the digraph induced by $P$ on $\pi[V(S)]$.
    Note that $\pi[V(S)]$ is order-uniform as all its members are strictly increasing. 
    We will show $\pi \colon S \to \pi(S)$ is an isomorphism.
	
	It is clear from the definition of $\pi$ that it is injective. Let $(a, b), (c, d) \in V(S)$. Let $v= \pi((a, b)) = (1-2^{-a}, 1-2^{-b}, 2-2^{-a})$ and $w = \pi((c, d)) = (1-2^{-c}, 1-2^{-d}, 2-2^{-c})$.
	
	If $((a, b), (c, d))$ is an edge of $S$, then $b = c$. So $\pi((a, b)) = (1-2^{-a}, 1-2^{-b}, 2-2^{-a}) = v$ and $\pi((c, d)) = (1-2^{-b}, 1-2^{-d}, 2-2^{-b}) = w$, and we verify that $P$ induces the edge $(v,w)$. Since $v_2 = w_1$, the required equality is satisfied. To see that $v_1 \neq w_2$, note that $v_1 \neq w_2$ since $a < b = c < d$. Similar arguments show $v_1 \neq w_1, v_2 \neq w_2$, and $v_3 \neq w_3$. Also, we have $v_1, v_2 \neq w_3$ since $v_1, v_2 \in (0,1)$ and $w_3 \in (1,2)$, and similarly, we have $v_3 \neq w_1, w_2$. Thus $\pi$ is a homomorphism.
	
	If $((a,b), (c,d))$ is a non-edge in $S$, then $b \neq c$. So $v_2 \neq w_1$, and $P$ does not induce an edge from $v$ to $w$. 
\end{example}

In the previous example, in order to encode a 2-dimensional shift digraph, we considered vertices (in $\Qpld$) where every coordinate was associated to an element of the input. More generally, in order to encode a $D$-dimensional shift digraph, we will want to consider vertices where every coordinate is associated to an interval of elements of the input, as we will see in the next example.

\begin{example} \label{ex:encodeshift2}
	Let $d = 3$, and let $P$ be a path clause with a single positive literal $q_{2,1}$, and $Z = ((\{1\}, 3),( \{2\}, 3))$.
	We will show that $\dcY_{P,Z}$ contains the class of $3$-dimensional shift digraphs.
	
	Let $\pi_0 \colon \N^3 \to \N^2 \times \N^2 \times \N$ be defined by $\pi_0((a, b, c)) = ((a, b), (b, c), b)$. Let $g_1 \colon \N^2 \to (0,1)$ and $g_2 \colon \N \to (1,2)$ defined by $g_1(s,t) = 1 - 2^{-s}3^{-t}$ and $g_2(t) = 2 - 2^{-t}$, respectively.
    Note that both $g_j, j \in \{1,2\}$ are injective, strictly increasing in every coordinate, and the image of $g_1$ precedes the image of $g_2$.
    Let $\pi \colon \N^3 \to \Qpl^3$ be defined by $\pi((a, b, c)) = (g_1((a, b)), g_1((b, c)), g_2(b))$.
	
	Given a 3-dimensional shift digraph $S$ with representations for the vertices by increasing triples of integers, let $\pi(S)$ be the digraph induced by $P$ on $\pi[V(S)]$.
    By the properties of $g_1$ and $g_2$, the set $\pi[V(S)]$ is order-uniform as all its members are strictly increasing. 
    Moreover, $\pi_0[V(S)]$ is $((\{1\}, 3), (\{2\}, 3))$-functional since $b \in (a, b)$ and $b \in (b, c)$. We will show $\pi \colon S \to \pi(S)$ is an isomorphism.
	
	It is clear that $\pi_0$ is injective, and since $\pi$ is obtained by composing $\pi_0$ with the injections $g_1, g_2$, it is injective as well. Let $(a,b,c), (d,e,f) \in V(S)$. So $\pi((a,b,c)) = (g_1(a, b), g_1(b, c), g_2(b))$ and $\pi((d,e,f)) = (g_1(d, e), g_1(e, f), g_2(e))$.
	
	If $((a, b, c), (d, e, f))$ is an edge in $S$, then $b = d$ and $c = e$. Let $v = \pi((a,b,c)) = (g_1(a, b), g_1(b, c), g_2(b))$ and $w = \pi((d,e,f)) = (g_1(b, c), g_1(c, f), g_2(c))$, and we verify that $P$ induces the edge $(v,w)$. Since $v_2 = w_1$, the required equality is satisfied. To see $v_1 \neq w_2$, note that $a < c$, so $(a, b) \neq (c, f)$, and so $g_1(a, b) \neq g_1(c, f)$. Similar arguments show $v_1 \neq w_1, v_2 \neq w_2$, and $v_3 \neq w_3$. Also, we have $v_1, v_2 \neq w_3$ since $v_1, v_2 \in (0,1)$ and $w_3 \in (1,2)$, and similarly, we have $v_3 \neq w_1, w_2$. Thus $\pi$ is a homomorphism.
	
	If $((a, b, c), (d, e, f))$ is a non-edge in $S$, then either $b \neq d$ or $c \neq e$. So $v_2 \neq w_1$, and $P$ does not induce an edge from $v$ to $w$. 
\end{example}

Given a path clause $P$ and a set of functional constraints $Z$, we define an \emph{interval representation for $P$ over $Z$} that allows us to build shift digraphs within $\dcY_{P,Z}$.

\begin{notation}
	Given $X \subset \Z$ and $y \in \Z$, we let $X+y$ denote $\{x+y : x \in X\}$.
\end{notation}

\begin{definition}[Interval representation]\label{def:interval_system}
    Fix $D \in \N$ and consider a collection of non-empty intervals $\cI = \{I_c \subseteq [D]: c \in [d]\}$.
	We say that $\cI$ is a \emph{($D$-dimensional) interval representation for $P$ over $Z$} if it satisfies the following:
	\begin{enumerate}[label=(\roman*)]
		\item\label{item:equalities} For each positive literal $q_{c+1,c}$ of $P$, we have
		\[
		    I_{c+1} = I_c + 1.
		\]
		\item\label{item:functionality} For each functional constraint $(L, \lambda)$ in $Z$, we have
		\[
		    I_{\lambda} \subseteq \bigcup_{c \in L} I_{c}.
		\]
	\end{enumerate}
	An interval representation $\cI$ for $P$ over $Z$ is \emph{minimal} if it has the smallest possible dimension among interval representations for $P$ over $Z$.
\end{definition}

\begin{example} \label{example:int rep}
    Recall the setting of Example~\ref{ex:encodeshift2}.
    Let $d = 3$, and let $P$ be a path clause with a single positive literal $q_{2,1}$, and let $Z =((\{1\}, 3), (\{2\}, 3))$.
	Then $I_1 = \{1,2\}, I_2 = \{2,3\}, I_3 = \{2\}$ is a 3-dimensional interval representation for $P$ over $Z$. This interval representation gives rise to the map $\pi_0$ sending $(a,b,c) \mapsto ((a,b), (b,c), b)$ used in Example \ref{ex:encodeshift2}. 
\end{example}

The implication \ref{item:tropicals_have_solution}~$\Rightarrow$~\ref{item:contains_shift_graphs} of Theorem~\ref{thm:tropical_dichotomy} readily follows from the two lemmas below. \Cref{lem:interval_representation_iff_tropical_solutions} first reduces a system of tropical inequalities to an interval representation.  

\begin{lemma}\label{lem:interval_representation_iff_tropical_solutions}
	There exist matrices $\widehat{A}, \widehat{B} \in \Zmin^{(2d-2m+k) \times d}$ and $\widetilde{A}, \widetilde{B} \in \Zmax^{(2d-2m+k) \times d}$ such that $P$ admits an interval representation over $Z$ if and only if each tropical inequality $\widehat{A} \minmu x \geq \widehat{B} \minmu x$ and $\widetilde{A} \maxmu y \leq \widetilde{B} \maxmu y$ has a finite solution.
\end{lemma}

Then, to complete the implication, \Cref{lem:interval_representation_implies_shift_graphs} builds shift digraphs from this interval representation.

\begin{lemma}\label{lem:interval_representation_implies_shift_graphs}
	If $P$ admits a $D$-dimensional interval representation over $Z$, then the class $\dcY_{P,Z}$ contains~$\dcS_D$.
\end{lemma}

We first describe the tropical systems of Lemma \ref{lem:interval_representation_iff_tropical_solutions}.
The idea is that they should capture the requirements of an interval representation for $P$ over $Z$.
An interval $I_c \subseteq [D]$ can be represented by the positions of its endpoints $x_c$ and $y_c$.
The system $\widehat{A} \minmu x \geq \widehat{B} \minmu x$ encodes relative positions of the left endpoints $x$, while $\widetilde{A} \maxmu y \leq \widetilde{B} \maxmu y$ encodes the right endpoints $y$.
Both systems consist of two parts that correspond to the two conditions in Definition~\ref{def:interval_system}.
That is,
\begin{equation}\label{eq:defABhat}
	\widehat{A} =
	\begin{pmatrix}
		\widehat{A}^\shift\\
		\widehat{A}^\func
	\end{pmatrix}
	\;\;\;\textrm{and}\;\;\;
	\widehat{B} =
	\begin{pmatrix}
		\widehat{B}^\shift\\
		\widehat{B}^\func
	\end{pmatrix},
\end{equation}
where $\widehat{A}^\shift, \widehat{B}^\shift \in \Zmin^{(2d-2m) \times d}$ and $\widehat{A}^\func, \widehat{B}^\func \in \Zmin^{k \times d}$, and likewise for $\widetilde{A}$ and $\widetilde{B}$.

It is clear how to represent the first condition of Definition~\ref{def:interval_system}.
That is, we impose the conditions 
\begin{equation}\label{eq:equality_condition}
	\begin{split}
		x_{c+1} &= x_c + 1
		,\\
		y_{c+1} &= y_c + 1
		,
	\end{split}
\end{equation}
for each positive literal $q_{c+1,c}$ of $P$.
This gives $d-m$ pairs of
equalities for left and right endpoints, respectively.
After splitting each equality into two inequalities, we obtain the systems $\widehat{A}^\shift \minmu x \geq \widehat{B}^\shift \minmu x$ (for left endpoints) and $\widetilde{A}^\shift \maxmu y \leq \widetilde{B}^\shift \maxmu y$ (for right endpoints) with matrices of dimensions $(2d-2m) \times d$ (see \cref{ex:shift_graph_tropical_system}).

Capturing the condition
\[
    I_{\lambda} \subseteq \bigcup_{c \in L} I_{c} \qquad \textrm{ for each $(L,\lambda) \in Z$}
\]
is less straightforward.
Instead, we express a formally weaker condition that the interval $I_{\lambda}$ is enclosed both from left and right by the set $\bigcup_{c \in L} I_{c}$.
That is,
\begin{equation}\label{eq:endpoint_condition}
	\begin{split}
		x_{\lambda} &\geq \min_{c \in L} x_c
		,\\
		y_{\lambda} &\leq \max_{c \in L} y_c
		.
	\end{split}
\end{equation}
These $k$ conditions for left and right endpoints, one pair for each 
$(L, \lambda) \in Z$, comprise the systems $\widehat{A}^\func \minmu x \geq \widehat{B}^\func \minmu x$ and $\widetilde{A}^\func \maxmu y \leq \widetilde{B}^\func \maxmu y$.

Clearly, there is a close connection between matrices of the min-plus and max-plus systems.
To explain the connection, we introduce the following definition.
\begin{definition}\label{def:twin}
    For $x \in \R \cup \{-\infty, \infty\}$, we define its \emph{twin} $\twin{x} \in \R \cup \{-\infty, \infty\}$ by
    \begin{align*}
        \twin{x} =
        \begin{cases}
            x & \text{ if } x \in \R, \\
            +\infty & \text{ if } x = -\infty, \\
            -\infty & \text{ if } x = +\infty.
        \end{cases}
    \end{align*}
    Note that the twin operation maps values from the min-plus semiring to the max-plus semiring and vice versa.    
    For a matrix $A \in (\R \cup \{-\infty, \infty\})^{m \times n}$, we say that $\twin{A}$, where the twin operation is applied entry-wise, is the \emph{twin} matrix of $A$.
\end{definition}

\begin{observation}\label{obs:tropical_systems_are_the_same}
	The matrices $\widehat{A}$ and $\widetilde{A}$ are twins to each other; likewise for $\widehat{B}$ and $\widetilde{B}$.
\end{observation}

\begin{example}\label{ex:shift_graph_tropical_system}
We continue by considering the setting of Example \ref{ex:encodeshift2}, so we let $d = 3$, $P$ be a path clause with a single positive literal $q_{2,1}$, and $Z =((\{1\}, 3), (\{2\}, 3))$ be a set of functional constraints.
	Then we have $\widehat{A}^\shift \minmu x \geq \widehat{B}^\shift \minmu x$ of the form
	\begin{align*}
		\begin{pmatrix}
			\infty &0 & \infty \\
			0 & \infty &\infty
		\end{pmatrix}
		\minmu
		\begin{pmatrix}
			x_1 \\ x_2 \\ x_3
		\end{pmatrix}
		\geq
		\begin{pmatrix}
			1 & \infty & \infty \\
			\infty &-1 & \infty
		\end{pmatrix}
		\minmu
		\begin{pmatrix}
			x_1 \\ x_2 \\ x_3 
		\end{pmatrix}
		,
	\end{align*}
	and the system $\widehat{A}^\func \minmu x \geq \widehat{B}^\func \minmu x$ of the form
	\begin{align*}
		\begin{pmatrix}
			\infty & \infty & 0 \\
			\infty &\infty & 0 
		\end{pmatrix}
		\minmu
		\begin{pmatrix}
			x_1 \\ x_2 \\ x_3
		\end{pmatrix}
		\geq
		\begin{pmatrix}
			0 &\infty &\infty \\
			\infty & 0 &\infty \\
		\end{pmatrix}
		\minmu
		\begin{pmatrix}
			x_1 \\ x_2 \\ x_3
		\end{pmatrix}
		.
	\end{align*}
	
	By Observation~\ref{obs:tropical_systems_are_the_same}, the matrices in $\widetilde{A} \maxmu y \leq \widetilde{B} \maxmu y$ have the same form, except that $\infty$ is replaced by $-\infty$.
\end{example}

Before proving Lemma~\ref{lem:interval_representation_iff_tropical_solutions}, let us point out a simple observation. 
Let $\Zmin$ stand for $\Z \cup \{\infty\}$.

\begin{lemma}\label{obs:integer_solution}
	Let $A,B \in \Zmin^{m \times n}$.
	If the system $A \minmu x \geq B \minmu x$ has a finite solution in $\R^n$, then it has a finite  solution in $\Z^n$.
	The same holds true for max-plus systems.
\end{lemma}
\begin{proof}
    It is easy to check that if $x \in \R^n$ satisfies $A \minmu x \geq B \minmu x$ (respectively, $A \maxmu x \leq B \maxmu x$), then so does $x^*=(\lfloor x_1 \rfloor, \lfloor x_2 \rfloor, \ldots, \lfloor x_n \rfloor) \in \Z^n$.
\end{proof}

\begin{proof}[Proof of Lemma~\ref{lem:interval_representation_iff_tropical_solutions}]
	Fix $\widehat{A}, \widehat{B}, \widetilde{A}, \widetilde{B}$ to be the matrices associated with $P$ and $Z$ as defined above.
	
	``$\Rightarrow$''
	This is a trivial implication.
	Assume that $P$ admits an interval representation $\cI = \{I_c: c \in [d]\}$ over $Z$.
	We denote the left and right endpoints of $I_c$ by $x_c$ and $y_c$, respectively.
	Then the vectors $x, y \in \Z^d$ clearly satisfy the respective systems $\widehat{A} \minmu x \geq \widehat{B} \minmu x$ and $\widetilde{A} \maxmu y \leq \widetilde{B} \maxmu y$, as follows from the properties~\ref{item:equalities}~and~\ref{item:functionality} from Definition~\ref{def:interval_system}.
	
	``$\Leftarrow$''
	Suppose that we have finite solutions $x, y \in \R^d$ of the respective systems $\widehat{A} \minmu x \geq \widehat{B} \minmu x$ and $\widetilde{A} \maxmu y \leq \widetilde{B} \maxmu y$.
	By Lemma~\ref{obs:integer_solution}, we may assume that $x, y \in \Z^d$.
    Furthermore, we may assume that 
    \begin{equation}\label{eq:xs_leq_ys}
	    0 < x_i \leq y_j \qquad \text{ for every }i,j \in [d].
	\end{equation}
    Indeed, if $x$ and $y$ are solutions to the systems, then, for all $\alpha,\beta\in\Z$, the vectors $x+\alpha$ and $y+\beta$ are also solutions, see Observation~\ref{obs:shifting_solution}.
    Thus, by choosing $\alpha$ large enough that $x_i+\alpha>0$ for all $i\in[d]$, and then choosing $\beta$ large enough that $\min_{j\in[d]} y_j+\beta \geq \max_{i\in[d]} x_i+\alpha$, we may replace $x$ and $y$ by these shifted solutions and obtain \eqref{eq:xs_leq_ys}.  
    
	Now, for each $c \in [d]$ define the (closed) interval $I_c = [x_c, y_c] \subseteq \N$, and
	note that~\cref{eq:xs_leq_ys} implies that all $I_c$ are non-empty and $I_i \cap I_j \not= \emptyset$ for every $i,j \in [d]$.
	
	We prove that $\cI = \{I_c: c \in [d]\}$ is a valid interval representation for $P$ over $Z$ (of dimension $D=\max_c{y_c}$).
	Clearly, it follows from the system of inequalities in~\cref{eq:equality_condition} that Definition~\ref{def:interval_system}\ref{item:equalities} is satisfied.
	To verify Definition \ref{def:interval_system}\ref{item:functionality} for a particular functional constraint $(L, \lambda) \in Z$, we observe that the validity of the inequalities~\cref{eq:endpoint_condition} guarantees existence of some $c_\ell, c_r \in L$ satisfying
	\[
        	x_{c_\ell} \leq x_{\lambda} \leq y_{\lambda} \leq y_{c_r}
	    .
	\]
	Since $I_{c_\ell} \cap I_{c_r} \neq \emptyset$, the set $I_{c_\ell} \cup I_{c_r}$ is also an interval, namely $[x_{c_\ell}, y_{c_r}]$, and so contains $I_{\lambda}$.
	Thus, we have
	\[
        	I_{\lambda} \subseteq \bigcup_{c \in L} I_{c}
        	,
	\]
	as claimed.\end{proof}

\begin{example}[Continuation of Example~\ref{ex:shift_graph_tropical_system}]\label{ex:shift_graph_interval_representation}
	The system $\widehat{A} \minmu x \geq \widehat{B} \minmu x$ has a solution $x = (1, 2, 2) \in \Z^3$, while $\widetilde{A} \maxmu y \leq \widetilde{B} \maxmu y$ has a solution $y = (2, 3, 2) \in \Z^3$.
	By the proof of Lemma~\ref{lem:interval_representation_iff_tropical_solutions}, we define a $3$-dimensional interval representation for $P$ over $Z$ by
	\[
	    I_1 = \{1,2\}, \quad I_2 = \{2,3\}, \quad I_3 = \{2\}
        	.
	\] 
	This recovers the interval representation from Example \ref{example:int rep}.
	
	There are also other solutions to the system. For example, $x = (1, 2, 2)$ and $y = (0, 1, -1)$ is also a solution. Following the proof of Lemma \ref{lem:interval_representation_iff_tropical_solutions}, we would translate $y$ by at least 3, obtaining a solution $(3, 4, 2)$. This gives the 4-dimensional interval representation
    \[
	    I_1 = \{1,2, 3\}, \quad I_2 = \{2,3, 4\}, \quad I_3 = \{2\},
	\] 
	however, this is a non-minimal representation.
\end{example}

As our next step, we aim to prove Lemma~\ref{lem:interval_representation_implies_shift_graphs}.

We first explain how we can use an interval representation to produce a mapping from shift digraphs onto a convenient $Z$-functional set.
Let $\cI = \{I_c \subseteq [D]: c \in [d]\}$ be a $D$-dimensional interval representation for $P$ over $Z$, and let $x_c$ and $y_c$ be respectively the 
minimum and the maximum
elements of $I_c$.

We now generalize the passage from the interval representation provided in Example \ref{example:int rep} to the functions $\pi_0$ and $\pi$ used in Example \ref{ex:encodeshift2}. As before $\pi_0(v)$ will be a tuple of tuples, but this time obtained by taking subintervals of $v$ corresponding to the intervals of the interval representation. We then compose with injections from suitable powers of $\N$ into disjoint intervals of $\Qpl$, one injection for each maximal path $\rho_i$ of $P$, to obtain $\pi$. Unfortunately the notation in the general setting is somewhat tortuous.

We define $\pi_0 \colon \N^D \to \N^{|I_1|} \times \N^{|I_2|} \times \cdots \times \N^{|I_d|}$ by $\pi_0(v)_c = (v_{x_c}, \dots, v_{y_c})$ for every $c \in [d]$.
Recall that $p \colon [d] \to [m]$ is the function mapping an element of $[d]$ to the index of the maximal path containing it. For each $j \in [m]$ and for some/all $i$ such that $p(i) = j$, we choose $g_j \colon \N^{|I_{i}|} \to (j-1,j) \cap \Q$ to be an injection increasing in every coordinate.%
\footnote{For example, we may choose $g_j(n_1, \dots, n_{|I_i|}) = j - \prod_{k=1}^{|I_i|}\primef(k)^{-n_k}$, where $\primef(k)$ denotes the $k$th prime number.}
We then define $\pi \colon \N^D \to \Qpld$ by
\[
     \pi(v) = (g_{p(1)}(\pi_0(v)_1), g_{p(2)}(\pi_0(v)_2), \dots, g_{p(d)}(\pi_0(v)_d))
    .
\]

\begin{figure}
	\begin{tikzpicture}[
    x=1cm,
    y=1cm,
    >=stealth,
    every node/.style={font=\normalsize},
    heading/.style={
        font=\normalsize,
        anchor=west
    },
    box/.style={
        draw,
        minimum height=6.5mm,
        inner sep=1pt
    },
    maparrow/.style={
        ->,
        thin,
        draw=gray!65
    }
]
		\colorlet{shiftblue}{blue!70}
		\colorlet{shiftorange}{orange!70}
		\colorlet{shiftgray}{gray!70}
		
		\node[heading] at (0,5.4) {Interval representation};
		\draw[gray!25] (0,5.15)--(6.25,5.15);
		\foreach \x/\symbol in {1.25/$a$,1.95/$b$,2.65/$c$}
		\node at (\x,4.76) {\symbol};
		\foreach \x in {0.90,1.60,2.30,3.00}
		\draw[gray!19] (\x,4.51)--(\x,1.20);
		\node[anchor=west,font=\small] at (3.37,4.76)
		{Projected coordinate};
		
		\foreach \yy/\idx/\col/\left/\right/\value in {
			4.17/1/shiftblue/0.90/2.30/{$\pi_0(v)_1=(a,b)$},
			3.55/2/shiftblue/1.60/3.00/{$\pi_0(v)_2=(b,c)$},
			2.93/3/shiftgray/1.60/2.30/{$\pi_0(v)_3=b$},
			2.31/4/shiftorange/1.60/2.30/{$\pi_0(v)_4=b$},
			1.69/5/shiftorange/2.30/3.00/{$\pi_0(v)_5=c$}
		}{
			\node[anchor=east,font=\small,text=\col] at (0.69,\yy) {$I_{\idx}$};
			\draw[draw=\col,fill=\col!11] (\left,\yy-0.22)
			rectangle (\right,\yy+0.22);
			\node[anchor=west] at (3.37,\yy) {\value};
		}
		\foreach \x/\y/\symbol in {
			1.25/4.17/$a$,1.95/4.17/$b$,
			1.95/3.55/$b$,2.65/3.55/$c$,
			1.95/2.93/$b$,1.95/2.31/$b$,2.65/1.69/$c$
		} \node at (\x,\y) {\symbol};
		
		\draw[gray!18] (6.55,0.95)--(6.55,5.4);
		
		\node[heading] at (6.92,5.4) {Encoding $v=(a,b,c)$};
		\draw[gray!25] (6.92,5.15)--(15.2,5.15);
		\node[anchor=west,font=\small] at (6.94,4.76)
		{Coordinates $1,\ldots,5$ of the embedding};
		\node[anchor=east] at (8.02,3.88) {$\pi_0(v)=$};
		\node[anchor=east] at (8.02,1.73) {$\pi(v)=$};
		\foreach \cx/\col/\toplabel/\bottomlabel/\idx/\g in {
			8.75/shiftblue/{$(a,b)$}/{$g_1(a,b)$}/1/$g_1$,
			10.19/shiftblue/{$(b,c)$}/{$g_1(b,c)$}/2/$g_1$,
			11.63/shiftgray/{$b$}/{$g_2(b)$}/3/$g_2$,
			13.07/shiftorange/{$b$}/{$g_3(b)$}/4/$g_3$,
			14.51/shiftorange/{$c$}/{$g_3(c)$}/5/$g_3$
		}{
			\node[font=\small] at (\cx,4.35) {$\idx$};
			\node[box,minimum width=1.28cm,draw=\col,fill=\col!11]
			at (\cx,3.88) {\toplabel};
			\node[box,minimum width=1.28cm,draw=\col,fill=\col!11]
			at (\cx,1.73) {\bottomlabel};
			\draw[maparrow] (\cx,3.52)--node[right,font=\small,
			text=\col,inner sep=2pt]{\g}(\cx,2.08);
		}
		\node[anchor=west,font=\small] at (8.23,0.95)
		{$g_1,g_2,g_3$ are injections with ordered images.};
	\end{tikzpicture}
	\caption{Extension of Example~\ref{ex:shift_graph_projection_definition}, where $d=5$ and $P$ is a path clause with positive literals $q_{2,1}, q_{5,4}$, and $Z = ((\{1\}, 3),( \{2\}, 3), ( \{3\}, 4))$. This admits a 3-dimensional interval representation (shown on the left) yielding $\pi_0$ of the form $\pi_0(a,b,c) = (ab,bc,b,b,c)$ with an associated function $g_j$ for each connected component of the clause digraph (shown on the right). For a concrete realisation of the functions $g_j$, we take $g_1(s,t)=1-2^{-s}3^{-t}$, $g_2(t)=2-2^{-t}$, and $g_3(t)=3-2^{-t}$.
	}
\end{figure}

\begin{example}[Continuation of Example~\ref{ex:shift_graph_interval_representation}]\label{ex:shift_graph_projection_definition}
	Consider the first interval representation in Example~\ref{ex:shift_graph_interval_representation}, i.e.\ 
    \[
	    I_1 = \{1,2\}, \quad I_2 = \{2,3\}, \quad I_3 = \{2\}
        	.
	\] 
	Then the corresponding function $\pi_0\colon \N^3 \to \N^2 \times \N^2 \times \N$ is
	\[
	    \pi_0((a, b, c)) = \big( (a, b), (b, c), b \big)
	    .
	\]
	Let $g_1(s,t) = 1 - 2^{-s}3^{-t}$ and $g_2(t) = 2 - 2^{-t}$.
    Since the path clause $P$ only has a single positive literal $q_{2,1}$, we have $p \colon [3] \to [2]$ defined by $p(1) = p(2) = 1$ and $p(3) = 2$. Then the corresponding function $\pi\colon \N^3 \to \Qpl^3$~is
    \[
	    \pi((a, b, c)) = \big( g_1((a, b)), g_1((b, c)), g_2(b) \big)
	    .
	\]
	Thus we have recovered the functions $\pi_0$ and $\pi$ used in Example \ref{ex:encodeshift2}.
\end{example}

\begin{lemma}\label{lem:increasing_to_increasing}
    If $X \subseteq \N^D$ consists of strictly increasing vectors, then $\pi[X]$ consists of strictly increasing vectors.
\end{lemma}
\begin{proof}
    Let $v \in \N^D$ be strictly increasing, we want to show that $\pi(v)_c < \pi(v)_{c'}$ for every $c < c'$.
    If $p(c) = p(c')$, we have the entry-wise inequality $\pi_0(v)_c < \pi_0(v)_{c'}$.
    Since $g_{p(c)}$ is increasing (in every coordinate), we have that $\pi(v)_c < \pi(v)_{c'}$.
    
    If $p(c) \not= p(c')$, we have $p(c) < p(c')$ as $c < c'$.
    Then $\pi(v)_c \in (p(c)-1, p(c))$, while $\pi(v)_{c'} \in (p(c')-1,p(c'))$, so $\pi(v)_c < \pi(v)_{c'}$.
\end{proof}

Next we show that $\pi$ produced $Z$-functional sets.
We first illustrate this on a representative example.

\begin{example}[Continuation of Example~\ref{ex:shift_graph_projection_definition}] \label{ex:functionality}
	We observe that for an $X \subset \N^D$, the image $\pi_0[X]$ is a $((\{1\}, 3), (\{2\}, 3))$-functional set. Recall that $\pi_0((a,b,c)) = ((a, b), (b, c), b)$. Since $b$ is contained in the tuple $(a, b)$, the set is $((\{1\}, 3))$-functional, and since $b$ is contained in the tuple $(b, c)$, the set is $(\{2\}, 3)$-functional.
	
	We will now observe that $\pi[X]$ is also $((\{1\}, 3), (\{2\}, 3))$-functional. Indeed, given an element $\pi((a,b,c)) = (g_1((a,b)), g_1((b, c)), g_2(b))$ in $\pi[X]$, we may invert $g_1$ and $g_2$ to obtain the element $\pi_0((a,b,c))$ and the functional constraints of this element pass to $\pi((a,b,c))$. Writing it out explicitly, we have $h_1 \colon \N^2 \to \N$ sending $(x, y) \mapsto y$ and $h_2 \colon \N^2 \to \N$ sending $(x, y) \mapsto x$ giving the two functional constraints in $\pi_0[X]$. Then $g_2 \circ h_1 \circ g_1^{-1} \colon \N \to \N$ and $g_2 \circ h_2 \circ g_1^{-1} \colon \N \to \N$ give the two functional constraints in $\pi[X]$.
\end{example}

More generally, we have the following. 

\begin{lemma}\label{lem:pi_produces_functional_set}
	Let $X \subset \N^D$. Then the image $\pi[X]$ is a $Z$-functional set.
\end{lemma}
\begin{proof}
    We first consider $\pi_0[X]$. Given a particular functional constraint $(L, \lambda)$ in $Z$, Definition \ref{def:interval_system}\ref{item:functionality} gives the inclusion
\[
    I_\lambda \subseteq \bigcup_{c \in L} I_c
    .
\]

Which is to say that given $v \in \pi_0[X]$, every element of the tuple $v_\lambda$ is contained in some tuple $v_c$ (at some fixed index) for some $c \in L$. Thus $v_{|L}$ determines $v_\lambda$, and so $\pi_0[X]$ is $Z$-functional. Composing with the injections $g_j$ and their partial inverses on their images, similar to the gymnastics performed in Example \ref{ex:functionality}, gives that $\pi[X]$ is $Z$-functional as well. 
\end{proof}

Given a $D$-dimensional shift digraph $S$ with vertices labelled by increasing $D$-tuples, we let $\pi(S)$ (resp. $\pi_0(S))$ be the digraph induced by $P$ on $\pi[V(S)]$ (resp. $\pi_0[V(S)]$).
Note that by \cref{lem:increasing_to_increasing}, $\pi[V(S)]$ is order-uniform. 
Thus, together with \cref{lem:pi_produces_functional_set}, we have $\pi(S) \in \dcY_{P, Z}$.

We next wish to generalize the argument in Example~\ref{ex:encodeshift2} that $\pi$ is a homomorphism. As was the case there, we will consider three cases. The equality conditions in $P$ will be ensured by Definition~\ref{def:interval_system}\ref{item:equalities}. The non-equalities will split into two cases: for coordinates in the same maximal path of $\Delta_P$, we will use the fact that the vertices of $S$ are represented by increasing tuples, while for elements of distinct maximal paths of $\Delta_P$ we will use the fact that the injections $g_j$ have disjoint images.

\begin{lemma}\label{lem:pi_is_homomorphism}
	$\pi$ is a homomorphism from $S$ to $\pi(S)$.
\end{lemma}
\begin{proof}
	Consider $u,v \in V(S)$ such that $(u,v)$ is an edge of the shift digraph $S$.
	Our goal is to prove that $\pi(u), \pi(v)$ satisfy the path clause $P$.
	The clause $P$ contains three types of literals: positive literals $q_{i+1,i}$, negative literals $q_{i,j}$ where $p(i) = p(j)$, and negative literals $q_{i,j}$ where $p(i) \neq p(j)$.
	We show that each of these types evaluates to true.
	Recall that for each $c \in [d]$ we let $x_c, y_c$ denote the left and right endpoints of $I_c$.
	\paragraph{Case 1: positive literal $q_{i+1,i}$.}
	We suppose that $P$ contains the positive literal $q_{i+1,i}$, and want to show that $\pi(u)_{i+1} = \pi(v)_i$. Since this implies $p(i) = p(i+1)$, it suffices to prove that $\pi_0(u)_{i+1} = \pi_0(v)_i$, since the same injection $g_{p(i)}$ will be applied to both of them when passing to $\pi$.
	Expanding the definition of the function $\pi_0$, we want that
	\[
        	(u_{x_{i+1}}, \dots, u_{y_{i+1}}) = (v_{x_i}, \dots, v_{y_i})
        	.
	\]
	By Definition~\ref{def:interval_system}\ref{item:equalities}, we have $x_{i+1} = x_i + 1$ and $y_{i+1} = y_i + 1$.
	Thus, our goal is to verify that
	\[
	    (u_{x_i + 1}, \dots, u_{y_i + 1}) = (v_{x_i}, \dots, v_{y_i})
	    .
	\]
	This equality is indeed satisfied as it follows from the assumption that $(u, v)$ is an edge of $S$, i.e.\ $u_{j+1} = v_j$ for all $j \in [D-1]$.
	Therefore, the positive literal $q_{i+1,i}$ is satisfied.
	\paragraph{Case 2: negative literal $q_{i,j}$ with $p(i) = p(j)$.}
	We suppose that $P$ contains the negative literal $q_{i,j}$ for some $i,j$ with $p(i) = p(j)$, and we want to show that $\pi(u)_i \neq \pi(v)_j$. Again, since $p(i) = p(j)$, it suffices to prove that $\pi_0(u)_i \neq \pi_0(v)_j$, since the same injection $g_{p(i)}$ will be applied to both of them when passing to $\pi$.
    In particular, we show that $u_{x_i} \not= v_{x_j}$.

    First we consider the case when $j$ is the last element of the $p(j)$-th path of $\Delta_P$, i.e.\ either $j = d$ or $p(j+1) \not= p(j)$.
    As $p(i) = p(j)$, we have $i \leq j$ and consequently $x_i \leq x_j$. 
    If $x_j > 1$, then we have $u_{x_i} \leq u_{x_j} = v_{x_j-1} < v_{x_j}$, where the first and last inequality follows from the fact that $u$ and $v$ are increasing tuples, while the middle equality follows from the assumption that $(u,v)$ is an edge of $S$.
    Otherwise, if $x_i = x_j = 1$, then clearly $u_1 < u_2 = v_1$.
    
    Otherwise, we have that $p(j+1) = p(j)$, implying that $\pi_0(u)_{j+1} = \pi_0(v)_j$ by Case~1.
    In particular, $u_{x_{j+1}} = v_{x_j}$.
    Since $q_{i,j}$ is a negative literal, we have $i \not= j+1$.
    Therefore, $u_{x_i} \not= u_{x_{j+1}} = v_{x_j}$.
	
    \paragraph{Case 3: negative literal $q_{i,j}$ where $p(i) \neq p(j)$.}	
	We suppose that $P$ contains the negative literals $q_{i,j}$ for some $i,j$ where $p(i) \neq p(j)$, and we want to show that $\pi(u)_{i} \neq \pi(v)_j$. Since $p(i) \neq p(j)$, the injections $g_{p(i)}$ and $g_{p(j)}$ have disjoint images. Since $\pi(u)_{i}$ is in the image of $g_{p(i)}$ and $\pi(v)_{j}$ is in the image of $g_{p(j)}$, the inequality is immediate.
\end{proof}

Knowing that there is a homomorphism from every $D$-dimensional shift digraph to a digraph in $\dcY_{P, Z}$ is enough to conclude that $\dcY_{P, Z}$ has unbounded chromatic number.
Nevertheless, we aim to prove \cref{lem:interval_representation_implies_shift_graphs}, i.e.\ to show that $\dcY_{P, Z}$ \emph{contains} the class of $D$-dimensional shift digraphs for some $D$.
For this, it will suffice to show that $\pi$ is an isomorphism, rather than just a homomorphism.
This will not be true for every choice of $\pi$, but will be true for the mapping $\pi$ associated to a minimal interval representation, which is enough.
We break this into two steps.
In the first step, we identify a condition on interval representations that ensures the corresponding $\pi$ is an isomorphism (\cref{def:gapless_representation}).
We also introduce a stronger notion of a representation that will be useful when obtaining bounds on the dimension of a minimal representation (\cref{def:irreducible_representation}).
In the second step, we verify that minimal interval representations have this stronger property.

\begin{definition}\label{def:gapless_representation}
    A $D$-dimensional interval representation for $P$ over $Z$ is \emph{gapless} if for every $i \in [D-1]$ there exists $c_i \in [d-1]$ such that $i \in I_{c_i}$ and $p(c_i) = p(c_i+1)$.
\end{definition}

\begin{lemma} \label{lem:gapless isomorphism}
    Fix a gapless $D$-dimensional interval representation $\cI$ for $P$ over $Z$ with a corresponding function $\pi \colon \N^D \to \Qpld$. Then $\pi$ is an isomorphism from $S$ to $\pi(S)$.
\end{lemma}
\begin{proof}
    By Lemma~\ref{lem:pi_is_homomorphism}, $\pi$ is a homomorphism. We first check that $\pi$ is injective. Since $\cI$ is gapless, for every $i \in [D-1]$, there is some $c_i \in [d-1]$ such that $i \in I_{c_i}$.
    Further, we have that $p(c_{D-1}) = p(c_{D-1}+1)$, so Definition \ref{def:interval_system}\ref{item:equalities} ensures that $D \in I_{c_{D-1}+1}$. Thus for every element $i \in [D]$, there is some $c_i \in [d]$ such that $i \in I_{c_i}$. This ensures that the map $\pi_0$ corresponding to $\cI$ is injective, since every element of $v \in \N^D$ will appear at some determined place in $\pi_0(v)$, and so the tuple $v$ can be recovered from $\pi_0(v)$. Thus $\pi$ is injective as well, since it just coordinate-wise composes $\pi_0$ with further injections.

    Now suppose $(u,v)$ is a non-edge of $S$, and we wish to show that $(\pi(u), \pi(v))$ is a non-edge of $\pi(S)$. If $(u, v)$ is a non-edge of $S$ then there is some $i \in [D-1]$ such that $v_i \neq u_{i+1}$. Since $\cI$ is gapless, let $c_i \in [d-1]$ be such that $i \in I_{c_i}$ and $p(c_i) = p(c_i + 1)$; in particular, we let $\ell$ be such that $i$ is the $\ell^{th}$ element of $I_{c_i}$.
    By Definition \ref{def:interval_system}\ref{item:equalities}, $i+1$ is the $\ell^{th}$ element of $I_{c_i+1}$. But this implies that $\pi_0(v)_{c_i} \neq \pi_0(u)_{c_i+1}$, since the $\ell^{th}$ element of $\pi_0(v)_{c_i}$ is $v_i$ while the $\ell^{th}$ element of $\pi_0(u)_{c_i+1}$ is $u_{i+1}$, and we have assumed these values are distinct. Thus we also have $\pi(v)_{c_i} \neq \pi(u)_{c_i+1}$. Since $p(c_i) = p(c_{i}+1)$, these values must be equal for $P$ to induce the edge $(\pi(u), \pi(v))$ and so $(\pi(u), \pi(v))$ is also a non-edge in $\pi(S)$.
\end{proof}

\begin{definition}\label{def:irreducible_representation}
    We say that a coordinate $i \in [D-1]$ is \emph{guarded} by a coordinate $c \in [d-1]$ if $p(c) = p(c+1)$ and either $i = \min I_{c}$ or $i = \max I_{c}$.
    A $D$-dimensional interval representation for $P$ over $Z$ is \emph{irreducible} if for every $i \in [D-1]$ there exists $c_i \in [d-1]$ such that $i$ is guarded by $c_i$.
\end{definition}

Clearly, an irreducible representation is a gapless representation.

Next we prove that a minimal representation is irreducible.
The proof proceeds by contraposition: if an interval representation $\cI$ is \emph{not irreducible}, we can take advantage of the unguarded coordinate and find a new interval representation $\cI'$ with \emph{reduced} dimension.
This explains the term `irreducible'.

\begin{lemma}\label{lem:minimal_representation_is_irreducible}
    If $\cI$ is a minimal interval representation for $P$ over $Z$, then $\cI$ is irreducible.    
\end{lemma}
\begin{proof}
    Suppose that a minimal interval representation $\cI$ for $P$ over $Z$ is not irreducible, and let $D$ be such that $\cI$ is $D$-dimensional.
    Let $K \in [D-1]$ be an unguarded coordinate, i.e.\ $K$ is not guarded by any $c \in [d-1]$.
    We will now define a new interval representation for $P$ over $Z$ of dimension $D-1$, contradicting that $\cI$ is minimal.

    Consider the interval $I_c \subseteq [D]$ of $\cI$ for $c \in [d]$.
    We define a new interval $I'_c \subset [D-1]$ by
    \[
        I'_c = (I_c \cap [K]) \cup ((I_c \cap [K+1, D]) -1)
        .
    \]
    We now claim that the interval assignment $\cI' = \{I'_c : c \in [d]\}$ is also an interval representation of $P$ over $Z$ (and visibly has dimension $D-1$).
    Note that each $I'_c$ is non-empty.

    We first check Definition~\ref{def:interval_system}\ref{item:equalities}.
    Let $c \in [d-1]$ be such that $p(c) = p(c+1)$.
    Then $I_{c+1} = I_c + 1$, and we want to prove that $I'_{c+1} = I'_c + 1$.
    As these are intervals, it is enough to prove that their endpoints are equal.
    Let $I_c = [x_c,y_c]$, $I'_c = [x'_c,y'_c]$.
    Since $c$ does not guard $K$, we have that $K$ is distinct from both $x_c$ and $y_c$. 
    Then we distinguish three cases: $x_c \leq y_c < K$, $x_c < K < y_c$, and $K < x_c \leq y_c$.
    
    First, let $x_c \leq y_c < K$.
    Then, $I_c, I_{c+1} \subseteq [K]$, so $I'_c = I_c$ and $I'_{c+1} = I_{c+1}$.
    
    Second, let $x_c < K < y_c$.
    Then, $x_c,x_{c+1} \in [K]$, so $x'_c = x_c$ and $x'_{c+1} = x_{c+1}$.
    Moreover, $y_c,y_{c+1} \in [K+1,D]$, so $y'_c = y_c - 1$ and $y'_{c+1} = y_{c+1} - 1$.
    
    Third, $K < x_c \leq y_c$.
    Then, $I_c,I_{c+1} \subseteq [K+1,D]$, so $I'_c = I_c - 1$ and $I'_{c+1} = I_{c+1} - 1$.
    
    Thus, we conclude in all cases that $I'_{c+1} = I'_c + 1$ as we wanted.

    We next check Definition~\ref{def:interval_system}\ref{item:functionality}. 
    Let $(L, \lambda) \in Z$ be a functional constraint.
    We know $I_\lambda \subset \bigcup_{c \in L} I_c$. Let $i \in I'_\lambda$.
    Then either $i \in I_\lambda \cap [K]$ or $i \in (I_\lambda \cap [K+1, D])-1$ by definition of $I'_\lambda$.
    
    In the first case, we have $i \leq K$ and $i \in I_\lambda$.
    Then, by Definition~\ref{def:interval_system}\ref{item:functionality} for $\cI$, we have $i \in I_c$ for some $c \in L$.
    But then $i \in I'_c$.
    
    Now suppose that $i \in (I_\lambda \cap [K+1, D])-1$, but not $i \in I_\lambda \cap [K]$.
    Then, $i+1 \in I_\lambda$.
    Again, by Definition~\ref{def:interval_system}\ref{item:functionality} for $\cI$, there is $c \in L$ such that $i+1 \in I_c$.
    Thus, $i \in I'_c$.
    
    Therefore, in both cases we conclude that $I'_\lambda \subset \bigcup_{c \in L} I'_c$.
\end{proof}

With this, we are ready to prove \cref{lem:interval_representation_implies_shift_graphs} to complete the implication \ref{item:tropicals_have_solution}~$\Rightarrow$~\ref{item:contains_shift_graphs} of Theorem~\ref{thm:tropical_dichotomy}.

\begin{proof}[Proof of \cref{lem:interval_representation_implies_shift_graphs}]
    Let $\cI$ be an interval representation of dimension $D$ for $P$ over $Z$ guaranteed by the assumption.
    We pass to a \emph{minimal} interval representation $\cI'$ for $P$ over $Z$ of dimension $D' \leq D$.
    Then, by Lemmas~\ref{lem:increasing_to_increasing},~\ref{lem:pi_produces_functional_set},~\ref{lem:gapless isomorphism}~and~\ref{lem:minimal_representation_is_irreducible}, the class $\dcY_{P,Z}$ contains the class $\dcS_{D'}$ of $D'$-dimensional shift graphs.
    Hence also $\dcS_{D} \subseteq \dcY_{P,Z}$ as $\dcS_{D} \subseteq \dcS_{D'}$ by the third item of Fact~\ref{fact:shift}. 
\end{proof}

\subsubsection{Bounds on dimension of shift graphs}

Although the proof of the implication \ref{item:tropicals_have_solution}~$\Rightarrow$~\ref{item:contains_shift_graphs} of Theorem~\ref{thm:tropical_dichotomy} is now finished, we may further analyze the dimension $D$ of the class $\dcS_D$ in terms of $d$.
In the following statement, we use the irreducibility property to give a tight bound.

Recall that the clause digraph $\Delta_P$ of the $d$-dimensional path clause $P$ consists of $m$ maximal paths.
Since at least one of the maximal paths of $\Delta_P$ contains at least two vertices, we have $d \geq m + 1 \geq 2$.

\begin{corollary}\label{cor:functional_contains_SD}
    If the pair $(P,Z)$ admits an interval representation, then the class $\dcY_{P,Z}$ contains the class $\dcS_D$ of $D$-dimensional shift digraphs for some $2 \leq D \leq 2(d-m) + 1$. 
    Moreover, if $m = 1$, then $2 \leq D \leq 2d-2$.
\end{corollary}
\begin{proof}
    Let $\cI$ be a minimal interval representation of dimension $D$ for $P$ over $Z$; it is irreducible by \cref{lem:minimal_representation_is_irreducible}.
    That is, each $K \in [D-1]$ is guarded by some coordinate $c_K$ with $p(c_K) = p(c_K + 1)$ such that $K = \min I_{c_K}$ or $K = \max I_{c_K}$.
    By definition, the last coordinate of a path of $\Delta_P$ cannot guard. Thus there are at most $d-m$ guarding (non-last) coordinates, each of which may guard at most $2$ coordinates from $[D-1]$.
    Therefore, $D \leq 2(d-m) + 1$, where the $+1$ is for the coordinate $D$, which need not be guarded.
    
    Moreover, if $m = 1$, we claim that $\cI = \{[c, c+d-2] : c \in [d]\}$ is a valid interval representation of dimension $2d-2$ for $P$ over $Z$, whenever the pair $(P,Z)$ admits one.
    Note that since $P$ has a single path, all intervals of \emph{any} interval representation for $P$ over $Z$ have some fixed length $\ell$ and the intervals of adjacent coordinates are shifted by one;
    indeed, this follows directly from Definition~\ref{def:interval_system}\ref{item:equalities}.
    Consequently, if there is $(L, \lambda) \in Z$ with $\lambda \not\in [\min L, \max L]$, then no interval representation exists in the first place as Definition~\ref{def:interval_system}\ref{item:functionality} cannot be satisfied.
    Hence, it holds $\lambda \in [\min L, \max L]$ for each $(L, \lambda) \in Z$.
    It follows that $\cI$ is a valid interval representation for $P$ over $Z$ as for any $a < c < b$ we have $I_c \subseteq I_a \cup I_b$, which proves Definition~\ref{def:interval_system}\ref{item:functionality} holds.
    Definition~\ref{def:interval_system}\ref{item:equalities} is satisfied trivially.
\end{proof}

The following examples show that the bounds given by Corollary~\ref{cor:functional_contains_SD} are tight.
In both examples, we implicitly assume that the pair $(P,Z)$ admits an interval representation to comply with the assumption of Corollary~\ref{cor:functional_contains_SD}.

\begin{example}
    Let $P$ be the path clause with a single path and consider $Z$ consisting of a single functional constraint $(\{1,d\}, \lambda)$ for an arbitrary $1 < \lambda < d$.
    We claim that the minimal dimension $D$ of an interval representation $\cI$ for $P$ over $Z$ is $2d-2$.
    Assuming that $\min I_1 = 1$, there is $\ell \in \N$ such that $I_c = [c,c+\ell]$ for each $c \in [d]$; then, $D = d+\ell$.
    To satisfy Definition~\ref{def:interval_system}\ref{item:functionality}, we need that $I_1 \cup I_d$ is an interval.
    Hence, $\ell \geq d-2$, implying that $D \geq 2d-2$, which agrees with the bound from Corollary~\ref{cor:functional_contains_SD}.
\end{example}

\begin{example}
    Consider a path clause $P$ consisting of a path $\{1, \dots, d-m\}$ and other $m \geq 1$ single-vertex paths $\{i\}$ for $i \in [d-m+1, d]$.
    Let $Z = \big(\:(\{1\},d), (\{d-m\},d)\:\big)$.
    Note that any interval representation $\cI$ for $P$ over $Z$ needs to satisfy that $I_d \subseteq I_1$ and $I_d \subseteq I_{d-m}$.
    In particular, $I_1 \cap I_{d-m} \not= \emptyset$.
    Therefore, $|I_1| \geq d-m$ due to Definition~\ref{def:interval_system}\ref{item:equalities}, so $|I_1 \cup I_{d-m}| \geq 2(d-m)-1$.
    Consequently, the dimension of $\cI$ is at least $2(d-m)-1$.
    Since $\Delta_P$ has $m+1$ paths, Corollary~\ref{cor:functional_contains_SD} gives the matching upper bound of $2(d-m-1)+1 = 2(d-m)-1$.
\end{example}

\subsubsection{Preparation for \cref{th:main-set-defined}}

We now prove a preparatory lemma for \cref{th:main-set-defined}.
Since it is about interval representations, we prefer to prove it here rather than later.

\begin{lemma}\label{lem:extension_of_gapless_is_gapless}
    Let $\cI = \{I_c : c \in [d]\}$ be a $D$-dimensional gapless interval representation for $P$ over $Z$.
    Fix $t \in \N$ and define $\cI' = \{I'_c : c \in [d]\}$ by setting $I'_c = I_c + [0,t]$.
    Then $\cI'$ is a $(D+t)$-dimensional gapless interval representation for $P$ over $Z$.
\end{lemma}
\begin{proof}
    It is easy to see that $I_{c+1} = I_c + 1$ implies $I'_{c+1} = I'_c + 1$, and $I_{\lambda} \subseteq \bigcup_{c \in L} I_{c}$ implies $I'_{\lambda} \subseteq \bigcup_{c \in L} I'_{c}$.
    Thus, $\cI'$ is an interval representation for $P$ over $Z$.
    
    It remains to prove that $\cI'$ is gapless by \cref{def:gapless_representation}.
    Consider a coordinate $i \in [D+t-1]$.
    We choose the witnessing coordinate $c'_i$ to show that $\cI'$ is gapless as
    \begin{align*}
        c'_i =
        \begin{cases}
            c_i & \text{if }i \in [D-1],\\
            c_{D-1} & \text{if }i \in [D,D+t-1],
        \end{cases}
    \end{align*}
    where $c_j \in [d-1]$ is the witness of the gapless property for $\cI$ and $j \in [D-1]$.
    Since $p(c_j) = p(c_j+1)$ by the definition, we have $p(c'_i) = p(c'_i+1)$.
    Moreover, it holds $i \in I'_{c'_i}$, which is clear if $i \in [D-1]$.
    If $i \in [D,D+t-1]$, then $i \in D-1 + [0,t]$.
    Thus, since $D-1 \in I_{c_{D-1}}$, we have $i \in I'_{c'_i}$.
\end{proof}

\subsection{Bounding the chromatic number}\label{ssec:bounding-the-chromatic-number}

\colorlet{HighlightOrange}{orange!100}
\colorlet{HighlightBlue}{blue!100}
\colorlet{HighlightGray}{gray!100}
\colorlet{LightOrange}{orange!70}
\colorlet{LightBlue}{blue!70}
\colorlet{LightGray}{gray!70}
\newcommand{\orangeText}[1]{\textcolor{LightOrange}{#1}}
\newcommand{\blueText}[1]{\textcolor{LightBlue}{#1}}
\newcommand{\grayText}[1]{\textcolor{LightGray}{#1}}
\newcommand{\highlightOrangeText}[1]{\textcolor{HighlightOrange}{#1}}
\newcommand{\highlightBlueText}[1]{\textcolor{HighlightBlue}{#1}}
\newcommand{\highlightGrayText}[1]{\textcolor{HighlightGray}{#1}}

In this section, we prove the implication \ref{item:unbounded_chromatic_number} $\Rightarrow$~\ref{item:tropicals_have_solution} of \cref{thm:tropical_dichotomy}.
In the previous section, we showed that if certain tropical systems admit a finite solution, then the class $\dcY_{P,Z}$ is \chiunbounded, as it contains some class $\dcS_D$ of shift digraphs.
Now we prove the opposite: if one of these systems fails to have a finite solution, then class $\dcY_{P,Z}$ has bounded chromatic number.
We assume throughout that it is the corresponding max-plus system that has no finite solution.
The other case, i.e.\ that the min-plus system has no finite solution, is analogous and we comment on the necessary modifications of the proof at the end of this section.

We start with a couple of illuminating examples of classes $\dcY_{P,Z}$ of bounded chromatic number.
They illustrate basic ideas behind the proof.
Although the examples do not explicitly refer to the tropical machinery, let us remark that in both examples we keep the convention that it is the max-plus system with no finite solution.

\begin{example}[A simple example]
Let $d=2$, and $P$ be a path clause with a single positive literal $q_{2,1}$. 
Let $L = \{1\}, \lambda = 2$, and $Z = ((L, \lambda))$.
Let $G \in \dcY_{P,Z}$.
We argue that every vertex in $G$ has at most one out-neighbor, which clearly would imply that the chromatic number of $G$ is at most 3.
Let $\phi: \bN \rightarrow \bN$ be a function witnessing the $(L,\lambda)$-functionality of $V(G)$, i.e.\ for every vertex $v \in V(G)$ we have $v_2 = \phi(v_1)$. Suppose $u$ is an out-neighbor of $v$, then, due to the literal $q_{2,1}$, we have that $u_1 = v_2$. Furthermore, due to the $(L,\lambda)$-functionality we have $u_2 = \phi(u_1) = \phi(v_2)$. That is, $u = (u_1, u_2)$ is uniquely determined by $v$, and thus $v$ cannot have two or more distinct out-neighbors.
\end{example}

A digraph has maximum out-degree one if and only if it contains no subdigraph isomorphic to $\vec{\Lambda}_2$, where $\vec{\Lambda}_t$ is the digraph obtained from two copies of the directed $t$-vertex path $\vec{P}_t$ by identifying their initial (source) vertices. It turns out that, in general, classes $\dcY_{P,Z}$ have bounded chromatic number if and only if they exclude some $\vec{\Lambda}_t$ or its reverse\footnote{The reverse of a digraph $G$ is the digraph obtained from $G$ by reversing its edge orientations.} as a subdigraph. The if direction of this equivalence is a consequence of the following fact. 

\begin{fact}\label{fact:no_Lambda_small_chromatic_number}
    If a digraph $G$ does not contain $\vec{\Lambda}_t$ or the reverse of $\vec{\Lambda}_t$ as a subdigraph, then $\chi(G) \leq 2t-1$.
\end{fact}
\begin{proof}
    It is enough to consider $\vec{\Lambda}_t$ as a forbidden subdigraph, as the claim for the reverse of $\vec{\Lambda}_t$ follows by considering the reverse of $G$.
    If $G$ does not contain $\vec{\Lambda}_2$, the maximal outdegree of $G$ is $1$.
    Hence, $G$ is a pseudoforest and we have $\chi(G) \leq 3$.
    For $t \geq 3$, we even have a stronger bound $\chi(G) < 2t-1$ by a result of Addario-Berry, Havet, and Thomassé \cite{ADDARIOBERRY2007620}, which implies that any digraph $H$ with $\chi(H) \geq 2t-1$, where $t \geq 3$, contains as a subdigraph every orientation of the path on $2t-1$ vertices with 2 directed blocks (including~$\vec{\Lambda}_t$).
\end{proof}

This reduces showing bounded chromatic number of a digraph to showing that it does not contain some $\vec{\Lambda}_t$ or its reverse as a subdigraph. The following example illustrates a more involved situation.

\begin{example}[A more interesting example]\label{ex:chi_bnd_more_interesting}
Let $d=5$, and $P$ be a path clause whose positive literals are \highlightGrayText{$q_{2,1}$}, \highlightGrayText{$q_{4,3}$}, \highlightGrayText{$q_{5,4}$}.
In other words, the clause digraph $\Delta_P$ of $P$ consists of two directed paths: $(2,1)$ and $(5,4,3)$.
Let \blueText{$L_1 = \{2,3\}, \lambda_1 = 5$}, \orangeText{$L_2 = \{ 3 \}, \lambda_2 = 1$}, and $Z = (\blueText{(L_1, \lambda_1)}, \orangeText{(L_2, \lambda_2)})$, and
let $G \in \dcY_{P,Z}$. 
We will show that $G$ does not contain $\vec{\Lambda}_3$ as a subdigraph, and therefore, by Fact~\ref{fact:no_Lambda_small_chromatic_number}, $G$ has chromatic number at most $5$.

To show that $G$ does not contain $\vec{\Lambda}_3$, we will show that in any $3$-vertex directed path $(a,b,c)$ in $G$ the first vertex functionally determines the second vertex. More formally, for every $i \in [5]$, there exists a function $f_i : \bN^5 \rightarrow \bN$ such that $b_i = f_i(a)$. Clearly, this would imply that $G$ does not contain $\vec{\Lambda}_3$ as a subdigraph, as the two 3-vertex paths constituting $\vec{\Lambda}_3$ share the first vertex, but have different second vertices.

Let $\blueText{\phi_1} : \bN^2 \rightarrow \bN$ and $\orangeText{\phi_2} : \bN \rightarrow \bN$ be functions witnessing the respective $\blueText{(L_1,\lambda_1)}$ and $\orangeText{(L_2,\lambda_2)}$ functional constraints of $V(G)$.
Due to the literals \highlightGrayText{$q_{2,1}$}, \highlightGrayText{$q_{4,3}$}, \highlightGrayText{$q_{5,4}$}, we have that $b_1 = a_2$, $b_3 = a_4$, $b_4 = a_5$, so we can immediately define $f_1(a) = a_2$, $f_3(a) = a_4$, $f_4(a) = a_5$.
To determine the value of $b_2$, we observe that $b_2$ is equal to $c_1$, which in turn is functionally determined by $c_3$ via $\orangeText{\phi_2}$. 
Furthermore, $c_3$ is equal to $b_4$, which is equal to $a_5$ as we already know. Thus we can express $b_2$ as a function of $a$ as follows $b_2 = \orangeText{\phi_2}(a_5) = f_2(a)$. Finally, the value of $b_5$ is functionally determined by $b_2$ and $b_3$ via $\blueText{\phi_1}$. Thus, using already known derivations of $b_2$ and $b_3$, we can express $b_5$ as a function of $a$ as follows $b_5 = \blueText{\phi_1}(b_2,b_3) = \blueText{\phi_1}(\orangeText{\phi_2}(a_5),a_4) = f_5(a)$.
Figure~\ref{fig:tracking-ex1} illustrates how each coordinate of $b$ is `derived' from the coordinates of $a$; this picture will be formalized and becomes clearer later in \cref{sssec:multidigraph_cQ}.%
\footnote{Let us note the general proof strategy yields $f_1(a) = \orangeText{\phi_2}(a_4)$ instead of $f_1(a) = a_2$.}
\end{example}

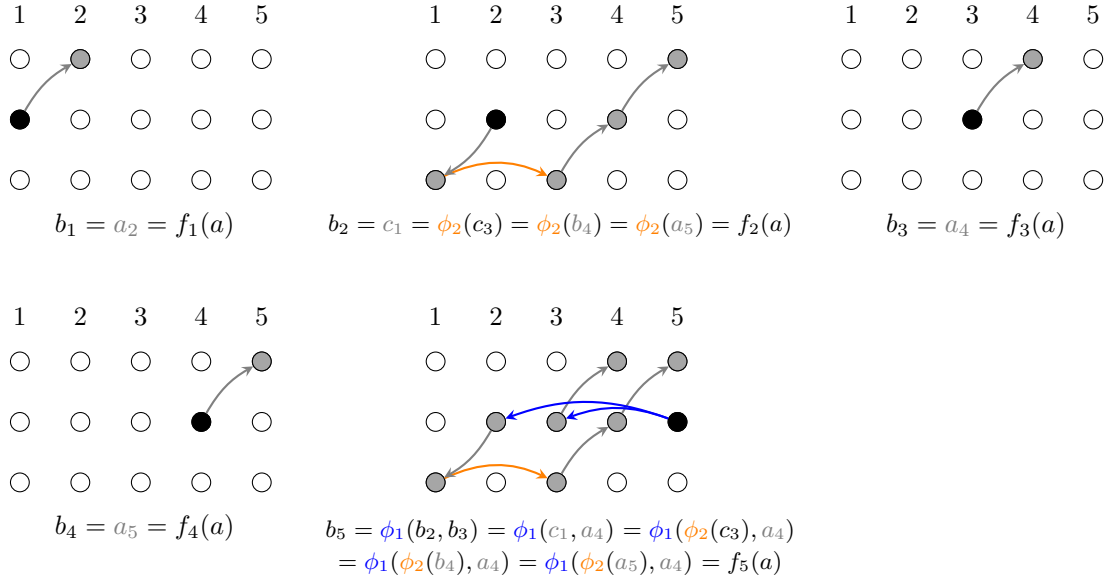
\begin{figure}
    \centering
    \begin{tikzpicture}[>=stealth]

  \begin{scope}[shift={(5,0)}, scale=1]
  \foreach \i in {1,...,3} {
    \node[draw, minimum size=0.1cm] (v\i) at (0,5-\i) {};
  }

  \node[left=4pt] at (v1) {$a$};
  \node[left=4pt] at (v2) {$b$};
  \node[left=4pt] at (v3) {$c$};

  \foreach \i in {1,...,2} {
    \pgfmathtruncatemacro{\next}{\i+1}
    \draw[->] (v\i) -- (v\next);
  }

  \foreach \row in {1,...,3} {
    \foreach \col in {1,...,5} {
      \node[circle, draw, inner sep=2.5pt]
        (s\row\col) at (1 + \col, 5-\row) {};
    }
  }

  \foreach \col in {1,...,5} {
    \node[above=10pt] at (s1\col) {$\col$};
  }

  \foreach \row in {1,...,2} {
    \pgfmathtruncatemacro{\nextrow}{\row+1}
    \draw[->,LightGray,bend left=15]  (s\row2)   to (s\nextrow1);
    \draw[->,LightGray,bend left=15]  (s\nextrow1) to (s\row2);
    \draw[->,LightGray,bend left=15]  (s\row5)   to (s\nextrow4);
    \draw[->,LightGray,bend left=15]  (s\nextrow4) to (s\row5);
    \draw[->,LightGray,bend left=15]  (s\row4)   to (s\nextrow3);
    \draw[->,LightGray,bend left=15]  (s\nextrow3) to (s\row4);
  }

  \foreach \row in {1,...,3} {
    \draw[->,LightBlue] (s\row5) to[bend right=20] (s\row2);
    \draw[->,LightBlue] (s\row5) to[bend right=20] (s\row3);

    \draw[->,LightOrange] (s\row1) to[bend left=25] (s\row3);
  }
  \end{scope}
  \node[align=center] at (5 + 4, 0+1) {Dependencies between coordinates of the\\ vertices of the directed path $(a,b,c)$ due to the literals \\ \highlightGrayText{$q_{2,1}$}, \highlightGrayText{$q_{4,3}$}, \highlightGrayText{$q_{5,4}$}, and functional constraints $\blueText{\phi_1}$ and $\orangeText{\phi_2}$};


  \begin{scope}[shift={(0,-4.5)}, scale=0.8]
    \foreach \row in {1,...,3} {
      \foreach \col in {1,...,5} {
        \node[circle, draw, inner sep=2.5pt]
          (c1s\row\col) at (1 + \col, 5-\row) {};
      }
    }
    \node[circle, draw, inner sep=2.5pt, fill=LightGray] at (c1s12) {};
    \node[circle, draw, inner sep=2.5pt, fill=black] at (c1s21) {};

    \foreach \col in {1,...,5} {
      \node[above=10pt] at (c1s1\col) {$\col$};
    }

    \foreach \row in {1,...,1} {
      \pgfmathtruncatemacro{\nextrow}{\row+1}
      \draw[->,thick,HighlightGray,bend left=15]  (c1s\nextrow1) to (c1s\row2);
    }

  \end{scope}
  \node at (0 + 3.25, -4.5 + 1) {$b_1 = \highlightGrayText{a_2} = f_1(a)$};

  \begin{scope}[shift={(5.5,-4.5)}, scale=0.8]
    \foreach \row in {1,...,3} {
      \foreach \col in {1,...,5} {
        \node[circle, draw, inner sep=2.5pt]
          (c2s\row\col) at (1 + \col, 5-\row) {};
      }
    }
    \node[circle, draw, inner sep=2.5pt, fill=black] at (c2s22) {};
    \node[circle, draw, inner sep=2.5pt, fill=LightGray] at (c2s31) {};
    \node[circle, draw, inner sep=2.5pt, fill=LightGray] at (c2s33) {};  
    \node[circle, draw, inner sep=2.5pt, fill=LightGray] at (c2s24) {};
    \node[circle, draw, inner sep=2.5pt, fill=LightGray] at (c2s15) {};

    \foreach \col in {1,...,5} {
      \node[above=10pt] at (c2s1\col) {$\col$};
    }

    \draw[->,thick,HighlightOrange] (c2s31) to[bend left=25] (c2s33);
    \draw[->,thick,HighlightGray,bend left=15]  (c2s22) to (c2s31);
    \draw[->,thick,HighlightGray,bend left=15]  (c2s33) to (c2s24);
    \draw[->,thick,HighlightGray,bend left=15]  (c2s24) to (c2s15);
  \end{scope}
  \node at (5.5 + 3.25, -4.5 + 1) {\small $b_2 = \highlightGrayText{c_1} = \highlightOrangeText{\phi_2}(c_3) = \highlightOrangeText{\phi_2}(\highlightGrayText{b_4}) = \highlightOrangeText{\phi_2}(\highlightGrayText{a_5}) = f_2(a)$};

  \begin{scope}[shift={(11,-4.5)}, scale=0.8]
    \foreach \row in {1,...,3} {
      \foreach \col in {1,...,5} {
        \node[circle, draw, inner sep=2.5pt]
          (c3s\row\col) at (1 + \col, 5-\row) {};
      }
    }
    \node[circle, draw, inner sep=2.5pt, fill=LightGray] at (c3s14) {};
    \node[circle, draw, inner sep=2.5pt, fill=black] at (c3s23) {};

    \foreach \col in {1,...,5} {
      \node[above=10pt] at (c3s1\col) {$\col$};
    }

    \foreach \row in {1,...,1} {
      \pgfmathtruncatemacro{\nextrow}{\row+1}
      \draw[->,thick,HighlightGray,bend left=15]  (c3s\nextrow3) to (c3s\row4);
    }
  \end{scope}
  \node at (11 + 3.25, -4.5 + 1) {$b_3 = \highlightGrayText{a_4} = f_3(a)$};

  \begin{scope}[shift={(0,-8.5)}, scale=0.8]
    \foreach \row in {1,...,3} {
      \foreach \col in {1,...,5} {
        \node[circle, draw, inner sep=2.5pt]
          (c4s\row\col) at (1 + \col, 5-\row) {};
      }
    }
    \node[circle, draw, inner sep=2.5pt, fill=LightGray] at (c4s15) {};
    \node[circle, draw, inner sep=2.5pt, fill=black] at (c4s24) {};

    \foreach \col in {1,...,5} {
      \node[above=10pt] at (c4s1\col) {$\col$};
    }

    \draw[->,thick,HighlightGray,bend left=15]  (c4s24) to (c4s15);
  \end{scope}
  \node at (0 + 3.25, -8.5 + 1) {$b_4 = \highlightGrayText{a_5} = f_4(a)$};

  \begin{scope}[shift={(5.5,-8.5)}, scale=0.8]
    \foreach \row in {1,...,3} {
      \foreach \col in {1,...,5} {
        \node[circle, draw, inner sep=2.5pt]
          (c5s\row\col) at (1 + \col, 5-\row) {};
      }
    }
    \node[circle, draw, inner sep=2.5pt, fill=LightGray] at (c5s22) {};
    \node[circle, draw, inner sep=2.5pt, fill=LightGray] at (c5s31) {};
    \node[circle, draw, inner sep=2.5pt, fill=LightGray] at (c5s33) {};  
    \node[circle, draw, inner sep=2.5pt, fill=LightGray] at (c5s23) {};
    \node[circle, draw, inner sep=2.5pt, fill=LightGray] at (c5s24) {};
    \node[circle, draw, inner sep=2.5pt, fill=LightGray] at (c5s14) {};
    \node[circle, draw, inner sep=2.5pt, fill=LightGray] at (c5s15) {};
    \node[circle, draw, inner sep=2.5pt, fill=black] at (c5s25) {};

    \foreach \col in {1,...,5} {
      \node[above=10pt] at (c5s1\col) {$\col$};
    }

    \draw[->,thick,HighlightGray,bend left=15]  (c5s24) to (c5s15);
    \draw[->,thick,HighlightGray,bend left=15]  (c5s23) to (c5s14);
    \draw[->,thick,HighlightBlue] (c5s25) to[bend right=20] (c5s22);
    \draw[->,thick,HighlightBlue] (c5s25) to[bend right=20] (c5s23);
    \draw[->,thick,HighlightOrange] (c5s31) to[bend left=25] (c5s33);
    \draw[->,thick,HighlightGray,bend left=15]  (c5s22) to (c5s31);
    \draw[->,thick,HighlightGray,bend left=15]  (c5s33) to (c5s24);
  \end{scope}
    \node[align=center] at (5.5 + 3.25, -8.5 + 0.75)
    { \small
      $b_5 = \highlightBlueText{\phi_1}(b_2,b_3)$  
      $= \highlightBlueText{\phi_1}(\highlightGrayText{c_1},\highlightGrayText{a_4})$ 
      $= \highlightBlueText{\phi_1}(\highlightOrangeText{\phi_2}(c_3),\highlightGrayText{a_4})$\\
      \small
      $= \highlightBlueText{\phi_1}(\highlightOrangeText{\phi_2}(\highlightGrayText{b_4}),\highlightGrayText{a_4})$
      $= \highlightBlueText{\phi_1}(\highlightOrangeText{\phi_2}(\highlightGrayText{a_5}),\highlightGrayText{a_4})$ 
      $= f_5(a)$
    };

\end{tikzpicture}
    \caption{Coordinate tracking using the functional constraints and the equality relations between coordinates in Example~\ref{ex:chi_bnd_more_interesting}. It shows how each \textbf{coordinate} of the second vertex $b$ of the directed path $(a,b,c)$ can be computed from the coordinates of its first vertex $a$ using the functional constraints \blueText{$\phi_1$}, \orangeText{$\phi_2$} and equalities between coordinates of adjacent vertices imposed by the literals \highlightGrayText{$q_{2,1}$}, \highlightGrayText{$q_{4,3}$}, \highlightGrayText{$q_{5,4}$}.}
    \label{fig:tracking-ex1}
\end{figure}

Example~\ref{ex:chi_bnd_more_interesting} reveals the most important ideas of the proof.
It remains to note that in full generality, we need to \emph{project} the given digraph $G \in \dcY_{P,Z}$, or more precisely its vertex set, to a subset of coordinates.

\begin{notation}
    Consider a non-empty $S \subseteq [d]$.
    Recall that $P_{|S}$ stands for the (path) clause obtained from $P$ by removing literals corresponding to variables $q_{i,j}$ with $\{i,j\} \not\subseteq S$ and define $Z_{|S}$ similarly to be the restriction of $Z$ to the set $S$, i.e.\ each $(L,\lambda) \in Z$ belongs to $Z_{|S}$ iff $L \subseteq S$ and $\lambda \in S$.
    Denote by $\dcY_{P,Z} \restriction S$ the class of realizations of the clause $P_{|S}$ on all (finite) $Z_{|S}$-functional vertex sets.
\end{notation}

The key lemma in full generality reads as follows.

\begin{lemma}\label{lem:first_determines_second}
    Suppose that the max-plus system $\widetilde{A}y \leq \widetilde{B}y$ has no finite solution. Then, there exists a non-empty set $S \subseteq [d]$ and an integer $2 \leq t \leq |S|+1$ such that for every $G \in \dcY_{P,Z} \restriction S$, there exists a function $f: \N^S \to \N^S$ that satisfies the following: 
    whenever $Q$ is a subdigraph of $G$ on vertices $v^1, \dots, v^t$ that is isomorphic to $\vec{P}_t$, we have $v^2 = f(v^1)$.
\end{lemma}

Let us remark that the function $f$ depends, in a sense, only on the path clause $P$ and the collection of functional constraints $Z$.
The proof of Lemma~\ref{lem:first_determines_second} shows that there is a uniform `blueprint' expression, which yields $f$ upon substituting the particular functions $\phi_s$ witnessing the functional constraints $(L_s,\lambda) \in Z$ in the vertex set of $G$.

With this, it is easy to prove that $\dcY_{P,Z}$ has bounded chromatic number.

\begin{observation}\label{obs:Lambda_not_subdigraph}
    Assuming Lemma~\ref{lem:first_determines_second}, the digraph $G \in \dcY_{P,Z} \restriction S$ does not contain $\vec{\Lambda}_t$ as a subdigraph.
\end{observation}
\begin{proof}
    For contradiction, consider a subdigraph of $G$ isomorphic to $\vec{\Lambda}_t$ on vertices $v^1, \dots, v^t$ and $u^1, \dots, u^t$ with $v^1 = u^1$.
    Since on both $v^1, \dots, v^t$ and $u^1, \dots, u^t$, we have a subgraph of $G$ isomorphic to $\vec{P}_t$, Lemma~\ref{lem:first_determines_second} implies that $v^2 = f(v^1) = f(u^1) = u^2$.
    However, this contradicts the definition of $\vec{\Lambda}_t$, which requires that $v^2$ and $u^2$ are distinct vertices.
\end{proof}

\begin{proof}[Proof of \ref{item:unbounded_chromatic_number} $\Rightarrow$~\ref{item:tropicals_have_solution} of Theorem~\ref{thm:tropical_dichotomy} assuming Lemma~\ref{lem:first_determines_second}]\label{prf:(i)->(iii)}
    Suppose that the max-plus system $\widetilde{A}y \leq \widetilde{B}y$ has no finite solution (we comment on the min-plus alternative in \cref{sssec:what_if_minplus_fails}) and let $S$ and $t$ be as in Lemma~\ref{lem:first_determines_second}.
    Consider an arbitrary digraph $G \in \dcY_{P,Z}$, which is the realisation of $P$ over a $Z$-functional set $V$, and denote by $G'$ the realization of $P_{|S}$ over $V_{|S}$. Note that $V_{|S}$ is a $Z_{|S}$-functional set, so $G' \in \dcY_{P,Z} \restriction S$.
    We claim that
    \[
        \chi(G) \leq \chi(G') \leq 2t-1
        .
    \]
    Indeed, the first inequality follows from the fact that $G$ is homomorphic to $G'$ (due to Lemma~\ref{lem:subclause}), while the second inequality is a consequence of Observation~\ref{obs:Lambda_not_subdigraph} (relying on Lemma~\ref{lem:first_determines_second}) and Fact~\ref{fact:no_Lambda_small_chromatic_number}.
\end{proof}

By the previous proof and the fact that $t \leq |S|+1 \leq d+1$ by Lemma~\ref{lem:first_determines_second}, we obtain the following quantitative corollary.

\begin{corollary}\label{cor:bound_on_chromatic_number}
    If the class $\dcY_{P, Z}$ has bounded chromatic number, then $\chi(G) \leq 2d+1$ for each $G \in \dcY_{P, Z}$.
\end{corollary}

Hence, the sole goal of the remainder of this section is to prove Lemma~\ref{lem:first_determines_second}.
We split the proof into several stages. 
First, in \cref{sssec:determining_set_S}, we determine the set $S \subseteq [d]$ using the connection between tropical algebra and mean payoff games; here we define important game digraphs $\Gamma$ and $\Gamma_{|S}$.
Next, in \cref{sssec:multidigraph_cQ}, we define a multidigraph $\cK$ on individual coordinates of the vertices of the directed path $Q$ whose edges show the relations among the coordinates, and state the conclusion of Lemma~\ref{lem:first_determines_second} in terms of these relations (see Observation~\ref{obs:trackable_vertices_are_function_of_first_level}).
Then, in \cref{sssec:relating_cQ_and_Gamma|S,sssec:utilizing_correspondence}, we establish a correspondence between the multidigraph $\cK$ (or rather its supergraph $\Theta$) with the game digraph $\Gamma_{|S}$, which allows us to further simplify the desired statement (see Observation~\ref{obs:trackable_iff_all_walks_in_cQOpt_reach_first_level}).
Finally, in \cref{sssec:finishing_proof}, we determine the value of $t$ and prove Lemma~\ref{lem:first_determines_second}.

\subsubsection{Determining the set $S$}\label{sssec:determining_set_S}

Recall the game digraph $\Gamma$ corresponding to the system $Ax \leq Bx$, where $A,B \in \Rmax^{m' \times n'}$.
It is a bipartite digraph with parts $\cC = [n']$ and $\cR = [m']$ indexed by columns of $A$ and rows of $B$ with edges between $\cC$ and $\cR$ weighted according to the entries of $A$ and $B$, respectively.

In our case, since the max-plus system $\widetilde{A}y \leq \widetilde{B}y$ with $\widetilde{A},\widetilde{B} \in \Rmax^{(2d-2m+k) \times d}$ from the previous section conveys an inherent interpretable meaning regarding the class $\dcY_{P,Z}$, we may also interpret the structure of the corresponding game digraph in greater detail.

Recall that both $\widetilde{A},\widetilde{B}$ are composed of two submatrices $\widetilde{A}^\shift, \widetilde{B}^\shift$ and $\widetilde{A}^\func, \widetilde{B}^\func$ of dimensions $(2d-2m) \times d$ and $k \times d$, respectively, given by \eqref{eq:defABhat}.
The matrices $\widetilde{A}^\shift, \widetilde{B}^\shift$ contain two rows for each positive literal in $P$ while $\widetilde{A}^\func, \widetilde{B}^\func$ contain a single row for each functional constraint in $Z$.
The entries of the vector $y$ and the columns of the matrices correspond to the coordinates of vertices of digraphs $G \in \dcY_{P,Z}$. 
These insights allow us to define the game digraph from scratch using only the path clause $P$ and the functional constraints $Z$ without mentioning the matrices $\widetilde{A}$ and $\widetilde{B}$.
We encourage the reader to verify that the following definition exactly corresponds to the original matrix-based definition of the game digraph.

\begin{definition}[Game digraph $\Gamma(P,Z)$]\label{def:game_graph}
    The bipartite (weighted) digraph $\Gamma = \Gamma(P,Z)$ has two parts $\cC = \cC(\Gamma)$ and $\cR = \cR(\Gamma)$, with $\cR$ consisting of two disjoint sets $\cR^\shift$ and $\cR^\func$.
   These are as follows:
    \begin{align*}
        \cC &= \{C_c : c \in [d]\},\\
        \cR^\shift &= \{ R^\shift_{c,+}, R^\shift_{c+1,-} : \text{for every positive literal $q_{c+1, c}$ of $P$}\},\\
        \cR^\func &= \{ R^\func_{s} : \text{for every $(L_s, \lambda_s) \in Z$}\}.
    \end{align*}
    For every positive literal $q_{c+1, c}$ of $P$, we have the following four edges
    \begin{alignat*}{2}
        C_c &\mapsto R^\shift_{c,+} &&\text{ with weight $0$}, \\
        R^\shift_{c,+} &\mapsto C_{c+1} &&\text{ with weight $-1$}, \\
        C_{c+1} &\mapsto R^\shift_{c+1,-} &&\text{ with weight $0$}, \\
        R^\shift_{c+1,-} &\mapsto C_c &&\text{ with weight $+1$}.
    \end{alignat*}
    Moreover, for each functional constraint $(L_s, \lambda_s) \in Z$, we have the edges
    \begin{alignat*}{2}
        C_{\lambda_s} &\mapsto R^\func_{s} &&\text{ with weight $0$}, \\
        R^\func_{s} &\mapsto C_c &&\text{ with weight $0$ \quad\quad\quad for every $c \in L_s$}.
    \end{alignat*}
\end{definition}

\begin{example}[Continuation of Example~\ref{ex:chi_bnd_more_interesting}]\label{ex:chi_bnd_more_interesting_matrices}
    Recall that Example~\ref{ex:chi_bnd_more_interesting} consists of a path clause with literals \highlightGrayText{$q_{2,1}$}, \highlightGrayText{$q_{4,3}$}, \highlightGrayText{$q_{5,4}$},  and has functional constraints $Z = (\blueText{(L_1, \lambda_1)}, \orangeText{(L_2, \lambda_2)})$ where \blueText{$L_1 = \{2,3\}, \lambda_1 = 5$}, \orangeText{$L_2 = \{ 3 \}, \lambda_2 = 1$}. 
    This corresponds to the system $\widetilde{A}y \leq \widetilde{B}y$ given by
    \begin{align*}
    y_2 &= y_1 + 1, \\
        y_4 &= y_3 + 1, \\
        y_5 &= y_4 + 1, \\
        y_5 &\leq \max\{y_2, y_3\}, \\
        y_1 &\leq \max\{y_3\},
    \end{align*}
    where the first three equalities come from the positive literals of $P$, while the last two inequalities come from the functional constraints $Z$.
    When written in a matrix form, we have
    \begin{align*}
        \begin{pmatrix}
        -\infty & 0 &-\infty &-\infty &-\infty \\
        0 &-\infty &-\infty &-\infty &-\infty \\
        -\infty &-\infty &-\infty & 0 &-\infty \\
        -\infty &-\infty & 0 &-\infty &-\infty \\
        -\infty &-\infty &-\infty &-\infty & 0 \\
        -\infty &-\infty &-\infty & 0 &-\infty \\
        -\infty &-\infty &-\infty & -\infty & 0 \\
        0 &-\infty &-\infty & -\infty & -\infty
        \end{pmatrix}
        \maxmu
        \begin{pmatrix}
        y_1 \\ y_2 \\ y_3 \\ y_4 \\ y_5
        \end{pmatrix}
         \leq
        \begin{pmatrix}
        1 &-\infty &-\infty &-\infty &-\infty \\
        -\infty & -1 &-\infty & -\infty &-\infty \\
        -\infty & -\infty & 1 & -\infty &-\infty \\
        -\infty &-\infty & -\infty & -1 &-\infty \\
        -\infty &-\infty & -\infty & 1 & -\infty \\
        -\infty &-\infty &-\infty & -\infty & -1 \\
        -\infty & 0 & 0 & -\infty & -\infty \\
        -\infty & -\infty & 0 & -\infty & -\infty
        \end{pmatrix}
        \maxmu
        \begin{pmatrix}
        y_1 \\ y_2 \\ y_3 \\ y_4 \\ y_5
        \end{pmatrix}
        .
    \end{align*}
    The corresponding game digraph $\Gamma$ is displayed in Figure~\ref{fig:game_graph-ex1}.
\end{example}

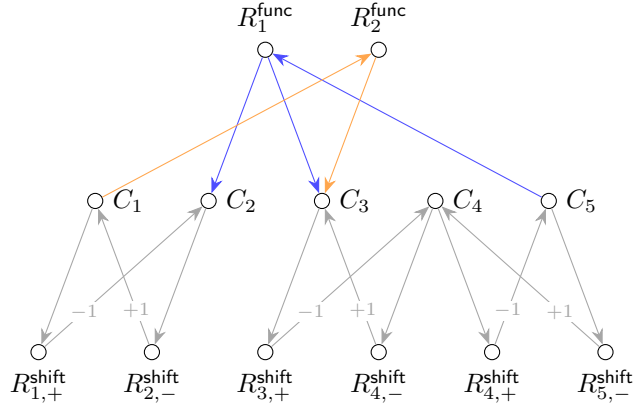
\begin{figure}[ht!]
    \centering
    \begin{tikzpicture}[
  dot/.style={circle,fill,inner sep=1.8pt},
  weight/.style={font=\scriptsize, fill=white, inner sep=1pt},
  >= {Stealth[scale=1.25]}
]

		\tikzset{
			circnode/.style={
				draw,
				circle,
				minimum size=0.2cm,
				inner sep=0pt
			}
		}

\def\Ybot{-2cm}
\def\Ymid{0cm}
\def\Ytop{2cm}

\def\hsep{1.5cm}

\node[circnode,label=right:$C_{1}$] (Cx1) at ({-2*\hsep},\Ymid) {};
\node[circnode,label=right:$C_{2}$] (Cx2) at ({-1*\hsep},\Ymid) {};
\node[circnode,label=right:$C_{3}$] (Cx3) at ({0*\hsep},\Ymid) {};
\node[circnode,label=right:$C_{4}$] (Cx4) at ({1*\hsep},\Ymid) {};
\node[circnode,label=right:$C_{5}$] (Cx5) at ({2*\hsep},\Ymid) {};

\node[circnode,label=below:$R^{\textsf{shift}}_{1,+}$] (Ls1) at ({-2.5*\hsep},\Ybot) {};
\node[circnode,label=below:$R^{\textsf{shift}}_{2,-}$] (Ls2) at ({-1.5*\hsep},\Ybot) {};
\node[circnode,label=below:$R^{\textsf{shift}}_{3,+}$] (Ls3) at ({-0.5*\hsep},\Ybot) {};
\node[circnode,label=below:$R^{\textsf{shift}}_{4,-}$] (Ls4) at ({0.5*\hsep},\Ybot) {};
\node[circnode,label=below:$R^{\textsf{shift}}_{4,+}$] (Ls5) at ({1.5*\hsep},\Ybot) {};
\node[circnode,label=below:$R^{\textsf{shift}}_{5,-}$] (Ls6) at ({2.5*\hsep},\Ybot) {};

\node[circnode,label=above:$R^{\textsf{func}}_{1}$] (Rf1) at ({-0.5*\hsep},\Ytop) {};
\node[circnode,label=above:$R^{\textsf{func}}_{2}$] (Rf2) at ({0.5*\hsep},\Ytop) {};

\newcommand{\edge}[4]{%
  \ifnum#3=0
    \draw[->] (#1) -- (#2);
  \else
    \draw[->] (#1) -- node[pos=#4, weight] {$#3$} (#2);
  \fi
}

\newcommand{\gedge}[4]{%
  \ifnum#3=0
    \draw[LightGray,->] (#1) -- (#2);
  \else
    \draw[LightGray,->] (#1) -- node[pos=#4, weight] {$#3$} (#2);
  \fi
}

\def\pos{0.25}

\gedge{Cx1}{Ls1}{0}{\pos}
\gedge{Cx2}{Ls2}{0}{\pos}
\gedge{Cx3}{Ls3}{0}{\pos}
\gedge{Cx4}{Ls4}{0}{\pos}
\gedge{Cx4}{Ls5}{0}{\pos}
\gedge{Cx5}{Ls6}{0}{\pos}

\gedge{Ls1}{Cx2}{-1}{\pos}
\gedge{Ls2}{Cx1}{+1}{\pos}
\gedge{Ls3}{Cx4}{-1}{\pos}
\gedge{Ls4}{Cx3}{+1}{\pos}
\gedge{Ls5}{Cx5}{-1}{\pos}
\gedge{Ls6}{Cx4}{+1}{\pos}

\draw[->, LightBlue] (Cx5) -- (Rf1);
\draw[->, LightBlue] (Rf1) -- (Cx2);
\draw[->, LightBlue] (Rf1) -- (Cx3);

\draw[->, LightOrange] (Cx1) -- (Rf2);
\draw[->, LightOrange] (Rf2) -- (Cx3);

\end{tikzpicture}
    \caption{Game digraph $\Gamma$ from Example~\ref{ex:chi_bnd_more_interesting_matrices}. Only non-zero weights are displayed.}
    \label{fig:game_graph-ex1}
\end{figure}

Next, we want to exploit the connection between systems of tropical inequalities and mean payoff games from Section~\ref{sec:tropical-prelim}, in particular Theorem~\ref{thm:connection_of_tropical_and_mean_payoff}.
We need to point out, though, that the system $\widetilde{A}y \leq \widetilde{B}y$ might not necessarily satisfy Assumption~\ref{assump:valid_moves_exist}.
However, this is only a minor technical issue which is safe to ignore as we discuss in Remark~\ref{rem:system_assumption_might_be_violated} below.

Assuming that the system $\widetilde{A}y \leq \widetilde{B}y$ has no finite solution, Corollary~\ref{cor:no_solution_implies_negative_mean_payoff} states that there is a non-empty set $X$ of states winning for the column player (in a game on $\Gamma$), that is $X = \{s \in V(\Gamma) : (\Gamma, s) < 0\}$, which is witnessed by a positional strategy $\sigma^\opt$.
Let $S = X \cap \cC(\Gamma)$.
Abusing notation, this is the set $S$ from Lemma~\ref{lem:first_determines_second}, using the natural bijection between $\cC$ and $[d]$.

We consider the analogously defined digraph $\Gamma_{|S} = \Gamma(P_{|S}, Z_{|S})$ for the class $\dcY_{P,Z} \restriction S$.
That is, we have $\cC(\Gamma_{|S}) = S$, while $\cR(\Gamma_{|S})$ and the edges of $E(\Gamma_{|S})$ are defined by $P_{|S}$ and $Z_{|S}$ in place of $P$ and $Z$, respectively.

\begin{lemma}\label{lem:all_states_in_Gamma|S_are_winning}
    The digraphs $\Gamma_{|S}$ and $\Gamma[X]$ are isomorphic as weighted digraphs.
    Therefore, the restriction $\sigma^\opt_{|S} : S \to \cR$ of the strategy $\sigma^\opt$ witnesses that the column player wins the game $(\Gamma_{|S}, s_0)$ for all $s_0 \in V(\Gamma_{|S})$.
\end{lemma}
\begin{proof}
    Regarding the isomorphism, since both are induced subgraphs of $\Gamma$, we only need to verify that $V(\Gamma_{|S}) = V(\Gamma[X])$.
    That is, that $V(\Gamma_{|S})$ is exactly the set $X$ of winning states for the column player in a game on $\Gamma$.
    In particular, we only need to focus on states of the row player as $\cC(\Gamma_{|S}) = S = X \cap \cC(\Gamma)$ by definition.

    Indeed, for $s \in \cR(\Gamma_{|S})$, we have $\nu(\Gamma, s) < 0$ as all the out-edges of $s$ end in vertices of $S$; thus, $s \in X$.
    On the other hand, a state $s \in \cR(\Gamma) \setminus \cR(\Gamma_{|S})$ has an in-edge from or an out-edge to a state $s' \in \cC(\Gamma) \setminus S$.
    If all out-neighbors of $s$ belong to $S$ (hence, $s'$ is an in-neighbor of $s$) then $\nu(\Gamma, s) < 0$ as above.
    However, the edge $(s',s)$ witnesses that $\nu(\Gamma,s') \leq \nu(\Gamma,s)$ as traversing $(s',s)$ is one of column player's options in the first round of the game $(\Gamma,s')$.
    Therefore, $\nu(\Gamma,s') < 0$, which contradicts that $s'  \not\in S$.
    
    Thus, we may assume that $s$ has an out-neighbor $s' \in \cC(\Gamma) \setminus S$.
    As before (up to exchanging the minimizing column player by the maximizing row player), the edge $(s,s')$ witnesses $\nu(\Gamma, s) \geq \nu(\Gamma, s')$ since using the edge $(s,s')$ in the first round of the game $(\Gamma, s)$ is one of row player's options.
    Moreover, as $s' \not\in S$, we have $\nu(\Gamma, s') \geq 0$.
    Joining these inequalities, we conclude that $\nu(\Gamma, s) \geq 0$; thus, $s \not\in X$.
    
    The ``Therefore'' part is a direct consequence of Observation~\ref{obs:restricted_game_graph}.
\end{proof}

\begin{example}[Continuation of Example~\ref{ex:chi_bnd_more_interesting_matrices}]\label{ex:chi_bnd_more_interesting_set_S}
    With the setting from Example~\ref{ex:chi_bnd_more_interesting}, one may verify that all states of $\Gamma$ are winning for the column player.
    Thus, $X = V(\Gamma)$ and $S = \cC(\Gamma)$, and we have $\Gamma_{|S} = \Gamma$.
\end{example}

\begin{remark}\label{rem:system_assumption_might_be_violated}
    As mentioned above, the system $\widetilde{A}y \leq \widetilde{B}y$ does not necessarily satisfy \cref{assump:valid_moves_exist}.
    We first note that the part of \cref{assump:valid_moves_exist} requiring each row of $\widetilde{B}$ to contain a finite entry is always satisfied.
    This is clear for the rows of $\widetilde{B}^\shift$, while for the rows of $\widetilde{B}^\func$ this follows from \cref{rem:each_discrete_clause_has_a_loop}.
    However, the other part of \cref{assump:valid_moves_exist} requiring
    that each column of $\widetilde{A}$ contains a finite entry is not necessarily satisfied.
    This is a merely technical issue, which we address by the following correction.%
    \footnote{A similar technical correction is used in \cite{akian2012tropical}, see the discussion following Assumptions~2.1~and~2.2 in~\cite{akian2012tropical}.}
    
    For each $c \in [d]$, we add the trivial inequality $y_c \leq y_c$, expanding the matrices by $d$ rows and ensuring that Assumption~\ref{assump:valid_moves_exist} is satisfied.
    Such a modification has no influence on the solution set of the system.
    In particular, the modified system $A' y \leq B' y$ has no finite solution whenever $\widetilde{A}y \leq \widetilde{B}y$ has none.
    
    The game digraph $\Gamma' = \Gamma(A', B')$ contains for each $c \in [d]$ an additional vertex $R'_c$ and two edges $C_c \mapsto R'_c, R'_c \mapsto C_c$ with weight $0$.
    However, when focusing on the behavior of an optimal strategy on the set of winning states for the column player, the additional vertices play no role.
    Indeed, let $X' \subseteq V(\Gamma')$ be the set of states of $\Gamma'$ winning for the column player (in a game on $\Gamma'$) and let $\sigma'$ be the witnessing strategy.
    Clearly, $\sigma'(C_c) \not= R'_c$ for each $C_c \in X'$ as otherwise $\nu(\Gamma', C_c) = 0$ because the token would only move between $C_c$ and $R'_c$.
    Therefore, each state $C_c \in S' = X' \cap \cC(\Gamma')$ has at least one out-neighbor other than $R'_c$.
    Thus, the additional vertices $R'_c$ become irrelevant once we pass to the digraph $\Gamma'_{|S'}$ and we may remove them (retaining the property that each vertex of the resulting game digraph has an out-neighbor).
\end{remark}

\subsubsection{Constructing the multidigraph $\cK$}\label{sssec:multidigraph_cQ}

Generalizing the construction of the digraph in the top right of Figure \ref{fig:tracking-ex1} from the directed path $(a,b,c)$, we define the multidigraph $\cK$ on individual coordinates of the vertices $v^1, \dots, v^t$ of a directed path $Q$ whose edges capture the relations among the coordinates.
We refer to the vertices of $\cK$ as \emph{positions} to distinguish them from the vertices of $Q$.

\begin{definition}[Coordinate multidigraph $\cK$]\label{def:multidigraph_cK}
    We define the multidigraph $\cK$ as follows.
    We set
    \[
        V(\cK) = \{ v^\ell_c : c \in S, \ell \in [t]\}.
    \]
    The edges of $\cK$ come from two sources: the \emph{shift} edges for the positive literals of $P_{|S}$, and the \emph{functional} edges for the individual functional constraints from $Z_{|S}$.
    We define the following \emph{bunches} of edges
    \begin{align*}
        \Esh_{\ell, c, +}(\cK) &= \{v^\ell_c \mapsto v^{\ell-1}_{c+1}\} \quad &\text{for every positive literal $q_{c+1, c}$ of $P_{|S}$ and $\ell \in [2,t]$,}\\
        \Esh_{\ell, c+1, -}(\cK) &= \{v^\ell_{c+1} \mapsto v^{\ell+1}_c\} \quad &\text{for every positive literal $q_{c+1, c}$ of $P_{|S}$ and $\ell \in [t-1]$,}\\
        \Efn_{\ell, s}(\cK) &= \{v^\ell_{\lambda_s} \mapsto v^{\ell}_c : c \in L_s\} \quad &\text{for every $(L_s, \lambda_s) \in Z_{|S}$ and $\ell \in [t]$.}
    \end{align*}
    The edge-set of $\cK$ is the \emph{disjoint} union of all these bunches, creating a multidigraph (see Remark~\ref{rem:multidigraph} below).
    
    The \emph{level} of a position $u \in V(\cK)$ is the integer $\ell = \ell(u)$ such that $u = v^\ell_c$ for some $c \in S$.
    The \emph{weight} $w(e)$ of an edge $e$ from $u$ to $w$ is defined as $\ell(w) - \ell(u)$.
    The common source vertex of all edges in a bunch is the \emph{center} of the bunch.
\end{definition}

\begin{remark}\label{rem:multidigraph}
    Each bunch has a common source vertex, the center of the bunch, and distinct target vertices.
    Hence, within a bunch, an edge may be identified solely by its source and target vertex.
    However, two or more bunches may contain an edge with the same pair of source and target vertices.%
    \footnote{For example, let $Z$ contain the functional constraints $(\{2,3\},1)$ and $(\{3,4\},1)$. Then the bunches for both of these constraints contain an edge from $v^\ell_1$ to $v^\ell_3$ (for any given $\ell \in [t]$).}
    The disjoint union of bunches then creates a multidigraph, where each of these individual edges is present.
    Thus, an edge in $\cK$ may be identified by its source and target vertices, \emph{and} the bunch where it belongs.
\end{remark}

\begin{example}[Continuation of Example~\ref{ex:chi_bnd_more_interesting_set_S}]\label{ex:chi_bnd_more_interesting_bunches}
    With the setting from Example~\ref{ex:chi_bnd_more_interesting}, we eventually determine that $t = 3$.
    The multidigraph $\cK$ is then exactly the graph from the top of Figure~\ref{fig:tracking-ex1}.
    Each grey shift edge forms a singleton bunch, while the functional edges coming from the same functional constraint form a bunch at any given level.
    There are exactly 18 bunches in this digraph.
\end{example}

Our primary interest is to prove that we may express every coordinate of the vertex $v^2 \in V(Q)$ as a function of the vertex $v^1 \in V(Q)$.
This motivates the following inductive definition, which illuminates the significance of bunches of $\cK$.

\begin{definition}[Trackable positions]\label{def:trackable_vertices}
    Let $T$ be the minimal subset of $V(\cK)$ satisfying both:
    \begin{enumerate}[label=(\roman*)]
        \item all the positions at the first level of $V(\cK)$ belong to $T$,
        \item if there is a bunch $B$ centered in a position $u$ such that for each $(u,w) \in B$ it holds $w \in T$, then $u$ belongs to $T$.
    \end{enumerate}
    Members of $T$ are called \emph{trackable} positions.
\end{definition}

Note that we may construct the set of trackable positions $T$ by setting $T_0 = \{v_c^1 : c \in S\}$ and then iteratively adding positions that satisfy the premise of the second property with respect to the current set.
This construction implies that trackability of a position can be validated by a finite certificate, which we use in the following proof.

\begin{observation}\label{obs:trackable_vertices_are_function_of_first_level}
    If $u \in \cK$ is trackable, it can be expressed as a function of $v^1$.
\end{observation}
\begin{proof}
    We proceed by induction on the nesting depth of a certificate that $u$ is trackable.
    As a base case, if $u$ is at level $1$, it can be expressed by the identity function of the corresponding coordinate.
    So, for the induction step, suppose that there is a bunch $B$ centered in $u$ such that for each $(u,w) \in B$ the position $w$ is trackable.
    If $B$ is a singleton bunch containing a shift edge to a position $w = f(v^1)$, where $f$ is from the induction hypothesis, then we take the same function for $u$ as $u = w = f(v^1)$.
    If $B$ is the bunch corresponding to a functional constraint $(L_s, \lambda_s)$, we apply the witness $\phi_s$ for $(L_s, \lambda_s)$ to the expressions of the out-neighbor positions.
    That is, let $f_c(v^1)$ be the expression for the out-neighbor position $w_c, c \in L_s$ from the induction hypothesis, then
    \[
        u 
        = \phi_s\big(w_c :c \in L_s\big)
        = \phi_s\big(f_c(v^1) : c \in L_s\big)
        , 
    \]
    as claimed. 
\end{proof}

\subsubsection{Relating $\cK$ and $\Gamma_{|S}$}\label{sssec:relating_cQ_and_Gamma|S}

We start by embedding $\cK$ into a more convenient infinite multidigraph $\Theta$ which captures the relations between coordinates of a \emph{doubly infinite path} induced by $P_{|S}$ on a $Z_{|S}$-functional vertex set.
The advantage of working with $\Theta$ compared to $\cK$ is that all levels are the same, so we do not have to treat the first and last level separately.

\begin{definition}[Coordinate multidigraph $\Theta$]
    We define a multidigraph $\Theta$ as follows.
    We set
    \[
        V(\Theta) = \{ v^\ell_c : c \in S, \ell \in \Z\}.
    \]
    The edge set of $\Theta$ is the disjoint union of the following bunches (see Remark~\ref{rem:multidigraph})
    \begin{align*}
        \Esh_{\ell, c, +}(\Theta) &= \{v^\ell_c \mapsto v^{\ell-1}_{c+1}\} \quad &\text{for every positive literal $q_{c+1, c}$ of $P_{|S}$ and $\ell \in \Z$,}\\
        \Esh_{\ell, c+1, -}(\Theta) &= \{v^\ell_{c+1} \mapsto v^{\ell+1}_c\} \quad &\text{for every positive literal $q_{c+1, c}$ of $P_{|S}$ and $\ell \in \Z$,}\\
        \Efn_{\ell, s}(\Theta) &= \{v^\ell_{\lambda_s} \mapsto v^{\ell}_c : c \in L_s\} \quad &\text{for every $(L_s, \lambda_s) \in Z_{|S}$ and $\ell \in \Z$.}
    \end{align*}
    The definitions of the level of a position of $\Theta$, the weight of an edge, and the center of a bunch are as in Definition~\ref{def:multidigraph_cK}.
\end{definition}

\begin{example}[Continuation of Example~\ref{ex:chi_bnd_more_interesting_bunches}]\label{ex:chi_bnd_more_interesting_Theta}
    For the setting from Example~\ref{ex:chi_bnd_more_interesting}, the multidigraph $\Theta$ is displayed in Figure~\ref{fig:Theta_ex1}.
\end{example}

\begin{figure}[ht!]
    \centering
        \begin{tikzpicture}[
		>= {Stealth[scale=1.25]}]
	
	\begin{scope}[shift={(5,0)}, scale=1]

		\tikzset{
			circnode/.style={
				draw,
				circle,
				minimum size=0.2cm,
				inner sep=0pt
			}
		}
		
			\def\xpos{1}
		\def\ypos{8.5}
		\draw (\xpos+0.15,\ypos-0.15) -- (\xpos-0.15,\ypos+0.15);
		\node at (\xpos+0.18,\ypos+0.18) {$c$};
		\node at (\xpos-0.18,\ypos-0.18) {$\ell$};

		\foreach \col in {1,...,5} {
			\node[font=\normalsize] at (1+\col, 8.5) {$\col$};
		}

		\foreach \row in {-2,-1,0,1,2}{
			\node[font=\normalsize] at (1, 5-\row) {$\row$};
		}
		
		\foreach \row in {-2,-1,0,1,2} {
			\foreach \col in {1,...,5} {
				\node[circnode
				] (s\row\col) at (1 + \col, 5-\row) {};
			}
		}
		
		\foreach \col in {1,...,5} {
			\node[font=\large] at (1+\col, 8) {$\vdots$};
			\node[font=\large] at (1+\col, 2.5) {$\vdots$};
		}
		
		\foreach \row in {-2,...,1} {
			\pgfmathtruncatemacro{\nextrow}{\row+1}
			\draw[->,LightGray,bend left=15]  (s\row2)     to (s\nextrow1);
			\draw[->,LightGray,bend left=15]  (s\nextrow1) to (s\row2);
			
			\draw[->,LightGray,bend left=15]  (s\row5)     to (s\nextrow4);
			\draw[->,LightGray,bend left=15]  (s\nextrow4) to (s\row5);
			
			\draw[->,LightGray,bend left=15]  (s\row4)     to (s\nextrow3);
			\draw[->,LightGray,bend left=15]  (s\nextrow3) to (s\row4);
		}
		
		\foreach \row in {-2,...,2} {
			\draw[->,LightBlue]   (s\row5) to[bend right=30] (s\row2);
			\draw[->,LightBlue]   (s\row5) to[bend right=30] (s\row3);
			\draw[->,LightOrange] (s\row1) to[bend left=60]  (s\row3);
		}
	\end{scope}
	
\end{tikzpicture}
    \caption{The multidigraph $\Theta$ from Example~\ref{ex:chi_bnd_more_interesting_Theta}. The vertex in row $\ell$ and column $c$ is $v_{c}^\ell$.}
    \label{fig:Theta_ex1}
\end{figure}
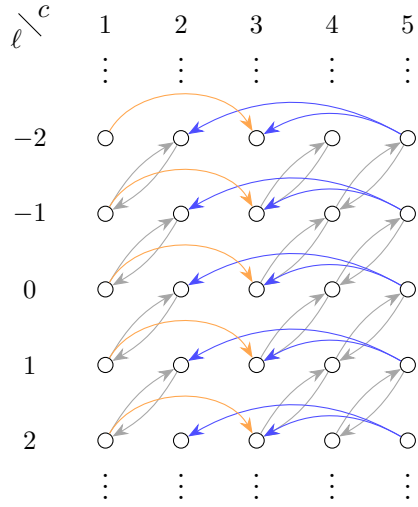

Clearly, the identity is an embedding of $\cK$ to $\Theta$.
While trivial, let us explicitly mention that the embedding respects levels, weights of edges, and membership to bunches.

We want to relate the (multi)digraphs $\Theta$ and $\Gamma_{|S}$.
To do so, we define another multidigraph $\Phi$ that will serve as an intermediary.
In short, we define $\Phi$ as the second power of $\Gamma_{|S}$ restricted to the vertices from $\cC(\Gamma_{|S})$.
More formally, we have $V(\Phi) = \cC(\Gamma_{|S})$ and an edge $ee'$ from $u \in \cC(\Gamma_{|S})$ to $v \in \cC(\Gamma_{|S})$ for every directed $2$-edge path $(e,e')$ from $u$ to $v$, where $e = u \mapsto r$ and $e' = r \mapsto v$ for some $r \in \cR(\Gamma_{|S})$.
Note that the middle vertex $r$ connecting $e$ and $e'$ is necessarily from $\cR(\Gamma_{|S})$ as $\Gamma_{|S}$ is bipartite.
We set the weight of $ee' \in E(\Phi)$ to $w(e) + w(e')$. We remark that this graph is described in \cite[Chapter 9.4]{Joswig2021}, where it is referred to as the \emph{reduced graph}.

However, as in the case of $\Gamma_{|S}$, we may define $\Phi$ directly, referring only to the clause $P_{|S}$ and constraints $Z_{|S}$.

\begin{definition}[Multidigraph $\Phi$]\label{def:multidigraph_Phi}
    We set $V(\Phi) = \cC(\Gamma_{|S})$.
    
    That is, the multidigraph $\Phi$ has the following bunches
    \begin{align*}
        \Esh_{c, +}(\Phi) &= \{C_c \mapsto C_{c+1} \text{ with weight $-1$}\} \qquad\qquad  &\text{for every positive literal $q_{c+1, c}$ of $P_{|S}$,}\\
        \Esh_{c+1, -}(\Phi) &= \{C_{c+1} \mapsto C_c \text{ with weight $+1$}\} \quad &\text{for every positive literal $q_{c+1, c}$ of $P_{|S}$,}\\
        \Efn_{s}(\Phi) &= \{C_{\lambda_s} \mapsto C_{c} : c \in L_s \text{ with weight $0$}\} \quad &\text{for every $(L_s,\lambda_s) \in Z_{|S}$.}
    \end{align*}
    The edge-set of $\Phi$ is the disjoint union of these bunches.
    Again, we say that the common source vertex of all edges in a bunch is the center of the bunch.
\end{definition}

\begin{remark}\label{rem:bunches_in_Phi}
    The edges of $\Phi$ can be viewed as pairs of edges $ee'$ from $\Gamma_{|S}$.
    The partition of $E(\Phi)$ into bunches is defined per the vertex of $\cR(\Gamma_{|S})$ joining $e$ and $e'$.
\end{remark}

\begin{example}[Continuation of Example~\ref{ex:chi_bnd_more_interesting_Theta}]\label{ex:chi_bnd_more_interesting_Phi}
    For the setting from Example~\ref{ex:chi_bnd_more_interesting}, the multidigraph $\Phi$ is displayed in Figure~\ref{fig:second_power-ex1}.
\end{example}

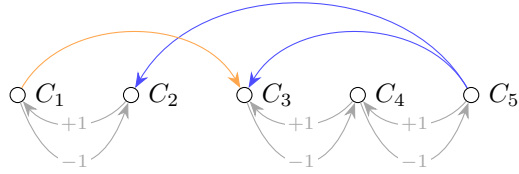
\begin{figure}[ht!]
    \centering
    	\begin{tikzpicture}[
		dot/.style={circle,fill,inner sep=1.8pt},
		weight/.style={font=\scriptsize, fill=white, inner sep=1pt},
		>= {Stealth[scale=1.25]}
		]

        \tikzset{
			circnode/.style={
				draw,
				circle,
				minimum size=0.2cm,
				inner sep=0pt
			}
		}

		\def\hsep{1.5cm}
		
		\node[circnode,label=right:$C_{1}$] (Cx1) at ({-2*\hsep},0) {};
		\node[circnode,label=right:$C_{2}$] (Cx2) at ({-1*\hsep},0) {};
		\node[circnode,label=right:$C_{3}$] (Cx3) at ({0*\hsep},0) {};
		\node[circnode,label=right:$C_{4}$] (Cx4) at ({1*\hsep},0) {};
		\node[circnode,label=right:$C_{5}$] (Cx5) at ({2*\hsep},0) {};
		
	\def\looseN{2}   
	\def\looseP{1}   
	\def\looseR{2.2}   
	
	\def\pos{0.5}
	
	\newcommand{\edgegray}[4]{%
		\ifnum#3=-1
		\draw[LightGray, ->, out=290, in=250, looseness=\looseN]
		(#1) to node[pos=#4, weight] {$-1$} (#2);
		\else
		\draw[LightGray, ->, out=230, in=310, looseness=\looseP]
		(#1) to node[pos=#4, weight] {$+1$} (#2);
		\fi
	}
		
		\newcommand{\edgetop}[4]{%
			\draw[->, out=120, in=60, looseness=\looseR] (#1) to (#2);
		}
		
		\newcommand{\oedgetop}[4]{%
			\draw[LightOrange, ->, out=60, in=120] (#1) to (#2);
		}
		
		\newcommand{\bedgetop}[4]{%
			\draw[LightBlue, ->, out=120, in=60] (#1) to (#2);
		}
		
		\edgegray{Cx1}{Cx2}{-1}{\pos}
		\edgegray{Cx2}{Cx1}{+1}{\pos}
		\edgegray{Cx3}{Cx4}{-1}{\pos}
		\edgegray{Cx4}{Cx3}{+1}{\pos}
		\edgegray{Cx4}{Cx5}{-1}{\pos}
		\edgegray{Cx5}{Cx4}{+1}{\pos}
		
		\bedgetop{Cx5}{Cx2}{0}{\pos}
		\bedgetop{Cx5}{Cx3}{0}{\pos}
		\oedgetop{Cx1}{Cx3}{0}{\pos}
		
	\end{tikzpicture}
    \caption{The multidigraph $\Phi$ from Example~\ref{ex:chi_bnd_more_interesting_Phi}. Only non-zero weights are displayed. }
    \label{fig:second_power-ex1}
\end{figure}

The relation between $\Gamma_{|S}$ and $\Phi$ is clear from the definition of $\Phi$ (or rather the discussion preceding Definition~\ref{def:multidigraph_Phi}).
To relate $\Theta$ and $\Phi$, we show that the following `level-squeezing' mapping from $V(\Theta) \cup E(\Theta)$ to $V(\Phi) \cup E(\Phi)$ is an out-covering projection (see Section~\ref{sec:covering-projections}).

\begin{definition}[Out-covering projection $\zeta$]\label{def:zeta}
    Consider a function $\zeta: V(\Theta) \cup E(\Theta) \to V(\Phi) \cup E(\Phi)$  such that the vertices are mapped as
    \[
        v_c^\ell \text{ maps to } C_c \quad \text{for every $c \in S$ and $\ell \in \Z$},  
    \]
    The mapping of edges is defined per bunches. 
    We require that
    \begin{align*}
        \Esh_{\ell,c,+}(\Theta) & \text{ maps to } \Esh_{c,+}(\Phi) \quad \text{for every $c \in S$ for which it makes sense and $\ell \in \Z$.}\\
        \Esh_{\ell,c+1,-}(\Theta) & \text{ maps to } \Esh_{c+1,-}(\Phi) \quad \text{for every $c \in S$ for which it makes sense and $\ell \in \Z$,}\\
        \Efn_{\ell,s}(\Theta) & \text{ maps to } \Efn_{s}(\Phi) \quad \text{for every $(L_s,\lambda_s) \in Z_{|S}$ and $\ell \in \Z$.}
    \end{align*}
    Specifically, the mapping from $\Efn_{\ell,s}(\Theta)$ to $\Efn_{s}(\Phi)$ is as follows
    \[
        v^\ell_{\lambda_s} \mapsto v^{\ell}_c 
        \text{ maps to }
        C_{\lambda_s} \mapsto C_c 
        \quad
        \text{for every $c \in L_s$ and $\ell \in \Z$}.
    \]
\end{definition}

Note that the mapping between bunches of shift edges need not be specified further as each such bunch contains only a single edge.

We say that $\zeta$ \emph{respects bunches} if each bunch of $\Theta$ is bijectively mapped to a bunch of $\Phi$.

\begin{observation}\label{lem:level_squeeze_is_covering_projection}
    The mapping $\zeta$ is an out-covering projection from $\Theta$ to $\Phi$ that respects bunches and preserves edge-weights.
\end{observation}
\begin{proof}
    Follows from a direct comparison of definitions of $\Theta$ and $\Phi$.
\end{proof}

\subsubsection{Utilizing the correspondence}\label{sssec:utilizing_correspondence}

Having established a relationship between $\Theta$ and $\Gamma_{|S}$ via $\Phi$, we may modify the digraph $\Gamma_{|S}$ and transfer the modifications to the multidigraph $\Theta$ and further to the multidigraph $\cK$ of our main interest.

Recall the winning positional strategy $\sigma_{|S}^\opt$ for the column player from Lemma~\ref{lem:all_states_in_Gamma|S_are_winning}.
Since the strategy $\sigma_{|S}^\opt$ is positional, we can capture it by a subdigraph $\Gamma_{|S}^\opt$ of $\Gamma_{|S}$ on the same vertex set.

\begin{definition}[Optimal game digraph $\Gamma_{|S}^\opt$]
    The spanning subdigraph $\Gamma_{|S}^\opt$ of $\Gamma_{|S}$ is defined by setting
    \begin{align*}
        E(\Gamma_{|S}^\opt) &= \{ C_c \mapsto \sigma_{|S}^\opt(C_c) : c \in S \} \cup \Big( (\cR(\Gamma_{|S}) \times \cC(\Gamma_{|S}))  \cap E(\Gamma_{|S}) \Big).
    \end{align*}
    In words, for each vertex $C_c \in \cC(\Gamma_{|S})$, we keep in $\Gamma_{|S}^\opt$ only the edge $C_c \mapsto \sigma_{|S}^\opt(C_c)$, which is the outgoing edge that the column player chooses according to $\sigma_{|S}^\opt$ when the current state of the game is $C_c$.
    The edges from vertices of $\cR(\Gamma_{|S})$ remain in $\Gamma_{|S}^\opt$ the same as in $\Gamma_{|S}$.
    Edge-weights are inherited from $\Gamma_{|S}$.
\end{definition}

We propagate this restriction further to the multidigraph $\Phi$ by considering its subdigraph%
\footnote{A consequence of Observation~\ref{obs:single_bunch_remains_in_opt} is that $\Phi^\opt$ does not have parallel edges.}
$\Phi^\opt$ of $\Phi$ taken as the second power of $\Gamma_{|S}^\opt$ restricted to $\cC(\Gamma_{|S}^\opt)$.

\begin{definition}[Optimal digraph $\Phi^\opt$]\label{def:multidigraph_Phi^opt}
    Representing each edge of $\Phi$ as a pair of edges from $\Gamma_{|S}$, the spanning subdigraph $\Phi^\opt$ of $\Phi$ is obtained by setting
    \[
        E(\Phi^\opt) = \{ee' \in E(\Phi) : e \in E(\Gamma_{|S}^\opt)\}.
    \]
    Edge-weights and partition of edges to bunches in $\Phi^\opt$ are inherited from $\Phi$.
\end{definition}

Note that for $ee' \in E(\Phi)$, the edge $e' \in E(\Gamma_{|S})$ from $\cR(\Gamma_{|S})$ to $\cC(\Gamma_{|S})$ \emph{always} belongs to $E(\Gamma_{|S}^\opt)$ (as the edges from $\cR(\Gamma_{|S})$ to $\cC(\Gamma_{|S})$ remain unchanged in $\Gamma_{|S}^\opt$).
Thus, $\Phi^\opt$ is indeed the second power of $\Gamma_{|S}^\opt$ restricted to $\cC(\Gamma_{|S}^\opt)$.

Finally, we transfer this restriction of edges to $\Theta$, creating a subdigraph $\Theta^\opt$, via the out-covering projection~$\zeta$.

\begin{definition}[Optimal coordinate digraph $\Theta^\opt$]
    The spanning subdigraph $\Theta^\opt$ of $\Theta$ is obtained by setting
    \[
        E(\Theta^\opt) = \{e \in E(\Theta) : \zeta(e) \in E(\Phi^\opt) \}.
    \]
    Edge-weights and partition of edges to bunches in $\Theta^\opt$ are inherited from $\Theta$.
\end{definition}

We accordingly obtain the digraph $\cK^\opt$, a spanning subdigraph of $\cK$, as the subdigraph of $\Theta^\opt$ induced by $V(\cK)$.

\begin{example}[Continuation of Example~\ref{ex:chi_bnd_more_interesting_Phi}]
\label{ex:chi_bnd_more_interesting_restricted_graphs}
For the setting from Example~\ref{ex:chi_bnd_more_interesting}, the digraphs $\Gamma_{|S}^\opt, \Phi^\opt$, and $\Theta^\opt$ are displayed in Figure~\ref{fig:combinedOPT_example}.
\end{example}

\begin{figure}[ht!]
    \centering 
 

\begin{tikzpicture}[  dot/.style={circle,fill,inner sep=1.8pt},
	weight/.style={font=\scriptsize, fill=white, inner sep=1pt},
	>= {Stealth[scale=1.25]}]  

\tikzset{
	circnode/.style={
		draw,
		circle,
		minimum size=0.20cm,
		inner sep=0pt
	}
}

\begin{scope}[shift={(-1,5.25)}, scale=1]
 
 \def\Ybot{-2cm}
 \def\Ymid{0cm}
 \def\Ytop{2cm}
 
 \def\hsep{1.5cm}
 
 \node[circnode,label=right:$C_{1}$] (Cx1) at ({-2*\hsep},\Ymid) {};
 \node[circnode,label=right:$C_{2}$] (Cx2) at ({-1*\hsep},\Ymid) {};
 \node[circnode,label=right:$C_{3}$] (Cx3) at ({0*\hsep},\Ymid) {};
 \node[circnode,label=right:$C_{4}$] (Cx4) at ({1*\hsep},\Ymid) {};
 \node[circnode,label=right:$C_{5}$] (Cx5) at ({2*\hsep},\Ymid) {};
 
 \node[circnode,label=below:$R^{\textsf{shift}}_{1,+}$] (Ls1) at ({-2.5*\hsep},\Ybot) {};
 \node[circnode,label=below:$R^{\textsf{shift}}_{2,-}$] (Ls2) at ({-1.5*\hsep},\Ybot) {};
 \node[circnode,label=below:$R^{\textsf{shift}}_{3,+}$] (Ls3) at ({-0.5*\hsep},\Ybot) {};
 \node[circnode,label=below:$R^{\textsf{shift}}_{4,-}$] (Ls4) at ({0.5*\hsep},\Ybot) {};
 \node[circnode,label=below:$R^{\textsf{shift}}_{4,+}$] (Ls5) at ({1.5*\hsep},\Ybot) {};
 \node[circnode,label=below:$R^{\textsf{shift}}_{5,-}$] (Ls6) at ({2.5*\hsep},\Ybot) {};
 
 \node[circnode,label=above:$R^{\textsf{func}}_{1}$] (Rf1) at ({-0.5*\hsep},\Ytop) {};
 \node[circnode,label=above:$R^{\textsf{func}}_{2}$] (Rf2) at ({0.5*\hsep},\Ytop) {};
 
 \newcommand{\edge}[4]{%
 	\ifnum#3=0
 	\draw[->] (#1) -- (#2);
 	\else
 	\draw[->] (#1) -- node[pos=#4, weight] {$#3$} (#2);
 	\fi
 }
 
 \newcommand{\gedge}[4]{%
 	\ifnum#3=0
 	\draw[LightGray,->] (#1) -- (#2);
 	\else
 	\draw[LightGray,->] (#1) -- node[pos=#4, weight] {$#3$} (#2);
 	\fi
 }
 
 \def\pos{0.25}
 
 \gedge{Cx2}{Ls2}{0}{\pos}
 \gedge{Cx3}{Ls3}{0}{\pos}
 \gedge{Cx4}{Ls5}{0}{\pos}
 
 \gedge{Ls1}{Cx2}{-1}{\pos}
 \gedge{Ls2}{Cx1}{+1}{\pos}
 \gedge{Ls3}{Cx4}{-1}{\pos}
 \gedge{Ls4}{Cx3}{+1}{\pos}
 \gedge{Ls5}{Cx5}{-1}{\pos}
 \gedge{Ls6}{Cx4}{+1}{\pos}
 
 \draw[->, LightBlue] (Cx5) -- (Rf1);
 \draw[->, LightBlue] (Rf1) -- (Cx2);
 \draw[->, LightBlue] (Rf1) -- (Cx3);
 
 \draw[->, LightOrange] (Cx1) -- (Rf2);
 \draw[->, LightOrange] (Rf2) -- (Cx3);
 
	\end{scope}

	\begin{scope}[shift={(5,0)}, scale=1]

		\foreach \row in {-2,-1,0,1,2} {
			\foreach \col in {1,...,5} {
				\node[circnode] (s\row\col) at (1 + \col, 5-\row) {};
			}
		}
		
		\foreach \col in {0,1,...,5} {
			\node[font=\large] at (1+\col, 8) {$\vdots$};
			\node[font=\large] at (1+\col, 2.5) {$\vdots$};
		}

		\def\xpos{1}
		\def\ypos{8.5}
  \draw (\xpos+0.15,\ypos-0.15) -- (\xpos-0.15,\ypos+0.15);
\node at (\xpos+0.18,\ypos+0.18) {$c$};
\node at (\xpos-0.18,\ypos-0.18) {$\ell$};

		\foreach \col in {1,...,5} {
		\node[font=\large] at (1+\col, 8.5) {$\col$};
		}

		\foreach \row in {-2,-1,0,1,2}{
		\node[font=\large] at (1, 5-\row) {$\row$};
	}

		\foreach \row in {-2,...,1} {
			\pgfmathtruncatemacro{\nextrow}{\row+1}
			\draw[->,LightGray,bend left=15]  (s\row2)     to (s\nextrow1);

			\draw[->,LightGray,bend left=15]  (s\nextrow4) to (s\row5);

			\draw[->,LightGray,bend left=15]  (s\nextrow3) to (s\row4);
		}
		
		\foreach \row in {-2,...,2} {
			\draw[->,LightBlue]   (s\row5) to[bend right=30] (s\row2);
			\draw[->,LightBlue]   (s\row5) to[bend right=30] (s\row3);
			\draw[->,LightOrange] (s\row1) to[bend left=60]  (s\row3);
		}
	\end{scope}

\draw[
line width=1.4pt,
-{Stealth[length=7pt,width=9pt]}
] (8.5,1.5)
.. controls (8.5,.5) and  (8,0) ..
node[right=8pt, font=\LARGE] {$\zeta$}
(6,-0.75);

\draw[
line width=1.4pt,
-{Stealth[length=7pt,width=9pt]}
] (-2.5,2)
.. controls (-2.5,.5) and  (-2,0) ..
node[left=8pt, font=\large] {square +}
(0,-0.75);
\node[font=\large] at (-3.5,-.1) {restriction to $\cC$};

\node[font=\large] at (-4,8.5) {$\Gamma_{|S}^\opt$};
\node[font=\large] at (3,-.5) {$\Phi^\opt$};
\node[font=\large] at (5,8.5) {$\Theta^\opt$};

\begin{scope}[shift={(3,-2)}, scale=1]

\def\hsep{1.5cm}

\node[circnode,label=right:$C_{1}$] (Cx1) at ({-2*\hsep},0) {};
\node[circnode,label=right:$C_{2}$] (Cx2) at ({-1*\hsep},0) {};
\node[circnode,label=right:$C_{3}$] (Cx3) at ({0*\hsep},0) {};
\node[circnode,label=right:$C_{4}$] (Cx4) at ({1*\hsep},0) {};
\node[circnode,label=right:$C_{5}$] (Cx5) at ({2*\hsep},0) {};

\def\looseN{2}   
\def\looseP{1}   
\def\looseR{2.2}   

\def\pos{0.5}

\newcommand{\edgegray}[4]{%
	\ifnum#3=-1
	\draw[LightGray, ->, out=290, in=250, looseness=\looseN]
	(#1) to node[pos=#4, weight] {$-1$} (#2);
	\else
	\draw[LightGray, ->, out=230, in=310, looseness=\looseP]
	(#1) to node[pos=#4, weight] {$+1$} (#2);
	\fi
}

\newcommand{\edgetop}[4]{%
	\draw[->, out=120, in=60, looseness=\looseR] (#1) to (#2);
}

\newcommand{\oedgetop}[4]{%
	\draw[LightOrange, ->, out=60, in=120] (#1) to (#2);
}

\newcommand{\bedgetop}[4]{%
	\draw[LightBlue, ->, out=120, in=60] (#1) to (#2);
}


\edgegray{Cx2}{Cx1}{+1}{\pos}
\edgegray{Cx3}{Cx4}{-1}{\pos}

\edgegray{Cx4}{Cx5}{-1}{\pos}

\bedgetop{Cx5}{Cx2}{0}{\pos}
\bedgetop{Cx5}{Cx3}{0}{\pos}
\oedgetop{Cx1}{Cx3}{0}{\pos}

\end{scope}
	
\end{tikzpicture}
    \caption{The graphs $\Gamma_{|S}^\opt$, $\Phi^\opt$, and $\Theta^\opt$ from \cref{ex:chi_bnd_more_interesting_restricted_graphs}.}
    \label{fig:combinedOPT_example}
\end{figure}

Note that by keeping only those edges of $\Theta$ whose $\zeta$-image is present in $\Phi^\opt$, the restriction $\zeta^\opt : \Theta^\opt \to \Phi^\opt$ of $\zeta$ is an out-covering projection.
Obviously, $\zeta^\opt$ still respects bunches and preserves edge-weights.
Moreover, all the described restrictions to subgraphs behave well with bunches in the sense that each bunch is either entirely removed or entirely remains in the subdigraph.

\begin{observation}\label{obs:single_bunch_remains_in_opt}
    Each vertex of $\Phi^\opt$ is the center of a unique bunch of $\Phi^\opt$.
    The same holds true for each position of $\Theta^\opt$.
    Moreover, each position of $\cK^\opt$ is the center of \emph{at most one} bunch, and \emph{exactly one} in case of a position at an \emph{internal level} of $\cK$ (i.e.\ any level except $1$ and~$t$). 
\end{observation}
\begin{proof}
    For $\Phi^\opt$, this follows from the fact that each vertex $C_c \in \cC(\Gamma_{|S})$ has a unique outgoing edge in $\Gamma_{|S}^\opt$, and the definition of bunches in $\Phi$ (see Remark~\ref{rem:bunches_in_Phi}).
    This property of $\Phi^\opt$ is transferred to $\Theta^\opt$ by the out-covering projection $\zeta^\opt$.
    
    Since $\cK^\opt$ is the restriction of $\Theta^\opt$ to $V(\cK)$, some edges adjacent to $V(\cK)$ might be removed.
    Note, however, that each bunch from $\Theta^\opt$ is either entirely removed or entirely belongs to $\cK^\opt$.
    Thus, each position of $\cK^\opt$ is the center of at most one bunch.
    Moreover, since $\Theta^\opt$ has edges only within a level or between consecutive levels, all the edges of $\Theta^\opt$ adjacent to an internal level of $\cK^\opt$ belong to $\cK^\opt$.
    Thus, the internal vertices are centers of exactly one bunch.
\end{proof}

The uniqueness of bunches given by Observation~\ref{obs:single_bunch_remains_in_opt} allows for simplification of the defining condition of trackable positions (cf. Definition~\ref{def:trackable_vertices}).

\begin{observation}\label{obs:trackable_iff_all_walks_in_cQOpt_reach_first_level}
    Let $u \in V(\cK) = V(\cK^\opt)$.
    If either
    \begin{enumerate}[label=(\roman*)]
        \item the level of $u$ is $1$, or
        \item $u$ has at least one out-neighbor in $\cK^\opt$ and all of them are trackable,
    \end{enumerate}
    then $u$ is trackable.
    
    Therefore, if every maximal walk in $\cK^\opt$ starting from $u$ enters the first level, then $u$ is trackable.
\end{observation}
\begin{proof}
    The inductive assumption clearly implies trackability from Definition~\ref{def:trackable_vertices}.
    In particular, the second condition of the observation implies the second condition of Definition~\ref{def:trackable_vertices} since each position of $\cK^\opt$ is a center of at most one bunch by Observation~\ref{obs:single_bunch_remains_in_opt}.
    

    We prove the `Therefore' statement by contradiction.
    Suppose that every maximal walk in $\cK^\opt$ starting from $u$ enters the first level, but $u$ is not trackable.
    By the first part, every untrackable position with an outgoing edge in $\cK^\opt$ has an untrackable out-neighbor.
    Starting at $u$, repeatedly choose such a neighbor for as long as possible, constructing a maximal walk $W$ (either finite if it reaches a sink, or infinite).
    Since $W$ traverses only untrackable positions, it never enters the first level of $\cK^\opt$, a contradiction.
\end{proof}

\subsubsection{Finishing the proof}\label{sssec:finishing_proof}

In view of Observation~\ref{obs:trackable_iff_all_walks_in_cQOpt_reach_first_level}, our task is reduced to proving that if $t$, the number of levels of $\cK$, is sufficiently large, all maximal walks in $\cK^\opt$ from any position $v^2_c, c \in S$ reach the first level.
This will be an easy consequence of the following series of lemmas on weights of walks in the weighted digraphs $\Theta^\opt$ and $\Gamma_{|S}^\opt$.

\begin{lemma}\label{lem:weight_of_closed_walk}
    Let $W$ be a closed walk in $\Gamma_{|S}^\opt$, then $w(W) < 0$.
\end{lemma}
\begin{proof}
    If $w(W) \geq 0$, the row player has a non-losing strategy in the game $(\Gamma_{|S}^\opt, s)$ for every $s \in \cR(\Gamma_{|S}) \cap V(W)$ by following the edges of $W$.
    This is a contradiction with Lemma~\ref{lem:all_states_in_Gamma|S_are_winning} stating that all states of $\Gamma_{|S}^\opt$ are winning for the column player.  
\end{proof}

\begin{lemma}\label{lem:weight_of_simple_path}
    Let $W$ be a simple path in $\Gamma_{|S}^\opt$ from $\cC(\Gamma_{|S}^\opt)$ to $\cC(\Gamma_{|S}^\opt)$, then $w(W) \leq |S| - 1$.
\end{lemma}
\begin{proof}
    The maximal weight of an edge from $\cR(\Gamma_{|S}^\opt)$ to $\cC(\Gamma_{|S}^\opt)$ is $1$, while all edges from $\cC(\Gamma_{|S}^\opt)$ to $\cR(\Gamma_{|S}^\opt)$ have weight $0$.
    Therefore, the weight of $W$ may increase by $1$ only between visits of vertices of $\cC(\Gamma_{|S})$.
    Since the path is simple and $|\cC(\Gamma_{|S})| = |S|$, there are at most $|S|-1$ such occasions.
    Therefore, $W$ may accumulate weight of at most $|S|-1$.
\end{proof}

\begin{observation}\label{obs:total_weight_is_level_difference}
    Let $W = (u_0, \dots, u_n)$ be a walk in $\Theta^\opt$, then $w(W) = \ell(u_n) - \ell(u_0)$.
\end{observation}
\begin{proof}
    Indeed, $w(W) = \sum_{i=1}^n w(u_{i-1}, u_i) = \sum_{i=1}^n \ell(u_i) - \ell(u_{i-1}) = \ell(u_n) - \ell(u_0)$.
\end{proof}

\begin{lemma}\label{lem:max_level_of_walks_from_second_level}
    There is $t \leq |S| + 1$ such that all walks starting from the second level of $\Theta^\opt$ reach at most the level $t$ of $\Theta^\opt$.
\end{lemma}
\begin{proof}
    Let $W = (u_0, u_1, \dots)$ be a walk in $\Theta^\opt$ and consider the walk $\zeta^\opt(W) = (\zeta^\opt(u_0), \zeta^\opt(u_1), \dots)$ in $\Phi^\opt$.
    Since $\Phi^\opt$ is the second power of the digraph $\Gamma_{|S}^\opt$, see Definition~\ref{def:multidigraph_Phi^opt}, we may unfold the edges of $\Phi^\opt$ to pairs of edges of $\Gamma_{|S}^\opt$.
    Hence, we unfold the walk $\zeta^\opt(W)$ to a walk $W'$ in the digraph $\Gamma_{|S}^\opt$.
    
    Let $W_n$ be the prefix $(u_0, \dots, u_n)$ of $W$ and $W'_n$ be the corresponding prefix of $W'$ (consisting of $2n$ edges).
    Since both steps of the transition from $W$ to $W'$ preserve weights,%
    \footnote{Due to the definition of weights of $\Phi^\opt$ via $\Gamma_{|S}^\opt$ and the fact that $\zeta^\opt$ is weight-preserving.}
    we have $w(W_n) = w(W'_n)$.

    We decompose each $W'_n$ into a prefix path $W^0_n$, a closed walk $W^c_n$ (possibly trivial), and a suffix path $W^1_n$, such that the composition $W^p_n$ of $W^0_n$ and $W^1_n$ is a simple path.
    Note that since $W'_n$ starts and ends in $\cC(\Gamma_{|S})^\opt$, so does $W^p_n$.
    Thus, we \cref{lem:weight_of_closed_walk,lem:weight_of_simple_path}, we have $w(W_n) = w(W'_n) = w(W^p_n) + w(W^c_n) \leq w(W^p_n) \leq |S| - 1$.
    Therefore, by \cref{obs:total_weight_is_level_difference}, if the walk $W$ starts at a position at the second level, it reaches at most level $t \leq (|S|- 1) + 2 = |S| + 1$.
\end{proof}

\begin{lemma}\label{lem:Theta_and_cQ_acyclic}
    The digraph $\Theta^\opt$, and consequently $\cK^\opt$, is acyclic.
\end{lemma}
\begin{proof}
    Using the same transformation of a walk $W$ from $\Theta^\opt$ to a walk $W'$ from $\Gamma_{|S}^\opt$ as above, we see that a cycle in $\Theta^\opt$ corresponds to a closed walk of weight $0$ in $\Gamma_{|S}^\opt$.
    This is a contradiction with Lemma~\ref{lem:weight_of_closed_walk}.
\end{proof}

Finally, we are ready to prove Lemma~\ref{lem:first_determines_second}.

\begin{proof}[Proof of Lemma~\ref{lem:first_determines_second}]
    Let $S$ be the set from Lemma~\ref{lem:all_states_in_Gamma|S_are_winning} and $t$ be from Lemma~\ref{lem:max_level_of_walks_from_second_level}.
    Denote the set of maximal walks in $\cK^\opt$ from the second level by $\cW$ and the set of infinite walks in $\Theta^\opt$ from the second level by $\cW'$.
    Note that every walk from $\cW$ is finite by Lemma~\ref{lem:Theta_and_cQ_acyclic}, but it can be extended to a walk from $\cW'$, which follows from the fact that each position in $\Theta^\opt$ has an out-neighbor due to Observation~\ref{obs:single_bunch_remains_in_opt}.
    By Observation~\ref{obs:single_bunch_remains_in_opt}, if a position of $\cK^\opt$ does not have out-edges, then it lies at the first or last level of $\cK^\opt$.
Hence, by maximality, the last position of a walk $W \in \cW$ lies either at the first or last level of $\cK^\opt$.
    We want to refute the second possibility.
    Recall that the number of levels of $\cK^\opt$ is $t$; hence, we want to refute that last position of $W$ lies at level $t$.

    Consider a walk $W \in \cW$ and its arbitrary extension $W' \in \cW'$.
    If the last position of $W$ lies at level $t$, the next vertex in $W'$, which lies outside of $\cK^\opt$ due to maximality of $W$, is at level $t+1$.
    However, this is a contradiction with the choice of $t$ from Lemma~\ref{lem:max_level_of_walks_from_second_level}.
    
    Therefore, the last position of every $W \in \cW$ lies at level $1$.
    By \cref{obs:trackable_iff_all_walks_in_cQOpt_reach_first_level}, this implies that all the positions at the second level of $\cK$ are trackable.
    Thus, by Observation~\ref{obs:trackable_vertices_are_function_of_first_level}, there is a function $f$ such that $f(v^1) = v^2$, which concludes the proof.
\end{proof}

\subsubsection{What changes if the min-plus system has no finite solution}\label{sssec:what_if_minplus_fails}

In the case that the min-plus system $\widehat{A}x \geq \widehat{B}x$ has no finite solution, the overall structure of the proof is the same with the following modifications.
The following min-plus counterpart of \cref{lem:first_determines_second} asserts that in an appropriate restriction of the digraph, the last vertex of a sufficiently long directed path determines the second last vertex.

\begin{lemma}\label{lem:last_determines_second_last}
    Suppose that the min-plus system $\widehat{A}x \geq \widehat{B}x$ has no finite solution.
    Then, there exists a non-empty set $S \subseteq [d]$ and an integer $2 \leq t \leq |S|+1$ such that for every $G \in \dcY_{P,Z} \restriction S$, there exists a function $f: \N^S \to \N^S$ that satisfies the following: 
    whenever $Q$ is a subdigraph of $G$ on vertices $v^1, \dots, v^t$ that is isomorphic to $\vec{P}_t$, we have $v^{t-1} = f(v^t)$.
\end{lemma}

Having \cref{lem:last_determines_second_last} at hand, the proof proceeds as before with the observation that the chromatic number of the class $\dcY_{P,Z}$ is bounded as the graphs in $\dcY_{P,Z} \restriction S$ exclude the reverse of $\vec{\Lambda}_t$ as a subdigraph (cf. \cref{obs:Lambda_not_subdigraph}).

The proof of \cref{lem:last_determines_second_last} requires the following modifications.
To determine the set $S$, we use Corollary~\ref{cor:min_plus_no_solution_implies_positive_mean_payoff}.
In particular, we claim that the game digraphs $\Gamma_{\min}(\widehat{A},\widehat{B})$ from Section~\ref{sec:min_plus_prelims} and $\Gamma(P,Z)$ from Definition~\ref{def:game_graph} are isomorphic as weighted digraphs.
Indeed, writing $-H$ for the weighted digraph obtained from $H$ by reversing the signs of all its weights, and $\cong$ for isomorphism of weighted digraphs, we have
\[
    \Gamma_{\min}(\widehat{A},\widehat{B}) \cong -\Gamma_{\max}(A',B') \cong \Gamma_{\max}(\widetilde{A},\widetilde{B}) \cong \Gamma(P,Z)
    ,
\]
where $A' = -\widehat{A}$ and $B' = -\widehat{B}$.
The first isomorphism is by the definition of $\Gamma_{\min}(\widehat{A},\widehat{B})$ in Section~\ref{sec:min_plus_prelims}, the second is by Observation~\ref{obs:tropical_systems_are_the_same}, and the last isomorphism is justified just before Definition~\ref{def:game_graph}.

Thus, we create $\Gamma_{|S} = \Gamma(P_{|S},Z_{|S})$ from $\Gamma(P,Z)$ restricting to the set of winning states $X$ for the column player (in the `min-plus' game on $\Gamma_{\min}$, i.e.\ the column player tries to maximize their payoff), where $S = X \cap \cC(\Gamma_{\min})$, and proceed by establishing the correspondence with the digraph $\cK$ via $\Theta$ and $\Phi$ whose definitions are unchanged.
However, the definition of trackability is changed as the `trivially' trackable positions are at the last level of $\cK$.
Then, by restricting $\Gamma_{|S}$ to $\Gamma_{|S}^\opt$ and propagating the modification to obtain $\cK^\opt$, we reduce the trackability of level $t-1$ of $\cK$ to the condition that each maximal walk in $\cK^\opt$ from level $t-1$ reaches level $t$.

By Corollary~\ref{cor:min_plus_no_solution_implies_positive_mean_payoff}, the strategy $\sigma^\opt$, or rather $\sigma^\opt_{|S}$, guarantees \emph{positive} mean payoff for the column player in the game played on $\Gamma_{|S}$.
Hence, the weights of closed walks in $\Gamma_{|S}^\opt$ are \emph{positive}, and the weight of a $\cC(\Gamma_{|S}^\opt)$-to-$\cC(\Gamma_{|S}^\opt)$ simple path in $\Gamma_{|S}^\opt$ is lower-bounded by $-(|S|-1)$ (similarly to Lemmas~\ref{lem:weight_of_closed_walk}~and~\ref{lem:weight_of_simple_path}), which allows us to obtain a counterpart of Lemma~\ref{lem:max_level_of_walks_from_second_level}.
Therefore, it is possible to choose the value of $t \leq |S| + 1$ to be large enough so that the walks from level $t-1$ of $\Theta^\opt$ never go below level $1$, which proves Lemma~\ref{lem:last_determines_second_last}.

\section{Main results for \nf set-defined classes}\label{sec:deciding-chi-boundedness}

In this section, we show that Theorem~\ref{th:main-set-defined}, Corollary~\ref{cor:GS}, and Theorem~\ref{th:full-set-defined-decision} follow from the analysis in the previous sections.

\subsection{Induced shift digraphs in \texorpdfstring{$\chi$}{chi}-unbounded \nf set-defined classes}
\label{sec:induced-shift-graphs}

The main difficulty in proving Theorem~\ref{th:main-set-defined} is that when reducing down to a single path clause in Lemma~\ref{lem:path+discrete-clauses}, we passed to subgraphs rather than induced subgraphs. In the following key Lemma~\ref{lem:unbounded_Y_F_contains_shift_or_symmetrized_shift}, we will reprove (the contrapositive of) Lemma~\ref{lem:path+discrete-clauses} in strengthened form.
We recall that in the set-up for Lemma~\ref{lem:path+discrete-clauses}, we have a class $\dcY_F$ consisting of digraphs induced by an acyclic DNF $F$.
Also, we have written $F = \bigvee_{i \in [a]} P_i \vee \bigvee_{j \in [b]} D_j$, where each $P_i$ is a path clause and each $D_j$ is a discrete clause.
We let $F_i = P_i \vee \bigvee_{j \in [b]} D_j$ and $\dcY_{F_{i}}$ consist of digraphs induced by $F_i$, which is contained in the monotone closure of $\dcY_F$.

Since the loop-elimination of Lemma~\ref{lem:loop-elimination} occurred after Lemma~\ref{lem:path+discrete-clauses}, we will need to consider path clauses with loops. So we also recall some points about Lemma~\ref{lem:loop-elimination}.
Recall that $L(C) \subseteq [d]$ denotes the set of loop coordinates of a clause $C$.
For each $L \subset [d]$, we consider the set $B_L = \{ j \in [b] : L \subseteq L(D_j)\}$.
For each DNF $F_i$ we consider the reduced DNF $F'_i$ defined as 
\[
    F'_i = P'_i \vee \bigvee_{j \in B_{L(P_i)}} D'_j
    ,
\]
where $P'_i, D'_j, j \in B_{L(P_i)}$ are obtained from $P_i, D_j, j \in B_{L(P_i)}$, respectively, by removing literals corresponding to variables with an index in $L(P_i)$.
Then, Lemma~\ref{lem:loop-elimination} states that (i) $\dcY_{F'_i} \subseteq \dcY_{F_{i}}$; and (ii) $\dcY_{F_{i}}$ is \chiunbounded if and only if $\dcY_{F'_i}$ is \chiunbounded.
The idea of the construction for the first item is as follows.
Given a representation of a graph $G' \in \dcY_{F'_i}$, fix a vector $w \in \Ninj{L(P_i)}$ whose values do not appear in the representation of $G'$.
Then we extend each vertex $v' \in V(G')$ by the vector $w \in \Ninj{L(P_i)}$ on coordinates in $L(P_i)$.
Clearly, the clauses $P_i$ and $D_j, j \in B_{L(P_i)}$ induce on the resulting vertex set $V$ a digraph $G$ isomorphic to $G'$.
Moreover, each discrete clause $D_j$ with $j \not\in B_{L(P_i)}$ has at least one non-loop coordinate among $L(P_i)$; hence, $D_j$ induces no edges on $V$.

Let $P$ and $Q$ be path clauses with clause digraphs $\Delta_P, \Delta_Q$ on the same vertex-set $[d]$.
We say that $Q$ is the \emph{reverse} of $P$ if $\Delta_Q$ is obtained by reversing the direction of all edges of $\Delta_P$.

\begin{lemma}\label{lem:unbounded_Y_F_contains_shift_or_symmetrized_shift}
    Suppose the class $\dcY_F$ is \chiunbounded.
    Then there is some $i^* \in [a]$ and a collection $\cV$ of vertex sets $V \subseteq \Ninj{d}$ such that:
    \begin{enumerate}[label=(\roman*)]
        \item\label{it:i*_induces_shift_graphs} $F_{i^*}$ induces on $\cV$ exactly the class $\dcS_D$ of $D$-dimensional shift graphs for some $D \geq 2$; and
         \item\label{it:others_induce_bounded_chromatic_number} For every $j \in [a]$, if $P_j$ is distinct from $P_{i^*}$ and the reverse of $P_{i^*}$, then $P_j$ induces no edges on $\cV$.
    \end{enumerate}
\end{lemma}
\begin{proof}
    Let $L_{\max}$ be an arbitrary $\subseteq$-maximal member of $\{L(P_i) : i \in [a], \dcY_{F_i} \text{ is \chiunbounded}\}$.
    Since $\dcY_F$ is \chiunbounded, Lemma~\ref{lem:path+discrete-clauses} implies that $L_{\max}$ is well-defined.
    Consider the sets $A = \{i \in [a] : L_{\max} \subseteq L(P_i)\}$ and $B = \{j \in [b] : L_{\max} \subseteq L(D_j)\}$.
    The set $A$ is non-empty by the choice of $L_{\max}$.
    Then for each $i \in A$ and $j \in B$, respectively, consider the reduced clauses $P'_i$ and $D'_j$ obtained by removing literals corresponding to variables with an index in $L_{\max}$.
    For $i \in A$, denote by $F'_i$ the reduced DNF of the form
    \[
        F'_i = P'_i \vee \bigvee_{j \in B} D'_j
        .
    \]
    By Lemma~\ref{lem:loop-elimination}, the class $\dcY_{F_i}$ is \chiunbounded if and only if the class $\dcY_{F'_i}$ is \chiunbounded.

    Let $d' = d - |L_{\max}|$.  
    Without loss of generality, by permuting indices if needed, we may assume that $L_{\max} = [d' + 1, d]$.
    Thus, the coordinates of the reduced clauses are exactly $[d']$.

    We define a preorder $\preceq$ on $A$ as follows: set $j \preceq k$ if for each maximal path $\rho$ of $\Delta_{P'_j}$ there is a maximal path $\rho^\sharp$ of $\Delta_{P'_k}$ such that $V(\rho) \subseteq V(\rho^\sharp)$%
    \footnote{Recall that the clause digraphs $\Delta_{P'_j}, \Delta_{P'_k}$ have the same vertex set $[d']$.}.
    Consider the set $A^* = \{i \in A : \dcY_{F'_i} \text{ is \chiunbounded}\}$.
    By the definition of $A$ and Lemma~\ref{lem:loop-elimination}, the set $A^*$ is non-empty.
    Moreover, $L(P_i) = L_{\max}$ for all $i \in A^*$; hence, the corresponding reduced clauses $P'_i$ are loopless.
    Fix an arbitrary $\preceq$-minimal element $i^*$ of $A^*$.
    We will show that $i^*$ is as desired.
    Without loss of generality, by permuting the indices, we may assume that all the positive literals of $P'_{i^*}$ have the form $q_{c+1, c}$.
    
    Since $\dcY_{F'_{i^*}}$ is \chiunbounded, Lemma~\ref{th:diaganal-to-path} implies that there is a collection of functional constraints $Z = ((L(D'_j), \lambda_j), j \in B)$ with $\lambda_j \in [d'] \setminus L(D'_j)$\footnote{Recall that we may omit the discrete clause $D$ with $L(D) = [d]$ by the discussion at the beginning of Section~\ref{sec:all-discrete-clauses}. Then also $L(D'_j) \not= [d']$ for all $j \in B$.} such that $\dcY_{P'_{i^*}, Z}$ has unbounded chromatic number.
    By Theorem~\ref{thm:tropical_dichotomy}, in particular by the techniques from Section~\ref{ssec:construction-of-shift-graphs} summarized in Corollary~\ref{cor:functional_contains_SD}, the class $\dcY_{P'_{i^*}, Z}$ contains a class of $D$-dimensional shift graphs.
    We recall that the representation for these shift graphs is obtained by projecting the canonical representation of $D$-dimensional shift graphs $\dcS_D$ via a function $\pi \colon \N^D \to \Ninj{d'}$ constructed from a gapless interval representation $\cI$ of dimension $D$ for $P'_{i^*}$ over $Z$.%
    \footnote{Here we are slightly imprecise.
    In \cref{ssec:construction-of-shift-graphs}, the function $\pi$ is defined from $\N^D$ to $\Qpl^{\underline{d'}}$, which is a harmless generalization to allow for a uniform definition of $\pi$ (see \cref{rem:Q_instead_of_N}).
    Thus, if we require an integer representation, we should formally write $\pi$ as a mapping from \emph{finite subsets} of $\N^D$ to \emph{finite subsets} of $\Ninj{d'}$.
    However, this distinction plays no role in the proof as we are always concerned with the image of at most one finite subset of $\N^D$ (the vertex set of a shift graph $S \in \dcS_D$).
    Hence, we allow ourselves this harmless imprecision.}
    That is, let $\cI = \{I_c: c \in [d']\}$ be an interval representation for $P'_{i^*}$ over $Z$ as defined in Definition~\ref{def:interval_system}.
    We assume that
    \begin{equation}\label{eq:intervals_share_coordinate}
        I_c \cap I_{c'} \not= \emptyset \text{ for all } c,c' \in [d']
        ,\tag{$*$}
    \end{equation}
    which we may enforce by extending each interval by a sufficiently large constant, which preserves the gapless property by \cref{lem:extension_of_gapless_is_gapless} (see \cref{rem:resulting_shift_graph_dimension_in_full_classes} for a discussion on how this affects the dimension of the constructed shift graphs).
    The mapping $\pi$ is obtained by starting with a map $\pi_0 \colon \N^D \to \N^{|I_1|} \times \N^{|I_2|} \times \cdots \times \N^{|I_{d'}|}$ and then composing with further injections $g_j$.
    
    Assuming the canonical representation of shift graphs in $\dcS_D$, let $\cV' = \{\pi[V(S)] : S \in \dcS_D\}$.
    We define $\cV$ as follows. 
    For each $V' \subset \Ninj{d'}$ from $\cV'$, we fix a vector $w \in \Ninj{L_{\max}}$ whose values exceed all entries that appear in vectors of $V'$, and define $V \subset \Ninj{d}$ as a set of vectors obtained from $V'$ by extending each $v' \in V'$ by $w$.
    Then 
    \[
        \cV := \{ V~|~ V' \in \cV' \}.
    \]
    By construction, for each $V' \in \cV'$ such that $V' = \pi[V(S)]$ for some $S \in \dcS_D$, the path clause $P'_{i^*}$ induces on $V'$ a digraph isomorphic to $S$.
    Similarly, the path clause $P_{i^*}$ induces on $V$ a digraph isomorphic to $S$.
    Therefore, $\dcS_D \subseteq \dcY_{F_{i^*}}$ as claimed by item~\ref{it:i*_induces_shift_graphs} of the lemma.
    
    It remains to prove item~\ref{it:others_induce_bounded_chromatic_number}.
    That is, we prove that each $P_j, j \in [a]$, distinct both from $P_{i^*}$ and the reverse of $P_{i^*}$, induces no edges on any $V \in \cV$.
    We first observe that if $j \not\in A$, then $L_{\max} \setminus L(P_j) \not= \emptyset$, and thus
    $P_j$ has a non-loop coordinate in $L_{\max}$.
    Therefore, since each $V \in \cV$ is constant on the coordinates of $L_{\max}$ (by construction), $P_j$ induces no edges on any $V \in \cV$. 
    
    In the rest of the proof we consider the case $j \in A$.
    Here we shall focus on the path clause $P'_j$ instead of $P_j$.
    Note that for each $V \in \cV$, the graph induced by $P'_j$ on $V'$ is isomorphic to the graph induced by $P_j$ on $V$ as $L_{\max} \subseteq L(P_j)$ and we always extend the vertices $v' \in V'$ by fresh values $w \in \Ninj{L_{\max}}$ that do not appear in $V'$.
    Thus, our goal is to prove that $P'_j$ induces no edges on any $V' \in \cV'$.
    Towards this goal, we fix a canonical representation of a $D$-dimensional shift graph $S$ such that $V' = \pi[V(S)]$.
    We distinguish two cases.
    \paragraph{Case (I): $j \not\preceq i^*$.}
    By definition of $\preceq$, there is a path $\rho^j$ of $\Delta_{P'_j}$ such that for every path $\rho^{i^*}$ of $\Delta_{P'_{i^*}}$, it holds that $V(\rho^j) \not\subseteq V(\rho^{i^*})$.
    In particular, there is an edge $c \mapsto c'$ in $\rho^j$, such that $c,c'$ belong to different respective paths $\rho^{i^*}_k, \rho^{i^*}_\ell$ of $\Delta_{P'_{i^*}}$.
    Let $u,v \in V(S)$ be arbitrary.
    Then $\pi(u)_c \not= \pi(v)_{c'}$ as we obtained these values by applying injections $g_k$ and $g_\ell$ with disjoint images to $\pi_0(u)_c$ and $\pi_0(v)_{c'}$, respectively.
    Thus, $P'_j$ cannot induce any edges on $V'$.
    
    \paragraph{Case (II): $j \preceq i^*$.}
    By definition of $\preceq$, the partition of the clause digraph $\Delta_{P'_j}$ into connected components refines the partition of $\Delta_{P'_{i^*}}$ (these partitions might be equal).
    Recall that each positive literal of $P'_{i^*}$ is of the form $q_{c+1,c}$ for some $c \in [d'-1]$.
    We shall make the final case distinction.
    \begin{enumerate}[label=(II)-(\alph*), leftmargin=*]
        \item $P'_j$ contains some positive literals $q_{\beta, \alpha}$ and $q_{\delta, \gamma}$ with $\beta - \alpha \not= \delta - \gamma$,
        \item the difference $\eps = \beta - \alpha$ is the same for all positive literals $q_{\beta, \alpha}$ of $P'_j$. 
    \end{enumerate}
    
    \paragraph{Case (II)-(a):}
    Let $t \in [D]$ be contained in both $I_\beta$ and $I_\delta$; existence of such a $t$ is guaranteed by property~\eqref{eq:intervals_share_coordinate} of the interval representation $\cI$ that we used to define the mapping $\pi_0$.
    For $c \in [d']$, let $\idx(t,c) \in [|I_c|]$ be the index of $t$ within the interval $I_c$, i.e. such that $t$ is the $\idx(t,c)$-th smallest element of $I_c$.
    For $c$ and $x \in [|I_c|]$, let $\val(x,c)$ be the value of the $x$-th smallest element of $I_c$.
    Note that for coordinates $c, c'$ that belong to the same path of $\Delta_{P'_{i^*}}$, we have%
    \footnote{In words, $\val(\idx(t,c'),c)$ is the value that $I_c$ has at the index where $I_{c'}$ has $t$.}
    \begin{equation}\label{eq:val-idx}
        \val(\idx(t,c'),c) = t + c - c'.
    \end{equation}

    Now, let $u, w \in V(S)$, and, 
    towards a contradiction, suppose that $P'_j$ induces an edge $(\pi(u),\pi(w))$.    
    Then, due to $q_{\beta, \alpha}$ being satisfied, we have
    \[
        u_t = (\pi_0(u)_\beta)_{\idx(t,\beta)}
        = (\pi_0(w)_\alpha)_{\idx(t,\beta)}
        = w_{\val(\idx(t,\beta),\alpha)}
        = w_{t + \alpha - \beta}
    \]
    where the first equality is by the definition of $\idx$, the second one is due to the equality $\pi_0(u)_\beta = \pi_0(w)_\alpha$ asserted by the positive literal $q_{\beta, \alpha}$ of $P'_j$, 
    the third equality is from the definitions of $\val$ and $\idx$, and the last equality is from \cref{eq:val-idx}.
    By similar reasoning, due to $q_{\delta, \gamma}$ being satisfied, we have
    \[
        u_t 
        = (\pi_0(u)_\delta)_{\idx(t,\delta)}
        = (\pi_0(w)_\gamma)_{\idx(t,\delta)}
        = w_{\val(\idx(t,\delta),\gamma)}
        = w_{t + \gamma - \delta}.
    \]
    Therefore, we conclude that 
    $w_{t + \alpha - \beta} = u_t = w_{t + \gamma - \delta}$, and
    since $\alpha - \beta \not= \gamma - \delta$, this implies that $w$ contains two entries with the same value.
    But this is a contradiction with the fact that vertices of $S$ are (strictly) increasing tuples.
    Therefore, $P'_j$ induces no edges on $V'$.
    
    \paragraph{Case (II)-(b):}
    We show that this case does not arise for any $P'_j$ distinct from both $P'_{i^*}$ and its reverse.
    To do this, we shall prove that the class $\dcY_{F'_j}$ is \chiunbounded.
    Then, by the $\preceq$-minimality of $i^*$, we have $i^* \preceq j$.
    Together with the assumption $j \preceq i^*$, we infer that both digraphs $\Delta_{P'_{i^*}}$ and $\Delta_{P'_j}$ have the same partition of coordinates $[d']$ into connected components.
    This implies that either $\eps = 1$ or $\eps = -1$; otherwise, any adjacent coordinates $c, c+1$ in the same path in $\Delta_{P'_{i^*}}$ would be in different paths in $\Delta_{P'_j}$.
    To reach the conclusion, it is simple to observe that $\eps = 1$ if and only if $P'_j = P'_{i^*}$, while $\eps = -1$ if and only if $P'_j$ is the reverse of $P'_{i^*}$.
    
    First, let us note that the difference $\eps$ is distinct from zero as $P'_j$ is a path clause, which is required to contain a non-trivial path in $\Delta_{P'_j}$, and $\eps = 0$ corresponds to having only loops.
    To prove that the class $\dcY_{F'_j}$ is \chiunbounded, we shall show that the class $\dcY_{P'_j, Z}$ has unbounded chromatic number.
    By Theorem~\ref{thm:tropical_dichotomy}, this is equivalent to the fact that certain systems of tropical inequalities admit finite solutions.
    
    Recall the tropical systems $\widehat{A} \minmu x \geq \widehat{B} \minmu x$ and $\widetilde{A} \maxmu y \leq \widetilde{B} \maxmu y$ for the path clause $P'_{i^*}$ and collection of functional constraints $Z$ from Section~\ref{ssec:construction-of-shift-graphs}.
    We note that the analogous systems $\widehat{A}^{(j)} \minmu x \geq \widehat{B}^{(j)} \minmu x$ and $\widetilde{A}^{(j)} \maxmu y \leq \widetilde{B}^{(j)} \maxmu y$ for $P'_j$ and $Z$, are obtained by  replacing the submatrices $\widehat{A}^\shift \minmu x \geq \widehat{B}^\shift \minmu x$ and $\widetilde{A}^\shift \maxmu y \leq \widetilde{B}^\shift \maxmu y$ corresponding to the positive literals of $P'_{i^*}$ by rows corresponding to the positive literals of $P'_j$ (see \cref{eq:equality_condition}).
    That is, we replace them by the rows corresponding to the equalities
    \begin{equation}\label{eq:shift_of_P'_j}
        x_\beta = x_\alpha + 1
    \end{equation}
    for every positive literal $q_{\beta, \alpha}$ of $P'_j$.
    
    We will show that each of the systems $\widehat{A}^{(j)} \minmu x \geq \widehat{B}^{(j)} \minmu x$ and $\widetilde{A}^{(j)} \maxmu y \leq \widetilde{B}^{(j)} \maxmu y$ admits a finite solution.
    Let $x^*$ and $y^*$ be the respective finite solutions of $\widehat{A} \minmu x \geq \widehat{B} \minmu x$ and $\widetilde{A} \maxmu y \leq \widetilde{B} \maxmu y$, which exist by the assumption that $\dcY_{P'_{i^*},Z}$ is \chiunbounded and Theorem~\ref{thm:tropical_dichotomy}.
    If $\eps > 0$, we set
    \[
        x^{(j)} = \frac{x^*}{\eps}, \quad 
        y^{(j)} = \frac{y^*}{\eps}
        .
    \]
    We claim that the vectors $x^{(j)}$ and $y^{(j)}$ solve the respective systems $\widehat{A}^{(j)} \minmu x \geq \widehat{B}^{(j)} \minmu x$ and $\widetilde{A}^{(j)} \maxmu y \leq \widetilde{B}^{(j)} \maxmu y$.
    We show it for $x^{(j)}$; the argument for $y^{(j)}$ is similar and we omit it.
    The vector $x^{(j)}$ clearly solves the inequalities arising from functional constraints in \cref{eq:endpoint_condition} of the form 
    \[
		x_{\lambda} \geq \min_{c \in L} x_c \qquad \textrm{ for each $(L,\lambda) \in Z$}
    \]
    as the vector $x^*$ solves these inequalities and this is preserved under scaling $x^*$ by a positive factor.
    As for the equalities from \cref{eq:shift_of_P'_j} for the positive literals of $P'_j$, note that chaining the equalities for the literals $q_{\alpha+1,\alpha}, q_{\alpha+2,\alpha+1}, \dots, q_{\beta,\beta-1}$ of $P'_{i^*}$, we have
    \[
        x^*_\beta = x^*_\alpha + (\beta - \alpha)
        .
    \]
    Therefore, as $\eps = \beta - \alpha$, we have
    \[
        x^{(j)}_\beta = \frac{x^*_\beta}{\eps} = \frac{x^*_\alpha}{\eps} + \frac{\beta - \alpha}{\eps} = x^{(j)}_\alpha + 1,
    \]
    as desired.
    
    The case when $\eps < 0$ is similar, we only switch the role of $x^*$ and $y^*$, and all minima and maxima.
    That is, we set 
    \[
        x^{(j)} = \frac{y^*}{\eps}, \quad 
        y^{(j)} = \frac{x^*}{\eps}
        .
    \]
    Then the max-plus inequality
    \[
        y^*_{\lambda} \leq \max_{c \in L} y^*_c
    \]
    for a functional constraint $(L,\lambda) \in Z$ for $y^*$ turns into the min-plus inequality
    \[
        x^{(j)}_{\lambda} \geq \min_{c \in L} x^{(j)}_c
    \]
    for $x^{(j)}$.
    Moreover, the equality
    \[
        y^*_\beta = y^*_\alpha + (\beta - \alpha)
    \]
    for positive literals for $y^*$ turns into the equality
    \[
        x^{(j)}_\beta = x^{(j)}_\alpha + 1
    \]
    for $x^{(j)}$ as above.
    This shows that $x^{(j)}$ solves $\widehat{A}^{(j)} \minmu x \geq \widehat{B}^{(j)} \minmu x$.
    The justification that $y^{(j)}$ solves the max-plus system $\widetilde{A}^{(j)} \maxmu y \leq \widetilde{B}^{(j)} \maxmu y$ is analogous.
\end{proof}

\begin{remark}\label{rem:resulting_shift_graph_dimension_in_full_classes}
    From the proof, it follows that $D \leq 4d-5$.
    Indeed, Corollary~\ref{cor:functional_contains_SD} states that the dimension of a minimal interval representation for $P'_{i^*}$ over $Z$ is at most $2d'-2$.
    To satisfy property~\eqref{eq:intervals_share_coordinate}, it is thus enough to extend all intervals by $2d'-3$ as in \cref{lem:extension_of_gapless_is_gapless}, obtaining an interval representation of dimension at most $4d' - 5$.
    Since $d' \leq d$, the above bound follows.
\end{remark}

\begin{remark}\label{rem:no_blowup_in_dimension}
    The blowup of the dimension of the shift graphs obtained from a minimal interval representation for $P'_{i^*}$ over $Z$, which is performed to guarantee property~\eqref{eq:intervals_share_coordinate}, is not necessary if $P'_{i^*}$ contains a single positive literal.
    Indeed, Case~(I) does not rely on property~\eqref{eq:intervals_share_coordinate}.
    Thus, it remains to consider Case~(II).
    If $j \preceq i^*$ and $P_j$ is distinct from both $P_{i^*}$ and its reverse, then $P_j$ is equal to $P_{i^*}$ (or its reverse) with at least one additional loop at some coordinate $\gamma \not\in L(P_{i^*})$.
    Thus, $P'_j$ falls under Case~(II)-(a), which we need to adapt without assuming property~\eqref{eq:intervals_share_coordinate}.
    For this purpose, it is enough to find positive literals $q_{\beta,\alpha}$ and $q_{\delta, \gamma}$ of $P'_j$ with $\beta - \alpha \not= \delta - \gamma$, and a value $t \in [D]$ that belongs to both intervals $I_\beta$ and $I_\delta$ (or both $I_\alpha$ and $I_\gamma$; see the following footnote).
    Then we may proceed with the argument from Case~(II)-(a), possibly up to exchanging the roles of $u$ and $w$.%
    \footnote{If $t \in I_\alpha$ and $t \in I_\gamma$, we prove that $w_t = u_{\val(\idx(t,\alpha),\beta)} = u_{t + \beta - \alpha}$ and $w_t = u_{\val(\idx(t,\gamma),\delta)} = u_{t + \delta - \gamma}$ by an analogous argument as in Case~(II)-(a). This implies $u_{t + \beta - \alpha} = u_{t + \delta - \gamma}$, which contradicts that $u$ is a strictly increasing vector.}

    Without loss of generality, by permuting the indices, let $q_{2,1}$ be the single positive literal of $P'_{i^*}$.
    Hence, either $q_{2,1}$ or $q_{1,2}$ is the single positive non-loop literal of $P'_j$.
    Let $\cI = \{I_c : c \in [d']\}$ be a \emph{minimal} interval representation of dimension $D$ for $P'_{i^*}$ over $Z$.
    By minimality, we have $I_1 \cup I_2 = [D]$.
    Now fix $t \in I_\gamma$.
    Since $I_1 \cup I_2 = [D]$ we have either $t \in I_1$ or $t \in I_2$.
    Therefore, we reach the desired configuration of two positive literals $q_{\beta,\alpha}$ (where $\{\alpha, \beta\} = \{1,2\}$) and $q_{\gamma,\gamma}$ of the reduced path clause $P'_j$, and a value $t$ in $I_\gamma$ and either $I_\alpha$ or $I_\beta$, which recovers Case~(II)-(a).

    To summarize, $P'_{i^*}$ contains a single positive literal and admits a $D$-dimensional representation over $Z$, then either $\dcS_D \subseteq \dcY_F$ or $\cS_D \subseteq \dcY_F$.
    Note that in this case, Corollary~\ref{cor:functional_contains_SD} implies that $D \leq 3$.
\end{remark}

Finally, we prove the promised \chiunbounded part of \cref{th:main-set-defined}.

\begin{lemma} \label{lem:full_contains_SD}
    Let $\dcX$ be a \nf set-defined digraph class that is \chiunbounded. Then there exists $D \geq 2$ such that $\dcX$ contains the class $\dcS_D$ of $D$-dimensional shift digraphs or the class $\cS_D$ of symmetrized $D$-dimensional shift digraphs.
\end{lemma}
\begin{proof}
\cref{th:decomposition-uniform-Y} and \cref{th:reduction-to-acyclic} imply  that there exists an acyclic DNF $F = \bigvee_{i \in [a]} P_i \vee \bigvee_{j \in [b]} D_j$ such that $\dcY_F$ is \chiunbounded and $\dcY_F \subseteq \dcX$.
Applying  \cref{lem:unbounded_Y_F_contains_shift_or_symmetrized_shift} to $F$ we obtain some $i^* \in [a]$ and a family $\cV$ of subsets of $\Ninj{d}$ such that $F_{i^*} = P_{i^*} \vee \bigvee_{j \in [b]} D_j$ induces on $\cV$ the class $\dcS_D$ for some $D \geq 2$, while all path clauses distinct from $P_{i^*}$ and its reverse induce no edges on any $V \in \cV$.
Consequently $\dcY_F$ contains either $\dcS_D$ or its symmetrization, depending on whether the DNF $F$ does not or does contain the reverse of $P_{i^*}$.
\end{proof}

\MainSetDefined*
\begin{proof}
By Corollary \ref{cor:poly-chi-vs-unions-shift-subgraphs}, if $\dcX$ is \chibounded then it is polynomially \chibounded. Otherwise, by Lemma \ref{lem:full_contains_SD}, $\dcX$ contains a class of shift digraphs or symmetrized shift digraphs with unbounded chromatic number.
\end{proof}

\subsection{Gy\'arf\'as--Sumner for \nf set-defined graph classes}
The following lemma is more general than needed for nearly full classes, but shows that Conjecture \ref{conj:general-set-defined} would imply Gy\'arf\'as--Sumner for all set-defined classes.

\begin{lemma} \label{lem:GS}
    Let $\cX$ be a graph class containing shift graphs of unbounded chromatic number. Then $\cX$ contains all forests.
\end{lemma}
\begin{proof}
     Since every forest is an induced subgraph of a tree, it suffices to show that $\cX$ contains all trees.
      
      Given a shift graph $G$, we let $\vec{G}$ be its natural orientation (i.e. $\vec{G}$ is an induced subdigraph of a canonical shift digraph). Then $\vec{G}$ is the directed line graph of some directed graph $\vec{H}$ (since $\vec{G}$ is an induced subdigraph of some canonical shift digraph $\vec{S}$, which is the directed line graph of a directed clique $\vec{K}$, and removing vertices from $\vec{S}$ corresponds to removing edges from $\vec{K}$). By \cite[Lemma 2.21]{hell2004graphs}, there is some $f \colon \N \to \N$ such that $\lim_{n \to \infty} f(n) = \infty$ and if $\chi(\vec{G}) \geq n$ then $\chi(\vec{H}) \geq f(n)$. By \cite{burr1980subtrees}, if $\chi(\vec{H})$ is sufficiently large, then it contains every oriented tree of size $k+1$ as a subdigraph. 
      This together with the following claim then imply that if $\chi(\vec{G})$ is sufficiently large, then it contains an orientation of every tree of size $k$ as an induced subdigraph, and thus  $G$ contains any tree of size $k$ as an induced subgraph.

      \begin{claim*}
          For every tree $T$, there exists a directed tree $\vec{H}$ such that the directed line graph $\vec{L}(\vec{H})$ of $\vec{H}$ is an orientation of $T$.
      \end{claim*}
      \begin{claimproof}
           Let $T=(V,E)$ and root $T$ at an arbitrary vertex $r$. Let the vertex set of $\vec{H}$ be $V'=\{x_v : v\in V\} \cup \{x_0\}$. For $v\in V$, let $p(v)$ be the parent of $v$ in $T$, and define $p(r)=0$. Let the directed edge set of $\vec{H}$ be $E'=\{e_v=x_{p(v)}x_v : v\in V\}$. Note that the vertex set of $\vec{L}(\vec{H})$ is $\{e_v : v\in V\}$. Also, $e_ve_u$ is an edge in $\vec{L}(\vec{H})$ if and only if $x_v=x_{p(u)}$. This occurs if and only if $v=p(u)$ if and only if $v$ is the parent of $u$ in $T$. Therefore there is an edge in $\vec{L}(\vec{H})$ if and only if there is an edge in $T$, and so $T$ is the underlying graph of $\vec{L}(\vec{H})$.     
      \end{claimproof}
      
      This completes the proof.
\end{proof}

\GS*
\begin{proof}
    By Theorem \ref{th:main-set-defined}, either $\cX$ is polynomially $\chi$-bounded, or contains some $\cS_k$. In the latter case, Lemma \ref{lem:GS} implies $\cX$ contains all forests.
\end{proof}

\subsection{Deciding \texorpdfstring{$\chi$}{Chi}-boundedness in \nf set-defined digraph classes}\label{sec:deciding-chiboundedness}

\Decision*
\begin{proof}
Let $f \colon \{0,1\}^{d^2} \to \{0,1\}$ be given. By \cref{th:from-full-set-defined-to-shift-colorable}, we can compute finitely many pairs $(P_1,Z_1)$, $(P_2,Z_2), \dots, (P_r,Z_r)$
such that
$\dcX_f$ is \chibounded if and only if $\dcY_{P_i,Z_i}$ has bounded chromatic number for every $i \in [r]$. 
By \cref{thm:tropical_dichotomy}, for each $i \in [r]$, bounded chromatic number of $\dcY_{P_i,Z_i}$ can be decided by constructing and solving a pair of finite systems of tropical inequalities. Since the existence of a finite solution to a finite system of tropical inequalities is decidable (see \cref{sssec:solving_MPG_overview}), the theorem follows for the full class $\dcX_f$. 
The theorem follows for the nearly full variant $\dcY_f$ since the first step of the proof of \cref{thm:tropical_dichotomy} is to reduce to injective order-uniform classes (Lemma \ref{lem:fromXtoY_H}), so the statement also holds for $\dcY_f$ by skipping this step.
\end{proof}

\section{Systems of tropical inequalities via deciding \texorpdfstring{$\chi$}{Chi}-boundedness}
\label{sec:compressed_matrices}

In Section~\ref{sec:chi-dichotomy}, we reduced the question of \chiunboundedness of the class $\dcY_{P,Z}$, for a path clause $P$ and a collection of functional constraints $Z$, to the existence of finite solutions for two tropical inequalities (\cref{thm:tropical_dichotomy}).
In this section we show that the existence of a finite solution to an arbitrary tropical system is equivalent to the problem of \chiunboundedness for a set-defined class $\dcY_{P,Z}$ for some $P$ and $Z$. 
Importantly, the reduction is efficient: with an appropriate representation of the path clause and functional constraints, the reduction runs in strongly polynomial time.

We state the main theorem of this section in terms of max-plus systems.
The statement and proof of the min-plus version are analogous and we omit them.

\begin{theorem}\label{thm:tropics_to_chi_unboundedness}
    There is a strongly polynomial-time algorithm that, given arbitrary matrices $A,B \in \Zmax^{k \times m}$, produces a succinct representation (see below) of loopless path clause $P$ and a collection of functional constraints $Z$ such that the following are equivalent:
    \begin{enumerate}[label=(\roman*)]
        \item the chromatic number of $\dcY_{P,Z}$ is unbounded;
        \item the system $A \maxmu y \leq B \maxmu y$ has a finite solution.
    \end{enumerate}
\end{theorem}

We assume that both matrices are represented by $k\times m$ entries, with a special value of constant length indicating non-finite entries.
In order to achieve strongly polynomial time, we need to represent the path clause $P$ and the collection of functional constraints $Z$ succinctly.
To do so, we assume an ordering of coordinates such that the positive literals of $P$ form a subset of $\{q_{c+1,c} : c \in [d-1]\}$.
Under this assumption, a clause $P$ with the clause digraph $\Delta_P$ containing $m$ paths $\rho_1, \dots, \rho_m$ can be represented by the starting index of each path and the value of $d$.
Each individual functional constraint $(L,\lambda) \in Z$, where $L \subseteq [d], \lambda \in [d]$, is represented by $|L|$ values for the set $L$ and the value~$\lambda$.
See \cref{rem:efficient_representation_necessary} for a discussion why a naive representation is insufficient.

\subsection{Outline of the proofs}\label{sec:meanpayoff-outline}

To prove Theorem~\ref{thm:tropics_to_chi_unboundedness}, we first develop in Section~\ref{ssec:efficient_representation} an efficient tropical representation of the constraints stemming from the pair $(P,Z)$ by a min-plus system $\com{\widehat{A}} \com{x} \geq \com{\widehat{B}} \com{x}$ and a max-plus system $\com{\widetilde{A}} \com{y} \leq \com{\widetilde{B}} \com{y}$.
We may view these systems as a compression of the original representation from Section~\ref{sec:chi-dichotomy}, reducing the original matrices $\widehat{A},\widehat{B}, \widetilde{A},\widetilde{B}$ of dimension $(2d-2m+k)\times d$ to matrices $\com{\widehat{A}},\com{\widehat{B}}, \com{\widetilde{A}},\com{\widetilde{B}}$ of dimension $k\times m$, where $m$ is the number of components in the clause digraph  $\Delta_P$ and $k$ is the number of functional constraints in $Z$.
Analogously to Theorem~\ref{thm:tropical_dichotomy}, Lemma~\ref{lem:efficient_representation_solvable_iff_chiunbounded} states that the class $\dcY_{P,Z}$ is \chiunbounded if and only if each of the systems $\com{\widehat{A}} \com{x} \geq \com{\widehat{B}} \com{x}$ and $\com{\widetilde{A}} \com{y} \leq \com{\widetilde{B}} \com{y}$ has a finite solution.

On the other hand, it will be apparent from the definition, and we prove it in Section~\ref{ssec:reduction}, that, given systems $\widehat{A}' x \geq \widehat{B}' x$ and $\widetilde{A}' y \leq \widetilde{B}' y$ of an appropriate form (described in Remark~\ref{rem:shape_of_compressed_systems}), we are able to produce in \emph{strongly polynomial} time a loopless path clause $P$ and a collection of functional constraints $Z$ whose efficient tropical representation is given exactly by the systems $\widehat{A}' x \geq \widehat{B}' x$ and $\widetilde{A}' y \leq \widetilde{B}' y$ (Lemma~\ref{lem:polynomial_reconstruction}).

Finally, we exploit the fact that the structure of systems arising in an efficient representation admits considerable flexibility.
Therefore, we may modify (again, in strongly polynomial time) an arbitrary max-plus system to have the form of the system $\com{\widetilde{A}} \com{y} \leq \com{\widetilde{B}} \com{y}$ for some pair $(P,Z)$.
Therefore, by complementing the max-plus system by a trivially satisfiable min-plus system, we obtain the full efficient representation of a certain pair $(P,Z)$.
Since Lemma~\ref{lem:polynomial_reconstruction} allows us to find such a pair in strongly polynomial time, we obtain Theorem~\ref{thm:tropics_to_chi_unboundedness}.

\begin{notation}\label{not:notation_efficient_representation}
    Recall that the clause digraph $\Delta_P$ of the $d$-dimensional loopless path clause $P$ consists of $m$ paths $\rho_1, \dots, \rho_m$, at least one of which satisfies $|\rho_i| \geq 2$, as discussed above.
    Recall that $p$ is the function mapping an element of $[d]$ to the index of a path of $\Delta_P$ containing it.
    We assume that $Z$ contains $k$ functional constraints $(L, \lambda)$.
    We recall that in \cref{sec:chi-dichotomy} we assumed that functional constraints $(L,\lambda)$ satisfy  $\lambda \not\in L$. This assumption is without loss of generality as if $\lambda \in L$, then every set is $(L,\lambda)$-functional, so the constraint is vacuous and can be removed without affecting the analysis.
    In this section, to reduce the number of special cases, it will be convenient to allow functional constraints $(L,\lambda)$ with $\lambda \in L$.

    We identify a maximal path $\rho_i$ of $\Delta_P$ with the subset of $[d]$ of indices that $\rho_i$ spans.
    Then, it makes sense to speak of the first coordinate of $\rho_i$ as $\min \rho_i$, or to inquire whether $\rho_i \cap L \not= \emptyset$ for some $L \subseteq [d]$, etc.
\end{notation}

\subsection{Efficient representation}\label{ssec:efficient_representation}

To motivate the definition of the efficient representation (\cref{def:eff-rep} below), let us recall the original systems $\widehat{A}x \geq \widehat{B}x$ and $\widetilde{A} y \leq \widetilde{B} y$ from Section~\ref{sec:chi-dichotomy} capturing a loopless path clause $P$ and a collection of functional constraints $Z$.
We shall primarily focus on the system $\widetilde{A} y \leq \widetilde{B} y$, but the case of $\widehat{A}x \geq \widehat{B}x$ is analogous.
The system $\widetilde{A} y \leq \widetilde{B} y$ consists of two parts $\widetilde{A}^\shift y \leq \widetilde{B}^\shift y$ and $\widetilde{A}^\func y \leq \widetilde{B}^\func y$.
The first part encodes the restrictions stemming from $P$ (see \cref{eq:equality_condition}), and the second part encodes the restrictions stemming from $Z$ (see \cref{eq:endpoint_condition}).

While this representation offers some advantages, it is often wasteful.
Indeed, the resulting system $\widetilde{A} y \leq \widetilde{B} y$ has $d$ variables, but variables $y_c, y_{c+1}$ with $p(c) = p(c+1)$ satisfy $y_{c+1} = y_c + 1$ due to \cref{eq:equality_condition}.
Thus, a solution $y \in \Z^d$ can be encoded using only $m$ values as knowing $y_c$ gives
\begin{equation}\label{eq:shift_within_path}
    y_{c'} = y_c + c'-c
\end{equation}
for any $c'$ with $p(c') = p(c)$.
Consequently, for a set $S \subseteq \rho_i$, we have that $\max_{c \in S} y_c = y_{\max S}$.
Therefore, for each functional constraint $(L, \lambda) \in Z$, we only need to record at most one index in $[d]$ per path $\rho_i$ of $\Delta_P$ (unless $\rho_i \cap L = \emptyset$), namely $\max \rho_i \cap L$, to evaluate the max-plus inequality from \cref{eq:endpoint_condition} (and $\min \rho_i \cap L$ to evaluate the min-plus inequality).

This suggests a way of compressing the system $\widetilde{A} y \leq \widetilde{B} y$ to a smaller system $\com{\widetilde{A}} \com{y} \leq \com{\widetilde{B}} \com{y}$.
We will represent the conditions $\widetilde{A}^\shift y \leq \widetilde{B}^\shift y$ implicitly and focus only on the conditions $\widetilde{A}^\func y \leq \widetilde{B}^\func y$.
We reduce each group of columns corresponding to a maximal path in $\Delta_P$ to a single column, while still encoding the constraints of \cref{eq:endpoint_condition}.
We think of the values $\com{y}_i$ as the values $y_{\min \rho_i}$.
See the precise definition below.

\begin{definition}[Efficient representation]
\label{def:eff-rep}
    Let $P, Z$ be as above.
    We define the \emph{efficient representation} of $(P,Z)$ to be the systems $\com{\widehat{A}} \minmu \com{x} \geq \com{\widehat{B}} \minmu \com{x}$ and $\com{\widetilde{A}} \maxmu \com{y} \leq \com{\widetilde{B}} \maxmu \com{y}$ with matrices of dimensions $k \times m$ such that for each $(L_s, \lambda_s) \in Z$, they contain a row representing the respective inequalities
    \begin{equation}\label{eq:endpoint_condition_compressed}
    \begin{split}
        \com{x}_{p(\lambda_s)} + \delta_s 
        \geq 
        \min_{\substack{i \in [m] \text{ s.t.} \\ \rho_i \cap L_s \not= \emptyset}} \com{x}_i + \widehat{\eps}_{s,i}
        , \\
        \com{y}_{p(\lambda_s)} + \delta_s
        \leq
        \max_{\substack{i \in [m] \text{ s.t.} \\ \rho_i \cap L_s \not= \emptyset}} \com{y}_i + \widetilde{\eps}_{s,i},        
    \end{split}
    \end{equation}
    where $\delta_s = \lambda_s - \min \rho_{p(\lambda_s)}, \widehat{\eps}_{s,i} = \min (\rho_i \cap L_s) - \min \rho_i$, and $\widetilde{\eps}_{s,i} = \max (\rho_i \cap L_s) - \min \rho_i$.
\end{definition}

Note that $\delta_s$ is the distance of $\lambda_s$ from the first coordinate of the path to which it belongs.
Similarly, unless $\rho_i \cap L_s = \emptyset$, $\widehat{\eps}_{s,i}$ and $\widetilde{\eps}_{s,i}$ record the offset of the minimal and maximal elements of $L_s$ within $\rho_i$, respectively.

The difference in the definition of $\widehat{\eps}_{s,i}$ and $\widetilde{\eps}_{s,i}$ implies that, unlike for $\widehat{A} x \geq \widehat{B} x$ and $\widetilde{A} y \leq \widetilde{B} y$, the matrices of systems $\com{\widehat{A}} \com{x} \geq \com{\widehat{B}} \com{x}$ and $\com{\widetilde{A}} \com{y} \leq \com{\widetilde{B}} \com{y}$ are not necessarily twins to each other (cf.\ with Observation~\ref{obs:tropical_systems_are_the_same}).
Let us comment on the general form of systems that may arise as an efficient representation of a pair $(P,Z)$.

\begin{remark}\label{rem:shape_of_compressed_systems}
    Let $\com{\widehat{A}} \com{x} \geq \com{\widehat{B}} \com{x}$ and $\com{\widetilde{A}} \com{y} \leq \com{\widetilde{B}} \com{y}$ form an efficient representation of a pair $(P,Z)$.
    Then, each of these systems satisfies the following properties:
    \begin{enumerate}[label=(\Roman*)]
        \item\label{it:shape_of_compressed_systems_non-negative} the finite entries of both matrices are non-negative integers,
        \item\label{it:shape_of_compressed_systems_A_finite_entry} each row of the matrix on the left-hand side has exactly one finite entry,
        \item\label{it:shape_of_compressed_systems_B_finite_entry} each row of the matrix on the right-hand side has a finite entry.
    \end{enumerate}
    Moreover, the systems are related as follows
    \begin{enumerate}[label=(\Roman*)]
        \setcounter{enumi}{3}
        \item\label{it:shape_of_compressed_systems_A_twins} the matrices $\com{\widehat{A}}$ and $\com{\widetilde{A}}$ are twins to each other,
        \item\label{it:shape_of_compressed_systems_B_inequality} for each $i,j$ it holds that $\com{\widehat{B}}_{i,j}$ is finite iff $\com{\widetilde{B}}_{i,j}$ is finite, and in such a case $\com{\widehat{B}}_{i,j} \leq \com{\widetilde{B}}_{i,j}$.
    \end{enumerate}
\end{remark}

\begin{remark}\label{rem:compression_is_lossy}
    At this point, it should be clear that distinct pairs $(P,Z)$ may have the same efficient representation.
    Notably, we completely lose the information regarding the original path lengths.
    We only know for each $i \in [m]$ that 
    \[
        |\rho_i|
        \geq 
        \max \Big\{ 
            \max_{\substack{s \in [k] \text{ s.t.} \\ \lambda_s \in \rho_i}} \delta_s,~~            
            \max_{\substack{s \in [k] \text{ s.t.} \\ \rho_i \cap L_s \not= \emptyset}}
            \widetilde{\eps}_{s,i}
        \Big\}
        + 1
        .
    \]
    (The $+1$ compensates for the fact that the values $\delta_s$ and $\widetilde{\eps}_{s,i}$ record merely the \emph{offset} from the first position of $\rho_i$.)
    Moreover, we generally also lose the information regarding the exact form of the sets $L_s, s \in [k]$.
    We know for each $s \in [k]$ and $i \in [m]$ with $\rho_i \cap L_s \not= \emptyset$, that
    \begin{align*}
        \min \rho_i \cap L_s &= \min \rho_i + \widehat{\eps}_{s,i}
        , \\
        \max \rho_i \cap L_s &= \min \rho_i + \widetilde{\eps}_{s,i}
        ,
    \end{align*}
    but we cannot determine what happens within the interval $[\min \rho_i \cap L_s, \max \rho_i \cap L_s]$.
\end{remark}

\begin{example}
    Continuing from Example~\ref{ex:chi_bnd_more_interesting_matrices}, consider the efficient representation for the path clause $P$, with paths $\{1,2\}$ and $\{3,4,5\}$, and a collection of functional constraints $Z = ((L_1, \lambda_1), (L_2, \lambda_2))$, where $L_1 = \{2,3\}, \lambda_1 = 5, L_2 = \{ 3 \}, \lambda_2 = 1$. 
    The max-plus system has the form
    \begin{align*}
        \begin{pmatrix}
        -\infty & 2 \\
        0 &-\infty 
        \end{pmatrix}
        \maxmu
        \begin{pmatrix}
        \com{y}_1 \\ \com{y}_2
        \end{pmatrix}
         \leq
        \begin{pmatrix}
        1 & 0 \\
        -\infty & 0
        \end{pmatrix}
        \maxmu
        \begin{pmatrix}
        \com{y}_1 \\ \com{y}_2
        \end{pmatrix}
        .
    \end{align*}
    Since both
    $L_1$ and $L_2$ have at most one element in each path, the matrices of the min-plus system are the twins of $\com{\widetilde{A}}, \com{\widetilde{B}}$, respectively.
\end{example}

The following lemma shows that we may indeed use the efficient representation to determine whether the class $\dcY_{P,Z}$ has bounded chromatic number.

\begin{lemma}\label{lem:efficient_representation_solvable_iff_chiunbounded}
    Let $P$ be a loopless path clause, $Z$ a collection of functional constraints and let $\com{\widehat{A}} \com{x} \geq \com{\widehat{B}} \com{x}$, $\com{\widetilde{A}} \com{y} \leq \com{\widetilde{B}} \com{y}$ be the efficient representation of $(P,Z)$.
    The following are equivalent:
    \begin{enumerate}[label=(\roman*)]
        \item the chromatic number of $\dcY_{P,Z}$ is unbounded,
        \item each of the tropical inequalities $\com{\widehat{A}} \com{x} \geq \com{\widehat{B}} \com{x}$ and $\com{\widetilde{A}} \com{y} \leq \com{\widetilde{B}} \com{y}$ has a finite solution.
    \end{enumerate}
\end{lemma}
\begin{proof}
    In view of Theorem~\ref{thm:tropical_dichotomy}, it is enough to prove that the system $\widetilde{A} y \leq \widetilde{B} y$ from the statement of Theorem~\ref{thm:tropical_dichotomy} has a finite solution $y \in \Z^d$ if and only if the system $\com{\widetilde{A}} \com{y} \leq \com{\widetilde{B}} \com{y}$ has a finite solution $\com{y} \in \Z^m$; and likewise for the min-plus systems $\widehat{A} x \geq \widehat{B} x$ and $\com{\widehat{A}} \com{x} \geq \com{\widehat{B}} \com{x}$.
    We prove the equivalence only for max-plus systems.
    The min-plus version is analogous except that the inequalities are reversed and all maxima are replaced by minima, so we omit the details.

    ``$\Rightarrow$''
    Let $y \in \Z^d$ be a solution to $\widetilde{A} y \leq \widetilde{B} y$, i.e. $y$ satisfies the respective constraints in \cref{eq:equality_condition} and \cref{eq:endpoint_condition}.
    We define $\com{y} \in \Z^m$ by setting $\com{y}_i = y_{\min \rho_i}$ for each $i \in [m]$.
    Consider some $(L_s, \lambda_s) \in Z$ and the corresponding inequality from $\com{\widetilde{A}} \com{y} \leq \com{\widetilde{B}} \com{y}$ of the form 
    \[
        \com{y}_{p(\lambda_s)} + \delta_s
        \leq
        \max_{\substack{i \in [m] \text{ s.t.} \\ \rho_i \cap L_s \not= \emptyset}} \com{y}_i + \widetilde{\eps}_{s,i}
        .
    \]
    By definition of $\delta_s, \widetilde{\eps}_{s,i}$, and \cref{eq:shift_within_path}, this is equivalent to
    \[
        y_{\lambda_s}
        \leq
        \max_{\substack{i \in [m] \text{ s.t.} \\ \rho_i \cap L_s \not= \emptyset}} y_{\max \rho_i \cap L_s}
        .
    \]
    Since, due to \cref{eq:equality_condition}, $y_{\max \rho_i \cap L_s} = \max_{c \in \rho_i \cap L_s} y_c$, we may rewrite the right-hand side as follows
    \[
        y_{\lambda_s}
        \leq
        \max_{\substack{i \in [m] \text{ s.t.} \\ \rho_i \cap L_s \not= \emptyset}} \max_{c \in \rho_i \cap L_s} y_c
        ,
    \]
    which, by merging the two maxima, is the same as the constraint
    \[
        y_{\lambda_s}
        \leq
        \max_{c \in L_s} y_c
    \]
    from \cref{eq:endpoint_condition}.
    Hence, $\com{y}$ is a solution to $\com{\widetilde{A}} \com{y} \leq \com{\widetilde{B}} \com{y}$.
    
    ``$\Leftarrow$''
    Let $\com{y} \in \Z^m$ be a solution to $\com{\widetilde{A}} \com{y} \leq \com{\widetilde{B}} \com{y}$.
    We define $y \in \Z^d$ by setting $y_c = \com{y}_{p(c)} + c - \min \rho_{p(c)}$ for each $c \in [d]$.
    We immediately get that $y$ is a solution to $\widetilde{A}^\shift y \leq \widetilde{B}^\shift y$.
    Hence, it remains to verify that
    \[
        y_{\lambda_s}
        \leq
        \max_{c \in L_s} y_c
        .
    \]
    Same steps as above prove that this inequality is equivalent to
    \[
        \com{y}_{p(\lambda_s)} + \delta_s
        \leq
        \max_{\substack{i \in [m] \text{ s.t.} \\ \rho_i \cap L_s \not= \emptyset}} \com{y}_i + \widetilde{\eps}_{s,i}
    \]
    from \cref{eq:endpoint_condition_compressed}.
    Hence, $y$ is a solution to $\widetilde{A} y \leq \widetilde{B} y$.
\end{proof}

\subsection{Reduction}\label{ssec:reduction}

Here we prove Theorem~\ref{thm:tropics_to_chi_unboundedness}.
As the first step, we prove a key lemma stating that we may translate tropical systems of an appropriate form to a pair $(P,Z)$ such that the systems are the efficient representation of $(P,Z)$.
In view of Remark~\ref{rem:compression_is_lossy}, the choice of $(P,Z)$ is not unique.

\begin{lemma}\label{lem:polynomial_reconstruction}
    There is a strongly polynomial-time algorithm that, given systems $\widehat{A}' x' \geq \widehat{B}' x'$ and $\widetilde{A}' y' \leq \widetilde{B}' y'$ satisfying properties \ref{it:shape_of_compressed_systems_non-negative}-\ref{it:shape_of_compressed_systems_B_inequality} from Remark~\ref{rem:shape_of_compressed_systems}, produces a loopless path clause $P$ and a collection of functional constraints $Z$ such that $\widehat{A}' x' \geq \widehat{B}' x'$ and $\widetilde{A}' y' \leq \widetilde{B}' y'$ is the efficient representation of $(P,Z)$.
\end{lemma}
\begin{proof}
    Suppose that matrices $\widehat{A}', \widehat{B}', \widetilde{A}'$, $\widetilde{B}'$ have dimensions $k \times m$.
    We define a loopless path clause $P$ consisting of $m$ paths $\rho_1, \dots, \rho_m$ such that
    \[
        |\rho_i| =
        \max_{s \in [k]} 
        \Big\{
            \widetilde{A}'_{s,i},            
            \widetilde{B}'_{s,i},
            1
        \Big\}
        + 1
        .
    \]
    Thus, the starting vertex of $\rho_i$ is $1 + \sum_{j < i} |\rho_j|$, and $d = \sum_{i \in [m]} |\rho_i|$.
    
    We further define a set $Z$ of $k$
    functional constraints $(L_s, \lambda_s), s \in [k]$, one for each row of the systems.
    Suppose that $\widetilde{A}'_{s,i_s}$ is the unique finite entry of the row $\widetilde{A}'_{s,*}$, which exists due to property \ref{it:shape_of_compressed_systems_A_finite_entry} in Remark~\ref{rem:shape_of_compressed_systems}.
    We set $\lambda_s = \widetilde{A}'_{s,i_s} + \min \rho_{i_s}$ and define the set $L_s$ by defining its intersection with each path $\rho_i$ (which is enough to uniquely define $L_s$, since $\{\rho_i: i \in [m]\}$ partitions $[d]$).
    Specifically, for each $i\in [m]$, if $\widetilde{B}'_{s,i}$ is infinite, then we define $\rho_i\cap L_s=\emptyset$; otherwise both $\widehat{B}'_{s,i}$ and $\widetilde{B}'_{s,i}$ are finite by property \ref{it:shape_of_compressed_systems_B_inequality} and we define
    \[
        \rho_i \cap L_s = \{\widehat{B}'_{s,i} + \min \rho_i, \widetilde{B}'_{s,i} + \min \rho_i\}
        .
    \]
    
    Then, $P$ is a valid loopless path clause with $|\rho_i| \geq 2$ for each $i \in [m]$, and $Z$ a valid collection of functional constraints according to Notation~\ref{not:notation_efficient_representation}.      
    Clearly, the construction of $(P,Z)$ can be performed in strongly polynomial time as each of the linearly many defining expressions can be computed by polynomially many arithmetic operations.
    Moreover, by inspecting the definition of the efficient representation for $(P,Z)$, in particular of the constants $\delta_s, \widehat{\eps}_{s,i}$ and $\widetilde{\eps}_{s,i}$, it is evident that the efficient representation of $(P,Z)$ exactly corresponds to the systems $\widehat{A}' x' \geq \widehat{B}' x'$ and $\widetilde{A}' y' \leq \widetilde{B}' y'$.
\end{proof}

\begin{remark}\label{rem:efficient_representation_necessary}
    The natural representation of a full $d$-dimensional clause $C$ with $d^2$ variables is a sequence of $d^2$ bits indicating which variables appear positively or negatively in $C$.
    The path clause $P$ constructed in \cref{lem:polynomial_reconstruction} contains $m$ paths whose lengths are proportional to the numerical values in the matrices $\widetilde{A}',\widetilde{B}'$.
    Since the dimension of $P$ is proportional to the numerical values on the input, it can be exponential in the input's bit-length.
    Thus, we cannot allocate $d^2$ to store the path clause $P$ in this naive way.%
    \footnote{Recall from Section~\ref{sssec:algorithms} that a strongly polynomial algorithm performs at most polynomially many arithmetic operations with respect to the number of integers on the input, and uses at most polynomial space with respect to the bit-length of the input.}

    Instead, we use the succinct representation specified after statement of \cref{thm:tropics_to_chi_unboundedness}, storing only $m$ values per $P$ and $|L| + 1$ per each functional constraint $(L, \lambda)$ in $Z$.
    Since we construct each $(L, \lambda)$ so that $|L| \leq 2m$, we need only $m + k(2m+1)$ integers (whose sizes are proportional to the numbers on the input) to store the pair $(P,Z)$, which complies with the memory requirements on a strongly polynomial-time algorithm.
\end{remark}

As a next step we show that an arbitrary max-plus system can be efficiently turned into a system satisfying the first three properties in Remark~\ref{rem:shape_of_compressed_systems}.

\begin{lemma}\label{lem:reduction_of_arbitrary_to_reasonable}
    There is a strongly polynomial-time algorithm that, given an arbitrary max-plus system $Ay \leq By$ with $A,B \in \Zmax^{k \times m}$, produces a system $A'y \leq B'y$ with the same finite solutions, which satisfies properties~\ref{it:shape_of_compressed_systems_non-negative}-\ref{it:shape_of_compressed_systems_B_finite_entry} from Remark~\ref{rem:shape_of_compressed_systems}.
\end{lemma}
\begin{proof}
    We perform a series of steps reducing $Ay \leq By$ to the desired system $A'y \leq B'y$.
    First, we add a large enough constant to both $A$ and $B$ to ensure that all the finite entries of the resulting matrices $A^{(1)}, B^{(1)}$ are non-negative to comply with property \ref{it:shape_of_compressed_systems_non-negative}.
    This clearly does not change the solution set.
    
    As the next step, we obtain matrices $A^{(2)}, B^{(2)}$ from $A^{(1)}, B^{(1)}$ by splitting each constraint
    \[
        \max_{i \in [m]} A^{(1)}_{s,i} + y_i
        \leq
        \max_{i \in [m]} B^{(1)}_{s,i} + y_i
    \]
    of the system into $m$ constraints of the form
    \begin{align*}
        A^{(1)}_{s,1} + y_1 &\leq \max_{i \in [m]} B^{(1)}_{s,i} + y_i, \\
        A^{(1)}_{s,2} + y_2 &\leq \max_{i \in [m]} B^{(1)}_{s,i} + y_i, \\
        &\shortvdotswithin{\leq}
        A^{(1)}_{s,m} + y_m &\leq \max_{i \in [m]} B^{(1)}_{s,i} + y_i.
    \end{align*}
    This collection of constraints is clearly equivalent to the original one, so we do not change the solution set.
    Moreover, each row of the resulting matrix $A^{(2)}$ contains at most one finite entry. 
    Now, if there is a row of $A^{(2)}$ without a finite entry, we delete that row from $A^{(2)}$ and the respective row from $B^{(2)}$. This does not change the solution set as the corresponding inequality is trivially satisfied due to its left-hand side being $-\infty$. We denote the resulting matrices by $A^{(3)}$ and $B^{(3)}$, respectively.
    
    If $A^{(3)}$ and $B^{(3)}$ do not contain any rows, then we conclude that $Ay \leq By$ is solved by any $y \in \Z^m$.
    Hence, we output the trivially satisfiable system with $A' = B' = (0, -\infty, \dots, -\infty) \in \Zmax^{1 \times m}$, which corresponds to the inequality $y_1 \leq y_1$,   
    and satisfies properties \ref{it:shape_of_compressed_systems_non-negative}-\ref{it:shape_of_compressed_systems_B_finite_entry}.
    Otherwise, both matrices $A^{(3)}, B^{(3)}$ have at least one row, and the corresponding system satisfies properties \ref{it:shape_of_compressed_systems_non-negative}~and~\ref{it:shape_of_compressed_systems_A_finite_entry}, by construction. If the system also satisfies property \ref{it:shape_of_compressed_systems_B_finite_entry}, then we output it.
    
    If the system does not satisfy property \ref{it:shape_of_compressed_systems_B_finite_entry}, i.e.\ it contains a row where $B^{(3)}$ has all entries $-\infty$, we conclude that the system $Ay \leq By$ has no finite solution.
    Indeed, the left-hand side of the corresponding constraint is finite due to property \ref{it:shape_of_compressed_systems_A_finite_entry}, while the right-hand side is $-\infty$.
    Thus, we conclude by outputting the trivially unsatisfiable system with $A' = (1, -\infty, \dots, -\infty) \in \Zmax^{1 \times m}$ and $B' = (0, -\infty, \dots, -\infty) \in \Zmax^{1 \times m}$, which corresponds to the inequality $1 + y_1 \leq y_1$,
    and satisfies properties \ref{it:shape_of_compressed_systems_non-negative}-\ref{it:shape_of_compressed_systems_B_finite_entry}.

    The algorithm clearly runs in strongly polynomial time.
\end{proof}

Now we are ready to prove Theorem~\ref{thm:tropics_to_chi_unboundedness}.

\begin{proof}[Proof of Theorem~\ref{thm:tropics_to_chi_unboundedness}]
    Let $Ay \leq By$ be the given max-plus system with $m$ columns.
    We first apply Lemma~\ref{lem:reduction_of_arbitrary_to_reasonable} to obtain a system $\widetilde{A}'y \leq \widetilde{B}' y$ with the same solution set, that satisfies properties~\ref{it:shape_of_compressed_systems_non-negative}-\ref{it:shape_of_compressed_systems_B_finite_entry} from Remark~\ref{rem:shape_of_compressed_systems}.
    We consider the min-plus system $\widehat{A}'x \geq \widehat{B}' x$, where $\widehat{A}'$ is the twin of $\widetilde{A}'$ and $\widehat{B}'$ defined as
    \begin{align*}
        \widehat{B}'_{i,j} =
        \begin{cases}
            0 & \text{if $\widetilde{B}'_{i,j}$ is finite}, \\
            +\infty & \text{otherwise}.
        \end{cases}
    \end{align*}
    Note that the system $\widehat{A}'x \geq \widehat{B}'x$ is trivially satisfied by the vector $0 \in \Z^m$ as each row of the matrix $\widetilde{B}'$ contains a finite entry by property \ref{it:shape_of_compressed_systems_B_finite_entry}.
    
    The systems $\widehat{A}'x \geq \widehat{B}'x$ and $\widetilde{A}'y \leq \widetilde{B}' y$ satisfy all conditions from Remark~\ref{rem:shape_of_compressed_systems}.
    Hence, we may use Lemma~\ref{lem:polynomial_reconstruction} to find a loopless path clause $P$ and collection of functional constraints $Z$ whose efficient representation are exactly the systems $\widehat{A}'x \geq \widehat{B}'x$ and $\widetilde{A}'y \leq \widetilde{B}' y$.
    Therefore, by Lemma~\ref{lem:efficient_representation_solvable_iff_chiunbounded}, the class $\dcY_{P, Z}$ has unbounded chromatic number if and only if each of the systems $\widehat{A}'x \geq \widehat{B}'x$ and $\widetilde{A}'y \leq \widetilde{B}' y$ admits a finite solution.
    As noted above, the min-plus system $\widehat{A}'x \geq \widehat{B}'x$ admits a finite solution.
    Therefore, the second part of the equivalence reduces to the statement that the system $\widetilde{A}'y \leq \widetilde{B}' y$ admits a finite solution, which is true if and only if the initial system $Ay \leq By$ admits a finite solution.
    
    In conclusion, the class $\dcY_{P, Z}$ has unbounded chromatic number if and only if the system $Ay \leq By$ admits a finite solution.
    Since all the invoked algorithms, as well as the construction of the system $\widehat{A}'x \geq \widehat{B}'x$, require strongly polynomial time, the whole reduction runs in strongly polynomial time.
\end{proof}

\section{Further applications and results on set-defined classes}

    \subsection{Fractional \texorpdfstring{$\chi$}{χ}-boundedness}
    \label{sec:fractional}
    
    In this section we apply our decomposition \cref{th:main-decomposition} to deduce that any set-defined graph class is \emph{polynomially fractionally \chibounded}.

    The fractional chromatic number $\chi_f(G)$ of a graph $G$ is the optimal value of the linear programming relaxation of the standard integer programming formulation of the chromatic number $\chi(G)$. We omit the formal definition, as we will not need it, and use only a few basic properties of the fractional chromatic number. Both the definition and these properties can be found in \cite{SU11}.

    We say that a hereditary graph class $\cX$ is \emph{fractionally \chibounded}, denoted \chifbounded, if there exists a function $h \colon \bN \to \bN$ such that $\chi_f(G) \le h(\omega(G))$ for every $G \in \cX$. If $h$ can be taken to be polynomial, then we say that $\cX$ is polynomially \chifbounded.
    We remark that the fractional clique number is defined as the optimal value of the linear program dual to that defining the fractional chromatic number. Thus, by linear programming duality, these two parameters coincide, which explains the use of the ordinary clique number rather than its fractional counterpart in the above definition.

    Many of the standard constructions of triangle-free graphs of unbounded chromatic number also have unbounded fractional chromatic number. This includes Mycielski graphs \cite{larsen1995fractional}, Zykov graphs \cite{charbit2011fractional}, and Burling graphs \cite{walczak2015triangle}. However, shift graphs have bounded fractional chromatic number \cite{simonyi2011directed}; together with Theorem \ref{th:main-decomposition}, this quickly yields the following result.

    \CorFracChi*
    \begin{proof}
        The fractional chromatic number of every shift graph is at most $4$ (see, e.g., \cite[\S 3]{simonyi2011directed}). 
        To combine this bound with \cref{th:main-decomposition}, we recall three basic properties of the fractional chromatic number, analogous to those of the ordinary chromatic number (see, e.g., \cite{SU11}).

        First, if a graph $H$ admits a homomorphism to a graph $G$, then $\chi_f(H) \leq \chi_f(G)$. Second, the fractional chromatic number of the union of two graphs is at most the product of their fractional chromatic numbers. Finally, if $V_1,V_2$ form a partition of $V(G)$, then $\chi_f(G) \leq \chi_f(G[V_1]) + \chi_f(G[V_2])$.

        These properties, together with the bound for shift graphs and \cref{th:main-decomposition}, imply that there exist a constant $c \in \mathbb{N}$ and a polynomial $\tau \colon \mathbb{N} \to \mathbb{N}$ such that $\chi_f(G) \leq 4^c \cdot \tau(\omega(G))$ holds for every graph $G \in \cX$. 
    \end{proof}

    We conclude by discussing the relevance of a very special case of this result to quantum computing.
    A key ingredient in the efficient quantum algorithm for shadow tomography developed in~\cite{King2025triply} is the polynomial \chifboundedness of the set-defined graph classes whose vertices are $2k$-element sets, with two vertices adjacent if and only if their corresponding sets intersect in an odd number of elements (see~\cite[\S~1.3]{janzer2026sunflowers}).

    \cref{prop:frac chi} immediately recovers the polynomial \chifboundedness of these classes. Moreover, our results yield a binding polynomial of degree at most $2^{2k-1}$. Indeed, these classes can be defined by injective DNFs in which every clause has an odd number of positive literals. The discrete clauses are in bijection with the odd-cardinality subsets of $[2k]$, each of which specifies the loop coordinates of the corresponding clause. Thus, there are exactly $2^{2k-1}$ discrete clauses. From~\cref{th:diagonal-multi}, every graph $G$ in such a class admits a partition of its vertex set into a number of parts bounded by a polynomial in $\omega(G)$ of degree at most $2^{2k-1}$, with each part inducing a graph of bounded fractional chromatic number.
    
    This gives a polynomial $\chi_f$-binding function of degree at most $2^{2k-1}$, which is better than the bound $O((2k)^{k+1})$ obtained in \cite{King2025triply}, but significantly worse than the bound $k$ proved in \cite[Theorem~5.6]{janzer2026sunflowers}.
    The latter is optimal by the matching lower bound in~\cite[Theorem~5.19]{ABSZ24}.\footnote{Although \cite[Theorem~5.19]{ABSZ24} speaks of the ordinary chromatic number, the proof proceeds by computing the ratio between the number of vertices and independence number of the graph; hence, it gives a lower bound for the fractional chromatic number as well.}

\subsection{A case study: generalized Johnson graphs}
\label{sec:johnson}

Let $d \in \N$ and $\cA \subseteq [0,d-1]$. 
The class $\cJ(d,\cA)$ of $(d,\cA)$-generalized Johnson graphs is the class of graphs admitting a realization such that each vertex $v$ is assigned a set $S_v$ of size $d$ (all assigned sets are distinct), and vertices $u,v$ are adjacent iff $|S_u \cap S_v| \in \cA$.%
\footnote{The assumption that $d \not\in \cA$ is without loss of generality, as distinct sets of size $d$ cannot agree on $d$ elements.}
In this section, we apply our machinery to characterize the \chiboundedness of classes of generalized Johnson graphs.

The generalized Johnson graph classes are set-defined. More precisely, for every $d$ and $\cA$, the class $\cJ(d,\cA)$ is exactly the \nf $d$-dimensional set-defined class of graphs whose vertices are distinct vectors of length $d$ with some uniform ordering of their entries and the defining formula is the injective DNF $F(d,\cA)$ consisting of all the injective clauses with exactly $k$ positive literals, for each $k \in \cA$.

\johnson*

To prove these results, we begin by recalling how the reductions from the preceding sections
apply to $F(d,\cA)$. By \cref{th:reduction-to-acyclic}, the cyclic
clauses may be discarded. Every remaining clause is either a path
clause or a discrete clause, and \cref{lem:path+discrete-clauses}
shows that $\cJ(d,\cA)$ is \chiunbounded if and only if there is a
path clause $P$ of $F(d,\cA)$ such that the class defined by the DNF consisting of $P$
and all the discrete clauses of $F(d,\cA)$ is
\chiunbounded.

Consider such a DNF with a path clause $P$.
Let $R:=L(P)$, $\ell := |R|$, and $S := [d]\setminus R$, where $L(C)$ denotes the set of loop coordinates of a clause $C$. In the loop-elimination construction of
\cref{lem:loop-elimination}, we retain precisely those discrete
clauses $C$ for which $R\subseteq L(C)$ and restrict all retained
clauses to the coordinates in $S$. In particular, this produces the
loopless path clause $P'=P_{|S}$, which we call the
\emph{loopless reduction} of $P$. After relabelling the remaining
$d'=d-\ell$ coordinates, the reduced DNF consists of $P'$ together
with the discrete clauses of $F(d',\cA')$, where
$\cA'=\{t-\ell:t\in\cA,\ t\geq\ell\}$.
By permuting coordinates, we may assume that all positive literals of $P'$ are of the form $q_{i+1,i}$ for some $i \in [d'-1]$.

For each discrete clause $C$ in the reduced DNF, choosing a non-loop
coordinate $\lambda_C\in[d']\setminus L(C)$ gives the functional
constraint $(L(C),\lambda_C)$. A collection $Z$ of functional
constraints consists of one such choice for every reduced discrete
clause. An interval assignment $\cI=\{I_i:i\in[d']\}$ satisfies the
constraint $(L(C),\lambda_C)$ if
$I_{\lambda_C}\subseteq\bigcup_{i\in L(C)}I_i$. It is an interval
representation for $P'$ over $Z$ if it satisfies all these functional
constraints and, the shift condition $I_{i+1}=I_i+1$ for every positive literal $q_{i+1,i}$ of $P'$.

By \cref{th:diaganal-to-path}, the class defined by the reduced DNF is \chiunbounded if and
only if there is a collection $Z$ for which the class
$\dcY_{P',Z}$ has unbounded chromatic number. By
\cref{thm:tropical_dichotomy} and
\cref{lem:interval_representation_iff_tropical_solutions}, this is
equivalent to the existence of an interval representation for $P'$
over $Z$. Consequently, $\cJ(d,\cA)$ is \chiunbounded if and only if
there is a path clause $P$ of $F(d,\mathcal A)$ whose loopless reduction $P'$ admits an interval representation $\mathcal I$ over some collection $Z$ of functional constraints arising from the reduced discrete clauses. Moreover, by
\cref{lem:interval_representation_implies_shift_graphs}, if $\cI$ is $D$-dimensional, then $\dcS_D\subseteq\dcY_{P',Z}$ and therefore witnesses \chiunboundedness.
By extension, we say that such a path clause $P$ witnesses \chiunboundedness of $\cJ(d,\cA)$.

Proving that $\dcS_D\subseteq\dcY_{P',Z}$ does not imply that $\dcS_D\subseteq J(d,\cA)$ without further assumptions.
However, as the proof of \cref{lem:unbounded_Y_F_contains_shift_or_symmetrized_shift} shows, a sufficient additional assumption is that $P'$ contains only a single positive literal, see \cref{rem:no_blowup_in_dimension}.
In \cref{thm:johnson_graphs_main}, we show that \chiunboundedness of $\cJ(d,\cA)$ is always witnessed by a path clause $P$ whose path reduction $P'$ contains a single positive literal, so the remark applies.
In particular, since $F(d,\cA)$ contains the reverse of $P$, i.e. the path clause $Q$ with $\Delta_Q$ being the reverse of $\Delta_P$, we obtain $\cS_D \subseteq J(d,\cA)$ for $D \leq 3$, where $\cS_D$ denotes the class of symmetrized $D$-dimensional shift graphs.

Before proving \cref{thm:johnson_graphs_main}, we introduce the necessary notation and establish some auxiliary results.
Let $\cI = \{I_i : i \in [d]\}$ be an interval assignment.
For $S \subseteq \bigcup_i I_i$, we denote $\cC_S = \{i \in [d] : I_i \subseteq S\}$, and $N_S = |\cC_S|$.
We write $\#S$ for for the number of coordinates whose assigned interval is exactly $S$, that is, $\# S = |\{i \in [d] : I_i = S\}|$.
Clearly, $N_S = \sum_{T \subseteq S} \# T$.

The following lemma gives a criterion for an interval assignment to satisfy a collection of functional constraints associated with the discrete clauses of $F(d',\cA')$.

\begin{lemma}\label{lem:johnson_graphs_existence_of_Z}
    Let $\cI = \{I_i : i \in [d']\}$ be an interval assignment.
    There is a collection of functional constraints $Z$ for the discrete clauses of $F(d',\cA')$ so that $\cI$ satisfies every  functional constraint of $Z$ if and only if for each $S \subseteq \bigcup_i I_i$ we have $N_S \not\in \cA'$.
\end{lemma}
\begin{proof}
    Suppose first that there is a collection $Z$ of functional constraints for the discrete clauses of $F(d',\cA')$ and that $\cI$ satisfies every constraint in $Z$. Fix $S\subseteq\bigcup_i I_i$, and suppose for a contradiction that $N_S\in \cA'$. By the definition of $F(d',\cA')$, it contains a discrete clause $C$ with $L(C)=\cC_S$. Let $(L(C),\lambda_C)\in Z$ be the functional constraint corresponding to $C$. Since $\lambda_C\notin L(C)=\cC_S$, by the definition of $\cC_S$ we have  $I_{\lambda_C}\nsubseteq S$. On the other hand, $I_i\subseteq S$ for every $i\in L(C)=\cC_S$, and hence $\bigcup_{i\in L(C)}I_i\subseteq S$. It follows that $I_{\lambda_C}\nsubseteq\bigcup_{i\in L(C)}I_i$, contradicting the assumption that $(L(C), \lambda_C)$ satisfies $\cI$. Hence, $N_S\notin \cA'$. Since $S$ was arbitrary, this holds for every
    $S\subseteq\bigcup_i I_i$.

    Conversely, suppose that $N_S\notin \cA'$ for every
    $S\subseteq\bigcup_i I_i$. We construct the required collection
    $Z$. For each discrete clause $C$ of $F(d',\cA')$, let
    $S_C=\bigcup_{i\in L(C)}I_i$. Every $i\in L(C)$ satisfies
    $I_i\subseteq S_C$, so $L(C)\subseteq\cC_{S_C}$. Moreover,
    $|L(C)|\in \cA'$, whereas
    $|\cC_{S_C}|=N_{S_C}\notin \cA'$. Consequently,
    $L(C)\subsetneq\cC_{S_C}$. Choose
    $\lambda_C\in\cC_{S_C}\setminus L(C)$. Then $\lambda_C$ is a
    non-loop coordinate of $C$, and its membership in $\cC_{S_C}$
    gives
    $I_{\lambda_C}\subseteq S_C=\bigcup_{i\in L(C)}I_i$.
    Thus, $\cI$ satisfies the functional constraint
    $(L(C),\lambda_C)$. Making this choice for every discrete clause
    $C$ gives the required collection $Z$.
\end{proof}

\begin{observation}\label{obs:johnson_graphs_path_clause_minimality}
    Let $P$ be a path clause with $\ell$ loops witnessing that the class $\cJ(d,\cA)$ is \chiunbounded.
    Then $\ell \not\in \cA$.
    Moreover, if $Q \in F(d,\cA)$ is a path clause with $L(Q) = L(P)$, and the positive literals of $Q$ form a subset of the positive literals of $P$, then $Q$ witnesses \chiunboundedness of $\cJ(d,\cA)$.
    
    In particular, the \chiunboundedness of $\cJ(d,\cA)$ can be witnessed by a path clause $Q \in F(d,\cA)$ with $L(Q) = L(P)$ such that its loopless reduction $Q'$ contains $k = \min \cA'$ positive literals, where $\cA' = \{t - \ell : t \in \cA, t \geq \ell\}$.
\end{observation}
\begin{proof}
    Let $P'$ be the loop reduction of $P$, and $\cI$ be the interval representation for $P'$ over a collection $Z$ of functional constraints arising from the reduced discrete clauses.
    If $\ell \in \cA$, then $F(d,\cA)$ contains the discrete clause $C$ with $L(C) = L(P)$.
    Thus, the discrete clause $C'$ obtained by the loop reduction of $C$ contains no positive literal.
    Consequently, $Z$ contains an unsatisfiable constraint of the form $(\emptyset, \lambda)$, contradicting the existence of $\cI$, see \cref{rem:each_discrete_clause_has_a_loop}.

    Moreover, the path clause $Q$ witnesses \chiunboundedness of $\cJ(d,\cA)$, since the loop reduction $Q'$ of $Q$ imposes a subset of conditions (with respect to $P'$) on the interval representation $\cI$.
    Thus, $\cI$ is an interval representation for $Q'$ over $Z$.

    The `In particular' part describes the extremal case of a path clause $Q \in F(d,\cA)$ with $L(Q) = L(P)$ but otherwise minimal possible number of positive constraints.
    Note that $k \geq 1$ as $0 \not\in \cA'$ by the property $\ell \not\in \cA$ from the first paragraph; hence, $Q$ is indeed a path clause.
\end{proof}

With this preparation, we are ready to give a proof of \cref{thm:johnson_graphs_main}.

\begin{proof}[Proof of \cref{thm:johnson_graphs_main}]
    ``$\Leftarrow$'':
    Suppose that there are $a,b,c \in [0,d]$ according to the second part of the statement.
    Let $P$ be a path clause with $a$ loops and a single positive non-loop literal $q_{2,1}$; note that $P \in F(d,\cA)$ as $a+1 \in \cA$.
    Let $P'$ be the loopless reduction of $P$ of dimension $d' = d - a$.
    Let $\cA' = \{t - a : t \in \cA, t \geq a\}, b' = b - a$ and $c' = c - a$.
    Note that $1 \in \cA'$ and $c'-b' \in \{1, \dots, \lfloor \frac{d'-b'}{2} \rfloor\}$. 
    
    We define a $3$-dimensional interval assignment $\cI = \{I_i : i \in [d']\}$.
    We set $I_1 = \{1,2\}, I_2 = \{2,3\}$.
    To the other coordinates, we assign intervals $\{2\}, \{1,2\}, \{2,3\}, \{1,2,3\}$, so that we obtain an interval assignment $\cI$ with $\# \{2\} = b'$, $\# \{1,2\} = \# \{2,3\} =  c' - b'$ and $\# \{1,2,3\} = d' - 2c' + b'$.
    This is possible since $b' \geq 0$ by item~(i), $c' - b' \geq 1$ by item~(iv) (note that we need $\geq 1$ since $I_1$ and $I_2$ are `pre-assigned' intervals that contribute to $\# \{1,2\}$ and $\# \{2,3\}$, respectively), $d' - 2c' + b' \geq 0$ as $c'-b' \leq \frac{d'-b'}{2}$ by item~(iv), and $\sum_{S \subseteq [3]} \# S = d'$.

    We observe that $N_{\{1,3\}} = 0 \not\in \cA'$, $N_{\{2\}} = b' \not\in \cA'$, $N_{\{1,2\}} = N_{\{2,3\}} = c' \not\in \cA'$ by item~(ii), and $N_{\{1,2,3\}} = d' \not\in \cA'$ by the assumption that $d \not\in \cA$.
    Thus, \cref{lem:johnson_graphs_existence_of_Z} implies that there is a collection of functional constraints $Z$ so that $\cI$ satisfies every functional constraint in $Z$.
    Moreover, $\cI$ satisfies the single shift constraint $I_2 = I_1 + 1$ imposed by $P'$ arising from the positive literal $q_{2,1}$.
    Thus, $\cI$ is a 3-dimensional interval representation for $P'$ over $Z$.
    Therefore, $\cS_3 \subseteq \cJ(d,\cA)$ by the discussion after the statement of \cref{thm:johnson_graphs_main}.

    ``$\Rightarrow$'':
    Suppose that $P$ is a path clause with $\ell$ loops witnessing \chiunboundedness of the class $\cJ(d,\cA)$ such that its loopless reduction $P'$ contains $k \geq 1$ positive literals.
    Let $d' = d-\ell$ and $\cA' = \{t - \ell : t \in \cA, t \geq \ell\}$.
    In view of the `In particular' part of \cref{obs:johnson_graphs_path_clause_minimality}, we may assume that $k = \min \cA'$.
    Let $\cI = \{I_i : i \in [d']\}$ be a \emph{minimal} $D$-dimensional interval representation over $P'$ and some collection of functional constraints $Z$ for the discrete clauses of $F(d',\cA')$.

    We claim that $a = k-1 + \ell, b = \max\{a, N_{[2,D-1]} + \ell\}$, and $c = \min\{ N_{[1,D-1]}, N_{[2,D]} \} + \ell$ satisfy all the requirements (i)-(iv) from the theorem statement.
    \begin{enumerate}[label=(\roman*)]
        \item Clearly, $a \leq b$.
        To prove that $b < c$, we prove both $a < c$ and $N_{[2,D-1]} + \ell < c$.

        Starting with the former, observe that the condition $I_i + 1 = I_{i+1}$, arising from every positive literal $q_{i+1,i}$ in $P'$, implies that $I_i \in [1,D-1]$ and $I_{i+1} \subseteq [2,D]$.
        Thus, $i \in \cC_{[1,D-1]}$ and $i+1 \in \cC_{[2,D]}$.
        Since $P'$ has $k$ positive literals, we have $k \leq |\cC_{[1,D-1]}| = N_{[1,D-1]}$, and similarly $k \leq N_{[2,D]}$.
        Therefore, $k \leq \min\{ N_{[1,D-1]}, N_{[2,D]} \}$, so $a < c$.
        
        Next, we prove that $N_{[2,D-1]} < N_{[1,D-1]}$.
        For contradiction, if the difference is $0$, the value $1$ only belongs to intervals $I_i = [1,D]$.
        Since such intervals do not appear in any condition of the form $I_i + 1 = I_{i+1}$, it is easy to see that the interval representation $\cI' = \{I'_i : i \in [d']\}$, where $I'_i = (I_i \setminus \{1\}) -1$, is an interval representation of dimension $D-1$.
        This violates the minimality of the interval representation $\cI$.

        An analogous argument shows that $N_{[2,D-1]} < N_{[2,D]}$ as otherwise we may decrease the dimension of $\cI$ by removing the value $D$ from every interval.
        Therefore, $N_{[2,D-1]} < \min\{ N_{[1,D-1]}, N_{[2,D]} \}$, so $N_{[2,D-1]} + \ell < c$.
        This concludes the proof of $b < c$.
        
        \item Since $k = \min \cA'$ and $k \geq 1$, we have that $k-1 \not\in \cA'$ and $a = k-1 + \ell \not\in \cA$.
        By \cref{lem:johnson_graphs_existence_of_Z}, $N_S \not\in \cA'$ for every $S \subseteq [1,D]$.
        Thus, $N_S + \ell \not\in \cA$ for every $S \subseteq [1,D]$.
        Therefore, $b \not\in \cA$ as neither $a$ nor $N_{[2,D-1]} + \ell$ belongs to $\cA$.
        Similarly, $c \not\in \cA$ as neither $N_{[1,D-1]} + \ell$ nor $N_{[2,D]} + \ell$ belongs to $\cA$.
        
        \item The path clause $P$ with $k+\ell$ positive literals witnesses that $k+\ell = a+1\in \cA$.
        \item We have
        \[
            (N_{[1,D-1]} - N_{[2,D-1]}) + (N_{[2,D]} - N_{[2,D-1]}) + N_{[2,D-1]} \leq d',
        \]
        since the terms of the left-hand side are supported on disjoint sets of coordinates.
        Indeed, the three term express the respective sizes of $\{i \in [d'] : 1 \in I_i, D \not\in I_i\}$, $\{i \in [d'] : 1 \not\in I_i, D \in I_i\}$, and $\{i \in [d'] : 1 \not\in I_i, D \not\in I_i\}$.
        Thus, by taking the minimum of the first two terms and putting $N_{[2,D-1]}$ to the right-hand side, we obtain
        \[
            2\big(\min\{ N_{[1,D-1]}, N_{[2,D]} \} - N_{[2,D-1]} \big) \leq d' - N_{[2,D-1]}
            .
        \]
        Therefore, if $b = N_{[2,D-1]} + \ell$, it holds
        \[
            c-b
            = \min\{ N_{[1,D-1]}, N_{[2,D]} \} - N_{[2,D-1]}
            \leq \left\lfloor \frac{d' - N_{[2,D-1]}}{2} \right\rfloor
            = \left\lfloor \frac{d - b}{2} \right\rfloor
            .
        \]
        Increasing the value of $b$ makes the inequality even more true, so it also holds if $b = a \geq N_{[2,D-1]} + \ell$.
        Finally, $c-b \geq 1$ follows from the inequality $b < c$ that was established in~(i).
        Thus, $c-b \in [1, \lfloor \frac{d-b}{2} \rfloor]$ as we wanted.
    \end{enumerate}
    In conclusion, the values $a,b,c$ are as desired, which finishes the left-to-right implication.

    The `Moreover' part follows from the fact that if the class $\cJ(d,\cA)$ is \chiunbounded, the left-to-right implication guarantees existence of the values $a,b,c \in [0,d]$.
    Then, the right-to-left implication shows that if there exists such values $a,b,c \in [0,d]$, the class $\cJ(d,\cA)$ contains $\cS_3$.

    To prove the `Furthermore' part, recall that the $3$-dimensional interval representation $\cI = \{I_i : i \in [d']\}$ from the right-to-left implication satisfies $\# \{1\} = \# \{3\} = 0$, and $\# \{2\} = b' = b-a$.
    Thus, if $a = b$, we have $\# \{2\} = 0$.
    Consider the interval assignment $\cI' = \{I'_i : i \in [d']\}$ defined as
    \begin{align*}
        I'_i =
        \begin{cases}
            \{1\} & \text{if }I_i = \{1,2\},\\
            \{2\} & \text{if }I_i = \{2,3\},\\
            \{1,2\} & \text{if }I_i = \{1,2,3\}.
        \end{cases}
    \end{align*}
    Then the interval counts of $\cI'$ satisfy $N_{\{1\}} = N_{\{2\}} = c' \not\in \cA'$ and $N_{\{1,2\}} = d' \not\in \cA'$.
    Moreover, $\cI'$ satisfies the single shift constraint $I'_2 = I'_1 + 1$.
    Thus, \cref{lem:johnson_graphs_existence_of_Z} guarantees a collection of functional constraints $Z$ such that $\cI'$ is an interval representation for $P'$ over $Z$.
    Since the dimension of $\cI'$ is $2$, we conclude that $\cS_2 \subseteq J(d,\cA)$ by the same argument as at the end of the right-to-left implication.
\end{proof}

\begin{example}
    Suppose that $J(d,\cA)$ is \chiunbounded, and let $a,b,c$ be the witnessing values from \cref{thm:johnson_graphs_main}.
    We construct an explicit representation of $S(n,3) \in J(d,\cA)$, $n \geq 3$.
    Denote $G = S(n,3)$, $\alpha = b - a$, $\beta = c-b$, and $N = n^3$.
    We define a witnessing function $\pi : V(G) \to \binom{\N}{d}$ by setting $\pi(x,y,z), x < y < z$ to be the set
    \begin{align}
        &\;\{nx + y, ny + z\} \label{set:shift} \\
        \cup&\; \{mN + y : m \in [1,b-a]\} \label{set:padding_y} \\
        \cup&\; \{mN + nx+y : m \in [b-a+1, c-a-1]\} \label{set:padding_xy} \\
        \cup&\; \{mN + ny+z : m \in [c-a, 2c-b-a-2]\} \label{set:padding_yz} \\
        \cup&\; \{mN + n^2x+ny+z : m \in [2c-b-a-1, d-a-2]\} \label{set:padding_xyz} \\
        \cup&\; F \label{set:const}
        ,
    \end{align}
    where $F = \{mN : m \in [d-a-1, d-2]\}$.

    The function is clearly injective as the contribution from \cref{set:shift} allows to recover the triple $(x,y,z)$.
    The contribution from $F$ is constant (does not depend on the argument), and thus is part of every $\pi(x,y,z)$.
    Apart from that, the choice of $N = n^3$ ensures that the values obtained from \cref{set:padding_y}-\cref{set:padding_xyz} belong to disjoint intervals.
    Thus, $\pi(x,y,z) \in \binom{\N}{d}$.
    Moreover, two sets $\pi(x,y,z)$ and $\pi(x',y',z')$ do not share a value obtained from \cref{set:padding_y}-\cref{set:padding_xyz} unless $q = q'$ for some $q \subseteq \{x,y,z\}$.

    We first verify that $\pi$ maps edges to edges.
    Let $(x,y,z)(x',y',z') \in E(G)$.
    If $y = x'$ and $z = y'$, then
    \[
        \pi(x,y,z) \cap \pi(x',y',z') = \{ny + z\} \cup F
        ,
    \]
    where $ny+z = nx'+y'$ is from \cref{set:shift} and $F$ belongs to every constructed set.
    There are no further common points from \cref{set:shift}, as $nx+y < ny'+z'$, nor from \cref{set:padding_y}-\cref{set:padding_xyz} as the arguments $(x,y,z),(x',y',z')$ do not agree in any entry.
    Thus, $|\pi(x,y,z) \cap \pi(x',y',z')| = a+1 \in \cA$ by item~(iii) of \cref{thm:johnson_graphs_main}.
    If $y' = x$ and $z' = y$, we have $\pi(x,y,z) \cap \pi(x',y',z') = \{ny' + z'\} \cup F$ and proceed by the same argument.

    Now we verify that $\pi$ maps non-edges to non-edges.
    Let $(x,y,z)(x',y',z') \not\in E(G)$.
    We distinguish cases by the subset of entries in which $(x,y,z),(x',y',z')$ agree.
    If $y \not= y'$, then
    \[
        \pi(x,y,z) \cap \pi(x',y',z') = F
        ,
    \]
    since there is clearly no contribution from \cref{set:padding_y}-\cref{set:padding_xyz}, and also no contribution from \cref{set:shift} as $(x,y,z)(x',y',z') \not\in E(G)$ additionally implies that it holds neither $y = x'$ and $z = y'$, nor $y' = x$ and $z' = y$.
    Thus, $|\pi(x,y,z) \cap \pi(x',y',z')| = a \not\in \cA$ by item~(ii) of \cref{thm:johnson_graphs_main}.

    If $y = y'$, but $x \not= x'$ and $z \not= z'$, then
    \[
        \pi(x,y,z) \cap \pi(x',y',z') = \{mN + y : m \in [1,b-a]\} \cup F
        ,
    \]
    as among \cref{set:padding_y}-\cref{set:padding_xyz}, only  \cref{set:padding_y} can contribute, and there is still no contribution from \cref{set:shift}.
    Thus, $|\pi(x,y,z) \cap \pi(x',y',z')| = b \not\in \cA$ by item~(ii) of \cref{thm:johnson_graphs_main}.

    If $x = x'$ and $y = y'$, but $z \not= z'$, then
    \[
        \pi(x,y,z) \cap \pi(x',y',z') = \{nx+y\} \cup \{mN + y : m \in [1,b-a]\} \cup \{mN + nx+y : m \in [b-a+1, c-a-1]\} \cup F
        ,
    \]
    as among \cref{set:padding_y}-\cref{set:padding_xyz}, only  \cref{set:padding_y} and \cref{set:padding_xy} can contribute.
    Additionally, we have the contribution $nx+y = nx' + y'$ from \cref{set:shift}.
    Thus, $|\pi(x,y,z) \cap \pi(x',y',z')| = c \not\in \cA$ by item~(ii) of \cref{thm:johnson_graphs_main}.

    If $y = y'$ and $z = z'$, but $x \not= x'$, then
    \[
        \pi(x,y,z) \cap \pi(x',y',z') = \{ny+z\} \cup \{mN + y : m \in [1,b-a]\} \cup \{mN + ny+z : m \in [c-a, 2c-b-a-2]\} \cup F
        ,
    \]
    by the similar argument as above.
    Thus, $|\pi(x,y,z) \cap \pi(x',y',z')| = c \not\in \cA$ by item~(ii) of \cref{thm:johnson_graphs_main}.

    Finally, we cannot have $x = x'$, $y = y'$ and $z = z'$ for distinct vertices.
\end{example}

As a corollary, we characterize \chiboundedness in classes $\cJ(d,\cA)$ with $\cA = \{\ell\}$ for some $\ell \in [0, d-1]$.
We write $\cJ(d,\ell)$ instead of $\cJ(d,\{\ell\})$. 

\begin{corollary}
    Let $d \geq 2$.
    The class $\cJ(d, \ell)$ is \chibounded if and only if one of the following holds
    \begin{enumerate}[label=(\alph*)]
        \item $\ell = 0$ (Kneser graphs),
        \item $\ell = d-2$,
        \item $\ell = d-1$ (Johnson graphs).
    \end{enumerate}
    Moreover, if it is \chiunbounded, then $\cS_2 \subseteq \cJ(d, \ell)$.
\end{corollary}
\begin{proof}
    ``$\Leftarrow$'': We proceed by case analysis.
    In each case, we show that it is impossible to find the values $a,b,c$ from \cref{thm:johnson_graphs_main} for $\cJ(d,\cA)$ with $\cA = \{\ell\}$, proving that the class is \chibounded.
    \begin{enumerate}[label=(\alph*)]
        \item If $\ell = 0$, it is not possible to choose the value $a \in [0,d]$ so that $a+1 \in \cA$ as requires item~(iii).
        \item If $\ell = d-2$, we must choose $a = d-3$ to satisfy item~(iii).
        As $a \leq b < c$ and $b \not\in \cA$ by items~(i)~and~(ii), we must choose $b$ from $\{d-3,d-1\}$ and $c$ from $\{d-1,d\}$.
        If $b = d-3$, then $c-b \geq 2$, which violates item~(iv) as $c-b = 2 \not\in \{1\} = [1,\lfloor (d-b)/2 \rfloor]$.
        If $b = d-1$, then item~(iv) requires the impossible condition that $c-b \in \emptyset = [1,\lfloor (d-b)/2 \rfloor]$.
        Thus, it is not possible to find the suitable triple $a,b,c$.
        \item If $\ell = d-1$, the argument is a simplification of the previous case.
        To satisfy items~(i)-(iii), we must choose $a = b = d-2$ and $c = d$.
        But this violates item~(iv) as $c-b = 2 \not\in \{1\} = [1,\lfloor (d-b)/2 \rfloor]$.
    \end{enumerate}
    
    ``$\Rightarrow$'':
    Suppose that none of the cases applies; in particular $d \geq 4$.
    We prove that the class $\cJ(d,\ell)$ is \chiunbounded using \cref{thm:johnson_graphs_main}.
    We fix $a = b = \ell-1$, and $c = \ell + 1$ and verify the conditions stated in \cref{thm:johnson_graphs_main}.
    First, observe that $a,b,c \in [0,d]$ since $\ell \in [1,d-1]$ as the first item does not apply.
    The verification of items~(i)-(iii) is trivial.
    Hence, it remains to verify item~(iv).
    Since $\ell \leq d-3$, we have $d-b \geq 4$.
    Therefore, $c-b = 2 \in [1,\lfloor 4/2 \rfloor] \subseteq [1,\lfloor (d-b)/2 \rfloor]$.
\end{proof}

\begin{example}\label{ex:odd-johnson}
    Consider $\cJ(d, \cA)$ for even $d \geq 4$ and $\cA = \{2i-1 : i \in [d/2]\}$.
    The values $a,b = 0, c = 2$ prove, by \cref{thm:johnson_graphs_main}, that $\cS_2 \subseteq \cJ(d, \cA)$.
    This strengthens \cite[Theorem~5.7]{janzer2026sunflowers} by providing arbitrarily large \emph{triangle-free} graphs $S(n,2) = G \in \cJ(d,\cA)$ with $\chi(G) \geq \log_2 n \geq \frac{1}{2}\log_2(|V(G)|)$, see \cref{fact:shift}(ii).
\end{example}

\begin{example}
    Fix $k \geq 2$.
    Consider $d = 2^k$ and set $\cA = [0,d] \setminus S$, where $S = \{d - 2^i + 1 : i \in [0,k-1]\}$; note that $d \in S$.
    We show that the class $\cJ(d, \cA)$ is \chibounded.
    
    For contradiction, suppose that there are $a,b,c \in [0,d]$ according to \cref{thm:johnson_graphs_main}.
    Thus, $a,b,c \in S$.
    We claim that the values $b,c$ cannot satisfy item~(iv) of \cref{thm:johnson_graphs_main}.
    Indeed, suppose $b = d - 2^i + 1, c = d - 2^j + 1$ with $i > j$ as required by item~(i).
    Then we have $c-b \geq 2^{i-1} \not\in [1, 2^{i-1}-1] = [1, \frac{d-b}{2}]$.
    Thus, the values $a,b,c \in [0,d]$ compliant with \cref{thm:johnson_graphs_main} cannot exist, so class $\cJ(d, \cA)$ is \chibounded. 
\end{example}

\subsection{Random set-defined classes}

We now give an elementary proof showing that, with high probability in $d$, the  full set-defined class of digraphs given by a random $d$-dimensional function contains the class $\dcS_2$ of $2$-dimensional shift digraphs. 

\randomclass*

\begin{proof}
For a set
	$T\subseteq [d]\times [d]$, let $\mathbf 1(T)\in\{0,1\}^{d^2}$ denote its
	indicator vector. Define
	$I_d=\{(r,r):r\in[d]\}$, and for distinct $i,j\in[d]$, define the following \emph{equality patterns}:
	\[
	\beta_{ij}
	=
	\mathbf 1\Big(I_d\setminus\big\{(i,i),(j,j)\big\}\Big),
\qquad 
	\delta_i
	=
	\mathbf 1\Big( I_d\setminus\big\{(i,i)\big\}\Big),
	\qquad
	\varepsilon_{ij}
	=
	\mathbf 1\Big((I_d\setminus\big\{(i,i),(j,j)\big\})\cup \big\{(i,j)\big\}\Big).
	\]
	Call an ordered pair $(i,j)$ with $i<j$ \textit{good} if $f(\delta_i)=f(\delta_j)=f(\beta_{ij})=f(\varepsilon_{ij})=0$, 
	and $f(\varepsilon_{ji})=1$.
	\begin{claim}\label{clm:goodpairprob} $	\mathbb{P}(\text{there is no good pair})\to 0$ as $d \to \infty$. 
	\end{claim}
	\begin{poc} Let $A=\{i\in[d]: f(\delta_i)=0\}$. Since $f$ is chosen uniformly at random, the values 	$f(\delta_1),\dots,f(\delta_d)$ are independent fair bits, and  hence $|A|\sim \operatorname{Bin}(d,1/2)$. Thus, by Hoeffding's inequality \begin{equation}\label{eq:hoeff}	
        \mathbb{P}\Big(|A|<d/4 \Big)
        =
        \mathbb{P}\Big(|A|-\mathbb{E}[|A|]< - d/4\Big)
        \leq 
        \exp\left(-\frac{2(d/4)^2}{d}\right) 
        =  
        \exp(-d/8).  
    \end{equation}
	Conditional on $A$, for every unordered pair $\{i,j\}\subseteq A$ with
	$i<j$, the values $	f(\beta_{ij})$, $f(\varepsilon_{ij})$, and $f(\varepsilon_{ji})$ are independent fair bits, and these triples are independent for different
	pairs $\{i,j\}$. Therefore, conditional on $|A|=m$, the probability that no
	pair in $A$ is good is
	$	\left(1-\frac18\right)^{\binom m2}
	=
	\left(\frac78\right)^{\binom m2}$.
    It follows from \eqref{eq:hoeff} that
	\[
        \mathbb{P}\Big(\text{there is no good pair}\Big)
        \le 	
        \mathbb{P}\Big(\text{no pair in $A$ is good} \bigm| |A|\geq d/4\Big) + 
        \mathbb{P}\Big(|A|<d/4\Big) \\
        \leq \left(\frac78\right)^{\binom{\lfloor d/4\rfloor}{2}} +\exp(- d/8), 
        \] 
        which tends to $0$ as $d\rightarrow \infty$.
\end{poc}
	
	Now suppose $(i,j)$ is a good pair. We will show that $\dcX_f$ contains $\dcS_2$.
	Let $H$ be an arbitrary digraph in $\dcS_2$. Then $H$ is an induced subdigraph of $\vec S(n,2)$ for some
	$n$, so we may write $	V(H)\subseteq \{(x,y):1\le x<y\le n\}$,
	with a directed edge $(x,y)\mapsto (u,v)$ if and only if $y=u$. 
	
	Choose natural numbers $t_1,\dots,t_n$ and, for every $r\in[d]\setminus\{i,j\}$, choose a 
    natural number $c_r$ such that all of these numbers are pairwise distinct. Define an injective map $\phi:V(H)\to\mathbb N^d$ by
	\[
	   \phi(x,y)_i=t_x,
	   \qquad
	   \phi(x,y)_j=t_y,\qquad \text{and}\qquad 
	   \phi(x,y)_r=c_r, 
	   \text{ for every }r\in[d]\setminus\{i,j\}.
	\] 	
	Let $(x,y),(u,v)\in V(H)$ be distinct. Since $x<y$ and $u<v$, at most one of
	the following equalities can hold: $x=u$, $y=v$, $x=v$, $y=u$. Therefore the equality pattern between $\phi(x,y)$ and $\phi(u,v)$ is
	one of the following five patterns:
	\[
	Q_{\phi(x,y),\phi(u,v)}
	=
	\begin{cases}
		\varepsilon_{ji}, & \text{if } y=u,\\
		\varepsilon_{ij}, & \text{if } x=v,\\
		\delta_j, & \text{if } x=u,\\
		\delta_i, & \text{if } y=v,\\
		\beta_{ij}, & \text{otherwise.}
	\end{cases}
	\]
	Since $(i,j)$ is good, we have
	$
	f(\varepsilon_{ij})
	=
	f(\delta_i)
	=
	f(\delta_j)
	=
	f(\beta_{ij})
	=
	0$, and $f(\varepsilon_{ji})=1$.
	Thus, for distinct vertices $(x,y),(u,v)\in V(H)$,
	\[
	   f(Q_{\phi(x,y),\phi(u,v)})=1
	   \iff
	   y=u,
	\]
	which is exactly the adjacency rule of the  shift digraph
	$\vec S(n,2)$. Hence the digraph realised by $f$ on $\phi[V(H)]$ is isomorphic
	to $H$.
	
	Since $H\in\dcS_2$ was chosen to be arbitrary, we have $	\dcS_2\subseteq \dcX_f$	whenever a good pair exists. As the probability that a good pair exists tends
	to $1$ as $d\to \infty$ by \cref{clm:goodpairprob}, the result follows.
\end{proof}

\subsection{Stability and set-defined classes} \label{sec:stability}

We have mentioned in the introduction that set-defined graph classes are \emph{edge-stable}, meaning there is a bound on the size of semi-induced half-graphs in the class. (The half-graph of size $n$ is the bipartite graph $G=(V,E)$ with vertices $V = \{a_1, \dots, a_n\} \cup \{b_1, \dots, b_n\}$ and $(a_i,b_j) \in E \iff i \leq j$. A bipartite graph is semi-induced if it can be obtained by removing vertices and removing edges only between vertices in the same part.) Monadic stability, a strengthening of edge-stability, has played a significant role as a tameness property in generalizing sparsity theory to dense classes, as discussed in \cite[\S 4.1]{pilipczuk2025graph}. Since set-defined classes also generalize the sparsity notion of bounded degeneracy to dense classes, it seems interesting to ask how it interacts with monadic stability and edge-stability.

We first show that within the set-defined classes, and even within $\cS_2$, monadic stability and \chiboundedness are incomparable, before moving onto the interaction of edge-stability and set-defined classes. Rather than defining monadic stability, we state what we need in the following fact; the interested reader may refer to \cite[Definitions 16, 21]{pilipczuk2025graph} for the definitions of nowhere denseness and monadic stability. 

\begin{fact} \label{fact:monstab}
\begin{enumerate}
	\item If $\CC$ contains arbitrarily large $k$-subdivided balanced complete bipartite graphs for some $k \geq 1$, then $\CC$ is not monadically stable.
	\item If $\CC$ is the set of (directed) line graphs of a nowhere dense (directed) graph class $\CC_0$, then $\CC$ is monadically stable. (Since $\CC$ may be transduced from the class of 1-subdivisions of $\CC_0$, which is still nowhere dense.)
    \item If $\cC \subset \cD$ and $\cD$ is monadically stable, then so is $\cC$.
    \item A set of graphs is monadically stable if and only if its hereditary closure is monadically stable.
\end{enumerate}
\end{fact}

\begin{proposition} \label{prop:monstabvschi}
    There are subclasses $\cC, \cC' \subset \cS_2$ such that $\cC$ is monadically stable but \chiunbounded and $\cC'$ is \chibounded but not monadically stable.
\end{proposition}
\begin{proof}
    It is easy to check that $\cS_2$ contains the class of 1-subdivided complete bipartite graphs, whose hereditary closure we may take to be $\cC'$, by Fact \ref{fact:monstab}(1).

    We now will define $\cC$. Let $\DD_0$ be a nowhere dense class of unbounded chromatic number, e.g. we may take $\DD_0$ to be the class of all graphs whose girth is at least as large as their maximum degree (see \cite[Example 5.1]{nevsetvril2012sparsity}). Let $\DD_1$ be the class of all acyclic orientations of $\DD_0$, let $\DD_2$ be the class of directed line graphs of graphs in $\DD_1$. Finally, we let $\CC$ be the hereditary closure of the class of undirected graphs obtained by forgetting the orientations of graphs in $\DD_2$, and we claim $\CC$ is as desired.
	
	Since every acyclic directed graph may be extended to a linear order, $\CC \subset \cS_2$. Since $\DD_2$ is the class of directed line graphs of a class of unbounded chromatic number, $\DD_2$ (and thus $\CC$) has unbounded chromatic number as well \cite[Lemma 2.21]{hell2004graphs}. Since $\CC$ is a class of directed line graphs of a nowhere dense class, it is monadically stable by Fact \ref{fact:monstab}(2) and Fact \ref{fact:monstab}(4).
\end{proof}

It is easy to see that $\cS_k$ contains the class of $(k-1)$-subdivided complete bipartite graphs, and thus is not monadically stable by Fact \ref{fact:monstab}(1). So by Fact \ref{fact:monstab}(3), the class $\cC$ of Proposition \ref{prop:monstabvschi} does not contain any $\cS_k$. Thus, in addition to Example \ref{ex:shift_diag}, $\cC$ is another set-defined class that is \chiunbounded but contains no $\cS_k$ for any $k \geq 2$. In fact, the proof of Proposition \ref{prop:monstabvschi} shows that any nowhere dense class of unbounded chromatic number gives rise to such a class.

We now move on to edge-stability. In \cite{JiangNesOdM20}, the authors also introduced order-defined graph classes, which are like set-defined classes but may use both equality-checks and less-than-checks between the coordinates assigned to vertices. Edge-stability is an assumption that forbids the edges from encoding some sort of order, so since set-defined classes have no access to order-checks, they are edge-stable (and order-defined). 
The converse was posed as an open question in \cite[Problem 2]{JiangNesOdM20}: is every edge-stable order-defined class set-defined? The following provides the  negative answer, even when edge-stability is strengthened to weak sparsity.
\begin{proposition}
    There is a weakly sparse (and thus edge-stable) order-defined class that is not set-defined.
\end{proposition}
    \begin{proof}
        It was shown in \cite{basit2021zarankiewicz} that there exist $K_{2,2}$-free point-box incidence graphs (i.e. incidence graphs of points and axis-parallel rectangles) with superlinear number of edges\footnote{See also \cite[Lemma 4.3]{dallard2024functionality} for a graph-theoretic description of the construction.}. Consequently, the class of $K_{2,2}$-free point-box incidence graphs has unbounded degeneracy. This class is weakly sparse and order-defined, but it cannot be set-defined, because every weakly sparse set-defined class has bounded degeneracy by \cite[Theorem 35]{JiangNesOdM20}.
    \end{proof}

\iftoggle{anonymous}{
}{%
	\bigskip
	\textbf{Acknowledgments.}
We thank Patrice Ossona de Mendez for suggesting the argument in Proposition \ref{prop:monstabvschi}.
We thank Robert \v{S}\'amal for notifying us that some of our intermediate results might be expressed in the language of tropical algebra; this led us to the duality between tropical algebra and mean payoff games, which had a significant influence on our further work.
We thank Rob Sullivan for inspiring discussions at the beginning of this project.

This work was supported by a Royal Society International Exchanges grant IES$\backslash$R2$\backslash$242173 and a London Mathematical Society Scheme 7 grant SC7-2425-14. The preliminary discussions that led to this project took place during Samuel Braunfeld’s visit to the University of Liverpool, which was supported by the School of Computer Science \& Informatics Visiting Fellowship Scheme. Samuel Braunfeld and Tom\'a\v{s} Hons are further supported by Project 24-12591M of the Czech Science Foundation (GA\v{C}R), and Samuel Braunfeld is also supported by the long-term strategic development financing of the Institute of Computer Science (RVO: 67985807). Tomáš Hons is further supported by grant SVV–2025–260822. Viktor Zamaraev is supported by the Leverhulme Research Fellowship (RF-2026-309$\backslash$9).
}

\textbf{AI Acknowledgments.} In the course of this work, we originally had the inequalities in (\ref{eq:endpoint_condition}) written in terms of disjunctions rather than min and max, i.e. we had that $\bigvee_{c \in L} (x_\lambda \geq x_c)$ instead of $x_\lambda \geq \min_{c \in L} x_c$, and similarly for the inequality involving max. In July 2025, while asking around whether people were familiar with such ``linear programs'', Robert \v{S}\'{a}mal proposed asking various LLMs, and some version of Gemini (likely Gemini 2.5) noted that these inequalities could be recast in terms of tropical algebra. (As a snapshot of the times, the versions of ChatGPT and Claude queried did not give useful answers.) All further (re)casting of our methods and results in terms of tropical algebra is due to the authors.  
Various versions of ChatGPT were also used to produce many of the figures and together with Google Deep Review for mathematical proofreading.

\addcontentsline{toc}{section}{References}
\bibliography{references}

\newcommand{\etalchar}[1]{$^{#1}$}
\begin{thebibliography}{JNOdMS20}

\bibitem[ABD{\etalchar{+}}26]{abrishami2026burling}
Tara Abrishami, Marcin Bria{\'n}ski, James Davies, Xiying Du, Jana
  Masa{\v{r}}{\'\i}kov{\'a}, Pawe{\l} Rz{\k{a}}{\.z}ewski, and Bartosz Walczak.
\newblock Burling graphs in graphs with large chromatic number.
\newblock In {\em Proceedings of the 2026 Annual ACM-SIAM Symposium on Discrete
  Algorithms (SODA)}, pages 3978--3998. SIAM, 2026.

\bibitem[ABGJ15]{allamigeon2015tropicalizing}
Xavier Allamigeon, Pascal Benchimol, St{\'e}phane Gaubert, and Michael Joswig.
\newblock Tropicalizing the simplex algorithm.
\newblock {\em SIAM Journal on Discrete Mathematics}, 29(2):751--795, 2015.

\bibitem[ABHT07]{ADDARIOBERRY2007620}
L.~Addario-Berry, F.~Havet, and S.~Thomassé.
\newblock Paths with two blocks in $n$-chromatic digraphs.
\newblock {\em Journal of Combinatorial Theory, Series B}, 97(4):620--626,
  2007.

\bibitem[ABSZ26]{ABSZ24}
Sarosh Adenwalla, Samuel Braunfeld, John Sylvester, and Viktor Zamaraev.
\newblock Boolean combinations of graphs.
\newblock {\em Electronic Journal of Combinatorics}, 33(P4.12), 2026.

\bibitem[AFCK25]{AFCK25}
Sepehr Assadi, Martin Farach-Colton, and William Kuszmaul.
\newblock Tight bounds for monotone minimal perfect hashing.
\newblock {\em ACM Transactions on Algorithms}, 21(4):1--23, 2025.

\bibitem[AGG12]{akian2012tropical}
Marianne Akian, St{\'e}phane Gaubert, and Alexander Guterman.
\newblock Tropical polyhedra are equivalent to mean payoff games.
\newblock {\em International Journal of Algebra and Computation},
  22(01):1250001, 2012.

\bibitem[AGM22]{akian2022tropical}
Marianne Akian, St{\'e}phane Gaubert, and Andrea Marchesini.
\newblock Tropical linear regression and mean payoff games: Or, how to measure
  the distance to equilibria.
\newblock {\em SIAM Journal on Discrete Mathematics}, 36(4):2643--2671, 2022.

\bibitem[AKL26]{AKL25}
Andris Ambainis, Hartmut Klauck, and Debbie Lim.
\newblock A hierarchy of constant communication complexity.
\newblock {\em Information and Computation}, page 105416, 2026.

\bibitem[AN25]{AN25}
Benny Applebaum and Oded Nir.
\newblock The meta-complexity of secret sharing.
\newblock In {\em Proceedings of the 57th Annual ACM Symposium on Theory of
  Computing}, pages 965--976, 2025.

\bibitem[APP{\etalchar{+}}05]{APPRRS05}
Noga Alon, J{\'a}nos Pach, Rom Pinchasi, Rado{\v{s}} Radoi{\v{c}}i{\'c}, and
  Micha Sharir.
\newblock Crossing patterns of semi-algebraic sets.
\newblock {\em Journal of Combinatorial Theory, Series A}, 111(2):310--326,
  2005.

\bibitem[AY24]{AY24}
Daniel Avraham and Amir Yehudayoff.
\newblock On blocky ranks of matrices.
\newblock {\em Comput. Complex.}, 33(1):2, 2024.

\bibitem[BB25]{bodirsky2025structures}
Manuel Bodirsky and Bertalan Bodor.
\newblock Structures preserved by primitive actions of ${S}_\omega$.
\newblock {\em arXiv preprint arXiv:2501.03789}, 2025.

\bibitem[BBM25]{bodirsky2025taking}
Manuel Bodirsky, Bertalan Bodor, and Paolo Marimon.
\newblock Taking model-complete cores.
\newblock {\em arXiv preprint arXiv:2512.21278}, 2025.

\bibitem[BCS{\etalchar{+}}21]{basit2021zarankiewicz}
Abdul Basit, Artem Chernikov, Sergei Starchenko, Terence Tao, and Chieu-Minh
  Tran.
\newblock Zarankiewicz’s problem for semilinear hypergraphs.
\newblock {\em Forum of Mathematics, Sigma}, 9:e59, 2021.

\bibitem[BDSZ25]{BDSZ25}
{\'E}douard Bonnet, Julien Duron, John Sylvester, and Viktor Zamaraev.
\newblock Adjacency labeling schemes for small classes.
\newblock In {\em 16th Innovations in Theoretical Computer Science Conference
  (ITCS 2025)}, pages 21--1. Schloss Dagstuhl--Leibniz-Zentrum f{\"u}r
  Informatik, 2025.

\bibitem[BDW24]{BDW24}
Marcin Bria{\'n}ski, James Davies, and Bartosz Walczak.
\newblock Separating polynomial $\chi$-boundedness from $\chi$-boundedness.
\newblock {\em Combinatorica}, 44(1):1--8, 2024.

\bibitem[BGH{\etalchar{+}}26]{BGHJN26}
Csongor Beke, Marcel~K Goh, Hamed Hatami, Sean Jaffe, and Daniel Naylor.
\newblock A characterization of idempotent {Schur} multipliers.
\newblock {\em arXiv preprint arXiv:2607.14316}, 2026.

\bibitem[Bur80]{burr1980subtrees}
Stefan~A Burr.
\newblock Subtrees of directed graphs and hypergraphs.
\newblock In {\em Proceedings of the Eleventh Southeastern Conference on
  Combinatorics, Graph Theory and Computing, Boca Raton, Congr. Numer},
  volume~28, pages 227--239, 1980.

\bibitem[But03]{butkovic2003max}
Peter Butkovi{\v{c}}.
\newblock Max-algebra: the linear algebra of combinatorics?
\newblock {\em Linear Algebra and its Applications}, 367:313--335, 2003.

\bibitem[But10]{Butkovic2010}
Peter Butkovič.
\newblock {\em Max-linear Systems: Theory and Algorithms}.
\newblock Springer Monographs in Mathematics. Springer London, 2010.

\bibitem[BV07]{bjorklund2007combinatorial}
Henrik Bj{\"o}rklund and Sergei Vorobyov.
\newblock A combinatorial strongly subexponential strategy improvement
  algorithm for mean payoff games.
\newblock {\em Discrete Applied Mathematics}, 155(2):210--229, 2007.

\bibitem[CCDO26]{CCDO26}
Maria Chudnovsky, Linda Cook, James Davies, and Sang-il Oum.
\newblock Reuniting $\chi$-boundedness with polynomial $\chi$-boundedness.
\newblock {\em Journal of Combinatorial Theory, Series B}, 176:30--73, 2026.

\bibitem[Cha23]{Cha2023}
Maurice Chandoo.
\newblock Logical labeling schemes.
\newblock {\em Discrete Mathematics}, 346(10):113565, 2023.

\bibitem[CJK{\etalchar{+}}17]{Calude2017}
Cristian~S. Calude, Sanjay Jain, Bakhadyr Khoussainov, Wei Li, and Frank
  Stephan.
\newblock Deciding parity games in quasipolynomial time.
\newblock In {\em Proceedings of the 49th Annual ACM SIGACT Symposium on Theory
  of Computing}, STOC ’17, page 252–263. ACM, June 2017.

\bibitem[CKBD25]{CKB25}
Mert~Onur Cakiroglu, Hasan Kurban, Elham Buxton, and Mehmet Dalkilic.
\newblock A novel discrete time series representation with {de Bruijn} graphs
  for enhanced forecasting using {TimesNet}.
\newblock {\em IEEE Access}, 2025.

\bibitem[CS11]{charbit2011fractional}
Pierre Charbit and Jean-S\'ebastien Sereni.
\newblock The fractional chromatic number of {Zykov} products of graphs.
\newblock {\em Applied Mathematics Letters}, 24(4):432--437, 2011.

\bibitem[CSS21]{chudnovsky2021induced}
Maria Chudnovsky, Alex Scott, and Paul Seymour.
\newblock {Induced subgraphs of graphs with large chromatic number. V.
  Chandeliers and strings}.
\newblock {\em Journal of Combinatorial Theory, Series B}, 150:195--243, 2021.

\bibitem[DG06]{Dhingra2006}
Vishesh Dhingra and St\'{e}phane Gaubert.
\newblock How to solve large scale deterministic games with mean payoff by
  policy iteration.
\newblock In {\em Proceedings of the 1st International Conference on
  Performance Evaluation Methodologies and Tools}, valuetools '06, page
  12–es, New York, NY, USA, 2006. Association for Computing Machinery.

\bibitem[DGP26]{DGP26}
Jan Dreier, Jakub Gajarsk{\'y}, and Micha{\l} Pilipczuk.
\newblock Efficient reversal of transductions of sparse graph classes.
\newblock In {\em Proceedings of the 58th Annual ACM Symposium on Theory of
  Computing}, pages 1180--1191, 2026.

\bibitem[DKZ19]{dorfman_et_al}
Dani Dorfman, Haim Kaplan, and Uri Zwick.
\newblock {A Faster Deterministic Exponential Time Algorithm for Energy Games
  and Mean Payoff Games}.
\newblock In {\em 46th International Colloquium on Automata, Languages, and
  Programming (ICALP 2019)}, volume 132 of {\em LIPIcs}, pages 114:1--114:14,
  2019.

\bibitem[DLM{\etalchar{+}}24]{dallard2024functionality}
Cl{\'e}ment Dallard, Vadim Lozin, Martin Milani{\v{c}}, Kenny {\v{S}}torgel,
  and Viktor Zamaraev.
\newblock Functionality of box intersection graphs.
\newblock {\em Results in Mathematics}, 79(1):48, 2024.

\bibitem[DLR95]{DLR95}
Dwight Duffus, Hanno Lefmann, and Vojt{\v{e}}ch R{\"o}dl.
\newblock Shift graphs and lower bounds on {R}amsey numbers {$r_k(l;r)$}.
\newblock {\em Discrete Math.}, 137(1-3):177--187, 1995.

\bibitem[EH68]{EH68}
Paul Erd\H{o}s and Andr{\'a}s Hajnal.
\newblock On chromatic number of infinite graphs.
\newblock In {\em Theory of Graphs (Proc. Colloq., Tihany, 1966)}, pages
  83--98, 1968.

\bibitem[EHK22]{EHK22}
Louis Esperet, Nathaniel Harms, and Andrey Kupavskii.
\newblock Sketching distances in monotone graph classes.
\newblock {\em Approximation, Randomization, and Combinatorial Optimization.
  Algorithms and Techniques}, page~1, 2022.

\bibitem[EHS72]{erdos1972some}
Paul Erd\H{o}s, Andr{\'a}s Hajnal, and Saharon Shelah.
\newblock On some general properties of chromatic number.
\newblock In A.~Cs\'az\'ar, editor, {\em Topics in Topology}, pages 243--255,
  1972.

\bibitem[EM79]{Ehrenfeucht1979}
A.~Ehrenfeucht and J.~Mycielski.
\newblock Positional strategies for mean payoff games.
\newblock {\em International Journal of Game Theory}, 8(2):109--113, Jun 1979.

\bibitem[Enj23]{Enj23}
Rodrigo~Aparecido Enju.
\newblock Uma conjectura de {Erd\H{o}s e Hajnal}.
\newblock Master's thesis, Universidade de S{\~a}o Paulo, 2023.

\bibitem[Erd73]{Erd73}
Paul Erd\H{o}s.
\newblock Problems and results in combinatorial analysis.
\newblock In {\em Colloq. Internat. Theor. Combin. Rome}, pages 3--17, 1973.

\bibitem[Esp17]{Esp17}
Louis Esperet.
\newblock {\em Graph colorings, flows and perfect matchings}.
\newblock Habilitation thesis, Universit{\'e} Grenoble Alpes, 2017.

\bibitem[FGO20]{fijalkow_et_al}
Nathana\"{e}l Fijalkow, Pawe{\l} Gawrychowski, and Pierre Ohlmann.
\newblock {Value Iteration Using Universal Graphs and the Complexity of Mean
  Payoff Games}.
\newblock In {\em 45th International Symposium on Mathematical Foundations of
  Computer Science (MFCS 2020)}, volume 170 of {\em LIPIcs}, pages 34:1--34:15,
  2020.

\bibitem[FHHH24]{FHHH24}
Yuting Fang, Lianna Hambardzumyan, Nathaniel Harms, and Pooya Hatami.
\newblock No complete problem for constant-cost randomized communication.
\newblock In {\em Proceedings of the 56th Annual ACM Symposium on Theory of
  Computing}, pages 1287--1298, 2024.

\bibitem[FHRT92]{FHRT92}
Z.~F\"uredi, P.~Hajnal, V.~R\"odl, and W.~T. Trotter.
\newblock Interval orders and shift graphs.
\newblock In {\em Sets, graphs and numbers ({B}udapest, 1991)}, volume~60 of
  {\em Colloq. Math. Soc. J\'anos Bolyai}, pages 297--313. North-Holland,
  Amsterdam, 1992.

\bibitem[FKKN14]{FKKN14}
Ji\v{r}{\'{\i}} Fiala, Pavel Klav{\'{\i}}k, Jan Kratochv{\'{\i}}l, and Roman
  Nedela.
\newblock Algorithmic aspects of regular graph covers with applications to
  planar graphs.
\newblock In {\em Automata, Languages, and Programming - 41st International
  Colloquium, {ICALP} 2014}, Lecture Notes in Computer Science, pages 489--501.
  Springer, 2014.

\bibitem[FPS14]{FPS14}
Jacob Fox, J{\'a}nos Pach, and Andrew Suk.
\newblock Density and regularity theorems for semi-algebraic hypergraphs.
\newblock In {\em Proceedings of the twenty-sixth annual ACM-SIAM symposium on
  Discrete algorithms}, pages 1517--1530. SIAM, 2014.

\bibitem[FX14]{FX14}
Vitaly Feldman and David Xiao.
\newblock Sample complexity bounds on differentially private learning via
  communication complexity.
\newblock In {\em Conference on Learning Theory}, pages 1000--1019. PMLR, 2014.

\bibitem[GHR25]{GHR25}
Mika G{\"o}{\"o}s, Nathaniel Harms, and Artur Riazanov.
\newblock Equality is far weaker than constant-cost communication.
\newblock In {\em Approximation, Randomization, and Combinatorial Optimization.
  Algorithms and Techniques (APPROX/RANDOM 2025)}, pages 58--1, 2025.

\bibitem[GJ76]{garey1976complexity}
M.~R. Garey and D.~S. Johnson.
\newblock The complexity of near-optimal graph coloring.
\newblock {\em J. ACM}, 23(1):43–49, January 1976.

\bibitem[GKK88]{GURVICH198885}
V.A. Gurvich, A.V. Karzanov, and L.G. Khachiyan.
\newblock Cyclic games and an algorithm to find minimax cycle means in directed
  graphs.
\newblock {\em USSR Computational Mathematics and Mathematical Physics},
  28(5):85--91, 1988.

\bibitem[GPT22]{GPT22}
Jakub Gajarsk{\'y}, Micha{\l} Pilipczuk, and Szymon Toru{\'n}czyk.
\newblock Stable graphs of bounded twin-width.
\newblock In {\em Proceedings of the 37th Annual ACM/IEEE Symposium on Logic in
  Computer Science}, pages 1--12, 2022.

\bibitem[Gy{\'a}87]{Gya87}
Andr{\'a}s Gy{\'a}rf{\'a}s.
\newblock Problems from the world surrounding perfect graphs.
\newblock {\em Applicationes Mathematicae}, 19(3-4):413--441, 1987.

\bibitem[Har20]{Har20}
Nathaniel Harms.
\newblock Universal communication, universal graphs, and graph labeling.
\newblock In {\em 11th Innovations in Theoretical Computer Science Conference,
  {ITCS} 2020}, LIPIcs, pages 33:1--33:27. Schloss Dagstuhl - Leibniz-Zentrum
  f{\"{u}}r Informatik, 2020.

\bibitem[HH24]{HH24}
Hamed Hatami and Pooya Hatami.
\newblock Guest column: Structure in communication complexity and constant-cost
  complexity classes.
\newblock {\em ACM SIGACT News}, 55(1):67--93, 2024.

\bibitem[HHH23]{HHH23}
Lianna Hambardzumyan, Hamed Hatami, and Pooya Hatami.
\newblock Dimension-free bounds and structural results in communication
  complexity.
\newblock {\em Israel Journal of Mathematics}, 253(2):555--616, 2023.

\bibitem[HK84]{hajnal1984must}
Andr{\'a}s Hajnal and P{\'e}ter Komj{\'a}th.
\newblock What must and what need not be contained in a graph of uncountable
  chromatic number?
\newblock {\em Combinatorica}, 4(1):47--52, 1984.

\bibitem[HKS22]{halevi2022infinite}
Yatir Halevi, Itay Kaplan, and Saharon Shelah.
\newblock Infinite stable graphs with large chromatic number.
\newblock {\em Transactions of the American Mathematical Society},
  375(3):1767--1799, 2022.

\bibitem[HKS23]{halevi2023infinite}
Yatir Halevi, Itay Kaplan, and Saharon Shelah.
\newblock Infinite stable graphs with large chromatic number {II}.
\newblock {\em Journal of the European Mathematical Society},
  26(12):4585--4614, 2023.

\bibitem[HKS25]{halevi2025infinite}
Yatir Halevi, Itay Kaplan, and Saharon Shelah.
\newblock Infinite cliques in simple and stable graphs.
\newblock {\em Model Theory}, 4(3):231--249, 2025.

\bibitem[HMR{\etalchar{+}}26]{hambardzumyan2026spiky}
Lianna Hambardzumyan, Konstantin Myasnikov, Artur Riazanov, Morgan Shirley, and
  Adi Shraibman.
\newblock Spiky rank and its applications to rigidity and circuits.
\newblock {\em arXiv preprint arXiv:2602.23503}, 2026.

\bibitem[HN04]{hell2004graphs}
Pavol Hell and Jaroslav Ne\v{s}et\v{r}il.
\newblock {\em Graphs and Homomorphisms}, volume~28 of {\em Oxford Lecture
  Series in Mathematics and its Applications}.
\newblock Oxford University Press, 2004.

\bibitem[HWZ25]{HWZ25}
Nathaniel Harms, Sebastian Wild, and Viktor Zamaraev.
\newblock Randomized communication and implicit graph representations.
\newblock {\em TheoretiCS}, Volume 4, 2025.

\bibitem[HZ24]{HZ24}
Nathaniel Harms and Viktor Zamaraev.
\newblock Randomized communication and implicit representations for matrices
  and graphs of small sign-rank.
\newblock In {\em Proceedings of the 2024 Annual ACM-SIAM Symposium on Discrete
  Algorithms (SODA)}, pages 1810--1833. SIAM, 2024.

\bibitem[IW95]{IW95}
Ramana Idury and Michael Waterman.
\newblock A new algorithm for dna sequence assembly.
\newblock {\em Journal of Computational Biology}, 2:291--306, 1995.

\bibitem[JJSW26]{janzer2026sunflowers}
Barnab{\'a}s Janzer, Zhihan Jin, Benny Sudakov, and Kewen Wu.
\newblock Sunflowers and {Ramsey} problems for restricted intersections.
\newblock {\em Combinatorica}, 46(4):29, 2026.

\bibitem[JNOdMS20]{JiangNesOdM20}
Y.~Jiang, J.~Ne{\v{s}}et{\v{r}}il, P.~Ossona~de Mendez, and S.~Siebertz.
\newblock Regular partitions of gentle graphs.
\newblock {\em Acta Math. Hungar.}, 161(2):719--755, 2020.

\bibitem[Jos21]{Joswig2021}
Michael Joswig.
\newblock {\em Essentials of tropical combinatorics}, volume 219 of {\em
  Graduate Studies in Mathematics}.
\newblock American Mathematical Society, Providence, RI, 2021.

\bibitem[Juk06]{jukna2006graph}
Stasys Jukna.
\newblock On graph complexity.
\newblock {\em Combinatorics, Probability and Computing}, 15(6):855--876, 2006.

\bibitem[KGKB25]{King2025triply}
Robbie King, David Gosset, Robin Kothari, and Ryan Babbush.
\newblock Triply efficient shadow tomography.
\newblock In {\em Proceedings of the 2025 Annual ACM-SIAM Symposium on Discrete
  Algorithms (SODA)}, pages 914--946. SIAM, 2025.

\bibitem[KK03]{KK03}
M~Frans Kaashoek and David~R Karger.
\newblock Koorde: A simple degree-optimal distributed hash table.
\newblock In {\em International Workshop on Peer-to-Peer Systems}, pages
  98--107. Springer, 2003.

\bibitem[KK26a]{KK26arc}
J.~Kohlmeyer and L.~Kruer.
\newblock A counterexample to {E}rd{\H{o}}s problem 108 via arc graphs.
\newblock \url{https://jenwin.io/papers/erdos108-arc-graphs.pdf}, 2026.

\bibitem[KK26b]{KK26six}
L.~Kruer and J.~Kohlmeyer.
\newblock {E}rd{\H{o}}s problem 108: High chromatic number with six-colour
  {$C_4$}-free subgraphs.
\newblock \url{https://conjectures.io/papers/erdos108.pdf}, 2026.

\bibitem[KO22]{krokhin2022invitation}
Andrei Krokhin and Jakub Opr\v{s}al.
\newblock An invitation to the promise constraint satisfaction problem.
\newblock {\em ACM SIGLOG News}, 9(3):30–59, August 2022.

\bibitem[Kor21]{korhonen2021lower}
Tuukka Korhonen.
\newblock Lower bounds on dynamic programming for maximum weight independent
  set.
\newblock In {\em 48th International Colloquium on Automata, Languages, and
  Programming (ICALP 2021)}. Schloss Dagstuhl-Leibniz-Zentrum f{\"u}r
  Informatik, 2021.

\bibitem[KPT97]{Kratochvil_97}
Jan Kratochv{\'i}l, Andrzej Proskurowski, and Jan~Arne Telle.
\newblock Complexity of colored graph covers {I}. {C}olored directed
  multigraphs.
\newblock In Rolf~H. M{\"o}hring, editor, {\em Graph-Theoretic Concepts in
  Computer Science}, pages 242--257, Berlin, Heidelberg, 1997. Springer Berlin
  Heidelberg.

\bibitem[Lac87]{lachlan1987structures}
Alistair~H Lachlan.
\newblock Structures coordinatized by indiscernible sets.
\newblock {\em Annals of Pure and Applied Logic}, 34(3), 1987.

\bibitem[LPU95]{larsen1995fractional}
Michael Larsen, James Propp, and Daniel Ullman.
\newblock The fractional chromatic number of {Mycielski}'s graphs.
\newblock {\em Journal of Graph Theory}, 19(3):411--416, 1995.

\bibitem[LS07]{LS07}
Nati Linial and Adi Shraibman.
\newblock Lower bounds in communication complexity based on factorization
  norms.
\newblock In {\em Proceedings of the thirty-ninth annual ACM symposium on
  Theory of computing}, pages 699--708, 2007.

\bibitem[LZR{\etalchar{+}}10]{LZR10}
Ruiqiang Li, Hongmei Zhu, Jue Ruan, Wubin Qian, Xiaodong Fang, Zhongbin Shi,
  Yingrui Li, Shengting Li, Gao Shan, Karsten Kristiansen, et~al.
\newblock De novo assembly of human genomes with massively parallel short read
  sequencing.
\newblock {\em Genome research}, 20(2):265--272, 2010.

\bibitem[MS14]{MS14}
Maryanthe Malliaris and Saharon Shelah.
\newblock Regularity lemmas for stable graphs.
\newblock {\em Transactions of the American Mathematical Society},
  366(3):1551--1585, 2014.

\bibitem[MW22]{MW22b}
Bojan Mohar and Hehui Wu.
\newblock Subgraphs of {K}neser graphs with large girth and large chromatic
  number.
\newblock {\em The Art of Discrete and Applied Mathematics}, 6(2):\#P2.11,
  2022.

\bibitem[NMP{\etalchar{+}}21]{NdMPRS21}
Jaroslav Ne{\v{s}}et{\v{r}}il, Patrice Ossona~de Mendez, Micha{\l} Pilipczuk,
  Roman Rabinovich, and Sebastian Siebertz.
\newblock Rankwidth meets stability.
\newblock In {\em Proceedings of the 2021 ACM-SIAM Symposium on Discrete
  Algorithms (SODA)}, pages 2014--2033. SIAM, 2021.

\bibitem[NOdM12]{nevsetvril2012sparsity}
Jaroslav Ne{\v{s}}et{\v{r}}il and Patrice Ossona~de Mendez.
\newblock {\em {Sparsity: Algorithms and Combinatorics}}.
\newblock Springer, 2012.

\bibitem[NR76]{NR76}
Jaroslav Ne{\v{s}}et{\v{r}}il and Vojt{\v{e}}ch R{\"o}dl.
\newblock The {R}amsey property for graphs with forbidden complete subgraphs.
\newblock {\em Journal of Combinatorial Theory, Series B}, 20(3):243--249,
  1976.

\bibitem[NW26]{NW26}
Tung Nguyen and Bartosz Walczak.
\newblock On the solution to the {E}rd{\H{o}}s--{H}ajnal problem on high-girth
  high-chromatic subgraphs.
\newblock {\em arXiv preprint arXiv:2609.40192}, 2026.

\bibitem[OdMPS25]{dMPS25}
Patrice Ossona~de Mendez, Micha{\l} Pilipczuk, and Sebastian Siebertz.
\newblock Transducing paths in graph classes with unbounded shrubdepth.
\newblock {\em European Journal of Combinatorics}, 123:103660, 2025.

\bibitem[{Ope}26]{OAI_MPG}
{OpenAI}.
\newblock {Deterministic quasipolynomial-time mean-payoff games}.
\newblock OpenAI Math Release preprint
  \href{https://github.com/openai/math/blob/main/preprints/Deterministic-quasipolynomial-time-mean-payoff-games-September-25-2026/paper.pdf}{OAI:Deterministic-quasipolynomial-time-mean-payoff-games-September-25-2026},
  2026.

\bibitem[Pil25]{pilipczuk2025graph}
Micha{\l} Pilipczuk.
\newblock Graph classes through the lens of logic.
\newblock {\em arXiv preprint arXiv:2501.04166}, 2025.

\bibitem[PTW01]{PTW01}
Pavel~A Pevzner, Haixu Tang, and Michael~S Waterman.
\newblock An {E}ulerian path approach to {DNA} fragment assembly.
\newblock {\em Proceedings of the national academy of sciences},
  98(17):9748--9753, 2001.

\bibitem[PTW26]{PTW26}
Seth Pettie, G{\'a}bor Tardos, and Bartosz Walczak.
\newblock On a clique game and the {E}rd{\H{o}}s-{H}ajnal problem on
  high-chromatic high-girth subgraphs.
\newblock In {\em Proceedings of the 2026 Annual ACM-SIAM Symposium on Discrete
  Algorithms (SODA)}, pages 2903--2927. SIAM, 2026.

\bibitem[R{\"o}d77]{Rod77}
Vojt{\v{e}}ch R{\"o}dl.
\newblock On the chromatic number of subgraphs of a given graph.
\newblock {\em Proceedings of the American Mathematical Society},
  64(2):370--371, 1977.

\bibitem[SF26]{soleymani2026quasipolynomialalgorithmellinftycontractions}
Ashkan Soleymani and Gabriele Farina.
\newblock A quasipolynomial algorithm for $\ell_\infty$ contractions.
\newblock arXiv preprint arXiv:2610.30122, 2026.

\bibitem[SS20]{SS20}
Alex Scott and Paul Seymour.
\newblock A survey of $\chi$‐boundedness.
\newblock {\em Journal of Graph Theory}, 95(3):473–504, 2020.

\bibitem[ST11]{simonyi2011directed}
G{\'a}bor Simonyi and G{\'a}bor Tardos.
\newblock On directed local chromatic number, shift graphs, and {Borsuk}-like
  graphs.
\newblock {\em Journal of Graph Theory}, 66(1):65--82, 2011.

\bibitem[SU11]{SU11}
Edward~R Scheinerman and Daniel~H Ullman.
\newblock {\em Fractional graph theory: a rational approach to the theory of
  graphs}.
\newblock Courier Corporation, 2011.

\bibitem[Tay70]{taylor1970combinatorial}
Walter Taylor.
\newblock Problem 43.
\newblock In Richard Guy, Haim Hanani, Norbert Sauer, and Johanan Schonheim,
  editors, {\em Combinatorial structures and their applications. Proceedings of
  the Calgary International Conference on Combinatorial Structures and Their
  Applications held at the University of Calgary, Calgary, Alberta, Canada,
  June, 1969.}, page 508, 1970.

\bibitem[TW18]{TW18}
G{\'a}bor Tardos and Bartosz Walczak.
\newblock On an {E}rd{\H{o}}s--{H}ajnal conjecture.
\newblock Talk at ``Combinatorics: Extremal, Probabilistic and Additive'',
  2018.
\newblock Conference talk.

\bibitem[Wal15]{walczak2015triangle}
Bartosz Walczak.
\newblock Triangle-free geometric intersection graphs with no large independent
  sets.
\newblock {\em Discrete \& Computational Geometry}, 53(1):221--225, 2015.

\bibitem[Wal21]{walsberg2021notes}
Erik Walsberg.
\newblock Notes on trace equivalence.
\newblock {\em arXiv preprint arXiv:2101.12194}, 2021.

\bibitem[Wil14]{williams2014faster}
Ryan Williams.
\newblock Faster all-pairs shortest paths via circuit complexity.
\newblock In {\em Proceedings of the forty-sixth annual ACM symposium on Theory
  of computing}, pages 664--673, 2014.

\bibitem[Wil24]{williams2024orthogonal}
Ryan Williams.
\newblock The orthogonal vectors conjecture and non-uniform circuit lower
  bounds.
\newblock In {\em 2024 IEEE 65th Annual Symposium on Foundations of Computer
  Science (FOCS)}, pages 1372--1387, 2024.

\bibitem[ZB08]{ZB08}
Daniel~R Zerbino and Ewan Birney.
\newblock Velvet: algorithms for de novo short read assembly using de {B}ruijn
  graphs.
\newblock {\em Genome research}, 18(5):821--829, 2008.

\bibitem[ZP96]{ZWICK1996343}
Uri Zwick and Mike Paterson.
\newblock The complexity of mean payoff games on graphs.
\newblock {\em Theoretical Computer Science}, 158(1):343--359, 1996.

\end{thebibliography}

\end{document}